\documentclass[10pt, a4paper, reqno]{article}

\usepackage{array}
\usepackage{multicol}
\usepackage{floatrow}
\usepackage{caption}
\usepackage[margin=1in
]{geometry}
\usepackage{graphicx}
\graphicspath{TEMPLATE FILES/Images}
\usepackage{wrapfig}
\usepackage{hyperref}
\usepackage{longtable}
\usepackage{booktabs}
\usepackage{caption, lipsum}
\usepackage{makecell}
\usepackage[most]{tcolorbox}
\usepackage{setspace}
\usepackage{threeparttablex}
\usepackage{tikz}
\usetikzlibrary{shapes, arrows}
\usepackage{biblatex}
\usepackage{fancyhdr}
\usepackage{pdflscape}
\usepackage{subcaption}
\usepackage{systeme}
\usepackage{dirtytalk}

\usepackage{amssymb, amsmath, amsthm, bm}
\usepackage{mathtools}
\usepackage{mathrsfs}                                                         
\usepackage{setspace}                                                         
\usepackage{geometry}                                                         
\usepackage{slashed}
\usepackage{xcolor}                                                               
\usepackage{cancel}
\usepackage{comment}                                                          
\usepackage{physics}

\newtheorem{theorem}{Theorem}
\newtheorem{corollary}[theorem]{Corollary}
\newtheorem{lemma}[theorem]{Lemma}
\newtheorem{proposition}[theorem]{Proposition}
\newtheorem{definition}[theorem]{Definition}
\newtheorem{remark}[theorem]{Remark}
\newtheorem{assumption}{Assumption}

\newtheorem{example}{Example}
\newtheorem{problem}{Problem}

\numberwithin{equation}{section}
\numberwithin{theorem}{section}

\newcommand{\ol}[1]{\overline{#1}}                                            
\newcommand{\ul}[1]{\underline{#1}}                                           
\newcommand{\mf}[1]{\mathfrak{#1}}                                            
\newcommand{\mc}[1]{\mathcal{#1}}                                             
\newcommand{\ms}[1]{\mathsf{#1}}                                              
\newcommand{\mi}[1]{\mathscr{#1}}                                             

\newcommand{\N}{\mathbb{N}}                                                   
\newcommand{\R}{\mathbb{R}}                                                   

\newcommand{\Lie}{\mc{L}}

\newcommand{\paren}[1]{\left(#1\right)}                                       
\newcommand{\brak}[1]{\left[#1\right]}                                        
\newcommand{\brac}[1]{\left\{#1\right\}}                                      

\makeatletter
\newcounter{proofpart}[section]
\newcommand{\proofpart}[1]{%
  \par
  \addvspace{\medskipamount}%
  \stepcounter{proofpart}%
  \noindent\textbf{\theproofpart. #1}\par\nobreak\smallskip
  \@afterheading
}
\makeatother

\title{Near-Boundary Asymptotics and Unique Continuation for the AdS--Einstein--Maxwell System}
\author{Simon Guisset}
\date{}

\begin{document}

\maketitle

\begin{abstract}
In this article, we extend the results of both Shao \cite{shao:aads_fg} and Holzegel--Shao \cite{Holzegel22} to the AdS--Einstein--Maxwell system \((\mathcal{M}, g, F)\). In particular:  

\begin{enumerate}
    \item We study the asymptotics of the metric \(g\) and the Maxwell field \(F\) near the conformal boundary \(\mathcal{I}\) for the fully nonlinear coupled system, in a finite regularity setting. Furthermore, we characterise the holographic (boundary) data, which is used in the second part of this work.  
    \item We prove the local unique continuation property for solutions of the coupled Einstein equations from the conformal boundary. Specifically, the prescription of the coefficients \((\mathfrak{g}^{(0)}, \mathfrak{g}^{(n)})\) in the near-boundary expansion of \(g\), along with the boundary data for the Maxwell fields \((\mathfrak{f}^{0}, \mathfrak{f}^{1})\), on a domain \(\mathcal{D} \subset \mathcal{I}\) uniquely determines \((g, F)\) near \(\mathcal{D}\). Moreover, the geometric conditions required for unique continuation are identical to those in the vacuum case, regardless of the presence of the Maxwell fields.  
\end{enumerate}

This work is part of the author’s thesis \cite{thesis_Guisset}.  
\end{abstract}

\tableofcontents
\section{Introduction}\label{sec:intro}

Asymptotically Anti-de Sitter (aAdS) spacetimes \((\mathcal{M}, g)\) are solutions to the Einstein equations with a \emph{negative} cosmological constant \(\Lambda\):  
\begin{equation}\label{Einstein_intro}
    \operatorname{Rc}[g] - \frac{1}{2}\operatorname{Rs}[g] + \Lambda \cdot g = T, \quad \Lambda < 0,
\end{equation}  
where \(\operatorname{Rc}[g]\) and \(\operatorname{Rs}[g]\) denote the Ricci tensor and scalar curvature, respectively, and \(T\) represents the stress-energy tensor of a given collection of matter fields. The maximally symmetric solution of the vacuum case $(T=0)$ is given by the \emph{pure AdS} solution $(\mc{M}_{AdS}, g_{AdS})$, which can be represented as the following $(n+1)$--dimensional Lorentzian manifold: 
\begin{equation}\label{intro_pure_AdS}
    \mc{M}_{AdS} := \R_t \times [0,\infty)_r\times \mathbb{S}^{n-1}_\omega\text{,}\qquad g_{AdS} = -(1+r^2)dt^2 + \frac{1}{1+r^2}dr^2 +r^2 \slashed{g}(\omega)\text{,}
\end{equation}
where $\slashed{g}$ is the metric on the $(n-1)$--sphere. A \emph{conformal boundary} $\mc{I}$ can be attached to $(\mc{M}_{AdS}, g_{AdS})$ at $r=\infty$, which is of timelike nature. For general solutions to \eqref{Einstein_intro} which asymptotically behave like \eqref{intro_pure_AdS}, the existence of such a conformal timelike boundary is in fact a generic feature of such geometries \cite{Hawking:1973uf}. 

What sets aAdS spacetimes apart from their \(\Lambda \geq 0\) counterparts is precisely the presence of this timelike conformal boundary, which dramatically undermines global hyperbolicity. As such, one is forced to impose boundary conditions on \(\mathcal{I}\) to restore well-posedness for the standard Cauchy problem \cite{Warnick13, Holzegel12}.  

This unique feature of aAdS spacetimes gives rise to rich physical phenomena, including stable trapping \cite{Holzegel:2015swa} and superradiant instabilities \cite{Cardoso_2014, ganchev2016superradiantinstabilityads}, which depend heavily on the boundary conditions prescribed at \(\mathcal{I}\). Among these, the most natural boundary conditions are of reflective type, but their imposition typically leads to \emph{instabilities}. Examples include the instability of AdS--Einstein--null dust spacetimes \cite{moschidis2018}, AdS--Einstein--Vlasov spacetimes \cite{Moschidis_2020}, and the slow \((\log t)^{-1}\) decay-in-time of linear waves and the linearisation of \eqref{Einstein_intro} in Kerr--AdS spacetimes \cite{holzegel2013decaypropertieskleingordonfields, graf2024linearstabilityschwarzschildantidesitterI, graf2024linearstabilityschwarzschildantidesitterII, graf2024linearstabilityschwarzschildantidesitterIII}. It is also worth mentioning that if one prescribes \emph{optimally dissipative} boundary conditions on $\mc{I}$, one typically obtains superpolynomial decay-in-time for linear fields \cite{Holzegel:2015swa}. 

Another very interesting feature of aAdS solutions is so-called \emph{AdS/CFT correspondence}, which posits a one-to-one correspondence between (quantum) gravity on aAdS solutions in $n+1$ dimensions and (conformal) field theories, which do not contain gravity, in $n$ dimensions located on $\mc{I}$ \cite{Maldacena_1999, Brown:1986nw, witten1998anti}. This correspondence has, since its discovery, lead to many applications, from quantum gravity \cite{Polchinski:2010hw}, to condensed matter physics \cite{Sachdev_2011} and hydrodynamics \cite{Policastro_2002}. A mathematical formulation of such a correspondence, however, is far from being established, although much progress has been made in the previous years, as will be presented in the next lines. Let us first briefly justify how this correspondence, from a mathematical perspective, may lead to several issues. 

A poor man's linearisation of \eqref{Einstein_intro} is given by the scalar wave equation on a fixed manifold $(\mc{M}, g)$ solution to \eqref{Einstein_intro}: 
\begin{equation}
    \Box_g \phi + \alpha \phi = 0\text{,}\qquad \alpha \in \R\text{.}\label{wave_intro}
\end{equation}
For a one-to-one correspondence between the solution $\phi$ and its boundary data, it would be enough to have well-posedness of the Cauchy problem for \eqref{wave_intro} with data on the boundary: 
\begin{equation*}
    \mc{D}\phi\vert_\mc{I} = \phi^{0}\text{,}\qquad \mc{N}\phi\vert_{\mc{I}} = \phi^1\text{,}
\end{equation*}
with $\mc{D}\phi\vert_{\mc{I}}, \mc{N}\phi\vert_{\mc{I}}$ the Dirichlet and Neumann traces on $\mc{I}$. Due to the timelike nature of $\mc{I}$, such a problem is in fact \emph{ill-posed} \cite{GUISSET2024223}, and the question of a correspondence between $(\phi^0, \phi^1)$ and $\phi$ in a neighbourhood of $\mc{I}$ becomes far less trivial. 

In view of the impossibility of solving from $\mc{I}$, a more natural way of phrasing the problem would be in terms of a \emph{unique continuation} problem. During the last twenty years various authors have worked towards establishing unique continuation on various contexts related to the present one. In particular, in 2008, Biquard first showed that unique continuation holds in the Riemannian context, i.e. for asymptotically hyperbolic Einstein manifolds \cite{biq:uc_einstein}. This result was later extended by Chrusciel and Delay \cite{chru_delay:uc_killing}, showing that unique continuation holds for stationary vacuum spacetimes. In both cases, the equations are elliptic and one can make use of appropriate Carleman estimates derived by Mazzeo \cite{Mazzeo}, after writing the Einstein equations in the right gauge. The works of Holzegel-Shao in \cite{hol_shao:uc_ads, hol_shao:uc_ads_ns} mark the first unique continuation results for aAdS metric in a Lorentzian setting. Typically, they showed that unique continuation holds for (tensorial) wave equations assuming that the non-stationarity of $(\mc{M}, g)$ was not ``too strong", and a null--pseudoconvexity condition, allowing one to derive useful Carleman estimates. This was later extended by McGill-Shao in \cite{McGill20}, assuming more general boundary geometries, and relating the null-pseudoconvexity conditions to the inexistence of near--boundary null geodesics which remain arbitrarily close to $\mc{I}$. One could indeed use these null geodesics to construct counterexamples \emph{à la} Alinhac-Baouendi \cite{AB}, as has been done in \cite{GUISSET2024223} by the author and Shao. The final piece needed to obtain a full nonlinear unique continuation result for \eqref{Einstein_intro} was the result of Chatzikaleas--Shao \cite{Shao22}, generalising the null-pseudoconvexity condition in a gauge invariant way, as well as establishing appropriate Carleman estimates for the wave equations. All these works culminated in the work of Holzegel--Shao \cite{Holzegel22}, proving the following informal result: 

\begin{theorem}[Informal version of Theorem 6.7 of \cite{Holzegel22}]\label{informal_thm_vacuum}
    Let $(\mc{M}, g)$ and $({\mc{M}}, \check{g})$ be two near--boundary regions of aAdS solutions, with $g,\, \check{g}$ solving \eqref{Einstein_intro} with $T=0$, and let $\mi{D}\subset \mc{I}$. Assuming that $\mi{D}$ satisfies a certain null-pseudoconvexity condition, and if the boundary data for $(\mc{M}, g)$ and $({\mc{M}}, \check{g})$ on $\mc{I}$ are ``gauge--equivalent", one has necessarily, in a neighbourhood of $\mi{D}$ in $\mc{M}$: 
    \begin{equation*}
        g \equiv \check{g}\text{,}
    \end{equation*}
    up to a diffeomorphism preserving $\mc{I}$. 
\end{theorem}
Such a result also allowed the authors to show that the symmetries on the boundary can be extended on a neighbourhood of $\mi{D}$. 

The goal of this present paper is to extend Theorem \ref{informal_thm_vacuum} to the AdS--Einstein--Maxwell system, i.e. where one couples \eqref{Einstein_intro} to the following matter field: 
\begin{gather}\label{intro_EM}
    T_{\alpha\beta} = F_{\alpha\gamma}F_\beta{}^\gamma - \frac{1}{4}g_{\alpha\beta}\abs{F}^2_g\text{,}\\
    d F = 0\text{,}\qquad d\star_g F = 0\text{,}\notag
\end{gather}
where $F$ is a two-form on $\mc{M}$, $\abs{F}^2_g := \paren{g^{-1}}^{\alpha\beta}\paren{g^{-1}}^{\gamma\delta}F_{\alpha\gamma}F^{\beta\delta}$, and $\star_g$ the Hodge--star operator with respect to $g$. In particular, this marks the first generalisation of \cite{Holzegel22} and \cite{shao:aads_fg} to nontrivial matter fields. 

In the rest of the introduction, we will present the two main results of this work, namely: 
\begin{itemize}
    \item The derivation of an appropriate near-boundary expansion for the fields of interest $(g, F)$ in terms of their boundary data, 
    \item The proof of the unique continuation for the system $(\mc{M}, g, F)$ from the conformal boundary in the full nonlinear setting. Such a proof will require the results of the first part.
\end{itemize}

\subsection{Fefferman-Graham expansion}

In this article, since we will only be interested in the asymptotic region near $\mc{I}$, we will work on an $n+1$ dimensional Lorentzian manifold $(\mc{M}, g)$ of the form:
\begin{definition}[FG-aAdS segment, informal]\label{FG_aAdS_informal_intro}
Let $(\mc{I}, \mf{g}^{(0)})$ be an $n$--dimensional Lorentzian manifold, and $\rho_0>0$. We say that the manifold $(\mc{M}, g)$ is an FG-aAdS segment if it takes the form:
\begin{equation}\label{FG_aAdS_informal}
    \mc{M} := \left( 0, \rho_0\right]_\rho \times \mc{I}\text{,}\qquad g = \rho^{-2}\paren{d\rho^2 + \ms{g}(\rho)}\text{,}\qquad \rho \in (0, \rho_0]\text{,}
\end{equation}
where $\ms{g}(\rho)$ is a smooth one-parameter family of Lorentzian metrics on $\mc{I}$ such that $\lim\limits_{\rho\searrow 0}\ms{g}(\rho)=\mf{g}^{(0)}$, in an appropriate sense.

\begin{itemize}
    \item The manifold $(\mc{I}, \mf{g}^{(0)})$ is defined as the \emph{conformal boundary}.
    \item If $(\mc{M}, g)$ solves the vacuum Einstein equations, we refer to it as a \emph{vacuum FG-aAdS segment}.
    \item $(\mc{M}, g, F)$ will be defined as a Maxwell-FG-aAdS segment if the triplet solve \eqref{Einstein_intro} coupled to \eqref{intro_EM}, and with $F$ satisfying the following boundary limits, in $C^0$: 
    \begin{equation*}
    \begin{cases}
        F_{ab} \rightarrow \mf{f}^{1,{(0)}}_{ab}\text{,}\qquad &n\geq 4\\
        F_{ab} \rightarrow \mf{f}^{1,(0)}_{ab}\text{,}\; F_{\rho a}\rightarrow\mf{f}^{0,(0)}_a\text{,}\qquad &n=3\text{,}\\
        \rho\cdot F_{\rho a}\rightarrow \mathfrak{f}_a^{0, (0)}\text{,}\qquad &n=2\text{,}
    \end{cases}
    \end{equation*}
    as $\rho\searrow 0$, and where $\mf{f}^{0,(0)}$ and $\mf{f}^{1,(0)}$ are, respectively, a one-form and a two-form on $\mc{I}$.
\end{itemize}
\end{definition}

\begin{example}
    A basic example is given by the pure $AdS$ solution \eqref{intro_pure_AdS}. Note that the following coordinate change $4r =: \rho^{-1}(2+\rho)(2-\rho)$ yields the following metric: 
    \begin{equation}\label{pure_ads_transformed}
        g_{AdS} = \rho^{-2}\paren{d\rho^2 + \ms{g}(\rho)}\text{,}\qquad \rho \in (0, 2]\text{,}
    \end{equation}
    where: 
    \begin{equation*}
        \ms{g} =  (-dt^2 +\slashed{g}(\omega)) -\frac{1}{2}\rho^2 (dt^2 + \slashed{g}(\omega)) + \frac{1}{16}\rho^4(-dt^2 + \slashed{g}(\omega))\text{.}
    \end{equation*}
    Note that $\rho\searrow 0$ corresponds to $r\nearrow \infty$, while $\rho\nearrow 2$ corresponds to the centre of the manifold $r\searrow 0$.
\end{example}

Observe that the metric \eqref{pure_ads_transformed} can be written as the following: 
\begin{equation*}
    g = \rho^{-2}\paren{d\rho^2 + \mf{g}^{(0)} + \rho^2 \mf{g}^{(2)} + \rho^4\mf{g}^{(4)}}\text{,}
\end{equation*}
with
\begin{gather*}
    \mf{g}^{(0)} = -dt^2 +\slashed{g}(\omega)\text{,}\qquad \mf{g}^{(2)}:=-\frac{1}{2} (dt^2 + \slashed{g}(\omega))\text{,}\qquad \mf{g}^{(4)} := \frac{1}{16}(-dt^2 + \slashed{g}(\omega))\text{.}
\end{gather*}
Such a near-boundary expansion in $\rho^2$ of the metric is in fact a general property of vacuum asymptotically Anti-de Sitter spacetimes. This feature has first been discovered by Fefferman and Graham \cite{fef_gra:conf_inv, fefferman2008ambientmetric}, who constructed vacuum solutions from conformal data on a given null hypersurface. It turns out that the Einstein equations, as well as the data being ``analytic" enough allows one to construct a formal series expansion, which converges in an appropriate sense to a genuine vacuum solution \cite{KICHENASSAMY2004268}. 

The same story unfolds in aAdS, for which one can construct a formal series expansion for a metric \eqref{FG_aAdS_informal} solving the vacuum Einstein equations of the form: 
\begin{equation}\label{informal_FG_expansion}
    \ms{g}(\rho) = 
    \begin{cases}
        \sum\limits_{k=0}^{\frac{n-1}{2}}\rho^{2k}\mf{g}^{(2k)} + \rho^n \mf{g}^{(n)} + \dots\text{,}\qquad &n\text{ odd,}\\
        \sum\limits_{k=0}^{\frac{n-2}{2}}\rho^{2k}\mf{g}^{(2k)} +\rho^n \log\rho \cdot \mf{g}^\star + \rho^n \mf{g}^{(n)} + \dots \text{,}\qquad &n\text{ even,}
    \end{cases}
\end{equation}
where:
\begin{itemize}
    \item The coefficients $\mf{g}^{(0)}$ and $\mf{g}^{(n)}$ are the ``free-data", which are not constrained by the vaccum Einstein equations\footnote{In fact, the $\mf{g}^{(0)}$--divergence and the $\mf{g}^{(0)}$--trace of $\mf{g}^{(n)}$ are constrained.}.
    \item The coefficients $\mf{g}^{(2k)}$, with $2\leq 2k<n$, and $\mf{g}^\star$ are fully constrained by $\mf{g}^{(0)}$.
    \item The coefficients arising after $\mf{g}^{(n)}$ are fully constrained by $\mf{g}^{(0)}$ and $\mf{g}^{(n)}$.
\end{itemize}
Once again, assuming enough ``analyticity", this formal expansion converges to a vacuum Einstein metric \cite{fefferman2008ambientmetric, fef_gra:conf_inv}. 

Observe that the first coefficient $\mf{g}^{(0)}$, corresponding to the Dirichlet data, is simply the boundary metric. On the other hand, one interprets the divergence-- and trace--free parts of $\mf{g}^{(n)}$ as the expectation value of the CFT stress-energy tensor \cite{Hubeny_2015, SKENDERIS_2001}, and corresponds to the Neumann boundary data. 

Let us return to the unique continuation problem, and assume that two FG-aAdS segments $(\mc{M}, g), (\mc{M}, \check{g})$ solving the vacuum Einstein equations turn out to have identical boundary data\footnote{As mentioned earlier, one will have to make sense of ``identical" and ``different" given the gauge freedom.}: 
\begin{equation*}
    \mf{g}^{(0)} = \check{\mf{g}}^{(0)}\text{,}\qquad \mf{g}^{(n)} = \check{\mf{g}}^{(n)}\text{.}
\end{equation*}
Given an expansion of the form \eqref{informal_FG_expansion} up to $\rho^n$, corresponding to the Neumann branch, one would immediately have: 
\begin{equation*}
    \ms{g} - \check{\ms{g}} = o(\rho^n)\text{,}
\end{equation*}
near the boundary, since all the coefficients would depend on $\mf{g}^{(0)}$ and $\mf{g}^{(n)}$ only. These asymptotics would certainly turn out to be extremely useful for the unique continuation problem. In the analytic case \footnote{Of course, the analyticity assumption has to take into account the polyhomogeneity of the metric for $n$ even.}, the problem is trivial since the coefficients for the Taylor expansion only depend on $\mf{g}^{(0)}$ and $\mf{g}^{(n)}$. Such an assumption on the metric, however, is a far too strong, since solutions to \eqref{Einstein_intro} arising from generic initial data are not generally analytic. 

In the smooth class of metric, however, one would still like to obtain a finite expansion of the form \eqref{informal_FG_expansion}. Such a finite expansion in this regularity class of vacuum metrics has been obtained in \cite{shao:aads_fg}, where the formal expansion takes a similar form to \eqref{informal_FG_expansion}: 
\begin{equation}\label{informal_FG_expansion_2}
    \ms{g}(\rho) = 
    \begin{cases}
        \sum\limits_{k=0}^{\frac{n-1}{2}}\rho^{2k}\mf{g}^{(2k)} + \rho^n \mf{g}^{(n)} + \rho^n \ms{r}_{\ms{g}}\text{,}\qquad &n\text{ odd,}\\
        \sum\limits_{k=0}^{\frac{n-2}{2}}\rho^{2k}\mf{g}^{(2k)} + +\rho^n \log\rho \cdot \mf{g}^\star + \rho^n \mf{g}^{(n)} + \rho^n \ms{r}_{\ms{g}} \text{,}\qquad &n\text{ even,}
    \end{cases}
\end{equation}
where, as before: 
\begin{itemize}
    \item The coefficients $\mf{g}^{(2k)}$, with $2\leq 2k<n$, and $\mf{g}^\star$ depend on $\mf{g}^{(0)}$ and its derivatives. 
    \item The coefficients $\mf{g}^{(0)}$, as well as the divergence-- and trace--free parts of $\mf{g}^{(n)}$ are ``free".
\end{itemize}
Furthermore, the remainder $\ms{r}_{\ms{g}}(\rho)$ is a continuous one-parameter family of tensor fields on $\mc{I}$ such that $\ms{r}_{\ms{g}}\rightarrow 0$, as $\rho\searrow 0$, appropriately.

To prove such a result, the author assumed finite regularity assumptions on the metric, namely
\begin{assumption}\label{assumption_g} The metric $\ms{g}$ satsifies the following regularity assumptions: 
    \begin{itemize}
    \item $\ms{g}$ locally bounded in $C^{n+4}$ in the directions given by $\mc{I}$,
    \item $\Lie_\rho \ms{g}$ satisfies a weak local bound in $C^0$\text{.}
\end{itemize}
\end{assumption}

In Section \ref{sec:FG_expansion}, we will prove that a Maxwell-FG-aAdS segment $(\mc{M}, g, F)$ satisfies an identical expansion in $\rho^2$. Namely, we will show:

\begin{theorem}[Theorem \ref{theorem_main_fg}, informal]\label{thm_main_fg_ads_informal}
    Let $(\mc{M}, g, F)$ be a Maxwell-FG-aAdS segment and define: 
    \begin{gather*}
        \ol{\ms{f}}^0_a := \begin{cases}
            \rho^{-1}F_{\rho a}\text{,}\qquad &n \geq 4\text{,}\\
            F_{\rho a}\text{,}\qquad &n = 3\text{,}\\
            \rho F_{\rho a}\text{,}\qquad &n =2\text{,}
        \end{cases}\qquad \ol{\ms{f}}^1_{ab} := 
            F_{ab}\text{.}
    \end{gather*}
    Assuming the same local boundedness for $\ms{g}$ as in Assumption \ref{assumption_g}, with $\ol{\ms{f}}^0$ and $\ol{\ms{f}}^1$ satisfying similar regularity assumptions, then $\ms{g}$ satisfies an expansion of the form \eqref{informal_FG_expansion_2} where\footnote{For $n=2$, however, we will show that the anomalous coefficient $\mf{g}^\star$ depends on $\mf{f}^{0,(0)}$.}: 
\begin{itemize}
    \item The coefficients $\mf{g}^{(k)}$, $2\leq 2k<n$ depend on the first $k$--derivatives of $\mf{g}^{(0)}$ and on the first $(k-4)$--derivatives of $\mf{f}^{1,(0)}$.
    \item The coefficient $\mf{g}^{(0)}$, as well as the divergence-- and trace--free parts of $\mf{g}^{(n)}$ are not constrained by the equations of motion.
\end{itemize}
Furthermore, the tensor fields $\ol{\ms{f}}^0$ and $\ol{\ms{f}}^1$ also satisfy a near-boundary expansion, with free coefficients $\mf{f}^{0, ((n-4)_+)}$\footnote{Here, $\mf{f}^{0, ((n-4)_+)}$ is the coefficient located at the power $(n-4)_+$ for $\rho$ in the near-boundary expansion of $\ol{\ms{f}}^0$. Observe that for $2\leq n\leq 4$, it simply corresponds to $\mf{f}^{0,(0)}$, already introduced in Definition \ref{FG_aAdS_informal_intro}.} and $\mf{f}^{1,(0)}$, respectively.
\end{theorem}

This result allows us to define the \emph{holographic data} for the system AdS--Einstein--Maxwell system \eqref{Einstein_intro}, \eqref{intro_EM}, which consists in the following: 
\begin{equation*}
    (\mc{I}, \mf{g}^{(0)}, \mf{g}^{(n)}, \mf{f}^{1,(0)}, \mf{f}^{0, ((n-4)_+)})\text{.}
\end{equation*}
We also show that this data is subject to the following constraints, at least for $n\geq 3$: 
\begin{gather*}
    \mf{D}\cdot \mf{g}^{(n)} = \mf{G}^1_n\paren{\partial^{\leq n+1}\mf{g}^{(0)}, \partial^{\leq n-3}\mf{f}^{1,(0)}}\text{,}\qquad \mf{tr}_{\mf{g}^{(0)}}\mf{g}^{(n)} = \mf{G}_n^2\paren{\partial^{\leq n}\mf{g}^{(0)}, \partial^{\leq n-4}\mf{f}^{1,(0)}}\text{,}\\
    \star_{\mf{g}^{(0)}} \mf{d} \star_{\mf{g}^{(0)}} \mf{f}^{0, ((n-4)_+)} = \mf{F}_n^0\paren{\partial^{\leq n-2}\mf{g}^{(0)}, \partial^{\leq n-2}\mf{f}^{1,(0)}}\text{,}\qquad \mf{d}\mf{f}^{1,(0)} = 0\text{,}
\end{gather*}
where $\mf{D}$ is the $\mf{g}^{(0)}$--covariant derivative, $\mf{d}$ the $\mf{g}^{(0)}$--exterior derivative and $\star_{\mf{g}^{(0)}}$ is the Hodge--star operator with respect to $\mf{g}^{(0)}$. Also, $\mf{G}^0_n\text{, } \mf{G}^{1}_{n}\text{, } \mf{F}^{0}_n$ are universal functions depending on $n$ and vanishing for $n$ odd. 

The proof is based on a rather simple, although technical, ODE analysis for equations typically of the form: 
\begin{equation}\label{transport_intro}
    xf'(x) - c\cdot f(x) = h(x)\text{,}\qquad x\in (0, 1]\text{;}\qquad h \rightarrow h^0\text{,}\qquad \text{as }  x\rightarrow 0\text{,}
\end{equation}
where $f\text{, } h$ take value in some appropriate functional space, and $c>0$ is some constant. The basic observation, using the Frobenius method, is that $f$ takes the following asymptotic form: 
\begin{equation*}
    f = -\frac{1}{c}h^0 + o(1)\text{,}\qquad \text{as }x\rightarrow 0\text{.} 
\end{equation*}
Note that, by differentiating \eqref{transport_intro} with respect to $x$ $k$--times, one obtains the following transport equation: 
\begin{equation*}
    xf^{(k+1)}(x) - (c-k)f^{(k)}(x) = h^{(k)}(x)\text{,}
\end{equation*}
and one can repeat the same analysis as long as $h^{(k)}\rightarrow h^k$, for some $h^k$, and $c-k>0$. In particular, if $c=k_\star$ for some $k_\star$ integer, one will reach eventually the following: 
\begin{equation*}
    xf^{(k_\star +1)} = h^{(k_\star)}(x)\text{.}
\end{equation*}
One can show that if $h^{(k_\star)}$ converges to some $h^\star$ in an improved manner, then $f^{(k_\star)}$ has the following asymptotic behaviour: 
\begin{equation*}
    f^{(k_\star)} = f^\dag + \log x \cdot h^\star+o(1)\text{,}
\end{equation*}
where $f^\dag$ is free, i.e. not constrained by $h$ and its derivatives. 

When it comes to the Einstein and Maxwell equations \eqref{Einstein_intro}, \eqref{intro_EM}, one finds the following transport equations: 
\begin{align}
    \begin{aligned}
    &\rho\partial_\rho \ms{m} - (n-1)\ms{m} \sim \rho\ms{Rc} + \ms{tr}_{\ms{g}}\ms{m} \cdot \ms{g} + \rho\cdot (\ms{m})^2 + \rho^3(\bar{\ms{f}}^1)^2 + \rho^{2\alpha_n}(\bar{\ms{f}}^0)^2\text{,}\\
    &\rho\partial_\rho \ms{tr}_{\ms{g}}\ms{m} - (2n-1)\ms{tr}_{\ms{g}}\ms{m} \sim \rho\ms{Rs} + \rho\cdot (\ms{m})^2 + \rho^3(\bar{\ms{f}}^1)^2 + \rho^{2\alpha_n}(\bar{\ms{f}}^0)^2\text{,}\\
    &\rho\partial_\rho \bar{\ms{f}}^0- (n-4)_+\bar{\ms{f}}^0 \sim \rho\ms{m}\cdot \bar{\ms{f}}^0 + \rho^{1+\beta_n}\ms{D}\cdot \bar{\ms{f}}^1\text{,}\\
    &\partial_\rho\bar{\ms{f}}^1 \sim \ms{D}\bar{\ms{f}}^0\cdot \rho^{\gamma_n}\text{,}
    \end{aligned}
\end{align}
where $\ms{m}:= \partial_\rho \ms{g}$ and $\alpha_n$, $\beta_n$, $\gamma_n$ are some powers depending on the dimension; see Propositions \ref{prop_transport_m} and \ref{prop_transport_f} for more details. From the above considerations, one simply needs to analyse the right-hand side of these transport equations, and treat the constant $c$ associated to each equation carefully. We will typically find, as in \cite{shao:aads_fg}, that the limits will depend on the parity of the number of $\rho$--derivatives. 

In particular, note that for $k_\star=n-1$, $q_\star = n-4$, one finds equations of the form, for $n\geq 4$: 
\begin{equation*}
    \rho\partial^{k_\star +1}_\rho \ms{m} = \ms{G}_{n}\text{,}\qquad \rho\partial_\rho^{q_\star +1 }\bar{\ms{f}}^0 = \ms{H}_n
\end{equation*}
where $\ms{G}_n$, $\ms{H}_n$ are vertical tensor fields which satisfy the following limits: 
\begin{equation*}
    \ms{G}_n\rightarrow \begin{cases}
        0\text{,}\qquad &n \text{ even,}\\
        \mf{G}_n^\star\text{,}\qquad &n\text{ odd.}
    \end{cases}\qquad \ms{H}_n\rightarrow \begin{cases}
        \mf{H}_n^\star\text{,}\qquad &n \text{ even,}\\
        0\text{,}\qquad &n\text{ odd,}
    \end{cases}
\end{equation*}
as $\rho\searrow 0$. As a consequence, the polyhomogeneous terms in the Fefferman-Graham expansion for the metric $\ms{g}$, and in the near-boundary expansion for $\bar{\ms{f}}^0$, are nontrivial for $n$ even only. 

We also show, along the way, how one can recover constraints on the holographic data. This is because the vertical decomposition of the Bianchi identities with respect to $g$, as well as the Maxwell and Bianchi equations for $F$ will also lead to constraint equations of the form: 
\begin{gather*}
    \ms{D}\cdot \ms{m} = \ms{D}(\ms{tr}_{\ms{g}}\ms{m}) + \rho^{\alpha_n} \bar{\ms{f}}^0\cdot \bar{\ms{f}}^1\text{,}\\
    \ms{D}_{[a}\bar{\ms{f}}^1_{bc]} = 0\text{,}\qquad \ms{D}\cdot \bar{\ms{f}}^0 = \rho^{-\gamma_n}\ms{m}\cdot\bar{\ms{f}}^1\text{,}
\end{gather*}
and the constraints simply follow by computing the limits on the right-hand sides. 

The proof, mostly technical, allows for any dimension $n\geq 2$. In particular, since the free coefficients for the Maxwell fields and the metric are located at different powers of $\rho$, one has to treat the cases $n=2, 3$ and $4$ differently. Note also that since the free coefficient for $F_{ab}$ is located at order $0$, and for $F_{\rho a}$ at $n-3$, the case $n=3$ corresponds to the borderline case. For $n\geq 4$, the ``first" free coefficient is given by $\mf{f}^{1,(0)}$, while for $n=2$, it is given by $\mf{f}^{0,(0)}$. This feature of the Maxwell fields, which shows a strong dependency on the dimension, forces one to be indeed careful with the dimensions in the ODE analysis.

\subsection{Unique continuation}

\subsubsection{Setting-up the problem}

We have just seen that any solution $(\mc{M}, g, F)$ generates holographic data $(\mc{I}, \mf{g}^{(0)}, \mf{g}^{(n)}, \mf{f}^{0, ((n-4)_+)}, \mf{f}^{1, (0)})$. The converse question is natural: Does any holographic data on a domain $\mi{D}$ generate a unique solution $(g, F)$ on $\Omega\subset \mc{M}$ in a neighbourhood of $\lbrace 0 \rbrace \times \mc{I}$? From the previous considerations, and the counterexamples from \cite{GUISSET2024223}, solving \eqref{Einstein_intro}, in this context coupled with \eqref{intro_EM}, is generically not possible. As mentioned previously, this is due to the timelike nature of the conformal boundary $\mc{I}$ and the hyperbolic nature of the Einstein equations. As a consequence, one should formulate the correspondence between the holographic data and solutions as the following: 
\begin{problem}\label{problem_informal}
    Let $\mi{A} = (\mc{I}, \mf{g}^{(0)}, \mf{g}^{(n)}, \mf{f}^{0, (0)}, \mf{f}^{1, (0)})$ -- modulo gauge equivalence -- be given holograhic data on a domain $\mi{D}\subset \mc{I}$. Is there a unique aAdS solution $(\mc{M}, g, F)$ -- up to a diffeomorphism -- realising this data? 
\end{problem}
In this work, we prove a \emph{conditional} positive answer to Problem \ref{problem_informal}. In particular, the domain $\mi{D}$ has to satisfy a certain pseudoconvexity condition with respect to $\mf{g}^{(0)}$, referred to as the \emph{GNCC} -- \emph{Generalised Null Convexity Criterion}. Such a condition turns out to depend on the boundary metric $\mf{g}^{(0)}$ and the coefficient $\mf{g}^{(2)}$ in the Fefferman-Graham expansion \eqref{informal_FG_expansion_2}. In the vacuum case, this coefficient $\mf{g}^{(2)}$ is simply the Schouten tensor of $\mf{g}^{(0)}$. We show in particular that it remains the case the in AdS--Einstein--Maxwell setting, since the coefficient $\mf{g}^{(2)}$ for $n\geq 3$ will not depend on the Maxwell holographic data. In other words, it depends \emph{only} on the boundary geometry $(\mi{D}, \mf{g}^{(0)})$. 

In \cite{Holzegel22}, the authors prove also a positive answer to the above problem for the \emph{vacuum} case, i.e. \eqref{Einstein_intro} with $T=0$. Moreover, they prove an extension of symmetry result, namely that the bulk manifold $\mc{M}$ inherits the symmetries of $(\mc{I}, \mf{g}^{(0)}, \mf{g}^{(n)})$. 

Let us informally state the second result we will be proving in this work: 
\begin{theorem}[Theorem \ref{thm_UC_main}, informal version]\label{thm_UC_informal}
    Let $(\mc{M}, g, F)$, $(\check{\mc{M}}, \check{g}, \check{F})$ be two Maxwell-FG-aAdS segments solving \eqref{Einstein_intro}, coupled with \eqref{intro_EM}. Let also $\mi{A}$ and $\check{\mi{A}}$ be their corresponding holographic data, and assume that they are gauge--equivalent on a domain $\mi{D}\subset \mc{I}$ satisfying the GNCC. Then, $(\mc{M}, g, F)$ and $(\check{\mc{M}}, \check{g}, \check{F})$ must necessarily be isometric near $\mi{D}$. 
\end{theorem}
In the remainder of this section, we will informally describe the proof of such a theorem. 

\subsubsection{Wave--transport system and renormalisation}

In order to make use of the Carleman estimates derived from linear theory \cite{hol_shao:uc_ads, hol_shao:uc_ads_ns, Shao22, McGill20}, it will be useful to write the Einstein equations as a wave--transport system. In the same way that \cite{ionescu2015rigidityresultsgeneralrelativity, Ionescu:2011wx} extends Killing vector fields along a transverse direction to a null hypersurface, we will extend the boundary data using the directions given by $\partial_\rho$, therefore exploiting the geometry of \eqref{FG_aAdS_informal}. In order to exploit it even further, we will write equations \eqref{Einstein_intro} and \eqref{intro_EM} in the so-called \emph{vertical formalism}, developed in \cite{hol_shao:uc_ads, hol_shao:uc_ads_ns, McGill20}, by working on level sets of $\rho$. 

In this formalism, spacetime tensor fields will be decomposed as tensor fields on level sets of $\rho$. As an example, we will define the following decomposition for the Weyl curvature, the metric, and the Maxwell field: 
\begin{gather*}
    \ms{w}^0_{abcd} := \rho^2 W_{abcd}\text{,}\qquad \ms{w}^1_{abc} := \rho^2W_{\rho abc}\text{,}\qquad \ms{w}^2_{ab} := \rho^2 W_{\rho a \rho b}\text{,}\\
    \ms{g}_{ab} := \rho^2g_{ab}\text{,}\qquad \ms{m}_{ab} = \Lie_\rho (\rho^2g)_{ab}\text{,}\\
    \ms{f}^0_a := \rho F_{\rho a}\text{,}\qquad \ms{f}^1_{ab} := \rho F_{ab}\text{.}
\end{gather*}
Note that the powers of $\rho$, although seemingly arbitrary, will turn out to be critical when obtaining a useful wave--transport system. 

The idea is, as in \cite{Holzegel22}, to couple the transport equations satisfied by the $\ms{g}, \ms{m}, \ms{f}^0, \ms{f}^1$, which are consequence of the vertical decomposition of the Einstein and Maxwell equations, to the wave equation satisfied by $W$ which, after careful decomposition yield:
\begin{equation*}
    \paren{\ol{\Box} + \sigma_i}\ms{w}^i = \ms{NL}_i\paren{\ms{w}^j, \ms{Dw}^j, \ms{f}^j, \ms{Df}^j, \ms{D}^2\ms{f}^j, \ms{m}, \ms{Dm}}\text{,}\qquad \ms{w}^i \in \lbrace \ms{w}^\star, \ms{w}^1, \ms{w}^2 \rbrace\text{,}
\end{equation*}
where $\ms{w}^\star$ is the traceless--part of $\ms{w}^0$, $\sigma_i\in \R$ and $\ms{NL}_i$ some nonlinearities. The operator $\ol{\Box}$ is an appropriate covariant wave operator for vertical tensor fields; for more details, see Definition \ref{def_nabla_bar}. For the Carleman estimates, however, one typically needs to bound the nonlinearities by: 
\begin{equation*}
    \abs{\ms{NL}_i}\lesssim \sum\limits_{\ms{A}\in \lbrace (\ms{w}^j), (\ms{f}^j), \ms{m}\rbrace}(\rho^{K_\ms{A}}\abs{\ms{A}} + \rho^{K'_{\ms{A}}}\abs{\ms{DA}})\text{,}\qquad {K}_{\ms{A}}, K_{\ms{A}}'>0\text{,}
\end{equation*}
which is not the case due to the presence of $\ms{D}^2\ms{f}$ in the nonlinearities. The solution is therefore to introduce new vertical fields: 
\begin{gather}\label{intro_h}
    \ul{\ms{h}}^0 := \ms{D}\ms{f}^1\text{,}\qquad \ul{\ms{h}}^2 := \ms{Df}^0\text{,}
\end{gather}
and to derive wave equations for such fields by vertically decomposing the wave equation satisfied by $F$, commuted with $\nabla$: 
\begin{equation}
    \Box_g \nabla F = \left[\Box_g, \nabla\right]F + \nabla(\Box_g F)\text{.}
\end{equation}
It turns out that the choice \eqref{intro_h} leads indeed to wave equations of the form: 
\begin{equation*}
    (\ol{\Box} + \sigma_i)\ul{\ms{h}}^i = \ms{NL}_i\paren{(\ms{w}^j), (\ms{Dw}^j), (\ms{f}^j), (\ul{\ms{h}}^j), (\ms{D}\ul{\ms{h}}^j), \ms{m}, \ms{Dm}}\text{,}\qquad i=0, 2\text{,}
\end{equation*}
and one eventually obtains the desired wave--transport system: 
\begin{equation}\label{WT_intro}
    \begin{cases}
        \partial_\rho \ms{B} + \sigma_{\ms{B}}\ms{B} = \mf{T}(\ms{B}, \ms{A})\text{,}\\
        \ol{\Box}\ms{A} + \sigma_{\ms{A}}\ms{A} = \mf{W}\paren{\ms{A}, \ms{DA}, \ms{B}, \ms{DB}}\text{,}
    \end{cases}
\end{equation}
where $\ms{B}\in \lbrace \ms{g}, \ms{m}, \ms{f}^0, \ms{f}^1\rbrace$ and $\ms{A}\in\lbrace \ms{w}^\star, \ms{w}^1, \ms{w}^2, \ul{\ms{h}}^0, \ul{\ms{h}}^2\rbrace$, and $\mf{T}\text{, }\mf{W}$ some nonlinearities exhibiting appropriate decay in $\rho$. 

Since one would want to compare two solutions $(g, F), (\check{g}, \check{F})$ on aAdS segments $\mc{M}, \check{\mc{M}}$ with gauge-equivalent data $\mi{A}\text{, } \check{\mi{A}}$ on a domain $\mi{D}\subset \mc{I}$, it will be necessary to look at differences of the vertical fields. Note that we can take, without loss of generality, $\mc{M} =  \check{\mc{M}}$\footnote{Note that $\mc{M}$ and $\check{\mc{M}}$ share the same conformal boundary $\mc{I}$} and $\mi{A} = \check{\mi{A}}$, by applying an appropriate gauge-transformation. As a consequence, the goal is to show that for $\mi{A} = \check{\mi{A}}$ on $\mi{D}$, one has necessarily:
\begin{equation}
    \ms{g} - \check{\ms{g}} \equiv 0\text{,}\qquad \ms{f}^i - \check{\ms{f}}^i \equiv 0\text{,}\qquad i=0, 1\text{,}
\end{equation}
in a neighbourhood of $\mi{D}$. The idea is thus to study the wave--transport system \eqref{WT_intro} in terms of the differences $\ms{B}- \check{\ms{B}}$ and $\ms{A}- \check{\ms{A}}$.  

First, one can immediately see that taking $\delta \ms{B} := \ms{B} - \check{\ms{B}}$, with $\ms{B}\in \lbrace \ms{g}, \ms{m}, \ms{f}^0, \ms{f}^1\rbrace$, one has a transport equation of the form: 
\begin{equation*}
    \partial_\rho \delta\ms{B} + \sigma_{\ms{B}}\delta\ms{B} = \delta \mf{T}\text{,}
\end{equation*}
and one can simply express the nonlinearity in terms of the difference of fields. An issue shows up however, when one wants to perform the same reasoning on the wave equation: 
\begin{equation}\label{intro_diff}
    \ol{\Box}\delta\ms{A} = \underbrace{\delta \mf{W}}_{\text{good}} -\underbrace{(\ol{\Box} - \check{\ol{\Box}})\check{\ms{A}}}_{\text{bad}}\text{.}
\end{equation}
Note that the ``good terms" can be treated as in the transport case, namely by expressing the nonlinearities in terms of the difference of the fields. The second one, on the other hand, hides problematic terms of the form $\ms{D}^2\delta\ms{g}$, which are not controlled in the Carleman estimates, for similar reasons as to why $\ms{D}^2\ms{f}$ was not controlled initially. Furthermore, one cannot naively apply the same trick as with the Maxwell fields, namely by considering an new vertical field $\ms{D}\delta{\ms{g}}$, since $\ol{\Box}\ms{D}\check{\ms{g}}$ obeys no particular wave equation. 

This problem, located at the level of the Einstein equations, is also found in the vacuum case, and has been treated in \cite{Holzegel22} by considering a renormalisation of the difference for the wave fields\footnote{A similar renormalisation can also be found in the work of Ionescu-Klainerman \cite{Ionescu:2011wx}, in the context of the rigidity of black holes.}: 
\begin{equation}\label{renormalisation_intro}
    \Delta \ms{A} \sim \delta \ms{A} + \ms{A}\cdot (\delta\ms{g} + \ms{Q})\text{,}
\end{equation}
where $\ms{Q}$ is an auxiliary field, satisfying a nice transport equation. The authors of \cite{Holzegel22} realised furthermore that curl--type derivatives of $\ms{g}$ can more easily be controlled. As such, they also introduced another auxiliary field $\ms{B}\sim \operatorname{curl} \delta \ms{g}$ such that $\ms{DB}$ controls some of the second derivatives of $\delta\ms{g}$. In turn, $\ms{DB}$ is controlled through the difference of the Weyl curvature. Furthermore, the ``bad terms" in \eqref{intro_diff} are such that the renormalisation \eqref{renormalisation_intro} essentially sends them on the left-hand side of \eqref{intro_diff}; see Section 1.4.2 of \cite{Holzegel22} for more details. One eventually obtains a new wave--transport system of the form:  
\begin{equation*}
\begin{cases}
    \paren{\ol{\Box}+\sigma_{\ms{A}}} \Delta\ms{A} \sim \rho^2\cdot \mc{O}\paren{\Delta\ms{A}} + \rho^3\cdot \mc{O}\paren{\ms{D}\Delta\ms{A}} + \rho^2 \cdot\mc{O}\paren{\delta\ms{B}', \ms{D}\delta\ms{B}'}\\
    \partial_\rho\delta\ms{B}' +\sigma_{\ms{B}}\delta\ms{B}' \sim \rho \cdot \mc{O}(\delta \ms{B}') +\rho\cdot \mc{O}\paren{\Delta\ms{A}}\text{,}
\end{cases}   
\end{equation*}
as $\rho\searrow 0$, where now $\ms{B}'\in \lbrace \ms{g}, \ms{m}, \ms{f}^0, \ms{f}^1, \ms{Q}, \ms{B}\rbrace$.

\subsubsection{Carleman estimates and unique continuation}

In \cite{Shao22}, the authors derived a Carleman estimate adapted to the AdS geometry for solutions to wave equations. Such estimates are particularly useful when the nonlinearities on the right-hand side of the wave equations exhibit sufficient decay in $\rho$, and typically take the following form, for $\mi{D}$ satisfying the GNCC: 
\begin{equation}
    \norm{f^{-p/2}(\ol{\Box} +  \sigma_{\ms{A}}) \ms{A}}_{L^2_\omega(\Omega_\star)}^2 + \lambda \cdot \text{boundary terms on $\mc{I}$} \geq C\lambda \norm{ \ms{A}}_{H^1_\omega(\Omega_\star)}^2\text{,}\qquad \lambda \geq \lambda_0>0\label{Carleman_intro}\text{,}
\end{equation}
where the weighted--Sobolev norms are taken on a region $\Omega_\star:= \lbrace f(\rho, x) \leq f_\star\rbrace$, with $f$ a function defining pseudoconvex level sets away from $\mc{I}$, and $f_\star>0$ small enough. One also requires the vertical tensor field $\ms{A}$ and its derivatives to vanish on $\lbrace f = f_\star\rbrace$. Furthermore, these norms are weighted by a Carleman weight $\omega(\lambda)$, as usual, and $0<2p<1$; see Section \ref{sec_carleman} for more details. Before obtaining the unique continuation result, one will need to treat the boundary terms.  

This is where Theorem \ref{thm_main_fg_ads_informal} plays a crucial role. Since the vertical fields of interest share indeed the same data on $\mi{D}$, one has necessarily: 
\begin{equation*}
    \delta \ms{g} = \mc{O}(\rho^n)\text{,}\qquad  \delta \ms{f}^0 = \mc{O}(\rho^{n-2})\text{,}\qquad \delta\ms{f}^1 = \mc{O}(\rho^{n-1})\text{,}
\end{equation*}
to some regularity along the vertical directions, and one can obtain similar asymptotics for the Weyl curvature and the fields $(\ul{\ms{h}}^0, \ul{\ms{h}}^2)$. Such an order of vanishing can in fact be improved, as long as one has enough regularity in the vertical directions, by integrating the Bianchi and Maxwell equations. We will show that this higher-order of decay will be enough to show that the boundary terms in \eqref{Carleman_intro} vanish. 

In order to show that the fields $\Delta\ms{A}$ actually vanish in a neighbourhood of $\mi{D}$, one simply need to use standard methods by using a cutoff $\chi$, allowing $\Delta\ms{A}$ and its derivatives to indeed vanish on $\lbrace f = f_\star\rbrace$; see the proof of Proposition \ref{prop_unique_continuation1}. Note that since the wave equations also involve the transport fields $\delta\ms{B}'$, one has to couple the above estimates to some transport Carleman estimates, derived in \cite{Holzegel22}. 

At this point, one has indeed showed that whenever $(\ms{g}, \check{\ms{g}})$ and $(F, \check{F})$ have identical data on $\mi{D}\subset \mc{I}$ satisfying the GNCC, one has necessarily: 
\begin{equation*}
    g - \check{g}\equiv 0\text{,}\qquad F- \check{F}\equiv 0\text{,}
\end{equation*}
in a neighbourhood $\Omega_\star$ of $\mi{D}$. It remains, however, to show that the result hold in the full gauge invariant setting. 

\subsubsection{Gauge freedom}

In Section \ref{sec:full_result}, we characterise gauge-equivalent boundary data for the full Einstein--Maxwell system. Namely, we show that the boundary data $\mi{A}$ and $\check{\mi{A}}$ are equivalent if there exists a smooth function $\mf{a}$ on $\mc{I}$ such that: 
\begin{gather}
    \label{check_g_n_intro}\check{\mf{g}}^{(0)} =e^{2\mf{a}}\mf{g}^{(0)}\text{,}\qquad \check{\mf{g}}^{(n)} = \mf{H}_n(\mf{g}^{(n)},\partial^{\leq n-2} \mf{g}^{(0)},\mf{D}^{\leq  n-4}\mf{f}^{1, (0)}, \mf{D}^{\leq n}\mf{a})\text{,}\\
    \label{check_f_intro}\begin{cases}
        \check{\mf{f}}^{1, (0)} = \mf{f}^{1, (0)}\text{,}\qquad \check{\mf{f}}^{0, ((n-4)_+)} = e^{-\mf{a}}\mf{f}^{0, ((n-4)_+)}\text{,}\qquad &n =3\text{,}\\
        \check{\mf{f}}^{1, (0)} =\mf{f}^{1, (0)}\text{,}\qquad \check{\mf{f}}^{0, ((n-4)_+)} = \mf{H}'_n\paren{\mf{f}^{0, ((n-4)_+)}, \partial^{\leq n-3}\mf{g}^{(0)}, \mf{D}^{\leq n-3}\mf{f}^{1, (0)},\mf{D}^{\leq n-3} \mf{a}}\text{,}\qquad &n\geq 4\text{,}
    \end{cases}
\end{gather}
where $\mf{H}_n$ and $\mf{H}'_n$ are universal functions depending on $n$ only. 

Furthermore, since the GNCC is also gauge--invariant, one can easily formulate from there a general covariant unique continuation result; see Theorem \ref{thm_UC_main}. To prove it, we follow the method of \cite{Holzegel22}, which shows how one can recover the setting in which the boundary data $\mi{A}$ and $\check{\mi{A}}$ are equal and on which $\mc{M} = \check{\mc{M}}$. 

\subsection{Future directions of study}

\subsubsection{Generalisations to nontrivial matter models}

Of course, the natural problem one may try to solve would be to look at matter models such as the Klein--Gordon field or the Vlasov fields, to cite only a few. 

In the first one, the Klein-Gordon mass introduces a non-trivial scale in the problem, and one will have to be careful when deriving the Fefferman-Graham expansion in order to avoid power mixing. It is believed that, in the well-posed range for the mass \cite{Warnick13, Holzegel12}, the GNCC should not depend on the scalar field and, furthermore, unique continuation should also be obtained. 

On the other hand, the treatment of the AdS--Einstein--Vlasov case should also be interesting. In particular, such a model may hide surprising features when it comes to the counterexamples, since these propagate along null geodesics, like the massless Vlasov field $f$. It is also interesting to note the connection between the Einstein--Vlasov system and high--frequency solutions, highlighted by the Burnett conjecture, as studied in\footnote{There is now a rich literature initiated by Luk, Huneau, Touati for these problems, presented in the following review \cite{huneau2024review}.} \cite{huneau2024, toua:go_eve}. Furthermore, the constructions of \cite{AB, GUISSET2024223} also involve a high-frequency geometric optics solution. 

\subsubsection{Optimality of the GNCC}

Another natural question is what happens when the GNCC fails. A (very) partial answer has been given in \cite{GUISSET2024223} where the authors show that the GNCC is in fact optimal for linear solutions to the wave equations for AdS and planar AdS. In the case of planar AdS, the counterexamples can be constructed for an ``arbitrarily long time". In the pure AdS setting \eqref{intro_pure_AdS}, however, the GNCC imposes time slabs of the form $(-T, T)\times \mathbb{S}^{n-1}\subset \mc{I}$ to satisfy $T>\pi/2$. In \cite{GUISSET2024223}, the authors show, at least in the linear case, that the unique continuation property for $T\leq \pi/2$ does not hold generically. In the more general setting, the authors also show that counterexamples can in general be constructed in the neighbourhood of some appropriate open regions of $\mc{I}$ which do not satisfy the GNCC. Once again these results only hold in the linear case. 

Showing, however, the following statement: 
\begin{center}
    \emph{Let $\mi{A}$, $\check{\mi{A}}$ be two gauge-equivalent boundary data on a domain $\mi{D}$ which does not satisfy the GNCC. Is it possible to construct solutions $(\mc{M}, g)$, $(\check{\mc{M}}, \check{g})$ to \eqref{Einstein_intro}, with $T=0$, and boundary data $\mi{A}$ and $\check{\mi{A}}$, respectively, such that for any boundary--preserving diffeomorphism $\phi:\mc{M}\rightarrow \check{\mc{M}}$:, the two metrics $g$ and $\phi_\star g$ are not isometric near $\mi{D}$ in $\mc{M}$?}
\end{center}
is far from being established. The construction will require completely novel techniques, on the condition that such a conjecture holds, since the usual Alinhac--Baouendi methods crucially rely on the linearity of the hyperbolic operator considered. One should note however the nonlinear counterexmaple of Métivier \cite{Mtivier1993CounterexamplesTH}, exploiting the linear construction of \cite{AB}. 

Furthermore, the Alinhac--Baounedi counterexamples are complex in nature, which seems to be incompatible with a result of the form of the above. The construction of purely real counterexample in general is also far from being obtained. 

\subsubsection{Global rigidity result}

The results proved here, as well as those in \cite{Holzegel22, mcgill2021}, are only local by essence. As a consequence, the holographic data only characterise in a unique way solutions to the Einstein equations near the conformal boundary. It would be interesting however, in the context of the black hole rigidity problem, for example, to obtain global uniqueness results. These would be certainly challenging to obtain and will typically require global assumptions on the manifold. 

It is worth noting that in \cite{hol_shao:uc_ads}, the authors show a global uniqueness result for AdS in the context of linearised gravity. They use in particular the so-called Poincaré patch, which allows one to recover Minkoswki from a portion of AdS via a conformal transformation. It is not known yet how one could obtain a global result in the nonlinear case.

\section{Preliminaries}

\subsection{Basic definitions}

In this section, we will give basic definitions and prove preliminary propositions.

The following definition, very general, describes what one should expect the near-boundary region to look like:
\begin{definition}[aAdS Region]
An aAdS region is a smooth manifold with boundary $\mc{M}$ of dimension $n+1, \, n\in \N$, that can be decomposed as: 
\begin{equation}\label{spacetime}
    \mc{M} := \left( 0, \rho_0\right]_\rho \times \mc{I}\text{,}
\end{equation}
with $\mc{I}$ a smooth manifold of dimension $n$, $\rho_0>0$ and $\rho \in C^\infty(\mc{M})$ denoting the the $\left(0, \rho_0\right]$-component of $\mathcal{M}$. 
\end{definition}
Such a splitting of the spacetime, as in \eqref{spacetime}, suggests working with a vertical formalism, where one decomposes the spacetime tensors into their different projections in the level sets of $\rho$. The object appearing in such a decomposition will be referred to as \emph{vertical tensor fields}, defined in the following sense: 

\begin{definition}[Vertical Bundle/Vertical Tensor Fields]\label{sec:vertical}
    Let $\mc{M}$ be an aAdS region. We will define the vertical bundle as the manifold containing tensors fields on $\mc{M}$ on a fixed level-set of $\rho$: 
    \begin{equation*}
        \ms{V}^k{}_l\mc{M} := \bigcup\limits_{\sigma \in \left(0, \rho_0 \right]}  T^k{}_l\lbrace \rho = \sigma \rbrace\text{.}
    \end{equation*}
    In particular, any smooth section of $\ms{V}^k{}_l\mc{M}$ will be defined as a \emph{vertical tensor field}. 
\end{definition}

\begin{remark}
    We shall use the same notations as in \cite{Holzegel22}, for simplicity. In particular, tensor fields $g$ in $\mc{M}$, vertical tensor fields $\ms{g}$, and tensor fields $\mathfrak{g}$ on $\mc{I}$ will be denoted by italic, sans script, and fraktur fonts, respectively. 
\end{remark}

\begin{remark}\label{rmk_def1}
    For any aAdS segment, the following can be noted: \begin{itemize}
        \item Note that any vertical tensor field $\ms{A}$ can be identified with a spacetime vector field by imposing vanishing contractions with $\partial_\rho$ and $d\rho$.
        \item Any vertical tensor field $\mc{A}$ can be seen as a smooth one-parameter family of tensor fields on $\mc{I}$.
        \item In particular, for a vertical tensor field $\ms{A}$, we will denote by $\ms{A}\vert_\sigma$ the tensor field on $\mc{I}$ obtained by setting $\rho = \sigma$, with $\sigma>0$ fixed. 
        \item Similarly, any tensor field $\mathfrak{A}$ on $\mc{I}$ can be extended to a vertical tensor field, also denoted $\mathfrak{A}$, independent of $\rho$. 
    \end{itemize}
    In the following, unless specified otherwise, we will work in the smooth class of tensor fields. 
\end{remark}
We will also use the following notation: 
\begin{definition} Let $\mc{M}$ be an aAdS segment.
    \begin{itemize}
        \item We will use latin letters $a, b, c, \dots$ for the vertical direction and greek letters $\alpha,\beta,\gamma,\dots$ for spacetime ones. 
        \item We will denote by $\partial_a$ coordinate derivatives along the vertical directions.  
        \item For any tensor field $T\in T^k{}_l\mc{M}$, and any vector $X\in T\mc{M}$, we denote by $\Lie_X T$ the Lie derivative of $T$ with respect to $X$. In the particular case of $X = \partial_\rho$, we write: $\Lie_\rho := \Lie_{\partial_\rho}$. 
        \item For any vertical tensor field $\ms{A}$ and $X\in T\mc{M}$, we define $\Lie_X \ms{A}$ as the Lie derivative of the spacetime tensor $A$ corresponding to $\ms{A}$, as noted in Remark \ref{rmk_def1}. 
    \end{itemize}
\end{definition}

\begin{definition}
    Let $\mc{M}$ be an aAdS segment with conformal boundary $\mc{I}$. Let $(U,\varphi)$ be a local chart of $\mc{I}$.  
    \begin{itemize}
        \item  We will assume in the rest of this work that $\bar{U}$ is compact in $\mc{I}$ and that $\varphi$ extends smoothly to $\partial U$. Such a coordinate system will be referred to as \emph{compact}.
        \item We denote by $\varphi_\rho$ the trivial extension of $\varphi$ from $U$ to $\left(0, \rho_0\right]\times U$, i.e. $\varphi_\rho := (\rho, \varphi)$. 
        \item We define, for any vertical tensor $\ms{A}\in \ms{V}^k{}_l\mc{M}$ and with respect to $(U,\varphi)$: 
        \begin{equation}
            \abs{\ms{A}}_{M, \varphi} := \sum\limits_{m=0}^M\sum\limits_{\substack{a_1,\dots, a_m\\b_1, \dots, b_k\\c_1,\dots, c_l}} \abs{\partial^{m}_{a_1\dots a_m}\ms{A}^{b_1, \dots, b_k}_{c_1, \dots, c_l}}\text{,}
        \end{equation}
        where $(\partial_a)$ are coordinate derivatives on $U$. 
        \item We say that a tensor $\ms{A}$ is locally bounded in $C^M$ if:
        \begin{equation}
        \norm{\ms{A}}_{M,\varphi}:=\sup\limits_{\left(0, \rho_0\right]\times U} \abs{\ms{A}}_{M, \varphi} <+\infty\text{.}
        \end{equation}
        \item A vertical tensor $\ms{A}$ converges to a tensor field $\mf{A}$ in $C^M$, denoted $\ms{A}\rightarrow^M \mf{A}$ if for any compact chart $(U,\varphi)$ on $\mc{I}$: 
        \begin{equation}
            \lim\limits_{\sigma\searrow 0} \sup\limits_{\lbrace \sigma \rbrace \times U} \abs{\ms{A}- \mf{A}}_{M,\varphi} = 0\text{.} 
        \end{equation}
        \item A vertical tensor $\ms{A}$ converges rapidly to a vector field $\mf{A}$ in $C^M$, denoted $\ms{A}\Rightarrow^M \mf{A}$, if for any compact chart $(U,\varphi)$ on $\mc{I}$: 
        \begin{equation}
            \ms{A}\rightarrow^M \mf{A}\text{,}\qquad \sup\limits_U \int_0^{\rho_0}\sigma^{-1}\abs{\ms{A}- \mf{A}}_{M,\varphi}\vert_\sigma d\sigma <+\infty \text{.}
        \end{equation}
    \end{itemize}
\end{definition}

Let us now define the type of geometries we will be interested in: 
\begin{definition}\label{def:FG_gauge}
    Let $\mc{M}$ be an aAdS region and $g$ a Lorentzian metric on $\mc{M}$. We define $(\mc{M}, g)$ to be an FG-aAdS segment if there exists $\mf{g}^{(0)}\in T^0{}_2\mc{I}$, a symmetric Lorentzian metric, such that $g$ takes the form: 
    \begin{equation}\label{FG-gauge}
        g = \rho^{-2}\paren{d\rho^2 + \ms{g}(\rho)}\text{,}\qquad \ms{g} \rightarrow^0 \mathfrak{g}^{(0)}\text{,} 
    \end{equation}
    with $\ms{g}(\rho)\in \ms{V}^0{}_2\mc{M}$ a Lorentzian metric. 
\end{definition}

\begin{definition}[Strongly FG-aAdS segment]\label{def_strongly_FG_aAdS}
    A FG-aAdS segment $(\mc{M}, g)$ is qualified as \emph{strongly FG-aAdS} if there exists a symmetric tensor $\mf{g}^{(2)}$ on $\mc{I}$ such that the following limits hold: 
    \begin{gather*}
        \ms{g} \rightarrow^3 \mf{g}^{(0)}\text{,}\qquad \Lie_\rho\ms{g}\rightarrow^2 0\text{,}\qquad \Lie_\rho^2\ms{g} \rightarrow^1 2\mf{g}^{(2)}\text{,}\qquad |\Lie_\rho^3\ms{g}|_{0, \varphi}\lesssim 1\text{.}
    \end{gather*}
\end{definition}

\begin{remark}
We will show, in Section \ref{sec:FG_expansion}, that imposing $g$ from Definition \ref{def:FG_gauge} to satisfy the Einstein equations \eqref{Einstein_intro}, with $T$ the stress-energy tensor from the electrovacuum theory, implies that $g$ must satisfy automatically Definition \ref{def_strongly_FG_aAdS}. This is true in particular in the vacuum case.
\end{remark} 

A gauge condition such as \eqref{FG-gauge} will be referred to in this work as a \emph{Fefferman-Graham gauge}. The manifold $(\mc{I}, \mathfrak{g}^{(0)})$ will be referred to as the \emph{conformal boundary} of $\paren{\mc{M},\, g,\, \rho}$. The Definition \ref{def_strongly_FG_aAdS} will turn out to be useful in Section \ref{sec_carleman}

Finally, we will need the following notations -- useful to settle the geometry we will be interested in: 
\begin{remark}
    Let $(\mc{M},\, g)$ be a FG-aAdS segment, we will write 
    \begin{itemize}
        \item associated to $g$: $g^{-1}$, $\nabla$, $\operatorname{R}$, $\operatorname{Ric}$, $\operatorname{Rs}$, 
        \item associated to $\ms{g}$: $\ms{g}^{-1}$, $\ms{D}$, $\ms{R}$, $\ms{Ric}$, $\ms{Rs}$,
        \item associated to $\mathfrak{g}^{(0)}$: $(\mathfrak{g}^{(0)})^{-1}$, $\mathfrak{D}^{(0)}$, $\mathfrak{R}^{(0)}$, $\mathfrak{Ric}^{(0)}$, $\mathfrak{Rs}^{(0)}$, 
    \end{itemize}
    as the metric inverse, the Levi-Civita connection, the Riemann tensor, the Ricci tensor and the Ricci scalar, respectively. We will also write: 
    \begin{equation*}
        \Box_{\ms{g}} := \ms{g}^{ab}\ms{D}^2_{ab}\text{.}
    \end{equation*}
    
    Finally, we will denote as $W$ the Weyl tensor associated to ${g}$. 
\end{remark}

\subsection{Mixed covariant formalism}\label{sec:mixed_formalism}

In this section, we will briefly remind the reader about the formalism used in \cite{Holzegel22}, convenient when it comes to deal with tensorial quantities on spacetimes such as \eqref{FG-gauge}. 

We will use the same notational conventions, we will give here a quick summary: 
\begin{definition}
    We will write multi-indices as overlined latin letters, $\bar{a}:= a_1\dots a_l$, as an example. Futhermore, for any $1\leq i\leq l$, we will write $\hat{a}_i$ as the multi-index $\bar{a}$ with index $a_i$ removed and $\hat{a}_i\brak{b}$ as the multi-index $\bar{a}$ with index $a_i$ replaced by $b$, e.g. 
        \[
        \hat{a}_i\brak{b} := a_1 a_2 \dots a_{i-1} b a_{i+1}\dots a_l\text{.}
        \] 
    The generalisation to the replacement of more indices is straightforward. 

    Furthermore, for any (vertical) tensor $A$ of rank $(0, l)$, we will use the following convention for the anti-symmetrisation and symmetrisation of indices: 
    \begin{gather*}
        A_{[a_1\dots a_l]} = \frac{1}{l!}\sum\limits_{p\in S_l} \operatorname{sgn}(p) \cdot A_{p(a_1)\dots p(a_l)}\text{,}\qquad 
        A_{(a_1\dots a_l)} = \frac{1}{l!}\sum\limits_{p\in S_l}  A_{p(a_1)\dots p(a_l)}\text{.}
    \end{gather*}
    We will also write $A_{[a_1\dots a_{i-1}|a_i \dots a_{j-1}|a_j\dots a_l]}$, where the permutation acts only on the indices $\lbrace a_1,\dots a_{i-1}a_j,\dots a_l\rbrace$, and similarly for the symmetrisation. 
\end{definition}

\begin{definition}
    Let $\ms{A}$ be a vertical tensor field of rank $(0, k)$. With respect to any compact coordinate system $(U,\varphi)$ on $\mc{I}$, we define the divergence of $\ms{A}$ with respect to $\ms{g}$, the following vertical tensor of rank $(0, k-1)$: 
    \begin{align*}
        (\ms{D} \cdot \ms{A})_{\bar{a}} = \ms{g}^{cd}\ms{D} \ms{A}_{\hat{a}_1[d]}\text{.}
    \end{align*}
    Similarly, for any $\mf{A}$ tensor field of rank $(0, l)$ on $\mc{I}$, we define with respect to $(U,\varphi)$:
    \begin{equation*}
        (\mf{D}\cdot \mf{A})_{\bar{a}} = (\mf{g}^{(0)})^{cd}\mf{D}_c \mf{A}_{\hat{a}_1[d]}\text{.}
    \end{equation*}
\end{definition}

The following proposition allows one to construct a connection on the vertical bundle, allowing in particular differentiating along spacetime directions in a covariant manner vertical tensor fields. More precisely, the usual Levi-Civita covariant derivative for $\ms{g}$ allows to differentiate only along vertical direction. The connection $\overline{\ms{D}}$ will also allow one to differentiate along $\rho$. 

\begin{proposition}\label{sec:aads_prop_covariant_D}
    Let $(\mc{M}, g)$ be a FG-aAdS segment. There exists a unique family of connections $\ol{\ms{D}}$ on the vertical bundle $\ms{V}^k{}_l\mc{M}$, for all $k,\, l$, which can be written, for any vertical tensor field $\ms{A}$, and on any compact coordinate chart $(U,\varphi)$: 
    \begin{align}\label{vertical_derivative}
        &\notag \ol{\ms{D}}_a \ms{A}^{\bar{b}}{}_{\bar{c}} = \ms{D}_a \ms{A}^{\bar{b}}{}_{\bar{c}}\text{,}\\
        &\ol{\ms{D}}_\rho \ms{A}^{\bar{b}}{}_{\bar{c}} = \Lie_\rho \ms{A}^{\bar{b}}{}_{\bar{c}} + \frac{1}{2}\sum\limits_{i=1}^k \ms{g}^{a_i d}\Lie_\rho \ms{g}_{de}\ms{A}^{\hat{b}_i[e]}{}_{\bar{c}} - \frac{1}{2}\sum\limits_{j=1}^l \ms{g}^{de}\Lie_\rho\ms{g}_{b_j d}\ms{A}^{\bar{b}}{}_{\hat{c}_j[e]}\text{.}
    \end{align}
    Furthermore, such a connection satisfies, for any spacetime vector field $X$ on $\mc{M}$:
    \begin{itemize}
        \item For any vertical tensors $\ms{A}, \, \ms{B}$: 
        \[
            \ol{\ms{D}}_X\paren{\ms{A}\otimes \ms{B}} =  \ol{\ms{D}}_X A \otimes \ms{B} + \ms{A} \otimes  \ol{\ms{D}}_X \ms{B}\text{.}
        \]
        \item For any vertical tensor field $\ms{A}$, any contraction $\mc{C}$: 
        \[
         \ol{\ms{D}}_X\paren{\mc{C} \ms{A}} = \mc{C}\paren{ \ol{\ms{D}}_X \ms{A}}\text{.}
        \]
        \item $\ol{\ms{D}}$ is compatible with $\ms{g}$: 
        \[
        \ol{\ms{D}}_X \ms{g} = \ol{\ms{D}}_X \ms{g}^{-1} = 0\text{.}
        \]
    \end{itemize}
\end{proposition}

\begin{proof}
     See Appendix \ref{app:D_bar}
\end{proof}

We are now ready to define the mixed tensor tensor fields, namely, tensor fields that will contain both spacetime and vertical contributions. More precisely, 
\begin{definition}\label{def_mixed_connection}
    Let $(\mc{M}, g)$ be a FG-aAdS segment. We define the following bundle: 
    \begin{equation*}
        T^\kappa{}_\lambda\ms{V}^k{}_l\mc{M} := T^\kappa{}_\lambda\mc{M}\otimes \ms{V}^k{}_l\mc{M}\text{,}
    \end{equation*}
    as the \emph{mixed bundle} of rank $(\kappa, \lambda; k, l)$. Smooth sections of this tensor bundle are referred to as \emph{mixed tensor fields}. Moreover, we define the connection $\ol{\nabla}$ on $T^\kappa{}_\lambda\ms{V}^k{}_l\mc{M}$ as the following: for any mixed tensor field $T:=A\otimes \ms{B}$ and vector field $X$ in $T\mc{M}$, we define the following operator:
    \begin{equation*}
        \ol{\nabla}_XT := \nabla_X A \otimes \ms{B} + A \otimes \ol{\ms{D}}_X\ms{B}\text{,}
    \end{equation*}
    extended linearly on the mixed bundle. Namely $\ol{\nabla}$ is simply the tensor product of the connections $\nabla$ and $\ol{\ms{D}}$. 
\end{definition}

This definition will allow use to perform operations on both vertical and spacetime tensor fields simultaneously.
\begin{remark}\label{rmk_mixed}
    There is an obvious bijection between $T^\kappa{}_\lambda\mc{M}$ and $T^\kappa{}_\lambda\ms{V}^0{}_0\mc{M}$. Namely, any spacetime tensor field of order $(\kappa, \lambda)$ can be seen as a mixed tensor field of order $(\kappa, \lambda; 0, 0)$, and similarly for vertical tensor fields. 
\end{remark}

The following proposition shows that the connection $\ol{\nabla}$ defined above satisfies the right properties:

\begin{proposition}[Proposition 2.28 of \cite{McGill20}]
    Let $(\mc{M}, g)$ be a FG-aAdS segment and let $\mathbf{A}, \mathbf{B}$ be two mixed tensor fields. Then, for any spacetime vector field $X$:
    \begin{itemize}
        \item The connection $\ol{\nabla}$ satisfies the Leibniz rule: 
        \begin{gather*}
        \ol{\nabla}_X (\mathbf{A}\otimes \mathbf{B}) = (\ol{\nabla}_X \mathbf{A})\otimes \mathbf{B} +  \mathbf{A}\otimes (\ol{\nabla}_X\mathbf{B})\text{,}
        \end{gather*}
        \item The connection $\ol{\nabla}$ is compatible with both the vertical and the spacetime metric:
        \begin{gather*}
            \ol{\nabla}_X g = \ol{\nabla}_X g^{-1} = 0\text{,}\qquad \ol{\nabla}_X \ms{g} = \ol{\nabla}_X \ms{g}^{-1} = 0\text{.}
        \end{gather*}
    \end{itemize}
\end{proposition}

The above definition allows one to make sense of covariant derivatives of mixed tensors, and especially a covariant wave operator: 
\begin{definition}\label{def_nabla_bar}
    Let $(\mc{M}, g)$ be a FG-aAdS segment, $\mathbf{A}$ a mixed tensor field of order $(\kappa, \lambda; k, l)$ and $X$ a vector field in $\mc{M}$. We define: 
    \begin{itemize}
        \item The mixed covariant derivative of $\mathbf{A}$, $\ol{\nabla}\mathbf{A}$, as the mixed tensor field of order $(\kappa, \lambda + 1; k, l)$ mapping $X$ to $\nabla_X \mathbf{A}$. 
        \item The mixed Hessian $\ol{\nabla}^2\mathbf{A}$, a mixed tensor field of order $(\kappa, \lambda + 2; k,l)$ as the mixed covariant derivative of $\ol{\nabla}\mathbf{A}$. 
        \item The mixed d'Alembertian, or mixed wave operator, as the $g$--trace of the Hessian, namely, with respect to any system of coordinates, $\ol{\Box}\mathbf{A}:=g^{\alpha\beta}\ol{\nabla}^2_{\alpha\beta}\mathbf{A}$. 
    \end{itemize}
\end{definition}

\begin{remark}
    Observe that, although we will only apply $\ol{\Box}$ on vertical tensor fields, its definition is necessary since, for any vertical tensor field of order $(k,l)$, the vertical covariant derivative $\overline{\ms{D}}\ms{A}$ is a mixed tensor field of order $(0, 1; k,l)$.
\end{remark} 

Finally, we will need the expression of the connection coefficients: 
\begin{proposition}[Christoffel symbols]\label{prop_christoffel_FG}
    Let $(\mc{M}, g)$ be a FG-aAdS segment and let $(U, \varphi)$ be a local coordinate system on $\mc{I}$ and let $\Gamma$, $\ms{\Gamma}$ be the Christoffel symbols with respect to $\nabla$ and $\ol{\ms{D}}$, respectively: 
    \begin{gather*}
        \nabla_\alpha \partial_\beta := \Gamma^\gamma_{\alpha\beta}\partial_\gamma\text{,}\qquad \ol{\ms{D}}_\alpha\partial_\beta = \ms{\Gamma}_{\alpha\beta}^\gamma\partial_\gamma\text{.}
    \end{gather*}
    These symbols can be expressed as: 
    \begin{gather*}
        \Gamma^{\alpha}_{\rho\rho} = - \rho^{-1}\delta^\alpha_\rho\text{,}\qquad \Gamma^{\rho}_{a\rho} = 0\text{,}\\
        \Gamma_{ab}^\rho = \rho^{-1}\ms{g}_{ab} - \frac{1}{2}\Lie_\rho \ms{g}_{ab}\text{,}\qquad \ms{\Gamma}_{ab}^\rho = - \frac{1}{2}\Lie_\rho \ms{g}_{ab}\text{,}\\
        \Gamma^a_{b\rho} = -\rho^{-1}\delta^a_b + \frac{1}{2}\ms{g}^{ac}\Lie_\rho\ms{g}_{cb}\text{,}\qquad \ms{\Gamma}^a_{b\rho} = \frac{1}{2}\ms{g}^{ac}\Lie_\rho \ms{g}_{cb}\text{.}
    \end{gather*}
    Furthermore, for any mixed tensor field $\mathbf{A}$ of rank $(\kappa, \lambda; k, l)$, one has on $(U, \varphi)$: 
    \begin{align*}
        \ol{\nabla}_\gamma\mathbf{A}^{\bar{\alpha}}{}_{\bar{\beta}}{}^{\bar{a}}{}_{\bar{b}} = \partial_\gamma\mathbf{A}^{\bar{\alpha}}{}_{\bar{\beta}}{}^{\bar{a}}{}_{\bar{b}} &+ \sum\limits_{i=1}^{\kappa}\Gamma^{\alpha_i}_{\gamma \delta}\mathbf{A}^{\hat{\alpha}_i[\delta]}{}_{\bar{\beta}}{}^{\bar{a}}{}_{\bar{b}} - \sum\limits_{j=1}^{\lambda}\Gamma^{\delta}_{\gamma b_j}\mathbf{A}^{\bar{\alpha}}{}_{\hat{\beta}_j[\delta]}{}^{\bar{a}}{}_{\bar{b}} \\
        &+\sum\limits_{i=1}^{k}\ms{\Gamma}^{a_i}_{\gamma d}\mathbf{A}^{\bar{\alpha}}{}_{\bar{\beta}}{}^{\hat{a}_i[d]}{}_{\bar{b}}-\sum\limits_{j=1}^{l}\ms{\Gamma}^{d}_{\gamma b_j}\mathbf{A}^{\bar{\alpha}}{}_{\bar{\beta}}{}^{\bar{a}}{}_{\hat{b}_i[d]}\text{.}
    \end{align*}
\end{proposition}

\begin{proof}
    See Appendix \ref{app:christoffel}
\end{proof}

\subsection{General formulas and notations}

This section, technical, will be dedicated to describing the useful notations and formulas for vertical tensor fields used throughout this work. 

\begin{definition}\label{def_remainder}
    Let $(\mc{M}, g)$ be a FG-aAdS segment. 
    \begin{itemize}
        \item Let $K\geq 1$, for any vertical tensor $\ms{A}$, we will write: 
        \begin{equation*}
            \ms{A}^K := \bigotimes^K \ms{A}\text{.}
        \end{equation*}
        Furthermore, we will write for the inverse metric $\ms{g}^{-K} := \left(\ms{g}^{-1}\right)^K$.
        \item For simplicity, we define the following symmetric vertical $(0,2)$-tensor: 
        \begin{equation*}
            \ms{m} := \Lie_\rho \ms{g}\text{.}
        \end{equation*}
    \end{itemize}
\end{definition}

In this work, it will be convenient to write lower order terms in a clever way, exhibiting the dependence on the different fields. As a consequence, we will follow the notations of \cite{Shao22} and \cite{Holzegel22}: 
\begin{definition}\label{sec:aads_def_errors}
    Let $(\mc{M}, g)$ be a FG-aAdS segment, $N\geq 1$, and let $(A_j)_{j=1}^{N}$ be a collection of non-metric vertical tensor fields on $\mc{M}$. We write $\mathcal{S}(\ms{g};\ms{A}_1, \dots, \ms{A}_N)$ as the following expression: 
    \begin{equation}\label{def_error}
        \sum\limits_{m_1, m_2=0}^M \sum\limits_{j=1}^J Q_j\paren{\ms{g}^{-m_1}\otimes \ms{g}^{+m_2}\otimes \ms{A}_1\otimes \dots \otimes \ms{A}_N}\text{,}\quad J\geq 0\text{,} \qquad M\geq 0\text{,}
    \end{equation}
    where $Q_j$, $1\leq j\leq J$, is a combination of zero or more of the following trivial operations: 
    \begin{enumerate}
        \item a non-metric contraction,
        \item a component permutation,
        \item a multiplication by a constant.
    \end{enumerate}
    If $M=0$, we will simply write $\mc{S}(\ms{A}_1, \dots, \ms{A}_N)$. 
\end{definition}

\begin{proposition}\label{sec:aads_commutation_Lie_D}
    Let $(\mc{M}, g)$ be an FG-aAdS segment, let $(U,\varphi)$ be a coordinate system on $\mc{I}$ and let $\ms{A}\in \ms{V}^k_l \mc{M}$. For any vertical tensor field $\ms{A}$:
    \begin{align}
        \label{commutation_L_D}&\left[\Lie_\rho, \ms{D}\right] \ms{A} = \mc{S}\paren{\ms{g}; \ms{D}\ms{m}, \ms{A}}\text{,}\\
        \label{commutation_D_D}&\left[\bar{\ms{D}}_\rho, \ms{D}\right]\ms{A} = \mc{S}\paren{\ms{g}; \ms{m}, \ms{DA}} + \mc{S}\paren{\ms{g}; \ms{Dm}, \ms{A}}\text{.}
    \end{align}
    For any $p\in \R$, the following identity holds: 
    \begin{equation}\label{wave_power}
        \ol{\Box}(\rho^p\ms{A}) = \rho^p \ol{\Box}\ms{A} + 2p\rho^{p+1}\ol{\ms{D}}_\rho\ms{A} - p(n-p)\rho^p\ms{A}+p\rho\mc{S}\paren{\ms{g}; \ms{m}, \rho^p\ms{A}}\text{.}
    \end{equation}
    Furthermore, the covariant vertical wave operator $\ol{\Box}$ can be written as: 
    \begin{align}\label{wave_op}
        \ol{\Box} \ms{A} =& \rho^2 \Lie_\rho^2 \ms{A} -(n-1)\rho\cdot \Lie_\rho \ms{A} + \rho^2 {\Box}_{\ms{g}} \ms{A} \\
        \notag&+ \rho^2\mc{S}\paren{\ms{g}; \ms{m}, \Lie_\rho \ms{A}} + \rho \mc{S}\paren{\ms{g}; \ms{m}, \ms{A}} + \rho^2 \mc{S}\paren{\ms{g}; \Lie_\rho \ms{m}, \ms{A}} + \rho^2 \mc{S}\paren{\ms{g}; \ms{m}^2, \ms{A}}\text{,}
    \end{align}
    where we wrote $\Box_{\ms{g}} := \ms{g}^{ab}\ms{D}^2_{ab}$.
\end{proposition}

\begin{proof}
    See Appendix \ref{app:aads_commutation_Lie_D}
\end{proof}

\begin{proposition}\label{prop_higher_order_commutation}
    Let $(\mc{M}, g)$ be a FG-aAdS spacetime and let $k\geq 1$. Then, for any vertical tensor $\ms{A}$: 
    \begin{equation*}
        \left[ \Lie_\rho^k, \ms{D}\right]\ms{A} = \sum\limits_{\substack{j+j_0 + j_1+\dots +j_\ell=k\\j_p\geq 1,\, j<k}}\mc{S}\paren{\ms{g}; \Lie_\rho^j \ms{A}, \ms{D}\Lie_\rho^{j_0}\ms{g}, \Lie_\rho^{j_1}\ms{g}, \dots, \Lie^{j_\ell}_\rho \ms{g}}\text{.}
    \end{equation*}
\end{proposition}

\begin{proof}
    See \cite{shao:aads_fg}, Proposition 2.33.
\end{proof}

In this work, we will very often decompose derivatives vertically. As such, the next proposition allows one to do so for any mixed tensor field. 
\allowdisplaybreaks{
\begin{proposition}\label{sec:aads_derivatives_vertical}[Proposition 2.34 in \cite{Holzegel22}]
    Let $A\in T^0{}_l\mc{M}$ be a tensor field on $\mc{M}$, with $l = l_1+l_2$, and let $\ms{A}\in \ms{V}^0{}_{l_2}\mc{M}$ be the vertical tensor field corresponding to $A$, where in any coordinate system $(U, \varphi)$: 
    \begin{equation*}
        \ms{A}_{\bar{a}} := A_{\bar{\rho}\bar{a}}\text{,} \qquad (\bar{\rho}, \bar{a}) := (\underbrace{\rho,\dots, \rho}_{l_1}, \underbrace{a_1, \dots, a_{l_2}}_{l_2})\text{.}
    \end{equation*}
    The following hold for any coordinate system $(U,\varphi)$: 
    \begin{align}
        &\nabla_\rho A_{\bar{\rho}\bar{a}} = \rho^{-l}\ol{\ms{D}}_{\rho}(\rho^l\ms{A})_{\bar{a}}\text{,}\\
        &\nabla_c A_{\bar{\rho}\bar{a}} = (\ms{D}_c \ms{A})_{\bar{a}} + \rho^{-1}\sum\limits_{i=1}^{l_1}(\ms{A}^\rho_i)_{c\bar{a}} - \rho^{-1}\sum\limits_{j=1}^{l_2}\ms{g}_{ca_i}(\ms{A}^\nu_j)_{\hat{a}_j}\\
        &\notag\hspace{50pt} + \sum\limits_{i=1}^{l_1}\mc{S}\paren{\ms{g}; \ms{m}, \ms{A}^\rho_i} + \sum\limits_{j=1}^{l_2}\mc{S}\paren{ \ms{m}, \ms{A}^\nu_j}\text{,}\\
        &\ol{\Box}A_{\bar{\rho}\bar{a}} = \rho^{-l}\ol{\Box}(\rho^l \ms{A})_{\bar{a}} + 2\rho\sum\limits_{i=1}^{l_1}\ms{g}^{bc}\ms{D}_b (\ms{A}^{\rho}_i)_{c\bar{a}} - 2\rho \sum\limits_{i=1}^{l_2}\ms{D}_{a_j}(\ms{A}^\nu_j)_{\hat{a}_j}-(nl_1+l_2)\ms{A}_{\bar{a}}\\
        &\notag\hspace{50pt} -2\sum_{i=1}^{l_1}\sum\limits_{j=1}^{l_2}(\ms{A}^{\rho, \nu}_{i,j})_{a_j\hat{a}_j} + 2\sum\limits_{1\leq i<j\leq l_1}\ms{g}^{bc}(\ms{A}^{\rho, \rho}_{i,j})_{bc\bar{a}} + 2\sum\limits_{1\leq i<j\leq l_2}\ms{g}_{a_ia_j}(\ms{A}^{\nu\nu}_{i,j})_{\hat{a}_{i,j}}\\
        &\notag \hspace{50pt}+ \rho\mc{S}\paren{\ms{g}; \ms{m}, \ms{A}}_{\bar{a}} + \rho^2 \mc{S}\paren{\ms{g}; \ms{m}^2, \ms{A}}_{\bar{a}} + \rho^2 \sum\limits_{i=1}^{l_1}\paren{\mc{S}\paren{\ms{g}; \ms{D}\ms{m}, \ms{A}^\rho_i)}_{\bar{a}} + \mc{S}\paren{\ms{g}; \ms{m}, \ms{D}\ms{A}^\rho_i)}_{\bar{a}}}\\
        &\notag\hspace{50pt} + \rho^2\sum\limits_{j=1}^{l_2}\paren{\mc{S}\paren{\ms{g}; \ms{m},\ms{D}\ms{A}^\nu_i)}_{\bar{a}} + \mc{S}\paren{\ms{g}; \ms{D}\ms{m},\ms{A}^\nu_i)}_{\bar{a}}}\\
        &\notag\hspace{50pt}+ \rho\sum\limits_{i\leq i<j\leq l_1}\paren{\mc{S}\paren{\ms{g}; \ms{m}, \ms{A}^{\rho\rho}_{i,j}}_{\bar{a}} + \rho \mc{S}\paren{\ms{g}; \ms{m}^2, \ms{A}^{\rho\rho}_{i,j}}_{\bar{a}}}\\
        &\notag\hspace{50pt}+\rho\sum\limits_{i\leq i<j\leq l_2} \paren{\mc{S}\paren{\ms{g}; \ms{m}, \ms{A}^{\nu\nu}_{i,j}}_{\bar{a}} + \rho\mc{S}\paren{\ms{g}; \ms{m}^2, \ms{A}^{\nu\nu}_{i,j}}_{\bar{a}}} + \rho\sum\limits_{i=1}^{l_1}\sum\limits_{j=1}^{l_2} \mc{S}\paren{\ms{g}; \ms{m}, \ms{A}^{\rho\nu}_{i,j}}_{\bar{a}} \\
        &\notag\hspace{50pt} + \rho^2 \sum\limits_{i=1}^{l_1}\sum\limits_{j=1}^{l_2}\mc{S}\paren{\ms{g}; \ms{m}^2, \ms{A}^{\rho\nu}_{i,j}}_{\bar{a}}\text{,}
    \end{align}
    where we wrote: 
    \begin{itemize}
        \item for $1\leq i\leq l_1$, we defined $\ms{A}^\rho_i$ to be the $(0, l_2+1)$ vertical tensor defined by: 
        \begin{equation*}
            (\ms{A}^\rho_i)_{c\bar{a}} := A_{\hat{\rho}_i[c]\bar{a}}\text{,}
        \end{equation*}
        \item for $1\leq j \leq l_2$, we defined $\ms{A}^\nu_j$ to be the $(0, l_2-1)$ vertical tensor defined by:
        \begin{equation*}
            (\ms{A}^\nu_j)_{\hat{a}_j} := A_{\bar{\rho}\hat{a}_j[\rho]}\text{,}
        \end{equation*}
        \item for $1\leq i, j \leq l_1$, and $i\neq j$, we defined $\ms{A}^{\rho\rho}_{i,j}$ to be the $(0, l_2+2)$ vertical tensor defined by: 
        \begin{equation*}
            (\ms{A}^{\rho\rho}_{i,j})_{cb\bar{a}} := A_{\hat{\rho}_{i,j}[c,b]\bar{a}}\text{,}
        \end{equation*}
        \item for $1\leq i \leq l_1$, $1\leq j\leq l_2$, we defined $\ms{A}^{\rho\nu}_{i,j}$ to be the $(0, l_2)$ vertical tensor defined by: 
        \begin{equation*}
            (\ms{A}^{\rho\nu}_{i,j})_{b\hat{a}_j} := A_{\hat{\rho}_i[b]\hat{a}_j[\rho]}\text{,}
        \end{equation*}
        \item for $1\leq i,j\leq l_2$, we defined $\ms{A}^{\nu\nu}_{ij}$ to be the $(0, l_2-2)$ vertical tensor defined by: 
        \begin{equation*}
            (\ms{A}^{\nu\nu}_{i,j})_{\hat{a}_{i,j}} := A_{\bar{\rho}\hat{a}_{i,j}[\rho,\rho]}\text{.}
        \end{equation*}
    \end{itemize}
\end{proposition}}

\begin{proof}
    See \cite{Holzegel22}, Proposition 2.34.
\end{proof}

\subsection{Asymptotically Anti-de Sitter spacetimes}

\begin{definition}[Non-vacuum FG-aAdS Segment]
    Let $(\mc{M}, g)$ be an FG-aAdS segment and $\Phi$ a nonempty collection of fields. A triplet $(\mc{M}, g, \Phi)$ is called \emph{non-vacuum} if it satisfies the following system in $\mc{M}$: 
    \begin{equation}\label{aads_non-vacuum}
            \operatorname{Rc}[g] - \frac{1}{2}\operatorname{Rs}[g]\cdot g + \Lambda \cdot g = T\brak{g; \Phi}\text{,}\qquad \Lambda := -\frac{n(n-1)}{2}\text{,}
    \end{equation}
    where $T$ is the stress-energy tensor associated to the matter field $\Phi$.
\end{definition}

\begin{remark}
    Of course, the special case $\Phi=\emptyset$ corresponds to \emph{vacuum FG-aAdS segment} as defined in \cite{Holzegel22, Shao22, shao:aads_fg}. 
\end{remark}

\begin{example}(Vlasov-FG-aAdS segment)
    As an example, one could define the Vlasov-FG-aAdS segment by defining $\Phi = \lbrace f \rbrace$, with $f\, : \, \mc{TM}\rightarrow \R_{\geq 0}$ the Vlasov distribution (with its usual support, depending on the massive or massless case) and: 
    \begin{gather}
        T^V[f]|_x = \int_{\mc{T}_x\mc{M}} \abs{g}^{1/2} f p\otimes p\text{,} \qquad p^\mu \partial_\mu f - \Gamma^\alpha_{\beta\gamma}p^\beta p^\gamma \partial_{p^\alpha}f=0\text{,}
    \end{gather}
    the usual stress-energy tensor and transport equation of the Vlasov distribution $f$, where $\mc{T}_x\mc{M}$ is the fiber of $\mc{TM}$ on $x\in \mc{M}$. The Vlasov-FG-aAdS segment can therefore be defined as the triplet $(\mc{M}, g, \{f\})$ satisfying \eqref{aads_non-vacuum}. 
\end{example}

\begin{definition}\label{def_modif_T}
    Let $(\mc{M}, g, \Phi)$ be a non-vacuum FG-AdS segment. We define the following modifications of the stress-energy tensor: 
    \begin{gather*}
        \tilde{T} := T - \frac{1}{n}\operatorname{tr}_g T\cdot g\text{,}\qquad \hat{T} := T - \frac{\operatorname{tr}_g T}{n-1}\cdot g \text{.}
    \end{gather*}
\end{definition}

One immediately obtains the following: 
\begin{proposition}\label{prop_aads_einstein}
Let $(\mc{M}, g, \Phi)$ be a non-vacuum FG-aAdS segment. The following holds: 
\begin{gather}\label{aads_einstein}
    \operatorname{Rc}[g] = -n \cdot g +\hat{T}\text{,}\quad \operatorname{Rs}[g] = -n(n+1) - \frac{2\tr_g T}{n-1}\text{.}
\end{gather}
Furthermore, the Weyl tensor can be expressed as:
\begin{gather}\label{sec:aads_weyl_equation}
    W = \operatorname{R}[g] + \frac{1}{2} g\star g - \frac{1}{n-1}g\star \tilde{T}\text{\emph{,}}
\end{gather}
where the operation $\star$ is defined as: 
\begin{align*}
    (S \star T)_{\alpha\beta\gamma\delta} := S_{\alpha \gamma}T_{\beta\delta}- S_{\alpha\delta}T_{\beta\gamma} + S_{\beta\delta}T_{\alpha\gamma} - S_{\beta \gamma}T_{\alpha\delta}\text{,}
\end{align*}
for any $(0,2)-$tensors $S,\, T$. 
\end{proposition}

Here, we will be looking at the AdS--electrovacuum system, which is given by:
\begin{definition}[Maxwell-FG-aAdS Segment]\label{def_Maxwell_FG_aAdS}
    Let $(\mc{M}, g, F)$ be a non-vacuum FG-aAdS with $F\in \Lambda^2\mathcal{M}$ a closed 2-form. We will define the triplet $(\mc{M}, g,F)$ to be a (source-free) \emph{Maxwell-FG-aAdS} segment if: 
    \begin{gather}\label{stress_energy_tensor}
        T\brak{g;F}_{\alpha\beta}dx^\alpha \otimes dx^\beta := \brac{g^{\mu \nu}F_{\alpha \mu}F_{\beta\nu} - \frac{1}{4}g_{\alpha \beta}g^{\mu\nu}g^{\sigma\tau}F_{\mu\sigma}F_{\nu\tau}}dx^\alpha\otimes dx^\beta\text{\emph{,}}\\
        \notag \text{\emph{div}}_g F=0\text{\emph{.}}
    \end{gather}
    Furthermore, we will assume the existence of an antisymmetric 2-form $\mf{f}^{1,{(0)}}_{ab}$  and a one-form $\mathfrak{f}^{0, (0)}_a$ on $\mc{I}$ such that:
    \begin{equation*}
    \begin{cases}
        F_{ab} \rightarrow^0 \mf{f}^{1,{(0)}}_{ab}\text{,}\qquad &n\geq 4\\
        F_{ab} \rightarrow^0 \mf{f}^{1,(0)}_{ab}\text{,}\; F_{\rho a}\rightarrow^0\mf{f}^{0,(0)}_a\text{,}\qquad &n=3\text{,}\\
        \rho\cdot F_{\rho a}\rightarrow^0 \mathfrak{f}_a^{0, (0)}\text{,}\qquad &n=2\text{.}
    \end{cases}
    \end{equation*}
\end{definition}

\begin{remark}
    Note that these asymptotics are consistent with \cite{Carranza:2018wkp}, at least for $n=3$. 
\end{remark}

As explained in Section \ref{sec:vertical}, any tensor field can be decomposed into its vertical components. In particular, we will chose the following definition: 
\begin{definition}\label{def_vert_f_w}
    Let $(\mc{M}, g, F)$ be a Maxwell-FG-aAdS segment, we define the vertical decomposition of $F$ and $W$ as: 
    \begin{gather*}
        \ms{f}^0_a := \rho F_{\rho a} \in \ms{V}^0{}_1 \mc{M}\text{,}\qquad \ms{f}^1_{ab} := \rho F_{ab} \in \ms{V}^0{}_2\mc{M}\text{,}\\
        \ms{w}^0_{abcd} := \rho^2 W_{abcd}\in \ms{V}^0{}_4\mc{M}\text{,}\qquad \ms{w}^1_{abc} := \rho^2 W_{\rho abc} \in \ms{V}^0{}_3\mc{M}\text{,}\qquad \ms{w}^2_{ab} := \rho^2 W_{\rho a\rho b}\in \ms{V}^0{}_2 \mc{M}\text{.}
    \end{gather*}
    Furthermore, for convenience, we will define the following vertical fields:
    \begin{gather*}
        \ol{\ms{f}}^0_a := \begin{cases}
            \rho^{-1}F_{\rho a}\text{,}\qquad &n \geq 4\text{,}\\
            F_{\rho a}\text{,}\qquad &n = 3\text{,}\\
            \rho F_{\rho a}\text{,}\qquad &n =2\text{,}
        \end{cases}\qquad \ol{\ms{f}}^1_{ab} := F_{ab}\text{.}
    \end{gather*}
\end{definition}

\begin{remark}
    It is important to notice that the vertical decomposition of $F$ into $\ms{f}^0$ and $\ms{f}^1$, being made along a spacelike direction, contains both electric and magnetic contributions. Therefore, the boundary conditions studied in \cite{Holzegel:2015swa}, based on the decomposition of the electromagnetic field into magnetic and electric parts, will mix components of both $\ms{f}^0$ and $\ms{f}^1$. 
\end{remark}

\begin{remark}
    In the next section, we will show that $\ol{\ms{f}}^0$ and $\ol{\ms{f}}^1$  satisfy some near-boundary expansions in powers of $\rho$. For $n\geq 3$, we will show, in particular,  that the coefficients located at order $0$ for $\ol{\ms{f}}^1$, written $\mf{f}^{1,(0)}$, and $(n-4)_+$ for $\ol{{\ms{f}}}^0$, written $\mf{f}^{0, (n-4)_+}$,  will not be fixed by the equations of motion \eqref{aads_einstein}. Since the coefficients $\mf{g}^{(0)}$ and $\mf{g}^{(n)}$, for the near-boundary expansion of $\ms{g}$, also correspond to the boundary data, we will describe the following quadruplet $\mc{D} = (\mf{g}^{(0)}, \mf{g}^{(n)}, \mf{f}^{0,(n-4)_+}, \mf{f}^{1, (0)})$ as the \emph{holographic data}\footnote{In fact, the boundary data will correspond to a subset of $\mc{D}$. Typically, one has to consider the divergence-- and trace--free parts of $\mf{g}^{(n)}$, the divergence--free part for $\mf{f}^{0,(n-4)_+}$ and the curl--free part of $\mf{f}^{1,(0)}$.}. 
    
    Heuristically, the coefficient $\mf{f}^{1,(0)}$ will be generated by the magnetic charge, while $\mf{f}^{0, (n-4)_+}$ will be generated by the electric charge. 

\end{remark}

\begin{remark}
    For the proof of Theorem \ref{theorem_main_fg}, describing the asymptotic limits of both $\ms{g}$ and $\ms{f}^0$, $\ms{f}^1$, it will be convenient to work with $\ol{\ms{f}}^0$ and $\ol{\ms{f}}^1$ instead of $\ms{f}^0$ and $\ms{f}^1$. The reasons for this are the following: 
    \begin{itemize}
        \item Regarding $\ms{f}^0$, the first nontrivial coefficient in the near-boundary expansion appears at order $\min(0, n-2)$. In particular, for $n\geq 4$, $\ms{f}^0$ has always zero boundary limit. Furthermore, for $n=2$, its first nontrivial coefficient in the expansion appears at order $-1$. Working with $\bar{\ms{f}}^0$ allows for a more systematic analysis. 

        \item For any $n>2$, the boundary limit of $\bar{\ms{f}}^1$ will correspond to its free boundary data. For $n=2$, one will have to take into account a logarithmic contribution.
    \end{itemize}
\end{remark}
\section{Near-boundary asymptotics}\label{sec:FG_expansion}

\subsection{Results}

In this section, we will describe the behaviours of the different fields used in this work near $\mc{I}$. In particular, we will write the fully coupled Fefferman-Graham expansion of $(g, F)$, assuming enough regularity in the vertical directions. 

First, we will state the results, and will discuss them. The proofs will be found in Section \ref{sec:proofs}. Let us first introduce the following notation : 
\begin{definition}
    Let $(\mc{M}, g)$ be a FG-aAdS segment and let $A = \lbrace \mf{A}_1,\dots, \mf{A}_m\rbrace$ be a collection of tensors on $\mc{I}$, and $M_1, \dots, M_m\geq 0$ given. We will write $ \mf{A}=\mi{D}^{(M_1, \dots, M_m)}( \mf{A}_1,\dots, \mf{A}_m)$ if: 
    \begin{equation*}
        \mf{A} = \sum\limits_{j=1}^N \mc{S}\paren{\mf{B}_{1}\text{,}\dots, \mf{B}_{N_j}}\text{,} \qquad N\geq 0\text{,}
    \end{equation*}
    where, for any $0\leq j \leq N$, one has $N_j\geq 1$, and:
    \begin{equation*}
        \mf{B}_j \in \bigcup_{i\in \lbrace 1, \dots m\rbrace} \lbrace \mf{A}_i,\partial \mf{A}_i, \dots, \partial^{M_i} \mf{A}_i\rbrace\text{.}
    \end{equation*}
    Furthermore, we will write, for $l\geq 0$:
    \begin{equation}\label{notation_D_even}
        \mi{D}^{(M_1, \dots, M_n)}\paren{\mf{A}_1,\dots, \mf{A}_m ; l} = \begin{cases}
            \mi{D}^{(M_1, \dots, M_m)}\paren{\mf{A}_1, \dots, \mf{A}_m}\text{,}\qquad &l\text{ even,}\\
            0\text{,}\qquad  &l\text{ odd.}
        \end{cases}
    \end{equation}
    
\end{definition}

\begin{example}
    A tensor field $\mf{A}$ on $\mc{I}$ satisfying: 
    \begin{equation*}
        \mf{A} = \mi{D}^{(k, j)}\paren{\mf{B}_1, \mf{B}_2}\text{,}
    \end{equation*}
    with $\mf{B}_1$, $\mf{B}_2$ tensor fields on $\mc{I}$,
    will depend on $\mf{B}_1$ and its first $k$--derivatives, as well as $\mf{B}_2$ and its first $j$--derivatives.
\end{example}

\begin{remark}
    This notation will be used to write the near-boundary expansion in $\rho$ of the different vertical fields. In particular, the notation \eqref{notation_D_even} will turn out to be convenient since the Fefferman-Graham expansion will typically involve an expansion in $\rho^2$, and some coefficients will be non-trivial for odd powers of $\rho$ only. 
\end{remark}

Let us define the type of regularity that we will need for the rest of this chapter:
\begin{definition}[$M_0$--regular segment]\label{def_regular_M}
    Let $(\mc{M}, g, F)$ be a Maxwell-FG-aAdS segment. We say $(\mc{M}, g, F)$ is regular to order $M_0\geq 0$ if, for any compact coordinate system $(U, \varphi)$ on $\mc{I}$: 
    \begin{itemize}
        \item The following uniform bounds are satisfied: 
        \begin{equation}\label{assumption_metric}
            \norm{\ms{g}}_{M_0+2, \varphi} \lesssim 1\text{,}\qquad \begin{cases}
                \norm{\ol{\ms{f}}^1}_{M_0+1, \varphi}\lesssim 1\text{,}\qquad &n \geq 4\text{,}\\
                \norm{\ol{\ms{f}}^0}_{M_0+1, \varphi}\lesssim 1\text{,}\qquad &n=2\text{.}
            \end{cases}
        \end{equation}
        \item The vertical field $\ms{m}$ and the Maxwell fields $\ms{f}^0, \, \ms{f}^1$ satisfy the following integrability conditions:
        \begin{equation}\label{assumption_m_f}
            \begin{aligned}
                &\sup\limits_U \int_0^{\rho_0} \abs{\ms{m}}_{0, \varphi}\vert_\sigma d\sigma<\infty\text{,}\qquad &&n\geq 2\text{,}\\
                &\sup\limits_U \int_0^{\rho_0}\paren{\abs{\ol{\ms{f}}^0}_{M_0+1, \varphi} + \abs{\ol{\ms{f}}^1}_{M_0+1, \varphi}}\vert_\sigma d\sigma <\infty\text{,} \qquad &&n=3\text{.}
            \end{aligned}
            \end{equation}
        \end{itemize}
\end{definition}

\begin{remark}
    Note how, in Assumption \eqref{assumption_metric}, the roles of $\ol{\ms{f}}^1$ and $\ol{\ms{f}}^0$ get exchanged from $n\geq 4$ to $n=2$, while $n=3$ will correspond to a borderline case, and it will be enough to impose weak $C^{M_0+1}$ boundedness on both fields. This is due to the fact that the near-boundary asymptotics of these fields will largely depend on the dimension and, for $n=3$, the free data will be located at the same powers. 
\end{remark}

We are now ready to state the main result of this chapter: 

\begin{theorem}\label{theorem_main_fg}
    Let $(\mc{M}, g, F)$ be a Maxwell-FG-aAdS segment, $n\geq 2$, fix $M_0\geq n +2$, and assume that $(\mc{M}, g, F)$ is regular to order $M_0$. Then, 
        \begin{itemize}
            \item The following bounds hold for $\ol{\ms{f}}^0$ and $\ol{\ms{f}}^1$: 
            \begin{gather}\label{thm_bounds_f0_f1}
                \begin{cases}
                    \norm{\ol{\ms{f}}^0}_{M_0, \varphi}\lesssim 1\text{,}\qquad &n> 4\text{,}\\
                    \norm{\rho\cdot \ol{\ms{f}}^0}_{M_0,\varphi}\lesssim 1\text{,}\qquad &n=4\text{,}\\
                    \norm{\ol{\ms{f}}^0}_{M_0, \varphi} + \norm{\ol{\ms{f}}^1}_{M_0, \varphi} \lesssim 1\text{,}\qquad &n=3\text{.}
                \end{cases}
            \end{gather}
            \item The following boundary limits are satisfied: 
            \begin{gather}
                \label{thm_first_limits_g}\ms{g}\Rightarrow^{M_0} \mathfrak{g}^{(0)}\text{,}\qquad \ms{g}^{-1}\Rightarrow^{M_0} (\mathfrak{g}^{(0)})^{-1} = \mi{D}^0(\mf{g}^{(0)})\text{,}\\
                \begin{cases}
                     \label{thm_first_limits_f}\ol{\ms{f}}^1\Rightarrow^{M_0-1} \mathfrak{f}^{1,(0)}\text{,} \qquad &n\geq 4\text{,}\\
                     \ol{\ms{f}}^0 \Rightarrow^{M_0-1} \mf{f}^{0,(0)}\text{,}\qquad \ol{\ms{f}}^1 \Rightarrow^{M_0-1} \mf{f}^{1,(0)}\text{,}\qquad &n =3\text{,}\\
                     \ol{\ms{f}}^0 \Rightarrow^{M_0-1}\mf{f}^{0,(0)}\text{,}\qquad &n =2\text{.}
                \end{cases} 
            \end{gather}
            \item For all $0\leq k<n\text{,} \; 0\leq q < n-4,\; 0\leq l < n-2$ there exist tensor fields $\mathfrak{g}^{(k)}\text{,}\; \mathfrak{f}^{0, (q)}\text{,} \; \mathfrak{f}^{1, (l)}\text{,}$ on $\mc{I}$ satisfying: 
            \begin{gather*}
                \mathfrak{g}^{(k)} = \begin{cases}
                    \mi{D}^{(k, k-4)}\paren{\mathfrak{g}^{(0)}, \mathfrak{f}^{1, (0)}; k}\text{,}\qquad &k\geq 4\\
                    \mi{D}^{k}\paren{\mathfrak{g}^{(0)}; k}\text{,}\qquad &k<4\text{,}
                \end{cases}\\
                \mathfrak{f}^{0,(q)} = \mi{D}^{(q+1, q+1)} \paren{\mathfrak{g}^{(0)},\mathfrak{f}^{1,(0)}; q}\text{,}\qquad \mathfrak{f}^{1,(l)}= \mi{D}^{(l, l)}\paren{\mathfrak{g}^{(0)}, \mathfrak{f}^{1,(0)}; l}\text{,}
            \end{gather*}
            such that the following limits hold: 
            \begin{gather}
                \label{thm_limits_g}\Lie_\rho^k \ms{g}\Rightarrow^{M_0-k} k!  \mathfrak{g}^{(k)}\text{,} \qquad \rho\Lie_\rho^{k+1} \ms{g}\Rightarrow^{M_0-k} 0\text{,}\\
                \label{thm_limits_f0}\Lie_\rho^q \ol{\ms{f}}^0\Rightarrow^{M_0-q-2} q! \, \mathfrak{f}^{0, (q)}\text{,}\qquad \rho\Lie_\rho^{q+1} \ol{\ms{f}}^0 \Rightarrow^{M_0-q-2} 0 \text{,}\\
                \label{thm_limits_f1}\Lie_\rho^{l}\ol{\ms{f}}^1 \Rightarrow^{M_0-l-1} l!\, \ol{\mathfrak{f}}^{1, (l)}\text{,}\qquad \rho \Lie_\rho^{l+1} \ol{\ms{f}}^1 \Rightarrow^{M_0-l-1} 0\text{.}  
            \end{gather}
            \item Furthermore, there exist tensor fields on $\mc{I}$ of the form:
            \begin{gather*}
                \begin{cases}
                    \mathfrak{g}^{\star} = \mi{D}^{(n, n-4)}\paren{\mathfrak{g}^{(0)}, \mathfrak{f}^{1, (0)}; n }\text{,}\\
                    \mathfrak{f}^{0, \star} = \mi{D}^{(n-3, n-3)}\paren{\mathfrak{g}^{(0)}, \mathfrak{f}^{1, (0)}; n}\text{,}\qquad \mf{f}^{1, \star} = \mi{D}^{(n-2, n-2)}(\mf{g}^{(0)}, \mf{f}^{1, (0)}; n)\text{,}\qquad  &n \geq3 \text{,}\\
                    \mathfrak{g}^\star = \mi{D}^{(2,0)}\paren{\mf{g}^{(0)}, \mf{f}^{0, (0)}}\text{,}\qquad 
                    \mathfrak{f}^{1, \star} = \mi{D}^{(1,1)}\paren{\mf{g}^{(0)}, \mf{f}^{0, (0)}}\text{,}\qquad &n=2\text{,}
                \end{cases}
            \end{gather*}
            such that the following anomalous limits hold: 
            \begin{gather}\label{thm_anomalous_g}
            \rho \Lie_\rho^{n+1} \ms{g} \Rightarrow^{M_0-n} n!\,  \mathfrak{g}^\star\text{,}\\
                \begin{cases}\label{anomalous_f_01}
                    \rho \Lie_\rho^{n-3}\ol{\ms{f}}^0 \Rightarrow^{M_0-(n-2)} (n-4)!\, \mathfrak{f}^{0, \star}\text{,}\quad \rho\Lie_\rho^{n-1}\ol{\ms{f}}^1\Rightarrow^{M_0-(n-1)} (n-2)!\,\mf{f}^{1, \star} \qquad &n>3\text{,}\\
                    \qquad\rho \Lie_\rho \ol{\ms{f}}^1 \Rightarrow^{M_0-2} \mf{f}^{1,\star}\text{,}\qquad &n=2\text{.}
                \end{cases}
            \end{gather}
            \item There exist $\mathfrak{g}^{\dag}\text{,}\; \mathfrak{f}^{0, \dag}\text{,}\; \mathfrak{f}^{1, \dag}$ tensor fields on $\mc{I}$ such that: 
            \begin{gather}
                \label{thm_critical_g}\Lie_\rho^n \ms{g} - n! \, \log\rho\cdot  \mathfrak{g}^\star \rightarrow^{M_0-n}n! \, \mathfrak{g}^{\dag}\text{,}\\
                    \label{thm_critical_f0}\Lie_\rho^{n-4}\ol{\ms{f}}^0 - (n-4)! \log\rho\cdot \mathfrak{f}^{0, \star}\rightarrow^{M_0-(n-2)} (n-4)!\,\mathfrak{f}^{0, \dag}\text{,}\qquad n\geq 4\text{,}\\
                \begin{cases}
                    \label{thm_critical_f1}\Lie^{n-2}_\rho \ol{\ms{f}}^1 - (n-2)!\rho^{n-2}\log\rho\cdot  \mf{f}^{1, \star}\rightarrow^{M_0-(n-1)} \mi{D}^{(n-2, n-2, 1)}\paren{\mf{g}^{(0)}, \mf{f}^{1, (0)}, \mf{f}^{0,\dag}}\text{,}&n \geq 4\text{ even,}\\
                    \Lie^{n-2}_\rho \ol{\ms{f}}^1\rightarrow^{M_0-(n-1)}\mi{D}^{(1,1)}\paren{\mf{g}^{(0)}, \mf{f}^{0,\dag}}\text{,}\qquad &n\geq 5\text{ odd,}\\
                    \ol{\ms{f}}^1 - \log \rho \cdot \mf{f}^{1, \star} \rightarrow^{M_0-2} \mathfrak{f}^{1, \dag}\text{,}\qquad &n = 2\text{.}
                \end{cases}
            \end{gather}
        \end{itemize}
\end{theorem}

\begin{remark}
    There are a couple of observations to draw from the above result.
    \begin{enumerate}
        \item First, observe that the limits $\mf{g}^{(k)}$ do not depend on $\mf{f}^{1, (0)}$ nor $\mf{f}^{0, (0)}$ for $k<4$. In particular, $\mf{g}^{(2)}$ depends only on $\mf{g}^{(0)}$ to order 2. 
        \item The free data for $F$ correspond roughly to the fields $(\mf{f}^{0, \dag}$, $\mf{f}^{1, (0)})$ for $n\geq 4$, to $(\mf{f}^{0,(0)}, \mf{f}^{1,(0)})$ for $n=3$ and to $(\mf{f}^{0, (0)}$, $\mf{f}^{1, \dag})$ for $n=2$. These coefficients are not fixed by the equations of motion. We will see, however, in Proposition \ref{prop_constraints}, that only parts of them are unconstrained. The next Corollary clarifies this by showing the explicit near-boundary expansion of the vertical fields. 
    \end{enumerate}
\end{remark}

\begin{corollary}\label{prop_expansion_vertical_fields}
    Let $(\mc{M}, g, F)$ be a Maxwell-FG-aAdS segment, $n:= \text{dim }\mc{I}\geq 2$ and fix $M_0\geq n+2$. Assume that $(\mc{M}, g, F)$ is regular to order $M_0$. Then $\ms{g}\text{,}\,  \ms{f}^0\text{ and } \ms{f}^1$ satisfy the following partial expansions in $\rho$: 
    \begin{gather}
        \label{expansion_g}\ms{g} = \begin{cases}
            \mf{g}^{(0)} + \rho^2 \mf{g}^{(2)} + \dots + \rho^{n-1}\mf{g}^{(n-1)} + \rho^n \mf{g}^{(n)} + \rho^n \ms{r}_{\ms{g}}\text{,}\qquad &n\text{ odd,}\\
            \mf{g}^{(0)} + \rho^2 \mf{g}^{(2)} + \dots + \rho^{n}\log\rho \cdot \mf{g}^{\star} + \rho^n \mf{g}^{(n)} + \rho^n \ms{r}_{\ms{g}}\text{,}\qquad &n\text{ even,}
        \end{cases}\\
        \label{expansion_f0}\ms{f}^0=\begin{cases}
            \rho^2\mf{f}^{0, (0)} + \rho^4 \mf{f}^{0, (2)} + \dots + \rho^{n-3}\mf{f}^{0, ((n-4)_+-1)} + \rho^{n-2} \mf{f}^{0, ((n-4)_+)} + \rho^{n-2} \ms{r}_{\ms{f}^0}\text{,}\qquad &n>4 \text{ odd}\text{,}\\
            \rho^2 \mf{f}^{0, (0)} + \rho^4 \mf{f}^{0, (2)} + \dots + \rho^{n-2}\log\rho \cdot \mf{f}^{0, \star} + \rho^{n-2} \mf{f}^{0, ((n-4)_+)} + \rho^{n-2} \ms{r}_{\ms{f}^0}\text{,}\qquad &n\geq 4\text{ even}\text{,}\\
            \rho\mf{f}^{0,(0)} + \rho\ms{r}_{\ms{f}^0}\text{,}\qquad &n=3\text{,}\\
            \mf{f}^{0,(0)} + \ms{r}_{\ms{f}^0}\text{,}\qquad &n=2\text{,}
        \end{cases}\\
        \label{expansion_f1}\ms{f}^1 = \begin{cases}
            \rho \mf{f}^{1, (0)} + \rho^{3} \mf{f}^{1, (2)} + \dots + \rho^{n-2} \mf{f}^{1, (n-3)}+\rho^{n-1}\mf{f}^{1, (n-2)}+\rho^{n-1}\ms{r}_{\ms{f}^1}\text{,}\qquad &n \text{ odd}\text{,}\\
            \rho \mf{f}^{1, (0)} + \rho^{3} \mf{f}^{1, (2)} + \dots + \rho^{n-1} \log\rho\cdot\mf{f}^{1, \star}+\rho^{n-1}\mf{f}^{1, (n-2)}+\rho^{n-1} \ms{r}_{\ms{f}^1}\text{,}\qquad &n\text{ even}\text{,}
        \end{cases}
    \end{gather}
    where the remainders $\ms{r}_{\ms{g}}\text{,}\, \ms{r}_{\ms{R}}\text{,}\, \ms{r}_{\ms{f}^0}\text{,}\, \ms{r}_{\ms{f}^1}$ are vertical tensor fields satisfying
    \begin{gather}
        \ms{r}_{\ms{g}} \rightarrow^{M_0-n} 0\text{,}\\
        \ms{r}_{\ms{f}^0}\begin{cases}
            \rightarrow^{M_0-(n-2)} 0\text{,} \qquad &n\geq 4\text{,}\\
            \rightarrow^{M_0-1} 0\text{,}\qquad &n = 2, 3
        \end{cases}\text{,}\qquad \ms{r}_{\ms{f}^1}\begin{cases}
            \rightarrow^{M_0-(n-1)} 0\text{,}\qquad &n\geq 4\text{,}\\
            \rightarrow^{M_0-1} 0\text{,}\qquad &n =2,3\text{.}
        \end{cases}
    \end{gather}.
    
    The coefficients $(\mf{g}^{(n)},\mf{f}^{0, ((n-4)_+)}, \mf{f}^{1,(0)})$ satisfy:
    \begin{gather*}
        \mf{g}^{(n)} = \mf{g}^\dag-C_n\mf{g}^\star\text{,}\\  
            \mf{f}^{0, (n-4)_+} =\begin{cases} \mf{f}^{0, \dag} - C_{n-4}\mf{f}^{0, \star}\text{,}\qquad &n >4 \text{,}\\
            \mf{f}^{0, \dag} \text{,}\qquad &n =4\text{,}
        \end{cases}\qquad \mf{f}^{1,(n-2)}=\begin{cases}
            \mf{f}^{1,\dag} - C_{n-2}\mf{f}^{1,\star}\text{,}\qquad &n>2\text{,}\\
            \mf{f}^{1,\dag}\text{,}\qquad &n=2\text{,}
        \end{cases}
    \end{gather*}
    with $C_k := \sum\limits_{m=1}^k m^{-1}$, for all $k\geq 1$. 
\end{corollary}

\begin{proof}
    See Appendix \ref{app:fg_expansion}. 
\end{proof}

\begin{proposition}\label{prop_constraints}
     Let $(\mc{M}, g , F)$ be a Maxwell-FG-aAdS segment, $n\geq 2$ and $(U, \varphi)$ a compact coordinate chart on $\mc{I}$. The following identities are satisfied for the coefficients of the  metric $\ms{g}$: 
    \begin{gather}
        \mathfrak{D}\cdot \mathfrak{g}^\star = \begin{cases}
        \mi{D}^{(n+1, n-3)}\paren{\mf{g}^{(0)}, \mf{f}^{1,(0)}; n}\text{,}\qquad &n\geq 3\\
        \mi{D}^{(3, 1)}\paren{\mf{g}^{(0)}, \mf{f}^{1,(0)}}\text{,}\qquad&n=2\text{.}
        \end{cases}\qquad 
        \mf{tr}_{\mathfrak{g}^{(0)}}\mathfrak{g}^\star = 0\text{,}\label{anomalous_tr}\\
        \mathfrak{D} \cdot \mathfrak{g}^{(n)} = \begin{cases}
            \mi{D}^{(n+1, n-3)}\paren{\mf{g}^{(0)}, \mf{f}^{1,(0)}; n}\text{,}\qquad &n\geq 3\text{,}\\
            \mi{D}^{(3, 1)}\paren{\mf{g}^{(0)}, \mf{f}^{0,(0)}}\text{,}\qquad &n=2\text{.}
        \end{cases}
        \label{critical_div}\text{,}\qquad 
        \mathfrak{tr}_{\mathfrak{g}^{(0)}}\mathfrak{g}^{(n)} = \begin{cases}
            \mi{D}^{(n, n-4 )}\paren{\mathfrak{g}^{(0)}, \mathfrak{f}^{1,(0)}; n}\text{,}\qquad &n\geq 3\text{,}\\
            \mi{D}^{(2,0)}\paren{\mf{g}^{(0)}, \mf{f}^{0,(0)}}\text{,}\qquad &n=2\text{,}
        \end{cases}
    \end{gather}
    as well as the following explicit expressions: 
    \begin{gather}\label{constraints_g}
        \begin{cases}
            -\mathfrak{g}^\star =\mf{f}^{0, (0)}\otimes \mf{f}^{0,(0)} - \frac{1}{2}(\mathfrak{f}^{0, (0)})^2_{\mathfrak{g}^{(0)}}\cdot \mathfrak{g}^{(0)}\text{,}\qquad &n =2\text{,}\\
            -\mathfrak{g}^{(2)} = \frac{1}{n-2}\paren{\mathfrak{Rc}^{(0)}- \frac{1}{2(n-1)}\mathfrak{Rs}^{(0)}\cdot \mathfrak{g}^{(0)}}\text{,}\qquad &n\geq 3\text{.}
        \end{cases}
    \end{gather}
    Similarly, the following identities are satisfied for the Maxwell field: 
    \begin{gather}\label{constraint_f}
            \begin{cases}
                \mathfrak{D}\cdot\mathfrak{f}^{0, (n-4)_+} = \mi{D}^{(n-2, n-2)}\paren{\mathfrak{g}^{(0)}, \mathfrak{f}^{1, (0)}; n}\text{,}\qquad \mf{D}\cdot \mf{f}^{0,\star} =0\text{,}\qquad 
                \mathfrak{D}_{[a}\mathfrak{f}^{1, (0)}_{bc]} = 0\text{,}\qquad &n\geq 3\text{,}\\
                \mathfrak{D}\cdot \mathfrak{f}^{0, (0)} = 0\text{,}\qquad \mf{D}_{[a}\mf{f}^{1,(0)}_{bc]} = 0\text{,}\qquad &n=2\text{.}
            \end{cases}
    \end{gather}
\end{proposition}
\begin{remark}
    \begin{itemize}
        \item One can interpret such relations as the fact that the boundary data for $(g, F)$ have to satisfy constraints. Namely, the trace and divergence of $\mf{g}^{(n)}$, the divergence of $\mf{f}^{0, ((n-4)_+)}$, as well as the $\mf{g}^{(0)}$--exterior derivative of $\mf{f}^{1, (0)}$, seen as a two-form on $\mc{I}$, are all fixed by the equations of motion. 
        \item Note that the divergence of the anomalous term $\mf{g}^\star$ does not generically vanish, as opposed to the vacuum case. A direct way to see it is given by the $n=2$ case where, in this case $\mf{D}\cdot \mf{g}^{\star}=0$ if and only if $\mf{f}^{1, \star}_{ab}=2\mf{D}_{[a}\mf{f}^{0,(0)}_{b]}=0$. More generally, this divergence anomaly is simply generated by the Maxwell fields, since we know that it has to vanish in the vacuum case. 
    \end{itemize}
\end{remark}

\begin{corollary}[Expansion of the Weyl tensor]\label{weyl_expansion}
    Let $(\mc{M}, g, F)$ be a Maxwell-FG-aAdS segment, $n=\operatorname{dim}\mc{I}\geq 3$, and assume that $(\mc{M}, g, F)$ is regular to order $M_0$. Then, the vertical decompositions of the Weyl tensor satisfy the following partial expansions in $\rho$: 
    \begin{align}
        &\label{expansion_w_0}\ms{w}^0_{abcd} = \begin{cases}
            \sum\limits_{k=1}^{\frac{n-1}{2}} \rho^{2k-2} \mf{w}^{0, (2k)}_{abcd} + \rho^{n-2}\mf{w}^{0,(n-2)}_{abcd}+ \rho^{n-2}\ms{r}_{\ms{w}^0}\text{,}\qquad n \text{ odd,}\\
            \sum\limits_{k=1}^{\frac{n-2}{2}}\rho^{2k-2}\mf{w}^{0, (2k)}_{abcd} + \rho^{n-2}\log\rho\cdot \mf{w}^{0,\star}_{abcd} + \rho^{n-2}\mf{w}^{0, (n-2)}_{abcd} + \rho^{n-2}\ms{r}_{\ms{w}^0}\text{,} \qquad n \text{ even,}
        \end{cases}\\
        &\label{expansion_w_1} \ms{w}^1_{ abc} =\begin{cases}
            \sum\limits_{k=1}^{\frac{n-1}{2}}\rho^{2k-1}\mf{w}^{1, (2k)}_{abc} + \rho^{n-1}\mf{w}^{1, (n-1)}_{abc} + \rho^{n-1}\ms{r}_{\ms{w}^1}\text{,}\qquad n \text{ odd,}\\
            \sum\limits_{k=1}^{\frac{n-2}{2}}\rho^{2k-1}\mf{w}^{1, (2k)}_{abc} + \rho^{n-1}\log\rho \cdot \mf{w}^{1, (n-1)}_{abc} + \rho^{n-1}\mf{w}^{1, (n-1)}_{abc} + \rho^{n-1}\ms{r}_{\ms{w}^1}\text{,} \qquad n\text{ even,}
        \end{cases}\\
        \label{expansion_w_2}& \ms{w}^2_{a b} = \begin{cases}
            \sum\limits_{k=2}^{\frac{n-1}{2}}\rho^{2k-2}\mf{w}^{2, (2k)}_{ab} + \rho^{n-2}\mf{w}^{2, (n-2)}_{ab} + \rho^{n-2}\ms{r}_{\ms{w}^2}\text{,}\qquad n\text{ odd,}\\
            \sum\limits_{k=2}^{\frac{n-2}{2}} \rho^{2k-2}\mf{w}^{2, (2k)}_{ab} + \rho^{n-2}\log\rho \cdot \mf{w}^{2, \star} + \rho^{n-2}\mf{w}^{2, (n-2)} + \rho^{n-2}\ms{r}_{\ms{w}^2}\text{,}\qquad n\text{ even,}
        \end{cases}
    \end{align}
    where: 
    \begin{itemize}
        \item The coefficients $\mf{w}^{0, (2k)}$, $\mf{w}^{1, (2k)}$ and $\mf{w}^{2, (2k)}$ are such that: 
        \begin{gather*}
            \mf{w}^{0,(2k)}=\mi{D}^{(2k, 2k-4)}\paren{\mf{g}^{(0)}, \mf{f}^{1,(0)}}\text{,}\\ \mf{w}^{1, (2k)} = \mi{D}^{(2k+1, 2k-3)}\paren{\mf{g}^{(0)}, \mf{f}^{1,(0)}}\text{,}\qquad \mf{w}^{2, (2k)} = \mi{D}^{(2k, 2k-4)}\paren{\mf{g}^{(0)}, \mf{f}^{1,(0)}}\text{.}
        \end{gather*}
        \item The anomalous coefficients $\mf{w}^{0, \star}, \mf{w}^{1, \star}$ and $\mf{w}^{2, \star}$ are such that:
        \begin{gather*}
            \mf{w}^{0, \star} = \mi{D}^{(n,n-4)}\paren{\mf{g}^{(0)}, \mf{f}^{1, (0)};n}\text{,} \\
            \mf{w}^{1, \star} = \mi{D}^{(n+1, n-3)}\paren{\mf{g}^{(0)}, \mf{f}^{1, (0)};n}\text{,}\qquad \mf{w}^{2, \star} = \mi{D}^{(n, n-4)}\paren{\mf{g}^{(0)}, \mf{f}^{1, (0)};n}\text{.}
        \end{gather*} 
        \item The coefficients $\mf{w}^{0, (n-2)}$, $\mf{w}^{1, (n-1)}$ and $\mf{w}^{2, (n-2)}$ are $C^{M_0-n}$, $C^{M_0-n-1}$ and $C^{M_0-n}$ tensor fields on $\mc{I}$.
        \item The remainder terms $\ms{r}_{\ms{w}^0}$, $\ms{r}_{\ms{w}^1}$ and $\ms{r}_{\ms{w}^2}$ satisfy: 
        \begin{gather*}
            \ms{r}_{\ms{w}^0}\rightarrow^{M_0-n}0\text{,}\qquad \ms{r}_{\ms{w}^1}\rightarrow^{M_0-n-1} 0\text{,}\qquad \ms{r}_{\ms{w}^2}\rightarrow^{M_0-n} 0\text{.}
        \end{gather*}
    \end{itemize}
\end{corollary}

\begin{proof}
    See Appendix \ref{app_weyl}
\end{proof}

\begin{remark}
    These results show that, in the context of the AdS-Einstein-Maxwell system, the correct holographic data to consider is given by $(\mf{g}^{(0)}, \mf{g}^{(n)}, \mf{f}^{0, ((n-4)_+)}, \mf{f}^{1, (0)})$. 
\end{remark}

\subsection{Proofs}\label{sec:proofs}

\subsubsection{Preliminary results}

First, let us define a useful vertical decomposition for the stress-energy tensor: 
\begin{definition}\label{def_decomposition_T}
    Let $(\mc{M}, g, \Phi)$ be a non-vacuum FG-aAdS segment and $(U, \varphi)$ a compact coordinate system on $\mc{I}$. We define the following vertical decomposition for the stress-energy tensor associated to the matter content $\Phi$: 
    \begin{gather}
        \ms{T}^0_{ab} = {T}_{ab}\text{,}\qquad \ms{T}^1_a = T_{\rho a}\text{,}\qquad \ms{T}^2 = T_{\rho\rho}\text{,}\\
        \tilde{\ms{T}}^0_{ab} = \tilde{T}_{ab}\text{,}\qquad \tilde{\ms{T}}^2 = \tilde{T}_{\rho\rho}\text{.}
    \end{gather}
\end{definition}

\begin{proposition}\label{prop_w_FG}
    Let $(\mc{M}, g, \Phi)$ be a non-vacuum aAdS segment and let $(U, \varphi)$ be a compact coordinate system on $\mc{I}$. Then, the following hold: 
    \begin{align}
        &\label{w_0_FG_asymp}\ms{w}^0 = \ms{R} - \frac{1}{8}\ms{m}\star \ms{m}  + \frac{1}{2\rho}\ms{g}\star \ms{m} - \frac{1}{n-1}\tilde{\ms{T}}^0\star \ms{g}\text{,}\\
        &\label{w_1_FG_asymp}\ms{w}^1_{abc} = \ms{D}_{[c}\ms{m}_{b]a} - \frac{2}{n-1} T_{\rho [b}\ms{g}_{c]a}\text{,}\\
        &\label{w_2_FG_asymp}\ms{w}^2_{ab} = -\frac{1}{2}\Lie_\rho \ms{m}_{ab} + \frac{1}{2\rho} \ms{m}_{ab}+\frac{1}{4} \ms{g}^{cd}\ms{m}_{ac}\ms{m}_{bd} - \frac{1}{n-1}\paren{\tilde{\ms{T}}^2\ms{g}_{ab} + \tilde{\ms{T}}^0_{ab}}\text{.}
    \end{align}
\end{proposition}

\begin{proof}
    First of all, the Gauss equation on level sets of $\rho$ gives the following identity: 
    \begin{equation}\label{Gauss_eq}
        \rho^2\operatorname{R}[g]_{abcd} = \ms{R}_{abcd} - \frac{1}{8}(\ms{m}\star\ms{m})_{abcd} + \frac{1}{2\rho}(\ms{g}\star\ms{m})_{abcd} - \frac{1}{2\rho^2}(\ms{g}\star\ms{g})_{abcd}\text{.}
    \end{equation}

    Next, the Codazzi equations give: 
    \begin{equation*}
        \rho\operatorname{R}[g]_{\rho a bc} = \rho^{-1}\ms{D}_{[c}\ms{m}_{b]a}\text{.}
    \end{equation*}

    Finally, the Ricci equation gives: 
    \begin{equation*}
        \rho^2 R_{\rho a\rho b} + \rho^{-2}\ms{g}_{ab} = -\frac{1}{2}\Lie_\rho\ms{m}_{ab} + \frac{1}{2\rho}\ms{m}_{ab} + \frac{1}{4}\ms{g}^{cd}\ms{m}_{ac}\ms{m}_{bd}\text{.}
    \end{equation*}

    The equations \eqref{w_0_FG}, \eqref{w_1_FG} and \eqref{w_2_FG} can therefore simply be obtained by decomposing \eqref{sec:aads_weyl_equation}. Starting with the purely vertical component, one finds: 
    \begin{align*}
        \rho^2 W_{abcd} = \rho^2\operatorname{R}[g]_{abcd} + \frac{1}{2\rho^2}(\ms{g}\star\ms{g})_{abcd} -\frac{1}{n-1}(\ms{g}\star\tilde{\ms{T}}^0)_{abcd}\text{.}
    \end{align*}
    Combining with \eqref{Gauss_eq}, the singular terms $\rho^{-2}\ms{g}\star\ms{g}$ cancel each other, and one is left with \eqref{w_0_FG_asymp}.

    The identities \eqref{w_1_FG_asymp} and \eqref{w_2_FG_asymp} are obtained similarly. 
\end{proof}

We are now ready to establish the transport equations satisfied by $\ms{m}$ and $\ms{tr}_{\ms{g}}\ms{m}$. 

\begin{proposition}[Transport equation for the metric]\label{prop_transport_m}
    Let $(\mc{M}, g, \Phi)$ be a non-vacuum FG-aAdS segment and let $(U, \varphi)$ be a compact coordinate system on $\mc{I}$. The following transport equations hold for $\ms{m}$: 
    \begin{align}
        &\begin{aligned}
        \label{transport_m}\rho\Lie_\rho \ms{m}_{ab} -(n-1)\ms{m}_{ab} =&\, 2\rho \ms{Rc}_{ab}+\ms{tr}_{\ms{g}}\ms{m}\cdot \ms{g}_{ab} +\rho\cdot \ms{g}^{cd}\ms{m}_{ac}\ms{m}_{bd} - \frac{1}{2}\rho\ms{tr}_{\ms{g}}\ms{m}\cdot \ms{m}_{ab} \\
        &-2\rho\ms{T}^0_{ab}+\frac{2\rho}{n-1}\ms{tr}_{\ms{g}}\ms{T}^0\cdot \ms{g}_{ab} +\frac{2\rho}{n-1}\ms{T}^2\cdot \ms{g}_{ab}\text{.}
        \end{aligned}\\
        &\begin{aligned}
            \label{transport_tr_m}\rho\Lie_\rho \ms{tr}_{\ms{g}}\ms{m} - (2n-1)\ms{tr}_{\ms{g}}\ms{m} = 2\rho\ms{Rs}  -\frac{1}{2}\rho(\ms{tr}_{\ms{g}}\ms{m})^2 + \frac{2}{n-1}\rho\ms{tr}_{\ms{g}}\ms{T}^0 + \frac{2n}{n-1}\rho\ms{T}^2\text{.}
        \end{aligned}
    \end{align}
    Furthermore, if $(\mc{M}, g, F)$ is a Maxwell-FG-aAdS segment, the decomposition of stress-energy tensor takes the following form: 
    \begin{align*}
        &\ms{T}^0_{ab} = \rho^2 \paren{\ms{g}^{bd}\ol{\ms{f}}^1_{ac}\ol{\ms{f}}^1_{bd}-\frac{1}{4}\ms{g}_{ab}(\ol{\ms{f}}^1)^2} + \paren{\ol{\ms{f}}^0_a\ol{\ms{f}}^0_b - \frac{1}{2}\ms{g}_{ab}(\ms{f}^0)^2}\cdot 
        \left\{\begin{aligned}
            &\rho^4 \text{,}\qquad &n\geq 4\text{,}\\
            &\rho^2\text{,}\qquad &n=3\text{,}\\
            &1\text{,}\qquad &n=2\text{,}
        \end{aligned}\right.\\
        &\ms{T}^1_a = \frac{1}{2}\ms{g}^{cd}\ol{\ms{f}}^0_c\ol{\ms{f}}^1_{ad}\cdot 
        \left\{\begin{aligned}
            &\rho^3\text{,}\qquad &&n\geq 4\text{,}\\
            &\rho^2\text{,}\qquad &&n=3\text{,}\\
            &\rho\text{,}\qquad &&n=2\text{,}
        \end{aligned}\right.\\
        &\ms{T}^2 = -\frac{1}{4}\rho^2(\ol{\ms{f}}^1)^2 + \frac{1}{2}(\ol{\ms{f}}^0)^2\cdot\left\{\begin{aligned}
            &\rho^4\text{,}\qquad &&n\geq 4\text{,}\\
            &\rho^2\text{,}\qquad &&n=3\text{,}\\
            &1\text{,}\qquad &&n=2\text{.}
        \end{aligned}\right.
    \end{align*}
\end{proposition}

\begin{proof}
    This is a consequence of the tracelessness of the Weyl curvature. Namely, \eqref{transport_m} is obtained from the following: 
    \begin{equation*}
        g^{\mu \nu} W_{\mu a \nu b} = 0\Rightarrow \ms{g}^{cd}\ms{w}^0_{cadb} = - \ms{w}^2_{ab}\text{.}
    \end{equation*}
    Equation \eqref{transport_tr_m} is obtained by taking the $\ms{g}$--trace of \eqref{transport_m}. In particular, one can use: 
    \begin{equation*}
        [\ms{tr}_{\ms{g}}, \Lie_\rho]\ms{m} = \mc{S}\paren{\ms{g}; (\ms{m})^2}\text{.}
    \end{equation*}
\end{proof}

The next proposition establishes the transport equations satisfied by the Maxwell fields:

\begin{proposition}[Transport equations for the Maxwell fields]\label{prop_transport_f}
    Let $(\mc{M}, g, F)$ be a Maxwell-FG-aAds segment and let $(U, \varphi)$ be a compact coordinate system on $\mc{I}$. The following transport equations hold for $\ol{\ms{f}}^0$ and $\ol{\ms{f}}^1$: 
    \begin{align}
        &\label{transport_f_0}
        \rho\Lie_\rho \ol{\ms{f}}^0_{a} - (n-4)_+\ol{\ms{f}}^0_{a} = \rho\cdot \ms{m}_a{}^b \ol{\ms{f}}^0_b - \frac{1}{2}\rho\ms{tr}_{\ms{g}}\ms{m}\cdot \ol{\ms{f}}^0_a -\ms{D}\cdot \ol{\ms{f}}^1_a\cdot \left\{\begin{aligned}
        &1\text{,}\qquad &&n\geq 4\text{,}\\
        &\rho \text{,}\qquad &&n=3\text{,}\\
        &\rho^2 \text{,}\qquad &&n=2\text{,}
        \end{aligned}\right.\\
        &\label{transport_f_1}\Lie_\rho \ol{\ms{f}}^1_{ab} = 2\ms{D}_{[a}\ol{\ms{f}}^0_{b]}\cdot \left\{\begin{aligned}
            &\rho\text{,}\qquad &&n\geq 4\text{,}\\
            &1\text{,}\qquad &&n=3\text{,}\\
            &\rho^{-1}\text{,}\qquad &&n=2\text{.}
        \end{aligned}\right.
    \end{align}
\end{proposition}

\begin{proof}
    These can be obtained from the following decomposition of the Maxwell and Bianchi equations:
    \begin{gather}\label{decomposition_Maxwell}
        \nabla^\mu F_{\mu a} = 0\text{,}\qquad \nabla_{[\rho}F_{ab]}=0\text{.}
    \end{gather}
    Next, one can use Proposition \ref{prop_christoffel_FG} to obtain indeed, for the first equation on the left: 
    \begin{equation*}
        \Lie_\rho(\rho^2\ms{F}^0)_a-(n-1)\rho\ms{F}^0_a = -\rho^2(\ms{D}\cdot \ms{F}^1)_a + \rho^2\cdot \ms{m}_a{}^b \ms{F}^0_b - \frac{1}{2}\rho^2\ms{tr}_{\ms{g}}\ms{m}\cdot \ms{F}^0_a\text{,}
    \end{equation*}
    where we defined, here for convenience: 
    \begin{equation*}
        \ms{F}^0_a := F_{\rho a}\text{,}\qquad \ms{F}^1_{ab} := F_{ab}\text{.}
    \end{equation*}
    The first equation in \eqref{decomposition_Maxwell} gives, combining it with Definition \ref{def_vert_f_w}, Equation \eqref{transport_f_0}. 

    The second equation in \eqref{decomposition_Maxwell} leads to: 
    \begin{equation*}
        \Lie_\rho \ms{F}^1_{ab} = 2\ms{D}_{[a}\ms{F}^0_{b]}\text{,}
    \end{equation*}
    which indeed gives \eqref{transport_f_1}, after using Definition \ref{def_vert_f_w}. 
\end{proof}

\begin{proposition}[Constraint equations]
    Let $(\mc{M}, g, \Phi)$ be a non-vacuum FG-aAdS segment and let $(U, \varphi)$ be a compact coordinate system on $\mc{I}$. Then, the following equation is satisfied for the metric on $\mc{U}$:
    \begin{equation}\label{eq_Dm}
        \ms{D}\cdot \ms{m}_a = \ms{D}_a \ms{tr}_{\ms{g}}\ms{m} + 4\ms{T}^1_a\text{.}
    \end{equation}
    Furthermore, if $(\mc{M}, g, F)$ is a Maxwell-FG-aAdS segment, the following constraint equations are satisfied for the Maxwell fields: 
    \begin{align}\label{eq_Df}
        \ms{D}_{[a}\ol{\ms{f}}^1_{bc]} = 0\text{,}\qquad \ms{D}\cdot \ol{\ms{f}}^0 = \frac{1}{2}\ms{g}^{ab}\ms{g}^{cd}\ms{m}_{ac}\ol{\ms{f}}^1_{bd}\cdot \left\{\begin{aligned}
            &\rho^{-1}\text{,}\qquad &&n\geq 4\text{,}\\
            &1\text{,}\qquad &&n=3\text{,}\\
            &\rho\text{,}\qquad &&n=2\text{.}
        \end{aligned}\right.
    \end{align}
\end{proposition}

\begin{proof}
    These constraints are consequences of the remaining components of the Bianchi and Maxwell equations, namely: 
    \begin{equation}
        g^{\mu\nu} W_{\mu \rho \nu a} = 0\text{,}\qquad \nabla^\mu F_{\mu \rho}=0\text{,}\qquad \nabla_{[a}F_{bc]} = 0\text{.}
    \end{equation}
    By using Definition \ref{def_vert_f_w}, Proposition \ref{prop_w_FG} and Proposition \ref{prop_christoffel_FG}, one obtains the desired equations. 
\end{proof}

\begin{proposition}[Proposition 2.28 of \cite{shao:aads_fg}]\label{prop_riemann}
    Let $(\mc{M}, g)$ be a FG-aAdS segment and let $(U, \varphi)$ be a compact coordinate system on $\mc{I}$. The following equation holds: 
    \begin{equation}
        \Lie_\rho \ms{R} = \mc{S}\paren{\ms{g}; \ms{D}^2\ms{m}}\text{,}
    \end{equation}
    and similarly for the Ricci tensor and scalar. 
\end{proposition}

\begin{proposition}[Higher-order transport for ${\ms{m}}$]\label{higher_order_m}
    Let $(\mc{M}, g, \Phi)$ be a non-vacuum FG-aAdS segment and $(U,\varphi)$ a compact coordinate system on $\mc{I}$. Then, for any $k\geq 0$:
    \begin{align}
        &\begin{aligned}\label{higher_transport_m}
            \rho\Lie_\rho^{k+1}\ms{m} - (n-1-k)\Lie_\rho^k \ms{m} = &2\rho\Lie_\rho^k \ms{Rc} + 2k\Lie_\rho^{k-1}\ms{Rc} + \ms{tr}_{\ms{g}}\Lie_\rho^k \ms{m} \cdot \ms{g} + \sum\limits_{\substack{j_1+\dots + j_\ell = k+1\\1\leq j_p<k+1}}\mc{S}\paren{\ms{g}; \Lie_\rho^{j_1-1}\ms{m}, \dots, \Lie_\rho^{j_\ell -1}\ms{m}} \\
            &+\sum\limits_{\substack{j_1+\dots+j_\ell=k+2\\1\leq j_p<k+2}}\rho\cdot \mc{S}\paren{\ms{g}; \Lie_\rho^{j_1-1}\ms{m}, \dots, \Lie_\rho^{j_\ell -1}\ms{m}}\\
            &+\sum\limits_{\substack{j_1+\dots + j_\ell + j = k}}\rho\cdot \mc{S}\paren{\ms{g}; \Lie_\rho^{j_1-1}\ms{m}, \dots, \Lie_\rho^{j_\ell - 1}\ms{m}, \Lie_\rho^j \ms{T}^0} \\
            &+\sum\limits_{\substack{j_1+\dots + j_\ell + j = k}}\rho\cdot \mc{S}\paren{\ms{g}; \Lie_\rho^{j_1-1}\ms{m}, \dots, \Lie_\rho^{j_\ell - 1}\ms{m}, \Lie_\rho^j \ms{T}^2} \\
            &+\sum\limits_{\substack{j_1+\dots + j_\ell + j = k-1}}\mc{S}\paren{\ms{g}; \Lie_\rho^{j_1-1}\ms{m}, \dots, \Lie_\rho^{j_\ell - 1}\ms{m}, \Lie_\rho^j \ms{T}^0} \\
            &+\sum\limits_{\substack{j_1+\dots + j_\ell + j = k-1}}\mc{S}\paren{\ms{g}; \Lie_\rho^{j_1-1}\ms{m}, \dots, \Lie_\rho^{j_\ell - 1}\ms{m}, \Lie_\rho^j \ms{T}^2}\text{,}
        \end{aligned}\\
        &\begin{aligned}\label{higher_transport_tr_m}
            \rho\Lie_\rho \ms{tr}_{\ms{g}}\paren{\Lie_\rho^k \ms{m}} - (2n-1-k)\ms{tr}_{\ms{g}}(\Lie_\rho^k \ms{m}) = &2\rho\Lie_\rho^k \ms{Rs} + 2k\Lie_\rho^{k-1}\ms{Rs} + \sum\limits_{\substack{j_1+\dots + j_\ell = k+1\\1\leq j_p<k+1}}\mc{S}\paren{\ms{g}; \Lie_\rho^{j_1-1}\ms{m}, \dots, \Lie_\rho^{j_\ell -1}\ms{m}} \\
            &+\sum\limits_{\substack{j_1+\dots+j_\ell=k+2\\1\leq j_p<k+2}}\rho\cdot \mc{S}\paren{\ms{g}; \Lie_\rho^{j_1-1}\ms{m}, \dots, \Lie_\rho^{j_\ell -1}\ms{m}}\\
            &+\sum\limits_{\substack{j_1+\dots + j_\ell + j = k}}\rho\cdot \mc{S}\paren{\ms{g}; \Lie_\rho^{j_1-1}\ms{m}, \dots, \Lie_\rho^{j_\ell - 1}\ms{m}, \Lie_\rho^j \ms{T}^0} \\
            &+\sum\limits_{\substack{j_1+\dots + j_\ell + j = k}}\rho\cdot \mc{S}\paren{\ms{g}; \Lie_\rho^{j_1-1}\ms{m}, \dots, \Lie_\rho^{j_\ell - 1}\ms{m}, \Lie_\rho^j \ms{T}^2} \\
            &+\sum\limits_{\substack{j_1+\dots + j_\ell + j = k-1}}\mc{S}\paren{\ms{g}; \Lie_\rho^{j_1-1}\ms{m}, \dots, \Lie_\rho^{j_\ell - 1}\ms{m}, \Lie_\rho^j \ms{T}^0} \\
            &+\sum\limits_{\substack{j_1+\dots + j_\ell + j = k-1}}\mc{S}\paren{\ms{g}; \Lie_\rho^{j_1-1}\ms{m}, \dots, \Lie_\rho^{j_\ell - 1}\ms{m}, \Lie_\rho^j \ms{T}^2}\text{.}
        \end{aligned}
    \end{align}
    Furthermore, if $(\mc{M}, g, F)$ is a Maxwell-FG-aAdS of dimension $n+1$, with $n\geq 4$, one can write, for all $k\geq 0$: 
    \begin{align}
        \Lie^k_\rho \ms{T}^0, \Lie_\rho^k \ms{T}^2 = &\notag\sum\limits_{q=0}^4\sum\limits_{j+j'+j_0+\dots + j_\ell = k-q} \rho^{4-q}\cdot \mc{S}\paren{\ms{g}; \Lie_\rho^{j_0-1}\ms{m}, \dots, \Lie_\rho^{j_\ell-1} \ms{m}, \Lie_\rho^j \ol{\ms{f}}^0, \Lie_\rho^{j'} \ol{\ms{f}}^0}\\
        &+\sum\limits_{q=0}^2\sum\limits_{j+j'+j_0+\dots + j_\ell = k-q} \rho^{2-q}\cdot \mc{S}\paren{\ms{g}; \Lie_\rho^{j_0-1}\ms{m}, \dots, \Lie_\rho^{j_\ell-1} \ms{m}, \Lie_\rho^j \ol{\ms{f}}^1, \Lie_\rho^{j'} \ol{\ms{f}}^1}\text{.}\label{T_higher_order}
    \end{align}
\end{proposition}

\begin{proof}
    See Appendix \ref{app:higher_transport_m}.
\end{proof}

\begin{proposition}[Higher-order transport for $\ol{\ms{f}}$]\label{higher_order_f} Let $n\geq 4$ and let $(\mc{M}, g, F)$ be a Maxwell FG-aAdS segment and $(U,\varphi)$ a compact coordinate system on $\mc{I}$. Then, for any $k\geq 0$, the following transport equations hold: 
    \begin{align}
        &\begin{aligned}\label{higher_transport_f_0}
            \rho\Lie_\rho^{k+1}\ol{\ms{f}}^0 - (n-4-k)\Lie_\rho^k \ol{\ms{f}}^0 = -\ms{D}\cdot \Lie_\rho^k \ol{\ms{f}}^1 &+  \sum\limits_{j+j_0+j_1+\dots + j_\ell = k} \rho\cdot\mc{S}\paren{\ms{g}; \Lie_\rho^{j_0} \ms{m}, \Lie_\rho^{j_1 -1}\ms{m}, \dots, \Lie_\rho^{j_\ell - 1}\ms{m}, \Lie_\rho^j \ol{\ms{f}}^0}\\
            &+\sum\limits_{j+j_0+j_1+\dots + j_\ell = k-1} \mc{S}\paren{\ms{g}; \Lie_\rho^{j_0} \ms{m}, \Lie_\rho^{j_1 -1}\ms{m}, \dots, \Lie_\rho^{j_\ell - 1}\ms{m}, \Lie_\rho^j \ol{\ms{f}}^0}\\
            &+\sum\limits_{\substack{j+j_0+j_1+\dots +j_\ell = k\\j<k, j_p\geq 1}} \mc{S}\paren{\ms{g}; \ms{D}\Lie_\rho^{j_0-1}\ms{m}, \Lie_\rho^{j_1-1} \ms{m}, \dots, \Lie_\rho^{j_\ell -1}\ms{m}, \Lie_\rho^{j}\ol{\ms{f}}^1}\text{,}
        \end{aligned}\\
        &\begin{aligned}\label{higher_transport_f_1}
            \Lie_\rho^{k+1}\ol{\ms{f}}^1_{ab} =  2\rho\ms{D}_{[a}\Lie_\rho^k \ol{\ms{f}}^0_{b]} + 2k \ms{D}_{[a}\Lie_\rho^{k-1}\ol{\ms{f}}^0_{b]} &+  \sum\limits_{\substack{j + j_0 +j_1+ \dots + j_\ell = k\\j<k, j_p\geq 1}} \rho\cdot\mc{S}\paren{\ms{g}; \ms{D}\Lie_\rho^{j_0-1}\ms{m}, \Lie_\rho^{j_1-1}\ms{m}, \dots, \Lie_\rho^{j_\ell-1}\ms{m}, \Lie_\rho^j \ol{\ms{f}}^0}_{ab}\\
            &+\sum\limits_{\substack{j + j_0 +j_1+ \dots + j_\ell = k-1\\j<k-1, j_p \geq 1}} \mc{S}\paren{\ms{g}; \ms{D}\Lie_\rho^{j_0-1}\ms{m}, \Lie_\rho^{j_1-1}\ms{m}, \dots, \Lie_\rho^{j_\ell-1}\ms{m}, \Lie_\rho^j \ol{\ms{f}}^0}_{ab}\text{.}
         \end{aligned}
    \end{align}
\end{proposition}

\begin{proof}
See Appendix \ref{app:higher_transport_f}.
\end{proof}

\begin{proposition}[Higher-order transport for the Riemann tensor]\label{higher_order_riemann}
    Let $(\mc{M}, g)$ be a FG-aAdS segment, $(U, \varphi)$ a compact coordinate system and $k\geq 0$, then: 
    \begin{equation}\label{transport_higher_riemann}
        \Lie_\rho^k \ms{R} = \sum\limits_{\substack{ j_1 + \dots + j_\ell = k\\i_1 + \dots + i_\ell=2, j_p \geq 1}}\mc{S}\paren{\ms{g}; \ms{D}^{i_1}\Lie_\rho^{j_1-1}{\ms{m}}\text{,}\dots, \ms{D}^{i_\ell}\Lie_\rho^{j_\ell-1}{\ms{m}}}\text{,}
    \end{equation}
    and similarly for the Ricci tensor and scalar. 
\end{proposition}

\begin{proof}
See \cite{shao:aads_fg}, Proposition 2.34.
\end{proof}

Finally, we will need the higher-order equations for the constraints: 

\begin{proposition}[Higher-order constraints]\label{prop_higher_constraints}
    Let $(\mc{M}, g, F)$ be a Maxwell-FG-aAdS segment. Let $(U,\varphi)$ be a compact local chart on $\mc{I}$ and let $k\geq 0$, then:
    \begin{align}
        &\begin{aligned}\label{higher_constraint_Dm}
            \ms{D}\cdot \Lie_\rho^k \ms{m} = \ms{D} \paren{\ms{tr}_{\ms{g}}\Lie_\rho^k \ms{m } }&+ \sum\limits_{\substack{j + j_0 + j_1 + \dots + j_\ell = k\\j<k, j_p\geq 1}} \mc{S}\paren{\ms{g}; \Lie_\rho^j \ms{m}, \ms{D}\Lie_\rho^{j_0-1}\ms{m}, \Lie_\rho^{j_1-1}\ms{m}, \dots, \Lie_\rho^{j_\ell-1}\ms{m}}+4\Lie_\rho^k \ms{T}^1 \text{.}
        \end{aligned}
    \end{align}
    If $(\mc{M}, g, F)$ is a Maxwell-FG-aAdS segment of dimension $n\geq 4$, then, for all $k\geq 0$:
    \begin{align*}
        \Lie_\rho^k \ms{T}^1 = &\sum\limits_{q=0}^3\sum\limits_{j + j' + j_1 + \dots + j_\ell = k-q} \rho^{3-q}\cdot \mc{S}\paren{\ms{g}; \Lie_\rho^{j_1-1}\ms{m}, \dots, \Lie_\rho^{j_\ell-1}\ms{m}, \Lie_\rho^j \ms{f}^0, \Lie_\rho^{j'}\ol{\ms{f}}^1}\text{.}
    \end{align*}
    Furthermore, the following higher-order constraints hold for the Maxwell fields, for $n\geq 4$: 
    \begin{align}
        &\begin{aligned}\label{equation_d_f1}
            \ms{D}_{[a}\Lie_\rho^k \ol{\ms{f}}^1_{bc]} = \sum\limits_{\substack{j+j_0+j_1+\dots+j_\ell = k\\j<k, j_p\geq 1}}\mc{S}\paren{\ms{g}; \ms{D}\Lie_\rho^{j_0-1}\ms{m}, \Lie_\rho^{j_1-1}\ms{m}, \dots, \Lie_\rho^{j_\ell -1}\ms{m}, \Lie_\rho^{j}\ol{\ms{f}}^1}_{abc}\text{,}
        \end{aligned}\\
        &\begin{aligned}\label{equation_div_f0}
            k\ms{D}\cdot \Lie_\rho^{k-1} \ol{\ms{f}}^0 + \ms{D}\cdot \rho\Lie_\rho^k\ol{\ms{f}}^0 = &\sum\limits_{j+j_0+j_1+\dots + j_\ell = k}\mc{S}\paren{\ms{g}; \Lie_\rho^{j_0}\ms{m}, \Lie_\rho^{j_1-1}\ms{m}, \dots, \Lie_\rho^{j_\ell - 1}\ms{m},\Lie_\rho^j \ol{\ms{f}}^1} \\
            &+\sum\limits_{\substack{j+j_0+j_1+\dots + j_\ell = k\\j<k, j_p\geq 1}}\rho\cdot \mc{S}\paren{\ms{g}; \ms{D}\Lie_\rho^{j_0-1}\ms{m}, \Lie_\rho^{j_1-1}\ms{m}, \dots, \Lie_\rho^{j_\ell - 1}\ms{m}, \Lie_\rho^{j}\ol{\ms{f}}^0}\\
            &+\sum\limits_{\substack{j+j_0+j_1+\dots + j_\ell = k-1\\j<k-1, j_p\geq 1}}\mc{S}\paren{\ms{g}; \ms{D}\Lie_\rho^{j_0-1}\ms{m}, \Lie_\rho^{j_1-1}\ms{m}, \dots, \Lie_\rho^{j_\ell - 1}\ms{m}, \Lie_\rho^{j}\ol{\ms{f}}^0}\\
        \end{aligned}
    \end{align}
\end{proposition}

\begin{proof}
    These are simple consequences of Proposition \ref{sec:aads_commutation_Lie_D} applied to \eqref{eq_Dm} and \eqref{eq_Df}. 
\end{proof}

The next proposition establishes the ODE analysis that we will use throughout the proof:

\begin{proposition}\label{prop_integration}
    Let $(\mc{M}, g)$ be a FG-aAdS segment, $(U, \varphi)$ a compact coordinate system and let $\ms{A}$, $\ms{G}$ be vertical tensor fields of the same rank satisfying: 
    \begin{equation}\label{lemma_transport}
        \rho\Lie_\rho \ms{A} -c\cdot\ms{A} = \ms{G}\text{,}
    \end{equation}
    for some $c\geq 0$ a constant. Let $M\geq 0$ and $q_0\geq 0$. The following properties hold:
    \begin{enumerate}
        \item\label{local_estimate} \textbf{[Local estimate]} If $c>0$, the following bound is satisfied on $U$, for all $\rho, \tilde{\rho}\in (0, \rho_0]$: 
        \begin{equation}\label{bound_integration}
            \abs{\ms{A}}_{M, \varphi} \lesssim \frac{\rho^c}{\tilde{\rho}^c}\abs{\ms{A}}_{M, \varphi}\vert_{\tilde{\rho}} + \rho^c\int_{\min(\rho, \tilde{\rho})}^{\max(\rho, \tilde{\rho})}\sigma^{-c-1} \abs{\ms{G}}_{M, \varphi}\vert_\sigma d\sigma\text{.}
        \end{equation}
        \item\textbf{[Local boundedness]}\label{lemma_item_1}  If $c>0$ and the derivatives of $\ms{G}$ satisfy the following bounds: 
            \begin{equation}\label{lemma_bound_G}
                \norm{\rho^q\Lie_\rho^{q}\ms{G}}_{M, \varphi}\lesssim 1\text{,}
            \end{equation}
        for all $q\leq q_0$, then:
            \begin{equation}\label{lemma_bound_A_k_p}
                \norm{\rho^q\Lie_\rho^{q}\ms{A}}_{M, \varphi}\lesssim 1\text{,}
            \end{equation}
        for all $q\leq q_0+1$.
    \item \textbf{[Convergence]} \label{lemma_item_2} If $c>0$ and the derivatives of $\ms{G}$ satisfy the following limits: 
    \begin{equation}\label{assumption_limit_G}
        \rho^q\Lie_\rho^{q}\ms{G}\Rightarrow^{M}\begin{cases}
             \mf{G}\text{,}\qquad &q=0\\
              0\text{,}\qquad &1\leq q\leq q_0\text{,}
        \end{cases}
    \end{equation}
    with $\mf{G}$ a tensor field on $\mc{I}$, then the following limits hold for $\ms{A}$: 
    \begin{equation}\label{lemma_limit_z_q}
        \rho^{q}\Lie_\rho^{q}\ms{A}\Rightarrow^{M}\begin{cases}
            -\frac{1}{c}\mf{G}\text{,}\qquad &q=0\text{,}\\
            0\text{,}\qquad &1\leq q\leq q_0+1\text{.}
        \end{cases} 
    \end{equation}
    \item\label{lemma_item_3} \textbf{[Fractional convergence]} If $c\in (0,1)$, and $\ms{G}$ satisfies the following bounds: 
    \begin{equation}\label{bound_G_frac}
        \norm{\rho^q\Lie_\rho^{1+q}\ms{G}}_{M, \varphi}\lesssim 1\text{,}
    \end{equation}
    for all $q\leq q_0$, then, there exists a tensor field $\mf{A}^{(c)}$ on $\mc{I}$ such that, for all $q\leq q_0+1$: 
    \begin{equation}\label{lemma_limit_anomalous}
        \rho^{1-c+q}\Lie_\rho^{1+q}\ms{A}
        \rightarrow^{M} b_q \mf{A}^{(c)}\text{,}
    \end{equation}
    with $b_q$ constants. 
    \item\label{lemma_item_4} \textbf{[Anomalous limit]} If $c=0$ and if $\ms{G}$ satisfies the limits \eqref{assumption_limit_G} for all $q\leq q_0$,
        then, there exists a tensor field $\mf{A}^\dag$ on $\mc{I}$ such that, for all $q\leq q_0$: 
        \begin{equation}\label{lemma_log_limit}
            \rho^{1+q}\Lie_\rho^{1+q}\ms{A}\Rightarrow^{M} b_q'\mf{G}\text{,}\qquad \ms{A} - \log\rho\cdot \mf{G} \rightarrow \mf{A}^\dag\text{,}
        \end{equation}
        with $b_q'$ constants, and $b_0 = 1$.
    \end{enumerate}
\end{proposition}

\begin{remark}
    This proposition is a generalisation of Propositions 2.40 and 2.41 of \cite{shao:aads_fg}. More precisely:
    \begin{itemize}
        \item \ref{lemma_item_1} and \ref{lemma_item_2} show that one can obtain higher-order (in $\rho$) boundedness and convergence at the price of higher $\rho$--derivatives for $\ms{G}$, without loss of vertical regularity. Such a feature for the Einstein equations has been used in the author's thesis \cite{thesis_Guisset}, although this will not exploited beyond the second order in this work. 
        \item \ref{lemma_item_4} allows one to obtain higher-order limits for the anomalous case $(c=0)$, which differ only by a numerical constant. The second limit in \eqref{lemma_log_limit} is a simple consequence of Proposition 2.41 of \cite{shao:aads_fg}. 
        \item Finally, \ref{lemma_item_3} allows one to treat cases where $c$ is noninteger. This will not be particularly useful in the context of this work, but may turn out to be when one considers more general matter fields; see Example \ref{ex:scalar_field} below. 
    \end{itemize}
\end{remark}

\begin{proof}
    See Appendix \ref{app:prop_integration}.
\end{proof}

\begin{example}\label{ex:scalar_field}
    Let $(\mc{M}, g, \phi)$ be a Klein-Gordon-FG-aAdS segment with $\phi$ satisfying the following wave equation: 
        \begin{equation*}
            (\Box_g -m^2)\phi = 0\text{,}\qquad m^2\geq -\frac{n^2}{4}\text{.}
        \end{equation*}
    Such a wave equation can easily be written as:
    \begin{equation}
        \begin{cases}\label{wave_scalar}
            \psi' = \Lie_\rho\psi\text{,}\\
            \rho\Lie_\rho\psi' - (2\nu-1)\psi' = - \rho\cdot \ms{D}^2 \psi - \frac{1}{2}\Delta_-\ms{tr}_{\ms{g}}\ms{m}\cdot \psi - \frac{1}{2}\rho\psi' \cdot \ms{tr}_{\ms{g}}\ms{m}\text{,}
        \end{cases}
    \end{equation}
    where we define $\Delta_- := \frac{n}{2} - \nu\text{,}\;\nu:= \sqrt{\frac{n^2}{4} + m^2}$ and $\psi$ is defined as the following rescaling of $\phi$: $\psi := \rho^{-\Delta_-}\phi$. Note that, for $m^2$ in the well-posed range \cite{Warnick13, Holzegel12}, the limit $\psi\rightarrow^0 \psi^{(D)}$ exists and is the Dirichlet trace of $\phi$ on $\mc{I}$. 

    Observe that \eqref{wave_scalar} is an equation of the form \eqref{lemma_transport} for $\Lie_\rho \psi$ with $c = 2\nu -1\geq 0$ for $\nu \geq 1/2$. In particular, $c$ does not need to be an integer. 
\end{example}

We will also need the following proposition: 
\begin{proposition}[Proposition 2.36 in \cite{shao:aads_fg}]\label{prop_bounds_metric}
    Let $(\mc{M}, g)$ be a FG-aAdS segment, $M\geq 0$ and $(U,\varphi)$ a compact coordinate system on $\mc{I}$. Assume that the following holds for some constant $C>0$: 
    \begin{gather}\label{bound_metric}
        \norm{\ms{g}}_{M, \varphi}\leq C\text{,}\qquad \norm{\ms{g}^{-1}}_{0, \varphi}\leq C\text{,}\qquad \norm{\mathfrak{g}^{(0)}}_{M, \varphi}\leq C\text{,}\qquad \norm{(\mathfrak{g}^{(0)})^{-1}}_{0,\varphi}\leq C\text{,}
    \end{gather}
    Then: 
    \begin{gather*}
        \norm{\ms{g}^{-1}}_{M, \varphi}\lesssim_C 1\text{,}\qquad \norm{(\mathfrak{g}^{(0)})^{-1}}_{M, \varphi}\lesssim_C 1\text{,} \qquad \ms{g}^{-1}- (\mathfrak{g}^{(0)})^{-1}= \mc{O}_M\paren{\ms{g}-\mathfrak{g}^{(0)}}\text{.}
    \end{gather*}
    Furthermore, if $M\geq 2$, the following estimates hold for the Riemann tensor, Ricci tensor and scalar: 
    \begin{gather*}
        \norm{\ms{R}}_{M-2, \varphi} \lesssim_C 1\text{,}\qquad \norm{\mathfrak{R}^{(0)}}_{M-2, \varphi}\lesssim_C 1\text{,}\qquad \ms{R}-\mathfrak{R}^{(0)}=\mc{O}_{M-2}(\ms{g}-\mathfrak{g}^{(0)})\text{,}\\
        \norm{\ms{Rc}}_{M-2, \varphi} \lesssim_C 1\text{,}\qquad \norm{\mathfrak{Rc}^{(0)}}_{M-2, \varphi}\lesssim_C 1\text{,}\qquad \ms{Rc}-\mathfrak{Rc}^{(0)}=\mc{O}_{M-2}(\ms{g}-\mathfrak{g}^{(0)})\text{,}\\
        \norm{\ms{Rs}}_{M-2, \varphi} \lesssim_C 1\text{,}\qquad \norm{\mathfrak{Rs}^{(0)}}_{M-2, \varphi}\lesssim_C 1\text{,}\qquad \ms{Rs}-\mathfrak{Rs}^{(0)}=\mc{O}_{M-2}(\ms{g}-\mathfrak{g}^{(0)})\text{.}
    \end{gather*}
\end{proposition}
\begin{proof}
    See \cite{shao:aads_fg}, Proposition 2.36. 
\end{proof}

\begin{proposition}
    Let $(\mc{M}, g)$ be a FG-aAdS segment, $M\geq 0$ and $(U, \varphi)$ a compact coordinate system on $\mc{I}$, and assume \eqref{bound_metric} holds for $M$. Let $\ms{A}$ be vertical tensor field on $\mc{M}$ and $\mf{A}$ a tensor field on $\mc{I}$ of the same rank. Then, the following estimate holds: 
    \begin{equation}
        \label{estimate_DA}\abs{\ms{DA}- \mf{DA}}_{M-1, \varphi} \lesssim \abs{\ms{A}- \mf{A}}_{M, \varphi} + \abs{\ms{g}- \mf{g}^{(0)}}_{M, \varphi}\paren{\norm{\ms{A}}_{M-1, \varphi} + \norm{\mf{A}}_{M-1, \varphi}}\text{.}
    \end{equation}
    In particular, if $\ms{A}\rightarrow^M \mf{A}$, $\ms{A}\Rightarrow^{M}\mf{A}$ respectively, then: 
    \begin{equation}
        \label{limits_DA}\ms{DA}\rightarrow^{M-1} \mf{DA}\text{,}\qquad \ms{DA}\Rightarrow^{M-1}\mf{DA}\text{.}
    \end{equation}
\end{proposition}

\begin{proof}
    The following estimate follows from the definition of $\ms{D}$: 
    \begin{equation*}
        \abs{\ms{DA}-\mf{DA}}_{M-1, \varphi}\lesssim \abs{\ms{A}-\mf{A}}_{M, \varphi} + \abs{\mc{P}([\partial \ms{g}], [\ms{g}^{-1}])\cdot \ms{A}-\mc{P}([\partial\mf{g}], [\mf{g}^{-1}])\cdot \mf{A}}_{M-1, \varphi}\text{,}
    \end{equation*}
    where $\mc{P}$ is a quadratic polynomial. One has, at least schematically:
    \begin{align*} &\abs{\mc{P}([\partial \ms{g}], [\ms{g}^{-1}])\cdot \ms{A}-\mc{P}([\partial\mf{g}], [\mf{g}^{-1}])\cdot \mf{A}}_{M-1, \varphi}\lesssim \abs{\ms{g}- \mf{g}}_{M, \varphi}\paren{\norm{\ms{A}}_{M-1, \varphi}+\norm{\mf{A}}_{M-1, \varphi}} + {\abs{\ms{A}-\mf{A}}_{M-1, \varphi}}\text{,}
    \end{align*}
    proving \eqref{estimate_DA}. The limits \eqref{limits_DA} follow immediately. 
\end{proof}

\subsubsection{Proof of Theorem \ref{theorem_main_fg}}

We will first consider the case $n\geq3$. The special dimension $n=2$ will be treated in the end. Throughout this proof, we will use the fact that the transport equations will typically be of the form: 
\begin{equation}\label{transport_proof}
    \rho\Lie_\rho \ms{A} - c_{\ms{A}, n} \cdot \ms{A} = \ms{G}\text{,}
\end{equation}
for some vertical tensors $\ms{A}, \ms{G}$ of the same rank and $c_{\ms{A}, n}\geq 0$ a constant depending on $\ms{A}$ and $n$. These types of equations can be treated using Proposition \ref{prop_integration}. 

\proofpart{First bounds and limits}

\begin{lemma}\label{lemma_first_bounds}
    Let $n\geq 3$. The following uniform bounds hold with respect to any compact chart $(U, \varphi)$ on $\mc{I}$: 
    \begin{gather}
            \norm{\ms{tr}_{\ms{g}}\ms{m}}_{M_0, \varphi} + \norm{\ms{m}}_{M_0, \varphi}\lesssim 1\text{,}\label{unif_bound_0_A01}\\
            \begin{cases}
                \norm{\ol{\ms{f}}^0}_{M_0, \varphi}+\norm{\ol{\ms{f}}^1}_{M_0+1, \varphi}\lesssim 1\text{,}\qquad &n> 4\text{,}\\
                \norm{\rho^\delta\cdot \ol{\ms{f}}^0}_{M_0,\varphi} + \norm{\ol{\ms{f}}^1}_{M_0+1, \varphi}\lesssim 1\text{,}\qquad &n=4\text{,}\\
                \norm{\ol{\ms{f}}^0}_{M_0, \varphi}+\norm{\ol{\ms{f}}^1}_{M_0, \varphi}\lesssim 1\text{,}\qquad &n=3\text{,}
            \end{cases}
        \end{gather}
        for any $0<\delta < 1$.
\end{lemma}

\begin{proof}

Let us first define the following vertical tensor field:
\begin{equation*}
    \tilde{\ms{f}}^0 :=
    \left\{\begin{aligned}
        &\ol{\ms{f}}^0\text{,}\qquad &&n\neq 4\text{,}\\
        &\rho^\delta\cdot \ol{\ms{f}}^0\text{,}\qquad &&n=4\text{,}
    \end{aligned}\right.
\end{equation*}
with $0<\delta<1$ fixed.
Observe that such a field satisfies the following transport equation: 
\begin{equation}
    \label{transport_delta}\rho\Lie_\rho \tilde{{\ms{f}}}^0 - c_0 \tilde{\ms{f}}^0 = \rho\cdot \mc{S}\paren{\ms{g}; \ms{m}, \tilde{\ms{f}}^0} - \ms{D}\cdot \ol{\ms{f}}^1 \cdot \left\{\begin{aligned}
        1\text{,}\qquad & n>4\text{,}\\
        \rho^{\delta} \text{,}\qquad & n=4\text{,}
    \end{aligned}\right.\text{,}\qquad c_0 :=\left\{ \begin{aligned}
        &n-4\text{,}\qquad &&n>4\text{,}\\
        &\delta\text{,}\qquad &&n=4\text{,}
    \end{aligned}\right.
\end{equation}
Since $c_0>0$ for $n\geq 4$, one can use \ref{local_estimate} of Proposition \ref{prop_integration} to obtain: 
\begin{align*}
    \begin{aligned}\abs{\tilde{\ms{f}}^0}_{0, \varphi} &\lesssim \frac{\rho^{c_0}}{\rho_0^{c_0}}\abs{\tilde{\ms{f}}^0}_{0, \varphi}\vert_{\rho_0} + \rho^{c_0}\int^\rho_{\rho_0} \sigma^{-c_0-1}\left[\sigma \abs{\tilde{\ms{f}}^0}_{0, \varphi}\abs{\ms{m}}_{0, \varphi} + \abs{\ol{\ms{f}}^1}_{1, \varphi}\right]\vert_\sigma d\sigma \\
    &\lesssim  1+\int^\rho_{\rho_0} \abs{\tilde{\ms{f}}^0}_{0, \varphi}\abs{\ms{m}}_{0, \varphi}\vert_\sigma d\sigma\text{,}
    \end{aligned}
\end{align*}
where we used the fact that $U$ is compact and \eqref{assumption_metric}. Finally, using Grönwall's inequality and \eqref{assumption_m_f}, one obtains indeed: 
\begin{equation}\label{bound0_f_proof}
    \abs{\tilde{\ms{f}}^0}_{0, \varphi}\lesssim \exp\left\{\int_0^{\rho_0 }\abs{\ms{m}}_{0, \varphi}\vert_\sigma d\sigma \right\}\Rightarrow\left\{\begin{aligned}
        &\norm{\ol{\ms{f}}^0}_{0, \varphi} \lesssim 1\text{,}\qquad &&n> 4\text{,}\\
        &\norm{\rho^\delta \cdot \ol{\ms{f}}^0}_{0, \varphi}\lesssim 1\text{,}\qquad &&n=4\text{.}
    \end{aligned}\right.
\end{equation}
Since $\ol{\ms{f}}^0$ and $\ol{\ms{f}}^1$ converge in $C^0$ for $n=3$, they must necessarily be bounded. This proves the uniform bounds for the Maxwell fields for $n\geq 3$. 

Observe that \eqref{transport_m} and \eqref{transport_tr_m} are also equations of the form \eqref{transport_proof}. Using therefore \ref{local_estimate} of Proposition \ref{prop_integration} yields, with $\tilde{\rho}:= \rho_0$: 
\begin{align*}
    &\begin{aligned}
        \abs{\ms{m}}_{0, \varphi} \lesssim \frac{\rho^{n-1}}{\rho_0^{n-1}}\abs{\ms{m}}_{0, \varphi}\vert_{\rho_0} + \rho^{n-1}\int^{\rho}_{\rho_0} \sigma^{-n} \left[\sigma \abs{\ms{Rc}}_{0, \varphi} + \abs{\ms{tr}_{\ms{g}}\ms{m}}_{0, \varphi} + \sigma \abs{\ms{m}}^2_{0, \varphi} + \sigma \abs{\ms{T}^0}_{0, \varphi} +  \sigma \abs{\ms{T}^2}_{0, \varphi}\right]\vert_\sigma d\sigma\text{,}
    \end{aligned}\\
    &\begin{aligned}
        \abs{\ms{tr}_{\ms{g}}\ms{m}}_{0, \varphi} \lesssim \frac{\rho^{2n-1}}{\rho_0^{2n-1}}\abs{\ms{m}}_{0, \varphi}\vert_{\rho_0} + \rho^{2n-1}\int^{\rho}_{\rho_0} \sigma^{-(2n-1)} \left[ \abs{\ms{Rs}}_{0, \varphi} + \abs{\ms{m}}^2_{0, \varphi} + \abs{\ms{T}^0}_{0, \varphi} +  \abs{\ms{T}^2}_{0, \varphi}\right]\vert_\sigma d\sigma\text{,}
    \end{aligned}
\end{align*}
where we used \eqref{assumption_metric}. Note also that, by Proposition \ref{prop_bounds_metric}, the Riemann tensor and scalar are automatically bounded in $C^{M_0}$. Furthermore, using the fact that $U$ is compact, the above becomes: 
\begin{align}
    &\label{intermediate_m}\begin{aligned}
        \abs{\ms{m}}_{0, \varphi} \lesssim 1 + \rho^{n-1}\int^{\rho}_{\rho_0} \sigma^{-n} \left[ \abs{\ms{tr}_{\ms{g}}\ms{m}}_{0, \varphi} + \sigma \abs{\ms{m}}^2_{0, \varphi} + \sigma \abs{\ms{T}^0}_{0, \varphi} +  \sigma \abs{\ms{T}^2}_{0, \varphi}\right]\vert_\sigma d\sigma\text{,}
    \end{aligned}\\
    &\label{intermediate_tr_m}\begin{aligned}
        \abs{\ms{tr}_{\ms{g}}\ms{m}}_{0, \varphi} \lesssim 1 + \int^{\rho}_{\rho_0}\left[ \abs{\ms{m}}^2_{0, \varphi} + \abs{\ms{T}^0}_{0, \varphi} +  \abs{\ms{T}^2}_{0, \varphi}\right]\vert_\sigma d\sigma\text{.}
    \end{aligned}
\end{align}
Inserting now \eqref{intermediate_tr_m} into \eqref{intermediate_m} gives, using Fubini's theorem: 
\begin{align*}
    \abs{\ms{m}}_{0, \varphi} \lesssim&\,  1 + \rho^{n-1}\int^{\rho}_{\rho_0} \sigma^{-n}\left(\int^{\sigma}_{\rho_0} \left[\abs{\ms{m}}^2_{0, \varphi} + \abs{\ms{T}^0}_{0, \varphi} +  \abs{\ms{T}^2}_{0, \varphi}\right]\vert_{\tau} d\tau +  \sigma \abs{\ms{m}}^2_{0, \varphi} + \sigma \abs{\ms{T}^0}_{0, \varphi} +  \sigma \abs{\ms{T}^2}_{0, \varphi}\right)\vert_\sigma d\sigma\\
    \lesssim&\,  1 + \int^\rho_{\rho_0} \left[ \abs{\ms{m}}^2_{0, \varphi} +  \abs{\ms{T}^0}_{0, \varphi} +  \abs{\ms{T}^2}_{0, \varphi}\right]\vert_\sigma d\sigma\text{.}
\end{align*}
Observe now that the following bound holds for the stress-energy tensor: 
\begin{align*}
    \int^{\rho}_{\rho_0}  \left[\abs{\ms{T}^0}_{0, \varphi} +  \abs{\ms{T}^2}_{0, \varphi}\right]\vert_\sigma d\sigma&\lesssim \int^\rho_{\rho_0}\left\{\begin{aligned}
        &\sigma^2\left[\sigma^2\abs{\ol{\ms{f}}^0}^2_{0, \varphi} + \norm{\ol{\ms{f}}^1}^2_{0, \varphi}\right]\vert_\sigma d\sigma\text{,}\qquad &n\geq 4\text{,}\\
        &\sigma^2\left[\abs{\ol{\ms{f}}^0}_{0, \varphi}^2 + \abs{\ol{\ms{f}}^1}_{0, \varphi}^2\right]\vert_\sigma d\sigma \text{,}\qquad &n=3\text{,}
    \end{aligned}\right.\\
    &\lesssim 1\text{.}
\end{align*}
where we used \eqref{assumption_metric} and \eqref{bound0_f_proof}. Note also that since the Maxwell fields converge in $C^0$ for $n=3$, they are necessarily bounded. Finally, using Grönwall's inequality, one has: 
\begin{align}
    \label{intermediate_m_gronwall}\begin{aligned}\abs{\ms{m}}_{0, \varphi} \lesssim 1 + \int^\rho_{\rho_0} \abs{\ms{m}}_{0, \varphi}\vert_\sigma \cdot \abs{\ms{m}}_{0, \varphi}\vert_\sigma d\sigma \Longrightarrow \abs{\ms{m}}_{0, \varphi}&\lesssim \exp\left\{ \int_0^{\rho_0}\abs{\ms{m}}_{0, \varphi}\vert_\sigma d\sigma \right\} \\
    &\lesssim 1\text{,}
    \end{aligned}
\end{align}
where we used \eqref{assumption_m_f}.

We are ready now to show higher-order bounds along vertical directions, which we will obtain by induction. Namely, let us assume the following bounds to hold for any $q<k$, with some $k\leq M_0$ fixed:
\begin{equation*}
    \norm{\ms{tr}_{\ms{g}}\ms{m}}_{q, \varphi} +  \norm{\ms{m}}_{q, \varphi} + \norm{\tilde{\ms{f}}^0}_{q, \varphi} + \norm{\ol{\ms{f}}^1}_{q, \varphi}\lesssim 1\text{.}
\end{equation*}
Note first that for $n\geq 4$, one has automatically $\norm{\ol{\ms{f}}^1}_{M_0, \varphi}\lesssim 1$. Starting with the metric, one has, by differentiating $k$--times along vertical directions and using Proposition \ref{prop_integration}: 
\begin{align}
    &\begin{aligned}
        \label{intermediate_m_k}\abs{{\ms{m}}}_{k, \varphi}\lesssim 1 + \rho^{n-1}\int^\rho_{\rho_0}\sigma^{-n}\left[\norm{\rho\ms{Rc}}_{k, \varphi} + \abs{\ms{tr}_{\ms{g}}\ms{m}}_{k, \varphi} + \sigma\abs{\ms{m}}_{k, \varphi} + \sigma \abs{\ms{T}^0}_{k, \varphi} +\sigma \abs{\ms{T}^2}_{k, \varphi}\right]\vert_\sigma d\sigma 
    \end{aligned}\\
    &\begin{aligned}
        \abs{\ms{tr}_{\ms{g}}{\ms{m}}}_{k, \varphi}\lesssim 1 + \int^\rho_{\rho_0}\left[\norm{\ms{Rs}}_{k, \varphi}  + \abs{\ms{m}}_{k, \varphi} + \abs{\ms{T}^0}_{k, \varphi} + \abs{\ms{T}^2}_{k, \varphi}\right]\vert_\sigma d\sigma\text{,}
    \end{aligned}
\end{align}
where we use the fact that lower order vertical derivatives are uniformly bounded, easing the treatment of nonlinear terms. 

The stress-energy tensor terms appearing on the right-hand side satisfy the following bounds:
\begin{align*}
    \abs{\ms{T}^0}_{k, \varphi} + \abs{\ms{T}^2}_{k, \varphi} \lesssim 1 + \left\{\begin{aligned}
        &\rho^4\abs{\ol{\ms{f}}^0}_{k, \varphi} \text{,}\qquad &&n>4\\
        &\rho^{4-2\delta}\abs{\sigma^\delta \cdot \ol{\ms{f}}^0}_{k, \varphi}\text{,}\qquad &&n=4\text{,}\\
        &\rho^2\abs{\sigma\cdot \ol{\ms{f}}^0}_{k, \varphi}+\rho^2\abs{\sigma\cdot \ol{\ms{f}}^1}_{k, \varphi}\text{,}\qquad &&n=3\text{,}
    \end{aligned}\right.
\end{align*}
where we used \eqref{assumption_metric}, yielding for the trace: 
\begin{equation*}
    \abs{\ms{tr}_{\ms{g}}\ms{m}}_{k, \varphi}\lesssim 1+\left\{\begin{aligned}
        &\int^\rho_{\rho_0} \left[\abs{\ms{m}}_{k, \varphi} +\sigma^4\abs{\ol{\ms{f}}^0}_{k, \varphi}\right]\vert_\sigma d\sigma\text{,}\qquad &&n>4\text{,}\\
        &\int^\rho_{\rho_0} \left[\abs{\ms{m}}_{k, \varphi} +\sigma^{4-2\delta}\abs{\sigma^\delta \cdot \ol{\ms{f}}^0}_{k, \varphi}\right]\vert_\sigma d\sigma\text{,}\qquad &&n=4\text{,}\\
        &\int^{\rho}_{\rho_0}\left[\abs{\ms{m}}_{k, \varphi} +\sigma^2\abs{\ol{\ms{f}}^0}_{k, \varphi} +\sigma^2\abs{ \ol{\ms{f}}^1}_{k, \varphi}\right]\vert_\sigma d\sigma\text{,}\qquad &&n=3\text{,}
    \end{aligned}\right.
\end{equation*}
which can be inserted into \eqref{intermediate_m_k}, using Fubini's theorem: 
\begin{equation}
    \label{intermediate_m_k_f}\abs{\ms{m}}_{k, \varphi}\lesssim 1+\left\{\begin{aligned}
        &\int^\rho_{\rho_0} \left[\abs{\ms{m}}_{k, \varphi} +\sigma^4\abs{\ol{\ms{f}}^0}_{k, \varphi}\right]\vert_\sigma d\sigma\text{,}\qquad &&n>4\text{,}\\
        &\int^\rho_{\rho_0} \left[\abs{\ms{m}}_{k, \varphi} +\sigma^{4-2\delta}\abs{\sigma^\delta \cdot \ol{\ms{f}}^0}_{k, \varphi}\right]\vert_\sigma d\sigma\text{,}\qquad &&n=4\text{,}\\
        &\int^{\rho}_{\rho_0}\left[\abs{\ms{m}}_{k, \varphi} +\sigma^2\abs{\ol{\ms{f}}^0}_{k, \varphi} +\sigma^2\abs{ \ol{\ms{f}}^1}_{k, \varphi}\right]\vert_\sigma d\sigma\text{,}\qquad &&n=3\text{.}
    \end{aligned}\right.
\end{equation}
Note that in the case $n=3$, one can use \eqref{assumption_m_f} to bound the source terms generated by $\ol{\ms{f}}^0$ and $\ol{\ms{f}}^1$. Next, in order to close the estimates for $n\geq 4$, one has to bound $\tilde{\ms{f}}^0$: 
\begin{equation}
    \label{intermediate_k_f}\abs{\tilde{\ms{f}}^0}_{k, \varphi}\lesssim 1+\int^\rho_{\rho_0} \left[\abs{\ms{m}}_{k, \varphi} + \abs{\tilde{\ms{f}}^0}_{k, \varphi}\right]\vert_\sigma d\sigma\text{,}
\end{equation}
where we used \eqref{assumption_metric}. Combining now \eqref{intermediate_m_k_f} with \eqref{intermediate_k_f} leads to, for $n\geq 4$: 
\begin{equation*}
    \abs{\ms{m}}_{k, \varphi} +\abs{\tilde{\ms{f}}^0}_{k, \varphi} \lesssim 1+\int^\rho_{\rho_0} \left[ \abs{\ms{m}}_{k, \varphi} +\abs{\tilde{\ms{f}}^0}_{k, \varphi}\right]\vert_\sigma d\sigma\text{,}
\end{equation*}
for which one can use Grönwall's inequality. Finally, one can use \ref{local_estimate} of Proposition \ref{prop_integration} to get the following estimate, for $n=3$: 
\begin{equation*}
    \begin{aligned}
        \abs{\ol{\ms{f}}^1}_{k, \varphi}\lesssim 1 + \int_\rho^{\rho_0} \abs{\ol{\ms{f}}^0}_{k+1, \varphi}\vert_\sigma d\sigma \lesssim 1\text{,}
    \end{aligned}
\end{equation*}
proving the lemma.
\end{proof}

\begin{lemma}\label{lemma_first_limits}
    Let $n\geq 3$. The following limits are satisfied on any compact chart $(U, \varphi)$:
    \begin{enumerate}
        \item \textbf{For the metric:}
        \begin{gather*}
            \ms{g}\Rightarrow^{M_0} \mf{g}^{(0)}\text{,}\qquad \ms{g}^{-1}\Rightarrow^{M_0} \paren{\mf{g}^{(0)}}^{-1}\text{,}\\
            \ms{m}\Rightarrow^{M_0}0\text{,}\qquad \rho\Lie_\rho\ms{m}\Rightarrow^{M_0}0\text{,}
        \end{gather*}
        \item \textbf{For Maxwell fields:}
        \begin{gather*}
            \ol{\ms{f}}^1 \Rightarrow^{M_0-1} \mf{f}^{1,(0)}\text{,}\qquad \rho\Lie_\rho  \ol{\ms{f}}^1 \Rightarrow^{M_0-1} 0  \\\left\{\begin{aligned}
                &\ol{\ms{f}}^0 \Rightarrow^{M_0-2} \mi{D}^{(1,1)}(\mf{g}^{(0)}, \mf{f}^{1,(0)})\text{,}\; \qquad &&\rho\Lie_\rho \ol{\ms{f}}^0\Rightarrow^{M_0-2} 0\text{,}\qquad   &&&n>4\text{,}\\
                &\rho^\delta\cdot \ol{\ms{f}}^0\Rightarrow^{M_0} 0\text{,}\qquad &&\rho\Lie_\rho(\rho^\delta\cdot \ol{\ms{f}}^0)\Rightarrow^{M_0 } 0\text{,}\qquad   &&&n=4\text{,}\\
                &\ol{\ms{f}}^0\Rightarrow^{M_0-1}\mf{f}^{0,(0)}\text{,}\qquad &&\rho\Lie_\rho \ol{\ms{f}}^0\Rightarrow^{M_0-1} 0\text{,}\qquad  &&&n=3\text{,}
            \end{aligned}\right.
        \end{gather*}
        for any $0<\delta < 1$. 

        Furthermore, for $n>4$:
        \begin{equation}\label{limit_Lie_square_f0}
            \rho^2\Lie_\rho^2\ol{\ms{f}}^0\Rightarrow^{M_0-2}0\text{.}
        \end{equation}
    \end{enumerate}
\end{lemma}
\begin{proof}
    Starting with the metric, one can simply estimate: 
    \begin{align}
        &\label{convergence_m}\sup\limits_{\lbrace \rho \rbrace \times U}\abs{\ms{g}- \mf{g}^{(0)}}_{M_0, \varphi}\lesssim \int_0^\rho \norm{\ms{m}}_{M_0, \varphi}d\sigma  \rightarrow 0\text{,}\\
        &\sup\limits_U \int_0^{\rho_0}\sigma^{-1}\cdot \abs{\ms{g}- \mf{g}^{(0)}}_{M_0, \varphi}\vert_\sigma d\sigma \lesssim \sup\limits_U \int_0^{\rho_0} \rho^{-1}\int_{0}^\rho \norm{\ms{m}}_{M_0, \varphi}d\sigma d\rho <\infty\text{.} 
    \end{align}
    This immediately implies, from Proposition \ref{prop_bounds_metric}, the following limits: 
    \begin{gather*}
        \ms{g}^{-1}\Rightarrow^{M_0}\paren{\mf{g}^{(0)}}^{-1}\text{,}\qquad \ms{R}\Rightarrow^{M_0-2}\mi{D}^{2}\paren{\mf{g}^{(0)}}\text{,}
    \end{gather*}
    and similarly for the Ricci tensor and scalar. 
    From Equations \eqref{transport_tr_m} and \eqref{transport_m}, one has, using Proposition \ref{prop_integration}: 
    \begin{equation*}
            \ms{tr}_{\ms{g}}\ms{m}\Rightarrow^{M_0} 0\text{,}\qquad \rho\Lie_\rho (\ms{tr}_{\ms{g}}\ms{m})\Rightarrow^{M_0} 0\text{,} \qquad \ms{m}\Rightarrow^{M_0} 0\text{,}\qquad \rho\Lie_\rho\ms{m}\Rightarrow^{M_0}0\text{.}
    \end{equation*}
    One can perform similarly for the Maxwell fields. Namely, starting with $\ol{\ms{f}}^1$, with $n\geq 3$: 
    \begin{align}
        &\begin{aligned}
            \label{convergence_f_1}\sup\limits_{\lbrace \rho \rbrace \times U} \abs{\ol{\ms{f}}^1-\mf{f}^{1,(0)}}_{M_0-1, \varphi}&\lesssim  \sup\limits_{U} \int_0^\rho \abs{\Lie_\rho\ol{\ms{f}}^1}_{M_0-1, \varphi} \vert_\sigma d\sigma \\
            &\lesssim  \sup\limits_{U}\int_0^\rho \abs{\ms{D}\ol{\ms{f}}^0}_{M_0-1, \varphi}\vert_\sigma d\sigma \\
            &\lesssim \int_0^\rho \sigma^{-\delta}\norm{\rho^\delta \cdot \ol{\ms{f}}^0}_{M_0, \varphi} d\sigma \rightarrow 0\text{,}
        \end{aligned}\\
        &\label{strong_convergence_f_1}\sup\limits_U \int_0^{\rho_0} \sigma^{-1}\cdot \abs{\ol{\ms{f}}^1-\mf{f}^{1,(0)}}_{M_0-1, \varphi}\vert_\sigma d\sigma \lesssim  \sup\limits_U \int_0^{\rho_0} \sigma^{-1} \int_0^\sigma \tau^{-\delta}\norm{\rho^\delta \cdot \ol{\ms{f}}^0}_{M_0, \varphi}\vert_\tau d\tau d\sigma<\infty\text{,} 
    \end{align}   
    with $0<\delta<1$ fixed, and where the anomalous limits vanish from \eqref{transport_f_1}:
    \begin{equation*}
        \sup\limits_{\lbrace \rho \rbrace \times U}\abs{\rho\Lie_\rho \ol{\ms{f}}^1}_{M_0-1, \varphi}\lesssim \rho^{1-\delta} \norm{\rho^\delta\cdot \ol{\ms{f}}^0}_{M_0, \varphi}=\mc{O}(\rho^{1-\delta})\rightarrow 0\text{.}
    \end{equation*}
    This implies for $\ol{\ms{f}}^0$, for $n> 4$, from \eqref{transport_f_0}: 
    \begin{equation*}
        {\ol{\ms{f}}^0} \Rightarrow^{M_0-2} \mi{D}^{(1, 1)}\paren{\mf{g}^{(0)}, \mf{f}^{1, (0)}}\text{,}\qquad \rho\Lie_\rho \ol{\ms{f}}^{0} \Rightarrow^{M_0-2} 0\text{.}
    \end{equation*}
    For $n=4$, since the right-hand side of Equation \eqref{transport_delta} has a vanishing limit in $C^{M_0}$ and by Proposition \ref{prop_integration}:
    \begin{equation*}
        \rho^\delta\cdot \ol{\ms{f}}^0\Rightarrow^{M_0} 0\text{,}\qquad \rho\Lie_\rho (\rho^\delta\cdot \ol{\ms{f}}^0) \Rightarrow^{M_0} 0\text{.}
    \end{equation*}
    For $n=3$, one can perform similarly as in \eqref{convergence_m}--\eqref{strong_convergence_f_1} to obtain \footnote{Note that the limits for $\ms{f}^1$ can be obtained in an identical manner.}: 
    \begin{align*}
        &\sup\limits_{\lbrace \rho \rbrace \times U}\abs{\ol{\ms{f}}^0 - \mf{f}^{0,(0)}}_{M_0-1, \varphi} \lesssim \int_0^\rho \left[\norm{\ms{m}}_{M_0-1, \varphi} + \norm{\ol{\ms{f}}^0}_{M_0-1, \varphi} +\norm{\ol{\ms{f}}^1}_{M_0, \varphi}\right]\vert_\sigma d\sigma \rightarrow 0\text{,}\\
        &\sup\limits_U \int_0^{\rho_0} \sigma^{-1}\abs{\ol{\ms{f}}^0 - \mf{f}^{0,(0)}}_{M_0-1, \varphi} \vert_\sigma d\sigma \lesssim \int_0^{\rho_0} \sigma^{-1}\int_0^{\sigma} \left[\norm{\ms{m}}_{M_0-1, \varphi} + \norm{\ol{\ms{f}}^0}_{M_0-1, \varphi} +\norm{\ol{\ms{f}}^1}_{M_0, \varphi}\right]\vert_\tau d\tau <\infty\text{.}
    \end{align*}
    The anomalous limit automatically vanishes from Equation \eqref{transport_f_0}:
    \begin{align*}
        &\sup\limits_{\lbrace \rho \rbrace \times U}\abs{\rho\Lie_\rho \ol{\ms{f}}^0}_{M_0-1, \varphi}\lesssim \rho\cdot \paren{\norm{\ms{m}}_{M_0-1, \varphi} + \norm{\ol{\ms{f}}^0}_{M_0-1, \varphi} + \norm{\ol{\ms{f}}^1}_{M_0, \varphi}}=\mc{O}(\rho)\rightarrow  0\text{.}
    \end{align*}
    Finally, the limit \eqref{limit_Lie_square_f0} is a straightforward consequence of \ref{lemma_item_2} of Proposition \ref{prop_integration}, together with the limits obtained so far. Namely, the right-hand side of \eqref{transport_f_0}, on which one applies the operator $\rho\Lie_\rho$, satisfies: 
    \begin{equation*}
        \ms{D}\cdot \rho\Lie_\rho\ol{\ms{f}}^1 + \rho \mc{S}\paren{\ms{g}; \ms{m}, \rho\Lie_\rho\ol{\ms{f}}^0} +   \rho \mc{S}\paren{\ms{g}; \rho\Lie_\rho \ms{m}, \ol{\ms{f}}^0} + \rho\mc{S}\paren{\ms{g}; \ms{m}, \ol{\ms{f}}^0} + \rho\mc{S}\paren{\ms{g}; \ms{D}\ms{m}, \ol{\ms{f}}^1}\Rightarrow^{M_0-2} 0\text{.}
    \end{equation*}
\end{proof}

\proofpart{Higher-order limits}

Using the higher-order transport equations from Propositions \ref{higher_order_m} and \ref{higher_order_f}, we will be able to use an induction argument to control higher-order derivatives. Since the criticality for the Maxwell fields is reached earlier than for the metric, we will have to be careful for near-criticality derivatives. 

\begin{lemma}\label{lemma_higher_limits}
    The following higher-order limits hold for $n>4$, for any $0\leq k < n-4$: 
    \begin{enumerate}
        \item \textbf{The metric\footnote{We will take, as a convention, that $\Lie_\rho^k\ms{m}\Rightarrow^{M_0-k-1}\mi{D}^{k+1}(\mf{g}^{(0)}; k+1)$ for $k<3$.}:} 
        \begin{gather}
            \Lie_\rho^k \ms{m}\Rightarrow^{M_0-k-1} \mi{D}^{(k+1, k-3)}\paren{\mf{g}^{(0)}, \mf{f}^{1,(0)}; k+1}\text{,}\qquad \rho\Lie_\rho^{k+1}\ms{m}\Rightarrow^{M_0-k} 0\text{,}\label{higher_limit_m}\\
            \Lie_\rho^{n-4}\ms{m}\Rightarrow^{M_0-n+3}\mi{D}^{(n-3, n-7)} \paren{\mf{g}^{(0)}, \mf{f}^{1,(0)}; n+1}\text{,}\qquad \rho\Lie_\rho^{n-3}\ms{m}\Rightarrow^{M_0-n+3}0\label{higher_limit_n-4}\text{.}
        \end{gather}
        \item \textbf{The Maxwell fields: }
        \begin{align}
            &\begin{aligned}
            &\Lie_\rho^{k}\ol{\ms{f}}^0 \Rightarrow^{M_0-k-2}\mi{D}^{(k+1, k+1)}\paren{\mf{g}^{(0)}, \mf{f}^{1,(0)}; k}\text{,}\qquad 
            &&\Lie_\rho^{k+1}\ol{\ms{f}}^1 \Rightarrow^{M_0-k-2}
            \mi{D}^{(k+1, k+1)}\paren{\mf{g}^{(0)}, \mf{f}^{1,(0)}; k+1}\text{,} \\
            &\rho\Lie_\rho^{k+1}\ol{\ms{f}}^0\Rightarrow^{M_0-k-2}0\text{,}\qquad &&\rho\Lie_\rho^{k+2}\ol{\ms{f}}^1\Rightarrow^{M_0-k-2}0\text{,}\label{higher_limit_f}
            \end{aligned}\\
            &\rho^2\Lie_\rho^{k+2}\ol{\ms{f}}^0\Rightarrow^{M_0-k-2}0\text{.}\label{limit_k+2_f_0}
        \end{align}
        Furthermore, 
        \begin{equation}\label{limit_n-3_f_1}
            \Lie_\rho^{n-3}\ol{\ms{f}}^1 \Rightarrow^{M_0-n+1} \mi{D}^{(n-3, n-3)}\paren{\mf{g}^{(0)}, \mf{f}^{1,(0)}}\text{,}\qquad \rho\Lie_\rho^{n-2}\ol{\ms{f}}^1 \Rightarrow^{M_0-n+1} 0\text{.}
        \end{equation}
    \end{enumerate}
\end{lemma}

\begin{remark}
    Here, it is necessary to restrict $k<n-4$ since one will reach criticality for $\ol{\ms{f}}^0$ otherwise. 
\end{remark}

\begin{proof}
    We will prove this lemma by induction on $k$. First of all, in order to start the argument, note that \eqref{transport_f_1}, as well as Lemma \ref{lemma_first_bounds}, yield: 
    \begin{gather*}
        \abs{\Lie_\rho \ol{\ms{f}}^1}_{M_0-1, \varphi} \lesssim  \rho\cdot\abs{\mc{S}\paren{\ms{g}; \ms{D}\ol{\ms{f}}^0}}_{M_0-1, \varphi} = \mc{O}(\rho)\rightarrow 0\text{,}
    \end{gather*}
    as $\rho\searrow 0$, and similarly for $\rho\Lie_\rho^2\ms{f}^1$. 

    Fix $k<n-4$ and assume that \eqref{higher_limit_m}, \eqref{higher_limit_f} and \eqref{limit_k+2_f_0} hold for any $q<k$. Since $k<n-4$, the following transport equations are satisfied: 
    \begin{equation*}
        \rho\Lie_\rho\ms{A}_k - c_{n,k} \Lie_\rho^k\ms{A}_k = \ms{G}_k\text{,}\qquad \ms{A}_k\in \lbrace \Lie_\rho^k\ms{m}, \ms{tr}_{\ms{g}}\Lie_\rho^k\ms{m}, \Lie_\rho^k \ol{\ms{f}}^0\rbrace\text{,}
    \end{equation*}
    with $c_{n,k}>0$ and $\ms{G}_k$ a tensor of the same rank as $\ms{A}$. As such, one can use \ref{lemma_item_2} of Proposition \ref{prop_integration} to obtain the desired limit. In fact, one only needs to study the right-hand sides of \eqref{higher_transport_m}, \eqref{higher_transport_tr_m} and \eqref{higher_transport_f_0}. 
    Starting with $\ol{\ms{f}}^0$, the right-hand side will contain the following terms:
    \begin{align*}
        &-\ms{D}\cdot \Lie_\rho^{(k-1)+1}\ol{\ms{f}}^1 \Rightarrow^{M_0-k-2} \mi{D}^{(k, k)}\paren{\mf{g}^{(0)}, \mf{f}^{1,(0)}; k}\text{,}\\
        & \sum\limits_{j+j_0+j_1+\dots + j_\ell = k} \rho\cdot\mc{S}\paren{\ms{g}; \Lie_\rho^{j_0} \ms{m}, \Lie_\rho^{j_1 -1}\ms{m}, \dots, \Lie_\rho^{j_\ell - 1}\ms{m}, \Lie_\rho^j \ol{\ms{f}}^0}\Rightarrow^{M_0-k-2}0\text{,}
    \end{align*}
    where we used the induction assumption, in particular to treat the following top-order terms:
    \begin{gather*}
        \rho\Lie_\rho^k\ms{m}\Rightarrow^{M_0-k-1}0\text{,}\qquad \rho\Lie_\rho^k \ol{\ms{f}}^0\Rightarrow^{M_0-k} 0\text{.}
    \end{gather*}
    However, the following terms will give nontrivial contributions depending on the parity of $k$: 
    \begin{align*}
        &I_1 := \sum\limits_{j+j_0+j_1+\dots + j_\ell = k-1} \mc{S}\paren{\ms{g}; \Lie_\rho^{j_0} \ms{m}, \Lie_\rho^{j_1 -1}\ms{m}, \dots, \Lie_\rho^{j_\ell - 1}\ms{m}, \Lie_\rho^j \ol{\ms{f}}^0}\text{,}\\
        &I_2:=\sum\limits_{\substack{j+j_0+j_1+\dots +j_\ell = k\\j<k, j_p\geq 1}} \mc{S}\paren{\ms{g}; \ms{D}\Lie_\rho^{j_0-1}\ms{m}, \Lie_\rho^{j_1-1} \ms{m}, \dots, \Lie_\rho^{j_\ell -1}\ms{m}, \Lie_\rho^{j}\ol{\ms{f}}^1}\text{.}
    \end{align*}
    We will here give once a precise treatment of such terms, for these will be handled more directly in the remaining part of the proof.

    Let $k$ be odd. Starting with $I_1$, since $k-1$ will be even, there will be two different cases for the partition $\lbrace j, j_0, \dots, j\rbrace$ of $k-1$ even:
    \begin{enumerate}
        \item $j_0$ is even. In this case, since $\Lie_\rho^{j_0}\ms{m} \Rightarrow^{M_0-j_0-1} 0$ by assumption, this term will give a trivial contribution at the boundary.
        \item $j_0$ is odd. In this case, either $j$ is odd, or at least one of the $j_p$'s is, with $p=1,\dots, \ell$. The former implies $\Lie_\rho^j\ol{\ms{f}}^0\Rightarrow^{M_0-j-2}0$, while the latter gives $\Lie_\rho^{j_p-1}\ms{m} \Rightarrow^{M_0-j_p} 0$.
    \end{enumerate}
    A similar analysis holds for $I_2$, where\footnote{Note that this time, $j+j_0+j_1+\dots + j_\ell=k$ odd.}:
    \begin{enumerate}
        \item At least one of the $j_p's$, $p=0, \dots, \ell$ is even, giving a zero-boundary limit. 
        \item All the $j_p$'s are odd. In this case, $j$ is odd and $\Lie_\rho^j\ol{\ms{f}}^1\Rightarrow^{M_0-j-1}0$.
    \end{enumerate}
    As a consequence, for $k$ odd, one has: 
    \begin{equation*}
        I_1, I_2 \Rightarrow^{M_0-k-1}0\text{.}
    \end{equation*}
    On the other hand, for $k$ even, one will be able to find terms with a non-trivial boundary limit and will find, using the induction assumptions: 
    \begin{gather*}
        I_1, I_2 \Rightarrow^{M_0-k-2} \mi{D}^{(k, k)}\paren{\mf{g}^{(0)}, \mf{f}^{1,(0)}}\text{.}
    \end{gather*}
    Such an analysis leads us to the following, using \ref{lemma_item_2} of Proposition \ref{prop_integration}:
    \begin{equation*}
        \Lie_\rho^k \ol{\ms{f}}^0 \Rightarrow^{M_0-k-2} \mi{D}^{(k+1, k+1)}\paren{\mf{g}^{(0)}, \mf{f}^{1,(0)};k}\text{,}\qquad \rho\Lie_\rho^{k+1}\ol{\ms{f}}^0\Rightarrow^{M_0-k-2}0\text{.}
    \end{equation*}
    One can perform a similar analysis for Equation \eqref{higher_transport_f_1}, which contains terms of the form: 
    \begin{align*}
        \rho\cdot \ms{D}_{[a}\Lie_\rho^{(k-1)+1}\ol{\ms{f}}^0_{b]} \Rightarrow^{M_0-k-2} 0\text{,}\qquad \ms{D}_{[a}\Lie_\rho^{k-1}\ol{\ms{f}}^0_{b]} \Rightarrow^{M_0-k-2} \mi{D}^{(k+1, k+1)}\paren{\mf{g}^{(0)}, \mf{f}^{1,(0)}; k+1 }_{ab}\text{.}
    \end{align*}
    The nonlinearities can be analysed as above to give: 
    \begin{align*}
         &\sum\limits_{\substack{j + j_0 +j_1+ \dots + j_\ell = k\\j<k, j_p\geq 1}} \rho\cdot\mc{S}\paren{\ms{g}; \ms{D}\Lie_\rho^{j_0}\ms{m}, \Lie_\rho^{j_1-1}\ms{m}, \dots, \Lie_\rho^{j_\ell-1}\ms{m}, \Lie_\rho^j \ol{\ms{f}}^0}_{ab}\Rightarrow^{M_0-k-1} 0\text{,}\\
         &\sum\limits_{\substack{j + j_0 +j_1+ \dots + j_\ell = k-1\\j<k-1, j_p \geq 1}} \mc{S}\paren{\ms{g}; \ms{D}\Lie_\rho^{j_0-1}\ms{m}, \Lie_\rho^{j_1-1}\ms{m}, \dots, \Lie_\rho^{j_\ell-1}\ms{m}, \Lie_\rho^j \ol{\ms{f}}^0}_{ab}\Rightarrow^{M_0-k-1} \mi{D}^{(k+1, k+1)}\paren{\mf{g}^{(0)}, \mf{f}^{1,(0)}; k+1}_{ab}\text{,}
    \end{align*}
    giving the desired limit\footnote{Note that the anomalous limit is out-of-reach at the moment. In order to derive it, we will need the boundary limits of the metric.}:
    \begin{equation*}
        \Lie_\rho^{k+1}\ol{\ms{f}}^1 \Rightarrow^{M_0-k-3} \mi{D}^{(k+1, k+1)}\paren{\mf{g}^{(0)}, \mf{f}^{1,(0)};k+1}\text{.}
    \end{equation*}

    Next, one can note, using Proposition \ref{higher_order_riemann}: 
    \begin{align}
        &\label{limit_r_k}\Lie_\rho^{k-1} \ms{R} = \sum\limits_{\substack{ j_1 + \dots + j_\ell = k-1\\i_1 + \dots + i_\ell=2, j_p \geq 1}}\mc{S}\paren{\ms{g}; \ms{D}^{i_1}\Lie_\rho^{j_1-1}{\ms{m}}\text{,}\dots, \ms{D}^{i_\ell}\Lie_\rho^{j_\ell-1}{\ms{m}}}\Rightarrow^{M_0-k-1} \left\{\begin{aligned}
            &\mi{D}^{(k+1, k-3)}\paren{\mf{g}^{(0)}, \mf{f}^{1,(0)}; k+1}\text{, } &&k\geq 3\text{,}\\
            &\mi{D}^{k+1}\paren{\mf{g}^{(0)};k+1}\text{, } &&k<3\text{,}
        \end{aligned}\right.\\
        &\label{limit_r_k+1}\rho\Lie_\rho^{k}\ms{R} = \sum\limits_{\substack{ j_1 + \dots + j_\ell = k\\i_1 + \dots + i_\ell=2, j_p \geq 1}}\rho\cdot \mc{S}\paren{\ms{g}; \ms{D}^{i_1}\Lie_\rho^{j_1-1}{\ms{m}}\text{,}\dots, \ms{D}^{i_\ell}\Lie_\rho^{j_\ell-1}{\ms{m}}}\Rightarrow^{M_0-k-1} 0\text{,}
    \end{align}
    and similarly for the Ricci tensor and scalar. 

    Furthermore, note the following limits for the components of the stress-energy tensor: 
    \begin{align}
        \Lie^{k-1}_\rho \ms{T}^0, \Lie_\rho^{k-1} \ms{T}^2 \notag&=\sum\limits_{q=0}^4\sum\limits_{j+j'+j_0+\dots + j_\ell = k-1-q} \rho^{4-q}\cdot \mc{S}\paren{\ms{g}; \Lie_\rho^{j_0-1}\ms{m}, \dots, \Lie_\rho^{j_\ell-1} \ms{m}, \Lie_\rho^j \ol{\ms{f}}^0, \Lie_\rho^{j'} \ol{\ms{f}}^0}\\
        &\qquad+\sum\limits_{q=0}^2\sum\limits_{j+j'+j_0+\dots + j_\ell = k-1-q} \rho^{2-q}\cdot \mc{S}\paren{\ms{g}; \Lie_\rho^{j_0-1}\ms{m}, \dots, \Lie_\rho^{j_\ell-1} \ms{m}, \Lie_\rho^j \ol{\ms{f}}^1, \Lie_\rho^{j'} \ol{\ms{f}}^1}\notag\\
        &\Rightarrow^{M_0-k} \mi{D}^{(k-3, k-3)}\paren{\mf{g}^{(0)}, \mf{f}^{1,(0)};k+1}\text{,}\label{limit_T_k-1}
    \end{align}
    where we used, in particular: 
    \begin{equation*}
        \rho^4 \Lie_\rho^{k-1} \ol{\ms{f}}^0 = \rho^3 \cdot \rho\Lie_\rho^{(k-2)+1}\ol{\ms{f}}^0 \Rightarrow^{M_0-k} 0\text{.}
    \end{equation*}
    Similarly, 
    \begin{align}
        \rho\Lie^{k}_\rho \ms{T}^0, \rho\Lie_\rho^{k} \ms{T}^2 \notag &=\sum\limits_{q=0}^4\sum\limits_{j+j'+j_0+\dots + j_\ell = k-q} \rho^{5-q}\cdot \mc{S}\paren{\ms{g}; \Lie_\rho^{j_0-1}\ms{m}, \dots, \Lie_\rho^{j_\ell-1} \ms{m}, \Lie_\rho^j \ol{\ms{f}}^0, \Lie_\rho^{j'} \ol{\ms{f}}^0}\\
        &\notag \qquad+\sum\limits_{q=0}^2\sum\limits_{j+j'+j_0+\dots + j_\ell = k-q} \rho^{3-q}\cdot \mc{S}\paren{\ms{g}; \Lie_\rho^{j_0-1}\ms{m}, \dots, \Lie_\rho^{j_\ell-1} \ms{m}, \Lie_\rho^j \ol{\ms{f}}^1, \Lie_\rho^{j'} \ol{\ms{f}}^1}\\
        &\Rightarrow^{M_0-k} 0\text{,}\label{limit_T_k}
    \end{align}
    where we used: 
    \begin{equation*}
        \rho^5\Lie_\rho^k \ol{\ms{f}}^0 = \rho^3\cdot \rho^2\Lie_\rho^{(k-2)+2}\ol{\ms{f}}^0 \Rightarrow^{M_0-k} 0\text{.}
    \end{equation*}
    Equations \eqref{limit_r_k}, \eqref{limit_r_k+1}, \eqref{limit_T_k-1} and \eqref{limit_T_k} allows one to obtain the following limits, using \ref{lemma_item_2} of Proposition \ref{prop_integration} applied to Equation \eqref{higher_transport_tr_m}: 
    \begin{equation*}
        \ms{tr}_{\ms{g}}\Lie_\rho^{k}\ms{m}\Rightarrow^{M_0-k-1} \mi{D}^{(k+1, k-3)}\paren{\mf{g}^{(0)}, \mf{f}^{1, (0)};k+1}\text{,}\qquad  \rho\Lie_\rho \paren{\ms{tr}_{\ms{g}}\Lie_\rho^{k}\ms{m}}\Rightarrow^{M_0-k-1} 0\text{,}
    \end{equation*}
    implying, for \eqref{higher_transport_m}:
    \begin{equation*}
        \Lie_\rho^{k}\ms{m}\Rightarrow^{M_0-k-1} \mi{D}^{(k+1, k-3)}\paren{\mf{g}^{(0)}, \mf{f}^{1, (0)};k+1}\text{,}\qquad  \rho\Lie_\rho^{k+1}\ms{m}\Rightarrow^{M_0-k-1} 0\text{.}
    \end{equation*}
    From this, one can obtain the anomalous limit \eqref{higher_limit_f} for $\ol{\ms{f}}^1$, using \eqref{higher_transport_f_1}.

    The limit \eqref{limit_k+2_f_0} can easily be obtained by applying $\rho\Lie_\rho$ to the right-hand side of \eqref{higher_transport_f_0}. 

    Finally, note that the limit \eqref{higher_limit_n-4} can also be obtained from \eqref{higher_limit_m} and \eqref{higher_limit_f} following an identical reasoning as above. 
\end{proof}

\begin{lemma}\label{lemma_critical_f}
    Let $n\geq 4$. There exists $\mf{f}^{0,\dag}$, a tensor field on $\mc{I}$, such that the following limits hold: 
    \begin{align}\label{critical_f_0}
        \rho\Lie_\rho^{n-3}\ol{\ms{f}}^0 \Rightarrow^{M_0-n+2}(n-4)!\mf{f}^{0, \star}\text{,}\qquad \Lie_\rho^{n-4}\ol{\ms{f}}^0 - (n-4)!\mf{f}^{0, \star}\rightarrow^{M_0-n+2} (n-4)! \mf{f}^{0, \dag}\text{,}
    \end{align}
    where:
    \begin{equation*}
        \mf{f}^{0, \star} = \mi{D}^{(n-3, n-3)}\paren{\mf{g}^{(0)}, \mf{f}^{1,(0)}; n}\text{.}
    \end{equation*}
\end{lemma}

\begin{proof}
    Equation \eqref{higher_transport_f_0} reaches criticality when $k=n-4$, where it takes the following form:
    \begin{align*}
        \rho\Lie_\rho^{n-3}\ol{\ms{f}}^0 = -\ms{D}\cdot \Lie_\rho^{n-4} \ol{\ms{f}}^1 &+  \sum\limits_{j+j_0+j_1+\dots + j_\ell = n-4} \rho\cdot\mc{S}\paren{\ms{g}; \Lie_\rho^{j_0} \ms{m}, \Lie_\rho^{j_1 -1}\ms{m}, \dots, \Lie_\rho^{j_\ell - 1}\ms{m}, \Lie_\rho^j \ol{\ms{f}}^0}\\
        &+\sum\limits_{j+j_0+j_1+\dots + j_\ell = n-5} \mc{S}\paren{\ms{g}; \Lie_\rho^{j_0} \ms{m}, \Lie_\rho^{j_1 -1}\ms{m}, \dots, \Lie_\rho^{j_\ell - 1}\ms{m}, \Lie_\rho^j \ol{\ms{f}}^0}\\
        &+\sum\limits_{\substack{j+j_0+j_1+\dots +j_\ell = n-4\\j<n-4, j_p\geq 1}} \mc{S}\paren{\ms{g}; \ms{D}\Lie_\rho^{j_0-1}\ms{m}, \Lie_\rho^{j_1-1} \ms{m}, \dots, \Lie_\rho^{j_\ell -1}\ms{m}, \Lie_\rho^{j}\ol{\ms{f}}^1}\text{,}
    \end{align*}
    and where each individual term has the following boundary limit: 
    \begin{align*}
        &\ms{D}\cdot \Lie_\rho^{n-4} \ol{\ms{f}}^1\Rightarrow^{M_0-n+2}\mi{D}^{(n-3, n-3)}\paren{\mf{g}^{(0)}, \mf{f}^{1,(0)};n}\text{,}\\
        &\sum\limits_{j+j_0+j_1+\dots + j_\ell = n-4} \rho\cdot\mc{S}\paren{\ms{g}; \Lie_\rho^{j_0} \ms{m}, \Lie_\rho^{j_1 -1}\ms{m}, \dots, \Lie_\rho^{j_\ell - 1}\ms{m}, \Lie_\rho^j \ol{\ms{f}}^0}\Rightarrow^{M_0-n+2}0\text{,}\\
        &\sum\limits_{j+j_0+j_1+\dots + j_\ell = n-5} \mc{S}\paren{\ms{g}; \Lie_\rho^{j_0} \ms{m}, \Lie_\rho^{j_1 -1}\ms{m}, \dots, \Lie_\rho^{j_\ell - 1}\ms{m}, \Lie_\rho^j \ol{\ms{f}}^0} \Rightarrow^{M_0-n+2}\mi{D}^{(n-4,n-4)}\paren{\mf{g}^{(0)}, \mf{f}^{1,(0)};n}\text{,}\\
        &\sum\limits_{\substack{j+j_0+j_1+\dots +j_\ell = n-4\\j<n-4, j_p\geq 1}} \mc{S}\paren{\ms{g}; \ms{D}\Lie_\rho^{j_0-1}\ms{m}, \Lie_\rho^{j_1-1} \ms{m}, \dots, \Lie_\rho^{j_\ell -1}\ms{m}, \Lie_\rho^{j}\ol{\ms{f}}^1}\Rightarrow^{M_0-n+2} \mi{D}^{(n-4, n-4)}\paren{\mf{g}^{(0)}, \mf{f}^{1,(0)};n}\text{,}
        \end{align*}
        implying indeed:
        \begin{equation*}
            \rho\Lie_\rho^{n-3}\ol{\ms{f}}^0\Rightarrow^{M_0-n+2}(n-4)!\ol{\mf{f}}^{0,\star}=\mi{D}^{(n-3, n-3)}\paren{\mf{g}^{(0)}, \mf{f}^{1,(0)}; n}\text{.}
        \end{equation*}
        From \ref{lemma_item_4} of Proposition \ref{prop_integration}, there exists a tensor field $\mf{f}^{0, \dag}$ such that \eqref{critical_f_0} is satisfied.
\end{proof}

Since the critical limit for $\ms{m}$ is located at order $n-1$, we will first need to state the limits of $\Lie_\rho^{k}\ms{m}$ with $k<n-1$, as well as to state higher-order anomalous limits for the Maxwell tensors. 
\begin{lemma}\label{lemma_sub_m}
    Let $n\geq 3$. The following limits hold, for $0\leq k<n-1$:
    \begin{align}\label{subcritical_m}
        \Lie_\rho^k \ms{m} \Rightarrow^{M_0-k-1}\mi{D}^{(k+1, k-3)}\paren{\mf{g}^{(0)}, \mf{f}^{1,(0)}; k+1}\text{,}\qquad \rho\Lie_\rho^{k+1}\ms{m} \Rightarrow^{M_0-k-1} 0\text{.}
    \end{align}
    Furthermore, the Maxwell tensors satisfy the following higher-order anomalous limits\footnote{The regularity at which those limits are satisfied can in fact be improved; see \cite{thesis_Guisset}, Chapter 9.}:
    \begin{gather}
        \rho^{2} \Lie_{\rho}^{n-2}\ol{\ms{f}}^0\Rightarrow^{M_0-n+2} \mi{D}^{(n-3, n-3)}\paren{\mf{g}^{(0)}, \mf{f}^{1,(0)}; n}\text{,}\qquad \rho^{2}\Lie_\rho^{n-2}\ol{\ms{f}}^1\Rightarrow^{M_0-n+2}0\text{,}\label{special_limit_f}\\
        \rho^{3} \Lie_{\rho}^{n-1}\ol{\ms{f}}^0\Rightarrow^{M_0-n+1} \mi{D}^{(n-3, n-3)}\paren{\mf{g}^{(0)}, \mf{f}^{1,(0)}; n}\text{,}\qquad \rho^{3}\Lie_\rho^{n-1}\ol{\ms{f}}^1\Rightarrow^{M_0-n+1}0\text{,}\\
        \rho^4 \Lie_\rho^{n}\ol{\ms{f}}^0 \Rightarrow^{M_0-n} \mi{D}^{(n-3, n-3)}\paren{\mf{g}^{(0)}, \mf{f}^{1,(0)}; n}\text{,}\qquad \rho^4\Lie_\rho^{n}\ol{\ms{f}}^1\Rightarrow^{M_0-n} 0\text{.}\label{limits_rho4_rhon}
    \end{gather}
\end{lemma}

\begin{proof}
    The limits \eqref{subcritical_m} have been obtained for $k<n-3$ in Lemma \ref{lemma_higher_limits}. Let us thus fix $k\in \lbrace n-3, n-2\rbrace$. 

    The analysis of terms involving the metric in \eqref{higher_transport_m} and \eqref{higher_transport_tr_m} will be identical to the proof of Lemma \ref{lemma_higher_limits} and we will therefore not cover them here. The only nontrivial limits will involve the stress-energy tensor terms, here for $n\geq 4$: 
    \begin{align}
        \Lie^{k-1}_\rho \ms{T}^0, \Lie_\rho^{k-1} \ms{T}^2 \notag&=\sum\limits_{q=0}^4\sum\limits_{j+j'+j_0+\dots + j_\ell = k-1-q} \rho^{4-q}\cdot \mc{S}\paren{\ms{g}; \Lie_\rho^{j_0-1}\ms{m}, \dots, \Lie_\rho^{j_\ell-1} \ms{m}, \Lie_\rho^j \ol{\ms{f}}^0, \Lie_\rho^{j'} \ol{\ms{f}}^0}\\
        &\qquad+\sum\limits_{q=0}^2\sum\limits_{j+j'+j_0+\dots + j_\ell = k-1-q} \rho^{2-q}\cdot \mc{S}\paren{\ms{g}; \Lie_\rho^{j_0-1}\ms{m}, \dots, \Lie_\rho^{j_\ell-1} \ms{m}, \Lie_\rho^j \ol{\ms{f}}^1, \Lie_\rho^{j'} \ol{\ms{f}}^1}\label{intermediate_stress_energy_k-1}\text{,}
    \end{align}
    as well as: 
    \begin{align}
        \rho\Lie^{k}_\rho \ms{T}^0, \rho\Lie_\rho^{k} \ms{T}^2 \notag &=\sum\limits_{q=0}^4\sum\limits_{j+j'+j_0+\dots + j_\ell = k-q} \rho^{5-q}\cdot \mc{S}\paren{\ms{g}; \Lie_\rho^{j_0-1}\ms{m}, \dots, \Lie_\rho^{j_\ell-1} \ms{m}, \Lie_\rho^j \ol{\ms{f}}^0, \Lie_\rho^{j'} \ol{\ms{f}}^0}\\
        & \qquad+\sum\limits_{q=0}^2\sum\limits_{j+j'+j_0+\dots + j_\ell = k-q} \rho^{3-q}\cdot \mc{S}\paren{\ms{g}; \Lie_\rho^{j_0-1}\ms{m}, \dots, \Lie_\rho^{j_\ell-1} \ms{m}, \Lie_\rho^j \ol{\ms{f}}^1, \Lie_\rho^{j'} \ol{\ms{f}}^1}\label{intermediate_stress_energy_k}\text{.}
    \end{align}
    Observe that such terms will involve limits which have not been covered from Lemma \ref{lemma_higher_limits} nor \ref{lemma_critical_f}, such as, for $n\geq 4$: 
    \begin{equation*}
        \sum\limits_{q=0}^4 \mc{S}\paren{\ms{g}; \rho^{5-q}\Lie^{k-q}_\rho\ol{\ms{f}}^0, \ol{\ms{f}}^0}\text{,}\qquad \sum\limits_{q=0}^2 \mc{S}\paren{\ms{g}; \rho^{3-q}\Lie^{k-q}_\rho\ol{\ms{f}}^1, \ol{\ms{f}}^1}\text{.}
    \end{equation*}
    In order to treat such terms, we will need to use the following hierarchy:
    \begin{equation*}
        \Lie_\rho^{n-3}\ms{m}\text{,}\, \rho \Lie_\rho^{n-2}\ms{m}\text{, }\rho^2\Lie_\rho^{n-2}\ol{\ms{f}}^0\text{, }\rho^2\Lie_\rho^{n-2}\ol{\ms{f}}^1\rightarrow \Lie_\rho^{n-2}\ms{m}\text{,}\, \rho\Lie_\rho^{n-1}\ms{m}\text{, } \rho^3\Lie_\rho^{n-1}\ol{\ms{f}}^0\text{, } \rho^3\Lie_\rho^{n-1}\ol{\ms{f}}^1\text{.}
    \end{equation*}
    
    Let first $n\geq 4$; the case $n=3$ follows an identical reasoning. Observe first that Equations \eqref{intermediate_stress_energy_k-1} and \eqref{intermediate_stress_energy_k}, with $k=n-3$, yield the following limits, from Lemmas \ref{lemma_higher_limits} and \ref{lemma_critical_f}:
    \begin{align*}
        &\Lie_\rho^{n-4}\ms{T}^0\text{,}\; \Lie_\rho^{n-4}\ms{T}^2\Rightarrow^{M_0-n+2} \left\{\begin{aligned}
            &\mi{D}^{(n-6, n-6)}\paren{\mf{g}^{(0)}, \mf{f}^{1,(0)}; n}\text{,}\qquad &&n\geq 6\text{,}\\
            &0\text{,}\qquad && n=4,5\text{,}
        \end{aligned}\right.\\
        &\rho\Lie_\rho^{n-3}\ms{T}^0\text{,}\; \rho\Lie_\rho^{n-3}\ms{T}^2\Rightarrow^{M_0-n+2}0\text{.}
    \end{align*}
    As a consequence, one finds immediately: 
    \begin{equation*}
        \Lie_\rho^{n-3}\ms{m}\Rightarrow^{M_0-n+2} \left\{
        \begin{aligned}
        &\mi{D}^{(n-2, n-6)}\paren{\mf{g}^{(0)}, \mf{f}^{1,(0)}; n}\text{,}\qquad &&n\geq 6\text{,}\\
        &\mi{D}^{n-2}\paren{\mf{g}^{(0)}; n+1}\text{,}\qquad &&n=4, 5\text{,}
        \end{aligned}\right.\qquad \rho\Lie_\rho^{n-2}\ms{m}\Rightarrow^{M_0-n+2}0\text{.}
    \end{equation*}
    Next, one can apply the operator $\rho\Lie_\rho$ to \eqref{higher_transport_f_0}, with $k=n-4$, to obtain: 
    \begin{align}
        &\begin{aligned}\label{higher_transport_f_0_intermediate}
            \rho^2\Lie_\rho^{n-2}\ol{\ms{f}}^0  &= - \rho\Lie_\rho^{n-3}\ol{\ms{f}}^0 -\ms{D}\cdot \rho\Lie_\rho^{n-3} \ol{\ms{f}}^1 +  \sum\limits_{j+j_0+j_1+\dots + j_\ell = n-3} \rho^2\cdot\mc{S}\paren{\ms{g}; \Lie_\rho^{j_0} \ms{m}, \Lie_\rho^{j_1 -1}\ms{m}, \dots, \Lie_\rho^{j_\ell - 1}\ms{m}, \Lie_\rho^j \ol{\ms{f}}^0}\\
            & \qquad +\sum\limits_{j+j_0+j_1+\dots + j_\ell = n-4} \rho\cdot \mc{S}\paren{\ms{g}; \Lie_\rho^{j_0} \ms{m}, \Lie_\rho^{j_1 -1}\ms{m}, \dots, \Lie_\rho^{j_\ell - 1}\ms{m}, \Lie_\rho^j \ol{\ms{f}}^0}\\
            & \qquad +\sum\limits_{\substack{j+j_0+j_1+\dots +j_\ell = n-3\\j<n-3, j_p\geq 1}} \rho\cdot \mc{S}\paren{\ms{g}; \ms{D}\Lie_\rho^{j_0-1}\ms{m}, \Lie_\rho^{j_1-1} \ms{m}, \dots, \Lie_\rho^{j_\ell -1}\ms{m}, \Lie_\rho^{j}\ol{\ms{f}}^1}\text{,}\\
            &\Rightarrow^{M_0-n+2} \mi{D}^{(n-3, n-3)}\paren{\mf{g}^{(0)}, \mf{f}^{1,(0)}; n} &&\text{.}
        \end{aligned}
    \end{align}
    Similarly, one can obtain the following limits: 
    \begin{align}\label{intermediate_critical_rho23}
        \rho^2 \Lie_\rho^{n-2}\ol{\ms{f}}^1 \Rightarrow^{M_0-n+2} 0\text{,}\qquad  \rho^3 \Lie_\rho^{n-1} \ol{\ms{f}}^0 \Rightarrow^{M_0-n+1}\mi{D}^{(n-3, n-3)}\paren{\mf{g}^{(0)},\mf{f}^{1,(0)}; n}\text{,}\qquad \rho^3 \Lie_\rho^{n-1} \ol{\ms{f}}^1 \Rightarrow^{M_0-n+1} 0\text{.}
    \end{align}
    This yields the following limits for the metric \footnote{Although we skipped this step for clarity, one needs to first compute the limit for $\ms{tr}_{\ms{g}}\Lie_\rho^{n-2}\ms{m}$.}: 
    \begin{equation*}
        \Lie_\rho^{n-2}\ms{m}\Rightarrow^{M_0-n+1} \left\{\begin{aligned}
            &\mi{D}^{(n-1, n-5)}\paren{\mf{g}^{(0)}, \mf{f}^{1,(0)}; n+1}\text{,}\qquad &&n\geq 5\text{,}\\
            &\mi{D}^{n-1}\paren{\mf{g}^{(0)}; n+1}\text{,}\qquad &&n=4\text{,}
        \end{aligned}\right.\qquad \rho\Lie_\rho^{n-1}\ms{m}\Rightarrow^{M_0-n+1} 0\text{.}
    \end{equation*}
    The final limits \eqref{limits_rho4_rhon} follow from the above and \eqref{intermediate_critical_rho23}. 
\end{proof}
The following lemma gives the final critical limit for the metric, as well as an anomalous limit for $\ol{\ms{f}}^1$, which completes the proof of Theorem \ref{theorem_main_fg}, for $n\geq 3$: 
\begin{lemma}
    Let $n\geq 3$. There exists $\mf{g}^\dag$, a tensor field on $\mc{I}$ such that the following limits hold for $\ms{m}$: 
    \begin{gather}
        \label{anomalous_m}\rho\Lie_\rho^{n} \ms{m}\Rightarrow^{M_0-n} n! \mf{g}^\star = \mi{D}^{(n,n-4)}\paren{\mf{g}^{(0)}, \mf{f}^{1,(0)}; n}\text{,}\\
        \label{critical_m}\Lie_\rho^{n-1}\ms{m} - n!\, \log\rho \cdot \mf{g}^\star \rightarrow^{M_0-n} n!\, \mf{g}^{\dag}\text{.}
    \end{gather}
    Furthermore, the following anomalous limits hold for $\ol{\ms{f}}^1$: 
    \begin{gather*}
        \rho\Lie_\rho^{n-1}\ol{\ms{f}}^1 \Rightarrow^{M_0-n+1}(n-3)!\, \mf{f}^{1,\star}=\mi{D}^{(n-2, n-2)}\paren{\mf{g}^{(0)}, \mf{f}^{1,(0)}; n}\text{,}\\
        \Lie_\rho^{n-2}\ol{\ms{f}}^1 - (n-3)!\log\rho\cdot \mf{f}^{1,\star}\rightarrow^{M_0-n+1} \mi{D}^{(n-2, n-2, 1)}\paren{\mf{g}^{(0)}, \mf{f}^{1,(0)},\mf{f}^{0, \dag}; n}\text{.}
    \end{gather*}
\end{lemma}

\begin{proof}
    The limits \eqref{anomalous_m} and \eqref{critical_m} are straightforward consequences of the following limits for the components of the stress-energy tensor, here for $n\geq 4$: 
    \begin{align*}
        \Lie_\rho^{n-2}\ms{T}^0\text{, }\Lie_\rho^{n-2}\ms{T}^2 &=\sum\limits_{q=0}^4\sum\limits_{j+j'+j_0+\dots + j_\ell = n-2-q} \rho^{4-q}\cdot \mc{S}\paren{\ms{g}; \Lie_\rho^{j_0-1}\ms{m}, \dots, \Lie_\rho^{j_\ell-1} \ms{m}, \Lie_\rho^j \ol{\ms{f}}^0, \Lie_\rho^{j'} \ol{\ms{f}}^0}\\
        &+\sum\limits_{q=0}^2\sum\limits_{j+j'+j_0+\dots + j_\ell = n-2-q} \rho^{2-q}\cdot \mc{S}\paren{\ms{g}; \Lie_\rho^{j_0-1}\ms{m}, \dots, \Lie_\rho^{j_\ell-1} \ms{m}, \Lie_\rho^j \ol{\ms{f}}^1, \Lie_\rho^{j'} \ol{\ms{f}}^1}\\
        &\Rightarrow^{M_0-n+1} \mi{D}^{(n-4, n-4)}\paren{\mf{g}^{(0)}, \mf{f}^{1,(0)}; n}\text{,}
    \end{align*}
    as well as: 
    \begin{align*}
        \rho\Lie_\rho^{n-1}\ms{T}^0\text{, }\rho\Lie_\rho^{n-1}\ms{T}^2 &=\sum\limits_{q=0}^4\sum\limits_{j+j'+j_0+\dots + j_\ell = n-1-q} \rho^{5-q}\cdot \mc{S}\paren{\ms{g}; \Lie_\rho^{j_0-1}\ms{m}, \dots, \Lie_\rho^{j_\ell-1} \ms{m}, \Lie_\rho^j \ol{\ms{f}}^0, \Lie_\rho^{j'} \ol{\ms{f}}^0}\\
        &+\sum\limits_{q=0}^2\sum\limits_{j+j'+j_0+\dots + j_\ell = n-1-q} \rho^{3-q}\cdot \mc{S}\paren{\ms{g}; \Lie_\rho^{j_0-1}\ms{m}, \dots, \Lie_\rho^{j_\ell-1} \ms{m}, \Lie_\rho^j \ol{\ms{f}}^1, \Lie_\rho^{j'} \ol{\ms{f}}^1}\\
        &\Rightarrow^{M_0-n+1}0\text{,}
    \end{align*}
    where we used the limits from Lemmas \ref{lemma_higher_limits}, \ref{lemma_critical_f} and \ref{lemma_sub_m}. These imply, using \eqref{higher_transport_tr_m} with $k=n-1$ and \ref{lemma_item_2} of Proposition \ref{prop_integration}: 
    \begin{gather*}
        \ms{tr}_{\ms{g}}\Lie_\rho^{n-1}\ms{m}\Rightarrow^{M_0-n} \mi{D}^{(n, n-4)}\paren{\mf{g}^{(0)}, \mf{f}^{1,(0)}; n}\text{,}\qquad \rho\Lie_\rho \paren{\ms{tr}_{\ms{g}}\Lie_\rho^{n-1}\ms{m}}\Rightarrow^{M_0-n}0\text{.}
    \end{gather*}
    This allows us to compute the following anomalous limit, using Equation \eqref{higher_transport_m}, with $k=n-1$: 
    \begin{equation*}
        \rho\Lie_\rho^{n}\ms{m}\Rightarrow^{M_0-n} n!\, \mf{g}^\star = \mi{D}^{(n, n-4)}\paren{\mf{g}^{(0)}, \mf{f}^{1,(0)}; n}\text{,}
    \end{equation*}
    implying, from \ref{lemma_item_4} of Proposition \ref{prop_integration}, the existence of a tensor field $\mf{g}^\dag$ on $\mc{I}$ satisfying: 
    \begin{equation*}
        \Lie_\rho^{n-1}\ms{m} - n!\log\rho \cdot \mf{g}^\star \rightarrow^{M_0-n} n! \, \mf{g}^\dag\text{.}
    \end{equation*}
    Finally, the anomalous limit of $\ol{\ms{f}}^1$ is obtained first from \eqref{higher_transport_f_1}, multiplied by $\rho$, with $k=n-2$:
    \begin{align*}
        \rho\Lie_\rho^{n-1}\ol{\ms{f}}^1_{ab} &= 2\ms{D}_{[a}(\rho^2\Lie_\rho^{n-2} \ol{\ms{f}}^0)_{b]} + 2(n-2) \ms{D}_{[a}(\rho\Lie_\rho^{n-3}\ol{\ms{f}}^0)_{b]} \\
        &\qquad +  \sum\limits_{\substack{j + j_0 +j_1+ \dots + j_\ell = n-2\\j<n-2, j_p\geq 1}} \rho^2\cdot\mc{S}\paren{\ms{g}; \ms{D}\Lie_\rho^{j_0-1}\ms{m}, \Lie_\rho^{j_1-1}\ms{m}, \dots, \Lie_\rho^{j_\ell-1}\ms{m}, \Lie_\rho^j \ol{\ms{f}}^0}_{ab}\\
        &\qquad +\sum\limits_{\substack{j + j_0 +j_1+ \dots + j_\ell = n-3\\j<n-3, j_p \geq 1}} \rho\cdot \mc{S}\paren{\ms{g}; \ms{D}\Lie_\rho^{j_0-1}\ms{m}, \Lie_\rho^{j_1-1}\ms{m}, \dots, \Lie_\rho^{j_\ell-1}\ms{m}, \Lie_\rho^j \ol{\ms{f}}^0}_{ab}\\
        &\Rightarrow^{M_0-n+1} \mi{D}^{(n-2, n-2)}\paren{\mf{g}^{(0)}, \mf{f}^{1,(0)}; n}_{ab}\text{, }
    \end{align*}
    where we used, in particular, \eqref{special_limit_f}. 
    Next, from \eqref{limits_DA}, the following critical limit holds: 
    \begin{equation*}
        \begin{aligned}
        &\Lie_\rho^{n-2}\ol{\ms{f}}^1_{ab}-2(n-3)!\log\rho\cdot \mf{D}_{[a}\mf{f}^{0, \star}_{b]}\\
        &\qquad = 2\ms{D}_{[a}(\rho\Lie_\rho^{n-3} \ol{\ms{f}}^0)_{b]} + 2(n-3) \underbrace{\paren{\ms{D}_{[a}(\Lie_\rho^{n-4}\ol{\ms{f}}^0)_{b]}-(n-4)!\log\rho\cdot \mf{D}_{[a}\mf{f}^{0, \star}_{b]}}}_{\rightarrow^{M_0-n+1}(n-4)!\mf{D}_{[a}\mf{f}^{0, \dag}_{b]}} \\
        &\qquad  \;+  \sum\limits_{\substack{j + j_0 +j_1+ \dots + j_\ell = n-3\\j<n-3, j_p\geq 1}} \rho\cdot\mc{S}\paren{\ms{g}; \ms{D}\Lie_\rho^{j_0-1}\ms{m}, \Lie_\rho^{j_1-1}\ms{m}, \dots, \Lie_\rho^{j_\ell-1}\ms{m}, \Lie_\rho^j \ol{\ms{f}}^0}_{ab}\\
        &\qquad \; +\sum\limits_{\substack{j + j_0 +j_1+ \dots + j_\ell = n-4\\j<n-4, j_p \geq 1}} \mc{S}\paren{\ms{g}; \ms{D}\Lie_\rho^{j_0-1}\ms{m}, \Lie_\rho^{j_1-1}\ms{m}, \dots, \Lie_\rho^{j_\ell-1}\ms{m}, \Lie_\rho^j \ol{\ms{f}}^0}_{ab}\\
        &\qquad \rightarrow^{M_0-n+1} \mi{D}^{(n-2, n-2, 1)}\paren{\mf{g}^{(0)}, \mf{f}^{1,(0)}, \mf{f}^{0, \dag}; n}_{ab}\text{.}
        \end{aligned}
    \end{equation*}
    Defining now: 
    \begin{equation*}
        \mf{f}^{1,\star}_{ab} := 2\mf{D}_{[a}\mf{f}^{0, \star}_{b]}
    \end{equation*}
    allows one to conclude for $n\geq 4$. 

    It remains to show the $n=3$ case, which reduces to showing: 
    \begin{equation}
        \Lie_\rho\ms{T}^0, \Lie_\rho\ms{T}^2, \rho\Lie_\rho^2 \ms{T}^0, \rho\Lie_\rho^2 \ms{T}^2 \Rightarrow^{M_0-3} 0\text{.}\label{limits_T}
    \end{equation}
    In order to obtain this, observe first that from \eqref{transport_f_0} and \eqref{transport_f_1}: 
    \begin{align*}
        &\Lie_\rho\ol{\ms{f}}^0 = \mc{S}\paren{\ms{g}; \ms{m}, \ol{\ms{f}}^0} + \mc{S}\paren{\ms{g}; \ms{D}\ol{\ms{f}}^1}\Rightarrow^{M_0-2}\mi{D}^{(1,1)}\paren{\mf{g}^{(0)}, \mf{f}^{1,(0)}}\\
        &\Lie_\rho \ol{\ms{f}}^1 = \mc{S}\paren{\ms{D}\ol{\ms{f}}^0}\Rightarrow^{M_0-2}\mi{D}^{(1,1)}\paren{\mf{g}^{(0)}, \mf{f}^{0,(0)}}\text{.}
    \end{align*}
    Furthermore, differentiating both \eqref{transport_f_0} and \eqref{transport_f_1} leads to: 
    \begin{align*}
        \Lie_\rho^2\ol{\ms{f}}^0 \Rightarrow^{M_0-3}\mi{D}^{(2, 2, 0)}\paren{\mf{g}^{(0)}, \mf{f}^{1,(0)}, \mf{f}^{0,(0)}}\text{,}\qquad  \Lie_\rho^2\ol{\ms{f}}^1 \Rightarrow^{M_0-3}\mi{D}^{(2, 0, 2)}\paren{\mf{g}^{(0)}, \mf{f}^{1,(0)}, \mf{f}^{0,(0)}}\text{.}
    \end{align*}
    Since in $ \Lie_\rho\ms{T}^0, \Lie_\rho\ms{T}^2, \rho\Lie_\rho^2 \ms{T}^0, \rho\Lie_\rho^2 \ms{T}^2$, the above limits are multiplied by a non-zero power of $\rho$, the limits \eqref{limits_T} are indeed trivial. This allows us to obtain the following limit for the trace of $\ms{m}$, from \eqref{higher_transport_tr_m}: 
    \begin{equation*}
         \ms{tr}_{\ms{g}}\Lie_\rho^2\ms{m}\Rightarrow^{M_0-3}0\text{,}\qquad \rho \Lie_\rho \paren{\ms{tr}_{\ms{g}}\Lie_\rho^2\ms{m}}\Rightarrow^{M_0-3}0\text{,}
    \end{equation*}
    yielding for the metric, using Equation \ref{higher_transport_m}: 
    \begin{equation*}
        \rho\Lie_\rho^3\ms{m}\Rightarrow^{M_0-3}0\text{,}\qquad \Lie_\rho^2\ms{m}\rightarrow^{M_0-3}3!\mf{g}^{\dag}\text{,}
    \end{equation*}
    with $\mf{g}^{\dag}$ some tensor field on $\mc{I}$. 
\end{proof}

\proofpart{$n=2$ case.}

It remains to compute the limits in Theorem \ref{theorem_main_fg} for the $n=2$ case.
For completeness, let us write the transport equations of interest in this setting: 
\begin{align}
    &\begin{aligned}\label{transport_m_n=2}
        \rho\Lie_\rho \ms{m} - \ms{m} = 2\rho\ms{Rc} + \ms{tr}_{\ms{g}}\ms{m}\cdot \ms{g} + \rho\mc{S}\paren{\ms{g}; \ms{m}^2}+\rho\mc{S}\paren{\ms{g}; (\ol{\ms{f}}^0)^2} + \rho^3\mc{S}\paren{\ms{g}; (\ol{\ms{f}}^1)^2}\text{,}
    \end{aligned}\\
    &\begin{aligned}\label{transport_tr_m_n=2}
        \rho\Lie_\rho \ms{tr}_{\ms{g}}\ms{m} - 3\ms{tr}_{\ms{g}}\ms{m} = 2\rho\ms{Rs} + \rho\mc{S}\paren{\ms{g}; \ms{m}^2}+\rho\mc{S}\paren{\ms{g}; (\ol{\ms{f}}^0)^2} + \rho^3\mc{S}\paren{\ms{g}; (\ol{\ms{f}}^1)^2}\text{,}
    \end{aligned}\\
    &\begin{aligned}\label{transport_f0_n=2}
        \Lie_\rho\ol{\ms{f}}^0 = \mc{S}\paren{\ms{g}; \ms{m}, \ol{\ms{f}}^0} + \rho\mc{S}\paren{\ms{g}; \ms{D}\ol{\ms{f}}^1} \text{,}
    \end{aligned}\\
    &\begin{aligned}\label{transport_f1_n=2}
        \rho\Lie_\rho\ol{\ms{f}}^1_{ab} = \mc{S}\paren{\ms{D}\ol{\ms{f}}^0}_{ab}\text{.}
    \end{aligned}
\end{align}
Observe that, for any $\delta$, Equation \eqref{transport_f1_n=2} can also be written as: 
\begin{equation}
    \rho\Lie_\rho(\rho^\delta\cdot \ol{\ms{f}}^1) - \delta \rho^\delta \cdot \ol{\ms{f}}^1 = \rho^\delta \mc{S}\paren{\ms{D}\ol{\ms{f}}^0}\text{.}\label{modified_transport_f_n=2}
\end{equation}
In particular, if $\delta>0$, one can use \ref{local_estimate} of Proposition \ref{prop_integration}. 

First, let us define useful uniform bounds. 

\begin{lemma}\label{lemma_unif_bounds_n=2}
    Let $n=2$ and let $(U, \varphi)$ be a compact coordinate system on $\mc{I}$. The following uniform bounds are satisfied \footnote{Note that, by assumption, $\norm{\ol{\ms{f}}^0}_{M_0+1, \varphi}\lesssim 1$.}: 
        \begin{align*}
            \norm{\ms{m}}_{M_0,\varphi} + \norm{\rho^\delta\cdot \ol{\ms{f}}^1}_{M_0, \varphi}\lesssim 1\text{,}
        \end{align*}
        for any $0<\delta<1$.
\end{lemma}
\begin{proof}
    The proof will follow the same reasoning as in Lemma \ref{lemma_first_bounds}. Using \ref{local_estimate} of Proposition \ref{prop_integration} for \eqref{transport_m_n=2}, \eqref{transport_tr_m_n=2} and \eqref{modified_transport_f_n=2} to obtain:
    \begin{align}
        &\label{intermediate_m_n=2}\abs{\ms{m}}_{0, \varphi} \lesssim 1 + \rho \int_{\rho_0}^\rho \sigma^{-2}\paren{\abs{\ms{tr}_{\ms{g}}\ms{m}}_{0, \varphi} + \sigma \abs{\ms{m}}^2_{0, \varphi} + \sigma^3 \abs{\ol{\ms{f}}^1}_{0, \varphi}^2}\vert_\sigma d\sigma\text{,}\\
        &\label{intermediate_tr_m_n=2}\abs{\ms{tr}_{\ms{g}}\ms{m}}_{0, \varphi} \lesssim 1 + \int_{\rho_0}^\rho \paren{ \abs{\ms{m}}^2_{0, \varphi} + \sigma^2 \abs{\ol{\ms{f}}^1}_{0, \varphi}^2}\vert_\sigma d\sigma\text{,}\\
        &\label{intermediate_f_n=2}\abs{\rho^\delta\cdot \ol{\ms{f}}^1}_{0, \varphi}\lesssim 1 + \rho^\delta \int_{\rho_0}^\rho  \sigma^{-1} d\sigma\text{,}
    \end{align}
    where we used:
    \begin{equation*}
        \norm{\ms{Rc}}_{M_0, \varphi} + \norm{\ms{Rs}}_{M_0, \varphi} + \norm{\ol{\ms{f}}^0}_{M_0+1, \varphi}\lesssim 1\text{.}
    \end{equation*}
    Note that \eqref{intermediate_f_n=2} immediately implies the uniform boundedness of $\ol{\ms{f}}^1$ in $C^0$. From there, one can use the same arguments in the proof of Lemma \ref{lemma_first_bounds} to obtain the boundedness of $\ms{m}$ and $\ms{tr}_{\ms{g}}\ms{m}$ in $C^0$. 

    Let us now assume that the following bounds hold, for some fixed $k\leq M_0$: 
    \begin{equation*}
        \norm{\ms{m}}_{k-1, \varphi} + \norm{\rho^\delta\ol{\ms{f}}^1}_{k-1, \varphi}\lesssim 1\text{.}
    \end{equation*}
    From \ref{local_estimate} of Proposition \ref{prop_integration}, the following bounds hold: 
    \begin{align*}
        &\begin{aligned}
            \abs{\ms{m}}_{k,\varphi} \lesssim 1 + \rho\int_{\rho}^{\rho_0} \sigma^{-2}[\abs{\ms{tr}_{\ms{g}}\ms{m}}_{k, \varphi} + \sigma\abs{\ms{m}}_{k, \varphi} + \sigma^{3-2\delta} |\sigma^\delta\cdot{\ol{\ms{f}}^1}|_{k,\varphi}]\vert_\sigma d\sigma
        \end{aligned}\\
        &\begin{aligned}
            \abs{\ms{tr}_{\ms{g}}\ms{m}}_{k,\varphi} \lesssim 1 + \int_{\rho}^{\rho_0} [\abs{\ms{m}}_{k, \varphi} + \sigma^{2-2\delta} |\sigma^\delta\cdot{\ol{\ms{f}}^1}|_{k,\varphi}]\vert_\sigma d\sigma
        \end{aligned}\\
        &\begin{aligned}
            \abs{\rho^\delta\ol{\ms{f}}^1}_{k, \varphi}\lesssim 1 + \rho^\delta \int_\rho^{\rho_0}\sigma^{-1}\norm{\ol{\ms{f}}^0}_{k+1, \varphi}\vert_\sigma d\sigma\lesssim 1\text{,}
        \end{aligned}
    \end{align*}
    where we used $\norm{\ol{\ms{f}}^0}_{M_0+1, \varphi}\lesssim 1$. Using the same techniques as in the proof of Lemma \ref{lemma_first_bounds}, one obtains indeed: 
    \begin{equation*}
        \norm{\ms{m}}_{k, \varphi} + \norm{\ol{\ms{f}}^1}_{k, \varphi}\lesssim 1\text{, }
    \end{equation*}
    proving the lemma. 
\end{proof}
\begin{lemma}
    Let $n=2$ and $(U, \varphi)$ be a compact coordinate chart on $\mc{I}$. The following limits hold: 
    \begin{itemize}
        \item For the metric: 
        \begin{gather}
            \ms{g}\Rightarrow^{M_0}\mf{g}^{(0)}\text{,}\qquad \ms{g}^{-1}\Rightarrow^{M_0}\paren{\mf{g}^{(0)}}^{-1}\text{,}\label{limit_g_n=2}\\
            \ms{m}\Rightarrow^{M_0-1} 0\text{,}\label{limit_m_n=2}\\
            \rho\Lie_\rho^2\ms{m}\Rightarrow^{M_0-2}2\mf{g}^\star = \mi{D}^{(2, 0)}\paren{\mf{g}^{(0)}, \mf{f}^{0, (0)}}\text{,}\qquad \Lie_\rho\ms{m} - 2\log\rho \cdot \mf{g}^\star \rightarrow^{M_0-2}2\mf{g}^\dag\text{,}
        \end{gather}
        with $\mf{g}^\dag$ a tensor field on $\mc{I}$. 
        \item For the Maxwell fields: 
        \begin{gather}
            \ol{\ms{f}}^0 \Rightarrow^{M_0-1} \mf{f}^{0, (0)}\text{,}\label{limit_f0_n=2}\\
            \rho\Lie_\rho\ol{\ms{f}}^1\Rightarrow^{M_0-2} \mi{D}^{(1,1)}\paren{\mf{g}^{(0)}, \mf{f}^{0,(0)}}\text{,}\qquad \ol{\ms{f}}^1 - \log \rho\cdot \mf{f}^{1,\star}\rightarrow^{M_0-2}\mf{f}^{1,\dag}\label{limit_f1_n=2}\text{.}
        \end{gather}
    \end{itemize}
\end{lemma}
\begin{proof}
    The limits \eqref{limit_g_n=2} and \eqref{limit_f0_n=2} are straightforward consequences of the fundamental theorem of calculus and the bounds from Lemma \ref{lemma_unif_bounds_n=2}: 
    \begin{align*}
        &\begin{aligned}
            \sup\limits_{\lbrace \rho \rbrace \times U}\abs{\ms{g}- \mf{g}^{(0)}}_{M_0, \varphi} &\lesssim \sup\limits_U\int_0^{\rho} \norm{\ms{m}}_{M_0, \varphi}\vert_\sigma d\sigma \rightarrow 0\text{,}\\
        \end{aligned}\\
        &\begin{aligned}
            \sup\limits_U\int_0^{\rho_0}\sigma^{-1}\abs{\ms{g}- \mf{g}^{(0)}}_{M_0, \varphi}\vert_\sigma d\sigma \lesssim \int_0^{\rho_0}\sigma^{-1}\int_0^\sigma \norm{\ms{m}}_{M_0, \varphi}\vert_\tau d\tau  d\sigma<\infty \text{,}
        \end{aligned}
    \end{align*}
    as well as
    \begin{align*}
        &\begin{aligned}
            \sup\limits_{\lbrace \rho \rbrace \times U}\abs{\ol{\ms{f}}^0- \mf{f}^{0, (0)}}_{M_0-1, \varphi} &\lesssim \sup\limits_U\int_0^{\rho} \left[\norm{\ms{m}}_{M_0-1, \varphi} + \norm{\ol{\ms{f}}^0}_{M_0-1, \varphi} + \sigma^{1-\delta}\norm{\rho^\delta\cdot \ol{\ms{f}}^1}_{M_0, \varphi}\right]\vert_\sigma d\sigma \rightarrow 0\text{,}\\
        \end{aligned}\\
        &\begin{aligned}
            \sup\limits_U\int_0^{\rho_0}\sigma^{-1}\abs{\ol{\ms{f}}^0- \mf{f}^{0, (0)}}_{M_0, \varphi}\vert_\sigma d\sigma \lesssim \int_0^{\rho_0}\sigma^{-1}\int_0^\sigma \left[\norm{\ms{m}}_{M_0-1, \varphi} + \norm{\ol{\ms{f}}^0}_{M_0-1, \varphi} + \tau^{1-\delta}\norm{\rho^\delta\cdot \ol{\ms{f}}^1}_{M_0, \varphi}\right]\vert_\tau d\tau  d\sigma<\infty \text{.}
        \end{aligned}
    \end{align*}
    These two limits immediately imply, using Proposition \ref{prop_bounds_metric}, Equations \eqref{transport_tr_m_n=2}, \eqref{transport_m_n=2} and \eqref{transport_f1_n=2}: 
    \begin{align}
        &\ms{g}^{-1}\Rightarrow^{M_0}(\mf{g}^{(0)})^{-1}\text{,}\qquad \ms{Rc}\text{, }\ms{Rs}\Rightarrow^{M_0-2}\mi{D}^{2}\paren{\mf{g}^{(0)}}\text{,}\label{limits_gR_n=2}\\
        &\label{intermediate_m_n=2_limit}\ms{m}\Rightarrow^{M_0-1}0\text{,}\qquad  \rho\Lie_\rho \ms{m}\Rightarrow^{M_0-1} 0\text{,}\\
        &\rho\Lie_\rho \ol{\ms{f}}^1 \Rightarrow^{M_0-2} \mf{f}^{1, \star}=\mi{D}^{(1, 1)}\paren{\mf{g}^{(0)}, \mf{f}^{0, (0)}}\text{.}
    \end{align}
    From \ref{lemma_item_4} of Proposition \ref{prop_integration}, there exists a tensor field $\mf{f}^{1,\dag}$ such that: 
    \begin{equation*}
        \ol{\ms{f}}^1 - \log\rho\cdot \mf{f}^{1,\star} \rightarrow^{M_0-2} \mf{f}^{1,\star}\text{.}
    \end{equation*}
    Furthermore, using \eqref{intermediate_m_n=2}, Lemma \ref{lemma_unif_bounds_n=2} and \eqref{limit_f0_n=2}:
    \begin{equation}
        \label{intermediate_f0_n=2}\Lie_\rho\ol{\ms{f}}^0 \Rightarrow^{M_0-1}0\text{.}
    \end{equation}
    The limits \eqref{limit_f0_n=2}, \eqref{limit_f1_n=2}, the bounds from Lemma \ref{lemma_unif_bounds_n=2} and \eqref{intermediate_f0_n=2} yield: 
    \begin{gather*}
        \Lie_\rho\paren{\rho\mc{S}\paren{\ms{g}; (\ol{\ms{f}}^0)^2} + \rho^3\mc{S}\paren{\ms{g}; (\ol{\ms{f}}^1)^2}}\Rightarrow^{M_0-1} \mi{D}^{(0, 0)}\paren{\mf{g}^{(0)}, \mf{f}^{0, (0)}}\text{.}
    \end{gather*}
    Using \eqref{limits_gR_n=2}, \eqref{intermediate_m_n=2_limit}
    \begin{equation*}
        \ms{tr}_{\ms{g}}\Lie_\rho\ms{m}\Rightarrow^{M_0-2} \mi{D}^{(2, 0)}\paren{\mf{g}^{(0)}, \mf{f}^{0,(0)}}\text{,}\qquad \rho\Lie_\rho\paren{ \ms{tr}_{\ms{g}}\Lie_\rho\ms{m}}\Rightarrow^{M_0-2}0\text{,}
    \end{equation*}
    implying immediately, from \eqref{transport_m_n=2}:
    \begin{equation*}
        \rho\Lie_\rho^2\ms{m}\Rightarrow^{M_0-2}2\mf{g}^\star =\mi{D}^{(2,0)}\paren{\mf{g}^{(0)}, \mf{f}^{0, (0)}}\text{.}
    \end{equation*}
    From \ref{lemma_item_4} of Proposition \ref{prop_integration}, there exists a tensor field $\mf{g}^{\dag}$ on $\mc{I}$ such that: 
    \begin{equation*}
        \Lie_\rho\ms{m} -2\log\rho\cdot \mf{g}^\star \rightarrow^{M_0-2} 2\mf{g}^\dag\text{,} 
    \end{equation*}
    concluding the proof. 
 \end{proof}

\subsubsection{Proof of Proposition \ref{prop_constraints}}

\begin{proof}
Let us start by proving the constraints \eqref{anomalous_tr}, \eqref{critical_div} on the trace of $\mf{g}^{(n)}$ and $\mf{g}^\star$. Such identities can be proven from \eqref{higher_transport_tr_m}, with $k=n-1$: 
\begin{equation*}
    \begin{aligned}
            \rho\Lie_\rho \ms{tr}_{\ms{g}}\paren{\Lie_\rho^{n} \ms{g}} - n\cdot \ms{tr}_{\ms{g}}(\Lie_\rho^n \ms{g}) = \Lie_\rho^{n-1} \ms{G}_0\text{,}
        \end{aligned}
\end{equation*}
where we replaced $\ms{m} =\Lie_\rho\ms{g}$, and with $\ms{G}_0$ a vertical tensor field satisfying \footnote{See proof of Theorem \ref{theorem_main_fg}}:
\begin{equation*}
    \Lie_\rho^{n-1} \ms{G}_0 \Rightarrow^{M_0-n} \begin{cases}
        \mi{D}^{(n, n-4)}\paren{\mf{g}^{(0)}, \mf{f}^{1, (0)}; n} \text{,}\qquad &n\geq 3\text{,}\\
        \mi{D}^{(2, 0)}\paren{\mf{g}^{(0)}, \mf{f}^{0, (0)}}\text{,}\qquad &n=2\text{,}
    \end{cases}
\end{equation*}
which implies, using \ref{lemma_item_2} of Proposition \ref{prop_integration}: 
\begin{equation}\label{intermediate_limit_tr_gn}
    \ms{tr}_{\ms{g}}\Lie_\rho^n \ms{g}\Rightarrow^{M_0-n} \begin{cases}
        \mi{D}^{(n, n-4)}\paren{\mf{g}^{(0)}, \mf{f}^{1, (0)}; n} \text{,}\qquad &n\geq 3\text{,}\\
        \mi{D}^{(2, 0)}\paren{\mf{g}^{(0)}, \mf{f}^{0, (0)}}\text{,}\qquad &n=2\text{,}
    \end{cases}\qquad \rho\Lie_\rho\ms{tr}_{\ms{g}}(\Lie_\rho^n \ms{g} )\Rightarrow^{M_0-n} 0\text{.}
\end{equation}
This immediately implies the following: 
\begin{align*}
    n! \log\rho \cdot \ms{tr}_{\ms{g}}\mf{g}^{\star} &=  \ms{tr}_{\ms{g}}\Lie_\rho^n\ms{g} -  \ms{tr}_{\ms{g}}\paren{\Lie_\rho^n\ms{g} - n!\log \rho\cdot \mf{g}^\star} \\
    &=:I_1 + I_2\text{,}
\end{align*}
with $I_1, I_2$ converging in $C^{M_0-n}$, and, as a consequence, let $\rho>0$ fixed: 
\begin{align*}
    \sup\limits_{ U} \abs{\mf{tr}_{\ms{g}^{(0)}}\mf{g}^{\star}}_{M_0-n, \varphi} &\lesssim  \sup\limits_{\lbrace \rho \rbrace \times U} \abs{\ms{tr}_{\ms{g}}\mf{g}^{\star}-\mf{tr}_{\mf{g}^{(0)}}\mf{g}^{\star}}_{M_0-n, \varphi} + \sup\limits_{\lbrace \rho \rbrace \times U} \abs{\ms{tr}_{\ms{g}}\mf{g}^{\star}}_{M_0-n, \varphi}\\
    &\lesssim (\log \rho)^{-1}\rightarrow 0\text{,}
\end{align*}
as $\rho\searrow 0$\text{,} and where the left-hand side does not depend on $\rho$. Also: 
\begin{align*}
    \ms{tr}_{\ms{g}}\paren{\Lie_\rho^n \ms{g} - n! \log\rho\cdot \mf{g}^\star} &=\ms{tr}_{\ms{g}}\paren{\Lie_\rho^n \ms{g}} - n! \log\rho\cdot \ms{tr}_{\ms{g}}{\mf{g}^\star} \\
    &\Rightarrow^{M_0-n} n!\mf{tr}_{\mf{g}^{(0)}}\mf{g}^\dag = \begin{cases}
        \mi{D}^{(n, n-4)}\paren{\mf{g}^{(0)}, \mf{f}^{1, (0)}; n} \text{,}\qquad &n\geq 3\text{,}\\
        \mi{D}^{(2, 0)}\paren{\mf{g}^{(0)}, \mf{f}^{0, (0)}}\text{,}\qquad &n=2\text{.}
        \end{cases}
\end{align*}
Since $\mf{g}^{(n)}$ is a linear combination of $\mf{g}^\star$ and $\mf{g}^\dag$, one can easily conclude.

Next, the constraints on the divergence can be deduced from \eqref{higher_constraint_Dm}, for which one has, for $k=n$: 
\begin{align*}
     \ms{D}\cdot (\rho\Lie_\rho^{n+1} \ms{g}) = \ms{D} \paren{\ms{tr}_{\ms{g}}\rho\Lie_\rho^{n+1} \ms{g } }&+ \sum\limits_{\substack{j + j_0 + j_1 + \dots + j_\ell = n\\j<n, j_p\geq 1}} \rho\cdot \mc{S}\paren{\ms{g}; \Lie_\rho^j \ms{m}, \ms{D}\Lie_\rho^{j_0-1}\ms{m}, \Lie_\rho^{j_1-1}\ms{m}, \dots, \Lie_\rho^{j_\ell-1}\ms{m}}\\
            &+4\rho\Lie_\rho^{n} \ms{T}^1\text{.}
\end{align*}
One has, for $n\geq 3$:
\begin{gather*}
\rho\Lie_\rho^n\ms{T}^1\Rightarrow^{M_0-n-1}0\text{,}
\end{gather*}
as well as, for $n=2$: 
\begin{equation*}
    \rho\Lie_\rho^2 \ms{T}^1 \Rightarrow^{M_0-2} \mi{D}^{(1, 1)}\paren{\mf{g}^{(0)}, \mf{f}^{0, (0)}}\text{.}
\end{equation*}
For $n\geq 4$, these limits are straightforward consequences of \eqref{limits_rho4_rhon}. For $n=3$, one can show the following: 
\begin{gather*}
    \rho\Lie_\rho^2\ol{\ms{f}}^0\Rightarrow^{M_0-2}0\text{,}\qquad  \rho\Lie_\rho^2\ol{\ms{f}}^1\Rightarrow^{M_0-2}0\\
    \rho^2\Lie_\rho^3\ol{\ms{f}}^0\Rightarrow^{M_0-3}0\text{,}\qquad  \rho^2\Lie_\rho^3\ol{\ms{f}}^1\Rightarrow^{M_0-3}0\text{.}
\end{gather*}
Let us briefly justify such limits. Equations \eqref{transport_f_0} and \eqref{transport_f_1}, on which one applies $\rho\Lie_\rho$, give: 
\begin{align*}
    &\rho\Lie_\rho^2\ol{\ms{f}}^0 = \mc{S}\paren{\ms{g}; \rho\Lie_\rho\ms{m}, \ol{\ms{f}}^0} + \mc{S}\paren{\ms{g};\ms{m}, \rho\Lie_\rho \ol{\ms{f}}^0} - \ms{D}\cdot \rho\Lie_\rho\ol{\ms{f}}^1 + \rho\mc{S}\paren{\ms{g}; \ms{Dm}, \ol{\ms{f}}^1}\Rightarrow^{M_0-3}0\text{,}\\
    &\rho\Lie_\rho^2 \ol{\ms{f}}^1_{ab} = 2\ms{D}_{[a}(\rho\Lie_\rho\ol{\ms{f}}^0)_{b]} + \rho\mc{S}\paren{\ms{g}; \ms{Dm}, \ol{\ms{f}}^0}{}_{ab}\Rightarrow^{M_0-3}0\text{,}
\end{align*}
where we used the limits and bounds from Theorem \ref{theorem_main_fg}. The higher-order limits follow from the same reasoning. 
Using Theorem \ref{theorem_main_fg}, and the freshly obtained limit \eqref{intermediate_limit_tr_gn}\footnote{Note that only the stress-energy tensor terms will yield nontrivial contributions.}: 
\begin{equation*}
    \ms{D}\cdot (\rho\Lie_\rho^{n+1}\ms{g})\Rightarrow^{M_0-n-1} n!\mf{D}\cdot \mf{g}^\star = 
    \begin{cases}
        \mi{D}^{(n-3,n-3)}\paren{\mf{g}^{(0)}, \mf{f}^{1,(0)}; n}\text{,}\qquad &n\geq 3\text{,}\\
        \mi{D}^{(3, 1)}\paren{\mf{g}^{(0)}, \mf{f}^{0,(0)}}\text{,}\qquad &n=2\text{.}
    \end{cases}
\end{equation*}
Note that in the vacuum case, such a limit is generally trivial \cite{shao:aads_fg}. As a consequence, the divergence of the anomalous term in the metric expansion is sourced by the matter fields. 
In Equation \ref{higher_constraint_Dm}, with $k=n-1$, one will find a divergent term of the form\footnote{Note that for $n$ odd, no such term appears.}: 
\begin{equation*}
    \begin{cases}
        \Lie_\rho^{n-1}\ms{T}^1\sim \mc{S}\paren{\ms{g}; \Lie_\rho^{n-4}\ol{\ms{f}}^0, \ol{\ms{f}}^1} + \text{l.o.t.}\text{,}\qquad &n\geq 4\\
        \ms{T}^1_a = \frac{1}{2}\ms{g}^{cd}\ol{\ms{f}}^0_c \ol{\ms{f}}^1_{ad}\text{,}\qquad &n=2\text{,}
    \end{cases}
\end{equation*}
where $l.o.t$ contains converging contributions, and where this term converges only if one removes the logarithmic contribution. More precisely, there exists $\tilde{\mf{g}}$ a symmetric two-tensor on $\mc{I}$ of the form: 
\begin{equation*}
    \tilde{\mf{g}} = 
    \begin{cases}
        \mi{D}^{(n+1, n-3)}\paren{\mf{g}^{(0)}, \mf{f}^{1,(0)}; n}\text{,}\qquad &n\geq 3\text{,}\\
        \mi{D}^{(3, 1)}\paren{\mf{g}^{(0)}, \mf{f}^{0,(0)}}\text{,}\qquad &n=2\text{,}
    \end{cases}
\end{equation*}
such that: 
\begin{equation*}
    \ms{D}\cdot \Lie_\rho^n \ms{g} - n! \log\rho\cdot \tilde{\mf{g}}\rightarrow^{M_0-n} \begin{cases}
        \mi{D}^{(n+1, n-3)}\paren{\mf{g}^{(0)}, \mf{f}^{1,(0)}; n}\text{,}\qquad &n\geq 3\text{,}\\
        \mi{D}^{(3, 1)}\paren{\mf{g}^{(0)}, \mf{f}^{0,(0)}}\text{,}\qquad &n=2\text{.}
    \end{cases}
\end{equation*}
By unicity of the limit, one must have: 
\begin{equation*}
    \sup\limits_{\lbrace \rho \rbrace \times U}\abs{\ms{D}\cdot \mf{g}^\star - \tilde{\mf{g}}}_{M_0-n-1, \varphi} \rightarrow 0\text{,}
\end{equation*}
proving indeed: 
\begin{equation*}
    \ms{D}\cdot \paren{\Lie_\rho^n \ms{g} - n! \log\rho\cdot {\mf{g}}^\star}\rightarrow^{M_0-n} n!\mf{D}\cdot \mf{g}^\dag=\begin{cases}
        \mi{D}^{(n+1, n-3)}\paren{\mf{g}^{(0)}, \mf{f}^{1,(0)}; n}\text{,}\qquad &n\geq 3\text{,}\\
        \mi{D}^{(3, 1)}\paren{\mf{g}^{(0)}, \mf{f}^{0,(0)}}\text{,}\qquad &n=2\text{.}
    \end{cases}
\end{equation*}
Since $\mf{g}^{(n)}$ is a linear combination of $\mf{g}^\star$ and $\mf{g}^{\dag}$, the conclusion follows.

The constraints \eqref{constraint_f} are direct consequences of Propositions \ref{prop_constraints}, \ref{prop_higher_constraints} and Equations \eqref{equation_d_f1} and \eqref{equation_div_f0}. 

\end{proof}
\section{Unique continuation}\label{chap:UC}

This section will be dedicated to proving the main theorem of this work, Theorem \ref{thm_UC_main}. In order to do so, we will derive first in Section \ref{sec_wave_transport} an appropriate wave-transport system of vertical fields involving the metric, the Weyl curvature, as well the Maxwell fields. Only certain vertical decomposition will give rise to an appropriate system, that one will be able to implement in the Carleman estimates derived in \cite{Shao22}, \cite{Holzegel22}. 

Since in Theorem \ref{thm_UC_informal}, one compares $(g,F)$ in $(\mc{M},g)$ and $(\check{g}, \check{F})$ in $(\mc{M}, \check{g})$, two Maxwell-FG-aAdS segments,  one may want to study the wave-transport system satisfied by the difference of the appropriate vertical decomposition. This is the purpose of Section \ref{sec:difference}. The main challenge of this section, also encountered in \cite{Holzegel22}, is the treatment of the wave operator. More precisely, given two metrics on $\mc{M}$: 
\begin{equation}\label{two_FG_aAdS}
    g = \rho^{-2}\paren{d\rho^2 + \ms{g}_{ab}dx^adx^b}\text{,}\qquad g = \rho^{-2}\paren{d\rho^2 + \check{\ms{g}}_{ab}dx^adx^b}\text{,}
\end{equation}
and two vertical fields $\ms{A}, \, \check{\ms{A}}$ on $(\mc{M}, g)$ and $(\mc{M}, \check{g})$ satisfying the following vertical wave equations: 
\begin{gather*}
    \Box_{\ms{g}} \ms{A} = \ms{N}\text{,}\qquad \Box_{\check{\ms{g}}}\check{\ms{A}} = \check{\ms{N}}\text{,}
\end{gather*}
with $\ms{N}, \, \check{\ms{N}}$ appropriate (non-linear) source terms, one has: 
\begin{align*}
    \Box_{\ms{g}}\paren{\ms{A} - \check{\ms{A}}} &= \Box_{\ms{g}}\ms{A} - \Box_{\check{\ms{g}}}\check{\ms{A}} + \paren{\Box_{\ms{g}} - \Box_{\check{\ms{g}}}}\check{\ms{A}}\\
    &=\underbrace{\ms{N} - \check{\ms{N}}}_{\text{good}} + \underbrace{\paren{\Box_{\ms{g}} - \Box_{\check{\ms{g}}}}\check{\ms{A}}}_{\text{bad}}\text{,}
\end{align*}
where ``bad" terms contain second order derivatives of $\ms{g}$, which can not be controlled in the Carleman estimates. In Section \ref{sec:difference}, we follow the renormalisation procedure introduced in \cite{Holzegel22} and show that one can indeed obtain a wave-transport system, closed for the purpose of the Carleman estimates. 

In Section \ref{sec:higher}, we improve the decay with respect to $\rho$ of the vertical fields, at the price of vertical regularity. The purpose of this procedure is to get rid of boundary terms in the Carleman estimates. 

We finally use the above results to show an intermediate proposition. Namely, for two FG-aAdS segment $(\mc{M}, g)$ $(\mc{M}, \check{g})$ as in \eqref{two_FG_aAdS}, such that ${g}$ and $\check{g}$, as well as $F$ and $\check{F}$ have identical boundary data on a domain $\mi{D}\subset\mc{I}$, implies necessarily: 
\begin{equation*}
    g \equiv g\text{,} \qquad F\equiv \check{F}\text{,}
\end{equation*}
in a neighbourhood of $\lbrace 0\rbrace \times \mi{D}$. To prove this proposition, we inject the wave-transport system obtained in the previous section in the Carleman estimates and follow a standard procedure. Note, however, that $\mc{D}$ has to satisfy a certain pseudoconvexity condition with respect to $\mf{g}^{(0)}$. 

The last part of this work, Section \ref{sec:full_result}, discusses the gauge freedom and presents the full result, namely Theorem \ref{thm_UC_main}. 

\subsection{Wave-transport system}\label{sec_wave_transport}

First of all, let us remind the different vertical fields that we will use throughout this section: 
\begin{gather*}
    \ms{m}_{ab} := \Lie_\rho \ms{g}_{ab}\text{,} \\
    {\ms{f}}^0_a := \rho F_{\rho a}\text{,}\qquad {\ms{f}}^1_{ab} := \rho F_{ab}\text{,}\\
    \ms{w}^0_{abcd} := \rho^2 W_{abcd}\text{,}\qquad \ms{w}^1_{abc}:= \rho^2 W_{\rho abc} \text{,}\qquad \ms{w}^2_{ab} := \rho^2 W_{\rho a \rho b}\text{.}
\end{gather*}

\begin{remark}
    Here, we will use the same decomposition of the stress-energy tensors as in Chapter \ref{sec:FG_expansion}, see Definition \ref{def_decomposition_T}. In the gauge \eqref{FG-gauge}, these take the form: 
    \begin{gather*}
        \tilde{\ms{T}}_{ab}^0 = T_{ab} - \frac{1}{n}\paren{T_{\rho\rho} + \ms{g}^{cd}T_{cd}}\ms{g}_{ab}\text{,}\qquad \tilde{\ms{T}}^2 =\frac{n-1}{n}T_{\rho\rho} - \frac{1}{n}\ms{g}^{ab}T_{ab}\text{,}\\
        \text{,}\\
        \ms{T}^0_{ab} = \paren{\ms{f}^0_a \ms{f}^0_b-\frac{1}{2}\ms{g}_{ab}(\ms{f}^0)^2} + \paren{\ms{g}^{cd}\ms{f}^1_{ac}\ms{f}^1_{bd}- \frac{1}{4}\ms{g}_{ab}(\ms{f}^1)^2}\text{,}\qquad  \ms{T}^2 = \frac{1}{2}(\ms{f}^0)^2 - \frac{1}{4} (\ms{f}^1)^2\text{.}
    \end{gather*}
\end{remark}

Furthermore, we will also need the following vertical field: 
\begin{definition}\label{def_w_star}
    Let $(\mc{M},g, F)$ be a Maxwell-FG-aAdS segment with $n>2$. We define $\ms{w}^\star$ to be the following $(0,4)-$tensor: 
    \begin{equation*}
        \ms{w}^\star := \ms{w}^0 + \frac{1}{n-2}\ms{w}^2\star \ms{g}\text{.}
    \end{equation*}
\end{definition}

We remind here the expressions for the vertical decomposition of the Weyl tensor:
\begin{proposition}
    Let $(\mc{M}, g, \Phi)$ be a non-vacuum aAdS segment and let $(U, \varphi)$ be a compact coordinate system on $\mc{I}$. Then, the following hold: 
    \begin{align}
        &\label{w_0_FG}\ms{w}^0 = \ms{R} - \frac{1}{8}\ms{m}\star \ms{m}  + \frac{1}{2\rho}\ms{g}\star \ms{m} - \frac{1}{n-1}\tilde{\ms{T}}^0\star \ms{g}\text{,}\\
        &\label{w_1_FG}\ms{w}^1_{abc} = \ms{D}_{[c}\ms{m}_{b]a} - \frac{2}{n-1}T_{\rho [b}\ms{g}_{c]a}\text{,}\\
        &\label{w_2_FG}\ms{w}^2_{ab} = -\frac{1}{2}\Lie_\rho \ms{m}_{ab} + \frac{1}{2\rho} \ms{m}_{ab}+\frac{1}{4} \ms{g}^{cd}\ms{m}_{ac}\ms{m}_{bd} - \frac{1}{n-1}\paren{\tilde{\ms{T}}^2\ms{g}_{ab} + \tilde{\ms{T}}^0_{ab}}\text{.}
    \end{align}
\end{proposition}

Before stating the useful transport equations, we will need to define the following higher-derivative tensors: 
\begin{definition}
    Let $(M, g, F)$ be a Maxwell-FG-aAdS segment. We will denote 
    \begin{equation*}
        H:= \nabla F\text{,}
    \end{equation*} 
    and define,  with respect to any local coordinate system, the following vertical fields: 
    \begin{gather}
        \ms{h}^0_{abc} := \rho H_{abc}\text{,}\qquad \ms{h}^1_{ab} :=\rho H_{\rho ab}\text{,}\\
        \notag\ms{h}^2_{ab}:= \rho H_{a\rho b}\text{,}\qquad \ms{h}^3_a := \rho H_{\rho\rho a}\text{.}
    \end{gather}
    Furthermore, define the trace of $\ms{h}^0$ as the following $(0,1)$--vertical tensor: 
    \begin{equation*}
        \ms{tr}_{\ms{g}}\ms{h}^0_c = \ms{g}^{ab}\ms{h}^0_{abc}\text{.}
    \end{equation*}
\end{definition}

\begin{remark}
    Observe that, using Bianchi's identity and Maxwell's equations for $F$, only two of these fields are independent. One indeed has: 
    \begin{equation}\label{property_h}
        \ms{h}^1_{ab} = 2\ms{h}^2_{[ab]}\text{,}\qquad \ms{h}^3_a = -\ms{tr}_{\ms{g}}\ms{h}^0_a\text{.}
    \end{equation}
    Furthermore, the antisymmetric part of $\ms{h}^{0}$ identically vanishes: 
    \begin{equation}\label{property_h_0}
        \ms{h}^0_{[abc]} =0\text{.}
    \end{equation}
\end{remark}

Let us also define the following vertical fields: 
\begin{definition}\label{def_h_bar}
     Let $(\mc{M}, g, F)$ be a Maxwell-FG-AdS segment. We define the following vertical tensors: 
     \begin{gather*}
         \ul{\ms{h}}^0_{abc} := \ms{D}_a\ms{f}^1_{bc}\text{,}\qquad \ul{\ms{h}}^2_{ab} := \ms{D}_a\ms{f}^0_{b}\text{.}
     \end{gather*}
\end{definition}
\begin{remark}
    In this work, we will use the vertical fields $\ul{\ms{h}}^0$ and $\ul{\ms{h}}^2$ in the Carleman estimates and thus correspond to the fields of interest. The fields $\ms{h}^i$ will mainly be used for convenience in the computations, since they directly correspond to the naive vertical decomposition of $H$. 
\end{remark}

The vertical tensors $\ms{h}$ and $\ul{\ms{h}}$ simply differ by a zeroth-order contribution and lower-order terms. 

\begin{proposition}\label{h_as_Df}
Let $(\mc{M}, g, F)$ be a Maxwell-FG-aAdS segment. Then, the following holds:
    \begin{gather}\label{expression_h_Df}
        {\ms{h}}^0_{abc} = \ul{\ms{h}}^0_{abc} - 2\rho^{-1}\ms{g}_{a[b}\ms{f}^0_{c]}+\mc{S}\paren{\ms{m}, {\ms{f}^0}}_{abc}\text{,}\qquad 
        \ms{h}^{2}_{ab}=\ul{\ms{h}}^2_{ab}+\rho^{-1}\ms{f}^1_{ab}+\mc{S}\paren{\ms{g}; \ms{m}, {\ms{f}}^1}_{ab}\text{.}
    \end{gather}
\end{proposition}

\begin{proof}
    Immediate from Proposition \ref{sec:aads_derivatives_vertical}.
\end{proof}

\begin{proposition}[Proposition \ref{prop_transport_m}]
    Let $(\mc{M}, g)$ be a Maxwell-FG-aAds segment and let $(U, \varphi)$ be a compact coordinate system on $\mc{I}$. The following transport equations hold for $\ms{m}$: 
    \begin{align*}
        \rho\Lie_\rho \ms{m}_{ab} -(n-1)\ms{m}_{ab} =&\, 2\rho \ms{Rc}_{ab}+\ms{tr}_{\ms{g}}\ms{m}\cdot \ms{g}_{ab} +\rho\cdot \ms{g}^{cd}\ms{m}_{ac}\ms{m}_{bd} - \frac{1}{2}\rho\ms{tr}_{\ms{g}}\ms{m}\cdot \ms{m}_{ab} \\
        &-2\rho\ms{T}^0_{ab}+\frac{2\rho}{n-1}\ms{tr}_{\ms{g}}\ms{T}^0\cdot \ms{g}_{ab} +\frac{2\rho}{n-1}\ms{T}^2\cdot \ms{g}_{ab}\text{.}
    \end{align*}
\end{proposition}

\begin{proposition}[Transport equations for the Maxwell fields]\label{transport_f_UC}
Let $(\mc{M}, g)$ be a Maxwell-FG-aAdS segment and let $(U, \varphi)$ be a compact coordinate system. The vertical Maxwell fields satisfy the following transport equations: 
    \begin{gather}
        \label{transport_f0_FG}\rho \ol{\ms{D}}_\rho {\ms{f}}^0_a - (n-2){\ms{f}}^0_a =- \rho \cdot \ms{tr}_{\ms{g}}\ul{\ms{h}}^0_a + \rho\mc{S}\paren{\ms{g}; \ms{m}, {\ms{f}}^0}_a \text{,}\\
        \label{transport_f1_FG}\rho \ol{\ms{D}}_\rho {\ms{f}}^1_{ab}-{\ms{f}}^1_{ab}= 2\rho\cdot \ul{\ms{h}}^2_{[ab]}+\rho\mc{S}\paren{\ms{g};\ms{m}, {\ms{f}}^1}_{ab}\text{.}
    \end{gather}
\end{proposition}

The vertical fields $\ul{\ms{h}}^0$, $\ul{\ms{h}}^2$ satisfy the following transport equations: 
\begin{proposition}\label{prop_transport_h_bar}
    Let $(\mc{M}, g, F)$ be a Maxwell-FG-aAdS segment and let $(U, \varphi)$ be a compact coordinate system on $\mc{I}$. Then, the following transport equations are satisfied: 
    \begin{align*}
        &\rho\ol{\ms{D}}_\rho\ul{\ms{h}}^0_{abc} -\ul{\ms{h}}^0_{abc}= 2\rho\ms{D}_a\ul{\ms{h}}^2_{[bc]} + \rho\mc{S}\paren{\ms{g}; \ms{m}, \ul{\ms{h}}^0}_{abc} + \rho\mc{S}\paren{\ms{g}; \ms{Dm}, {\ms{f}}^1}_{abc}\\
        &\rho\ol{\ms{D}}_{\rho}\ul{\ms{h}}^2_{ab} - (n-2)\ul{\ms{h}}^2_{ab} = -\rho\ms{D}_a\ms{tr}_{\ms{g}}\ul{\ms{h}}^0_b +\rho\mc{S}\paren{\ms{g}; \ms{m}, \ul{\ms{h}}^2}_{ab} + \rho\mc{S}\paren{\ms{g}; \ms{Dm}, \ms{f}^0}_{ab}\text{.}
    \end{align*}
\end{proposition}

\begin{proof}
    See Appendix \ref{app:transport_h_bar}.
\end{proof}

We will also need the following transport equations: 
    \begin{proposition}\label{prop_transport_h}
    Let $(\mc{M}, g, F)$ be a Maxwell-FG-aAdS segment and let $(U, \varphi)$ be a compact coordinate system on $\mc{I}$. The following transport equations hold for $\ms{h}^0$ and $\ms{h}^2$: 
    \begin{align*}
        \rho \cdot \ol{\ms{D}}_\rho \ms{h}^0 =\mc{R}_{T, \ms{h}^0}\text{,}\qquad \rho \cdot \ol{\ms{D}}_\rho\ms{h}^2 =\mc{R}_{T, \ms{h}^2}\text{,}
    \end{align*}
    where we defined: 
    \begin{align}
        &\label{sec:syst_transport_h_0}\begin{aligned}
        \paren{\mc{R}_{T, \ms{h}^0}}_{abc} :=  &-\ms{h}^0_{abc} + 2\ms{g}_{a[b}\ms{tr}_{\ms{g}}\ms{h}^0_{c]} + 2\rho \cdot \ms{D}_a\ms{h}^2_{[bc]}+2\ms{g}_{a[b}\rho^{-1}\ms{f}^0_{c]}\\
        &+ \rho\mc{S}\paren{\ms{g}; \ms{m}, \ul{\ms{h}}^0}_{abc} + \rho\mc{S}\paren{\ms{w}^2, \ms{f}^0}_{abc} + \rho\mc{S}\paren{\ms{g}; \ms{w}^1, \ms{f}^1}_{abc} +\rho\mc{S}\paren{\ms{g}; (\ms{m})^2, \ms{f}^0}_{abc}\\
        &+\mc{S}\paren{\ms{g}; \ms{m}, \ms{f}^0}_{abc} +\rho\mc{S}\paren{\ms{g}; (\ms{f}^0)^3}_{abc} + \rho\mc{S}\paren{\ms{g};  \ms{f}^1, \ms{f}^1, \ms{f}^0}_{abc}
        \end{aligned}\\
        &\label{sec:syst_transport_h_2}\begin{aligned}
        \paren{\mc{R}_{T, \ms{h}^2}}_{ab} := &-\ms{h}^2_{ab} +2\ms{h}^2_{[ab]} -\rho\ms{D}_a\ms{tr}_{\ms{g}}\ms{h}^0_{b} - \rho^{-1}\ms{f}^1_{ab}\\
        &+\rho\mc{S}\paren{\ms{g}; \ms{m}, \ul{\ms{h}}^{2}}_{ab} + \rho\mc{S}\paren{\ms{g}; \ms{w}^2, \ms{f}^1}_{ab} + \rho\mc{S}\paren{\ms{g}; \ms{w}^1, \ms{f}^0}_{ab}+\mc{S}\paren{\ms{g}; \ms{m}, \ms{f}^1}_{ab}\\
        &+\rho\mc{S}\paren{\ms{g};(\ms{m})^2, \ms{f}^1}_{ab} + \rho \mc{S}\paren{\ms{g}; \ms{f}^0, \ms{f}^0, \ms{f}^1}_{ab} + \rho\mc{S}\paren{\ms{g};  (\ms{f}^1)^3}_{ab}
        \end{aligned}
    \end{align}
    \end{proposition}

    \begin{proof}
        See Appendix \ref{app:transport_h}.
    \end{proof}

Finally, we will also need the transport equations satisfied by the Weyl tensors: 
\begin{proposition}[Transport equation for the Weyl curvature]\label{prop_transport_w}
    Let $(\mc{M}, g)$ be a FG-aAdS segment and let $(U, \varphi)$ be a compact coordinate system on $\mc{I}$. Then, the following transport equations are satisfied: 
    \begin{align}
        &\rho\bar{\ms{D}}_\rho \ms{w}^0_{abcd} = 2\rho \cdot \ms{D}_{[a}\ms{w}^1_{b]cd} - (\ms{g}\star  \ms{w}^2)_{abcd} + (\mc{R}_{T,\ms{w}^2})_{abcd}\label{transport_w_star}\\
        &\rho \bar{\ms{D}}_{\rho}\ms{w}^1_{abc} = 2\rho \ms{D}_{[b}\ms{w}^2_{c]a} + \ms{w}^1_{abc} + (\mc{R}_{T, \ms{w}^1})_{abc}\label{transport_w_1.1}\\
        &\rho \bar{\ms{D}}_\rho\ms{w}^1_{abc} = -\rho (\ms{D}\cdot \ms{w}^0)_{abc} +(n-2)\ms{w}^1_{abc} + (\mc{R}_{T, \ms{w}^{1}})_{abc}\label{transport_w_1.2}\\
        &\rho\bar{\ms{D}}_\rho \ms{w}^2_{ab} = -\rho\cdot  \ms{g}^{cd}\ms{D}_c \ms{w}^1_{adb} + (n-2)\ms{w}^2_{ab} + (\mc{R}_{T, \ms{w}^2})_{ab}\text{,}\label{transport_w_2}
    \end{align}
    where $\mc{R}_{T,w^1}, \mc{R}_{T,\ms{w}^2}$ are remainder terms that can be written as: 
    \begin{align}
        &\notag\mc{R}_{T, \ms{w}^1} =  \rho\mc{S}\paren{\ms{g}; \ms{m}, \ms{w}^1} + \rho\mc{S}\paren{\ms{g}; \ul{\ms{h}}^0, \ms{f}^1}+\rho\mc{S}\paren{\ms{g}; \ul{\ms{h}}^{2}, \ms{f}^0} \\
        &\hspace{60pt}+\rho\mc{S}\paren{\ms{g}; \ms{m}, \ms{f}^1, \ms{f}^0} +\mc{S}\paren{\ms{g}; \ms{f}^0, \ms{f}^1}\label{R_T_w_1}\text{,}\\
        &\notag \mc{R}_{T, \ms{w}^2} =  \rho\mc{S}\paren{\ms{g}; \ms{m}, \ms{w}^0} + \rho\mc{S}\paren{\ms{g}; \ms{m}, \ms{w}^2}+\rho\mc{S}\paren{\ms{g}; \ul{\ms{h}}^0, \ms{f}^0} + \mc{S}\paren{\ms{g}; (\ms{f}^1)^2} \\  
        &\hspace{60pt}+ \rho\mc{S}\paren{\ms{g}; \ul{\ms{h}}^{2}, \ms{f}^1}+\mc{S} \paren{\ms{g}; (\ms{f}^0)^2}\notag\\
        &\hspace{60pt} + \rho\mc{S}\paren{\ms{g};\ms{m},  (\ms{f}^0)^2} + \rho\mc{S}\paren{\ms{g}; \ms{m}, (\ms{f}^1)^2}\label{R_T_w_2}\text{.}
    \end{align}
\end{proposition}

\begin{proof}
    See Appendix \ref{app:transport_weyl_UC}
\end{proof}

To apply the Carleman estimates, we will also need the wave equations satisfied by the different vertical decompositions of the Weyl tensor $W$, electromagnetic tensor $F$ and the higher-derivative fields $H$. This is the purpose of the following proposition: 
\begin{proposition}\label{wave_weyl}
    Let $(\mc{M}, g, F)$ be a Maxwell-FG-aAdS segment. Then, with respect to any coordinate system on $(\mc{M}, g)$, the following equations hold: 
    \begin{align}\label{sec:system_weyl_wave}
        &(\Box_g + 2n)W_{\alpha\beta\gamma\delta} = -\frac{1}{n-1}(\Box_g + 2n)(g\star \tilde{T})_{\alpha\beta\gamma\delta} -2\paren{(W+\frac{1}{n-1}g\star \tilde{T})\odot (W+\frac{1}{n-1}g\star \tilde{T})}_{\alpha\beta\gamma\delta\notag}\\
        &\hspace{80pt} +2\, \hat{T}^\sigma{}_{[\gamma} W_{|\alpha\beta\sigma|\delta]} + \frac{2}{n-1}\, \hat{T}^\sigma{}_{[\gamma} (g\star \tilde{T})_{|\alpha\beta\sigma|\delta]}+\frac{2}{n-1}\paren{\operatorname{tr}_g (\tilde{T}\star g)\star g}_{\alpha\beta\gamma\delta}\notag\\
        &\hspace{80pt}- \paren{\hat{T}\star g}_{\alpha\beta\gamma\delta}+2\nabla_{[\gamma}\nabla_{|\alpha}\hat{T}_{\beta|\delta]}- 2\nabla_{[\gamma}\nabla_{|\beta}\hat{T}_{\alpha|\delta]} \text{,}\\
        \label{wave_maxwell}
        &\paren{\Box_g +2(n-1)}F_{\alpha\beta} =  -W_{\alpha\beta}{}^{\sigma\mu}F_{\sigma\mu}+\frac{2}{n-1}\tilde{T}^\mu{}_{[\alpha} F_{\beta]\mu} - 2\hat{T}^\mu{}_{[\alpha}F_{\beta]\mu}\text{,}
    \end{align}
    where we define the following operation: 
    \begin{equation}
        (A\odot B)_{\alpha\beta\gamma\delta} := A^{\sigma}{}_{\alpha\gamma\mu}B_{\sigma\beta\delta}{}^\mu + A^\sigma{}_{\beta\gamma\mu}B_{\alpha\sigma\delta}{}^\mu + A^{\sigma}{}_{\delta\gamma\mu}B_{\alpha\beta\sigma}{}^\mu\text{,}
    \end{equation}
    with $A,\, B$ any $(0,4)-$tensor. As a reminder:
    \begin{equation}
        \tilde{T} := T - \frac{\tr_g T}{n}\cdot g\text{,}\qquad \hat{T} := T - \frac{\tr_g T}{n-1}\cdot g\text{.}
    \end{equation}
\end{proposition}

\begin{remark}
    Observe that Equation \eqref{sec:system_weyl_wave} contains higher-order terms of the form 
    \begin{align*}
    \nabla^2T&\sim F\cdot \nabla^2 F + \nabla F \cdot \nabla F \\
    &= \nabla H \cdot F + H\cdot H\text{.}
    \end{align*}
    This is the reason why we introduced the field $H$. In order to derive the unique continuation result, we indeed need the wave equations to control zeroth-- and first--order derivatives. 
\end{remark}

\begin{proof}
    See Appendix \ref{app:wave_w_maxwell}. 
\end{proof}

Finally, we will also need the wave equation for the higher-derivative field $H$. 

\begin{proposition}\label{prop_wave_h_ST}
    Let $(\mc{M}, g, F)$ be a Maxwell-FG-aAdS segment. For any local coordinate system, the following wave equation hold: 
    \begin{align}\label{wave_higher_derivative}
        (\Box_g + 3n)H_{\mu\alpha\beta} = & - 2\nabla^\gamma W^\sigma{}_{[\alpha|\gamma\mu}F_{\sigma|\beta]} -4W^\sigma{}_{[\alpha}{}^\gamma{}_{|\mu}H_{\gamma\sigma|\beta]}- \nabla_\mu W_{\alpha\beta}{}^{\sigma\lambda} F_{\sigma\lambda}\\
        &\notag-W_{\alpha\beta}{}^{\sigma\lambda}H_{\mu\sigma\lambda} + \frac{2}{n-1}\tilde{T}^\sigma{}_{[\alpha}H_{\mu|\beta]\sigma}-2\hat{T}^\sigma{}_{[\alpha}H_{\mu|\beta]\sigma} \\
        &\notag - \frac{2}{n-1}(g\star \nabla^\gamma \tilde{T})^\sigma{}_{[\alpha|\gamma\mu}F_{\sigma|\beta]} + \frac{2}{n-1}\nabla_\mu \tilde{T}^\sigma{}_{[\alpha}F_{|\beta]\sigma}-2\nabla_\mu \hat{T}^\sigma{}_{[\alpha}F_{|\beta]\sigma}\\
        &\notag- \frac{4}{n-1}\paren{g_{\mu[\alpha}\tilde{T}^{\sigma\gamma}H_{|\gamma\sigma|\beta]} -\tilde{T}^{\gamma}{}_{[\alpha}H_{\gamma\mu|\beta]}}+ \frac{2}{n-1}\tilde{T}^\sigma{}_\mu H_{\sigma\alpha\beta} + \hat{T}^\sigma{}_\mu H_{\sigma\alpha\beta}\text{,}
    \end{align}
\end{proposition}

\begin{proof}
    See Appendix \ref{app:wave_h}.
\end{proof}

With the spacetime wave equations now at hand, we are ready to derive appropriate vertical wave equations.
\begin{proposition}\label{prop_wave_maxwell}
    Let $(\mc{M}, g, F)$ be a FG-aAdS segment with $n>2$ and let $(U, \varphi)$ be a coordinate system on $\mc{I}$. The following vertical wave equations are satisfied: 
    \begin{align}\label{wave_f_0}
        \paren{\ol{\Box} +2(n-2)}\ms{{f}}^0 = \mc{R}_{W, \ms{f}^0}\text{,}\\
        \label{wave_f_1}
        \paren{\ol{\Box}+(n-1)}\ms{f}^1 = \mc{R}_{W, \ms{f}^1}\text{,}
    \end{align}
    where we defined: 
    \begin{align*}
        \mc{R}_{W, \ms{f}^0} := &\;  \rho^2\mc{S}\paren{\ms{g}; \ms{m}, \ul{\ms{h}}^0} + \rho^2\mc{S}\paren{\ms{g}; \ms{D}\ms{m}, \ms{f}^1} + \rho\mc{S}\paren{\ms{g}; \ms{m}, \ms{f}^0}\\
        &\notag + \rho^2 \mc{S}\paren{\ms{g}; (\ms{m})^2, \ms{f}^0} + \rho^2\mc{S}\paren{\ms{g}; (\ms{f}^0)^3}+\rho^2\mc{S}\paren{\ms{g}; \ms{f}^0, (\ms{f}^1)^2} \\
        &\notag+\rho^2\mc{S}\paren{\ms{g}; \ms{w}^2, \ms{f}^0} + \rho^2\mc{S}\paren{\ms{g}; \ms{w}^1, \ms{f}^1}\text{,}\\
        \mc{R}_{W, \ms{f}^1} := &\;\rho^2\mc{S}\paren{\ms{g}; ({\ms{f}^0})^2, {\ms{f}^1}} + \rho^2\mc{S}\paren{\ms{g}; (\ms{f}^1)^3} + \rho^2\mc{S}\paren{\ms{g}; \ms{m}, \ul{\ms{h}}^{2}} \\
        &\notag+ \rho^2\mc{S}\paren{\ms{g}; \ms{D}\ms{m}, \ms{f}^0} + \rho\mc{S}\paren{\ms{g}; \ms{m}, \ms{f}^1} + \rho^2\mc{S}\paren{\ms{g}; (\ms{m})^2, \ms{f}^1}\\
        &\notag +\rho^2\mc{S}\paren{\ms{g}; \ms{w}^1, \ms{f}^0} + \rho^2\mc{S}\paren{\ms{g}; \ms{w}^0, \ms{f}^1}\text{.}\notag
    \end{align*}
    
\end{proposition}

\begin{proof}
    See Appendix \ref{app:prop_wave_maxwell}
\end{proof}

\begin{proposition}\label{prop_wave_h}
Let $(\mc{M}, g, F)$ be a Maxwell-FG-aAdS segment. For any local coordinate system, the following vertical wave equations hold for $\ul{\ms{h}}^0$ and $\ul{\ms{h}}^2$:   
    \begin{align}
        (\ol{\Box}+(n-1))\ul{\ms{h}}^0 = \mc{R}_{W, \ms{h}^0}\text{,} \qquad 
        (\ol{\Box} + 2(n-2))\ul{\ms{h}}^{2}
        =\mc{R}_{W, \ms{h}^2}\text{,}
    \end{align}
    where: 
    \allowdisplaybreaks{
    \begin{align}
        \label{remainder_h_0}\mc{R}_{W, \ms{h}^0}:=&\, \rho^{2}\mc{S}\paren{\ms{g}; \ms{m}, \ms{D}\ul{\ms{h}}^{2}} + \rho\mc{S}\paren{\ms{g}; \ms{w}^2, \ms{f}^0} + \rho\mc{S}\paren{\ms{g}; \ms{w}^1, \ms{f}^1} + \rho^2\mc{S}\paren{\ms{g}; \ms{Dm}, \ul{\ms{h}}^{2}}\\
        &\notag+\rho\mc{S}\paren{\ms{g}; \ms{w}^0, \ms{f}^0}+\rho^2 \mc{S}\paren{\ms{g};\ms{m}, \ms{w}^0, \ms{f}^0}  + \rho^2 \mc{S}\paren{\ms{g}; \ms{m}, \ms{w}^2, \ms{f}^0}+\rho^2\mc{S}\paren{\ms{g}; \ms{w}^2, \ul{\ms{h}}^0} \\
        &\notag +\rho^2 \mc{S}\paren{\ms{g}; \ms{w}^0, \ul{\ms{h}}^0} + \rho^2 \mc{S}\paren{\ms{g}; \ms{w}^1, \ul{\ms{h}}^{2}} + \rho^2 \mc{S}\paren{\ms{g}; \ul{\ms{h}}^0, (\ms{f}^0)^2}\\
        &\notag + \rho^2 \mc{S}\paren{\ms{g}; \ms{m}, \ms{w}^1, \ms{f}^1} + \rho^2 \mc{S}\paren{\ms{g};\ul{\ms{h}}^0, (\ms{f}^1)^2}+ \rho\mc{S}\paren{\ms{g}; \ms{m}, \ul{\ms{h}}^0} +\rho^2\mc{S}\paren{\ms{g}; (\ms{m})^2, \ul{\ms{h}}^0}  \\
        &\notag+ \rho^2\mc{S}\paren{\ms{g}; \ms{Dw}^1, \ms{f}^0} + \rho^2 \mc{S}\paren{\ms{g}; \ms{Dw}^2, \ms{f}^1} + \rho^2 \mc{S}\paren{\ms{g}; \ms{Dw}^0, \ms{f}^1}+\rho^2\mc{S}\paren{\ms{g}; \ul{\ms{h}}^{2}, \ms{f}^1, \ms{f}^0}\\
        &\notag+\rho\mc{S}\paren{\ms{g}; \ms{Dm}, \ms{f}^1} + \mc{S}\paren{\ms{g}; \ms{m}, \ms{f}^0} + \rho\mc{S}\paren{\ms{g};  (\ms{m})^2, \ms{f}^0}+\rho^2\mc{S}\paren{\ms{g}; \ms{m}, (\ms{f}^1)^2, \ms{f}^0}\\
        &\notag +\rho \mc{S}\paren{\ms{g}; (\ms{f}^1)^2, \ms{f}^0}  + \rho \mc{S}\paren{\ms{g}; (\ms{f}^0)^3}+\rho\mc{S}\paren{\ms{g}; \ms{w}^0, \ms{f}^0}+\rho^2\mc{S}\paren{\ms{g}; \ms{m}, (\ms{f}^0)^3}\\
        &\notag + \rho^2\mc{S}\paren{\ms{g}; (\ms{m})^3, \ms{f}^0} + \rho^2\mc{S}\paren{\ms{g}; \ms{m}, \ms{Dm}, \ms{f}^1}\text{,}\\
        \label{remainder_h_2}\mc{R}_{W, \ms{h}^2} :=&\rho^2\mc{S}\paren{\ms{g};\ms{m}, \ms{D}\ul{\ms{h}}^0} +\rho^2\mc{S}\paren{\ms{g}; \ms{Dw}^2, \ms{f}^0}+\rho^2\mc{S}\paren{\ms{g}; \ms{Dw}^0, \ms{f}^0}\\
        &\notag +\rho^2 \mc{S}\paren{\ms{g}; \ms{Dw}^1, \ms{f}^1} +\rho^2 \mc{S}\paren{\ms{g};  \ul{\ms{h}}^{2}, (\ms{f}^1)^2} + \rho^2\mc{S}\paren{\ms{g};  \ul{\ms{h}}^0, \ms{f}^0, \ms{f}^1}\\
        &\notag +\rho\mc{S}\paren{\ms{g}; \ms{m}, \ul{\ms{h}}^{2}} + \rho \mc{S}\paren{\ms{g}; \ms{w}^2, \ms{f}^1} + \rho\mc{S}\paren{\ms{g}; \ms{w}^1, \ms{f}^0} + \rho\mc{S}\paren{\ms{g}; \ms{w}^0, \ms{f}^1}\\
        &\notag+\rho^2 \mc{S}\paren{\ms{g}; \ul{\ms{h}}^{2} ,  (\ms{f}^0)^2}+\rho\mc{S}\paren{\ms{g}; (\ms{f}^0)^2, \ms{f}^1} + \rho\mc{S}\paren{\ms{g};  (\ms{f}^1)^3}+ \rho^2\mc{S}\paren{\ms{g}; \ms{m}, (\ms{f}^0)^2, \ms{f}^1} \\
        &\notag + \rho^2 \mc{S}\paren{\ms{g}; \ms{m}, (\ms{f}^1)^3}+\rho^2 \mc{S}\paren{\ms{g};\ms{m}, \ms{w}^0, \ms{f}^1} + \rho^2 \mc{S}\paren{\ms{g}; \ms{m}, \ms{w}^2, \ms{f}^1} \\
        &\notag + \rho^2 \mc{S}\paren{\ms{g}; \ms{m}^2, \ul{\ms{h}}^{2}} + \rho \mc{S}\paren{\ms{g}; \ms{m}^2, \ms{f}^1} + \mc{S}\paren{\ms{g}; \ms{m}, \ms{f}^1} + \rho^2\mc{S}\paren{\ms{g}; \ms{m}, \ms{w}^1, \ms{f}^0} \\
        &\notag + \rho\mc{S}\paren{\ms{g}; \ms{Dm}, \ms{f}^0}+ \rho^2 \mc{S}\paren{\ms{g}; \ms{Dm}, \ul{\ms{h}}^0} +\rho^2\mc{S}\paren{\ms{g}; \ms{w}^2, \ul{\ms{h}}^{2}} \\
        &\notag + \rho^2 \mc{S}\paren{\ms{g}; \ms{w}^0, \ul{\ms{h}}^{2}} + \rho^2 \mc{S}\paren{\ms{g}; \ms{w}^1, \ul{\ms{h}}^0}+\rho^2\mc{S}\paren{\ms{g}; (\ms{m})^3, \ms{f}^1} + \rho^2\mc{S}\paren{\ms{g}; \ms{m}, \ms{Dm}, \ms{f}^0} \text{,}
    \end{align}}
\end{proposition}

\begin{proposition}\label{prop_wave_w}
Let $(\mc{M}, g, F)$ be a FG-aAdS segment and let $(U, \varphi)$ be a local coordinate system on $\mc{I}$. Then, the following wave equations for the vertical Weyl curvatures hold with respect to $(U,\varphi)$:
    \begin{gather}\label{wave_w}
        \ol{\Box}\ms{w}^\star=\mc{R}_{W, \ms{w}^2}\text{,}\qquad
        (\ol{\Box}+2(n-2))\ms{w}^2 =\mc{R}_{W, \ms{w}^2}\text{,}\qquad 
        (\ol{\Box} + (n-1))\ms{w}^1 = \mc{R}_{W, \ms{w}^1}\text{,}
    \end{gather}
    where the remainder terms can be written as: \allowdisplaybreaks{
    \begin{align}
        \label{remainder_wave_w_2}\mc{R}_{W, \ms{w}^2} := &\rho^2\mc{S}\paren{\ms{g}; \ms{m}, \ms{D}\ms{w}^1} + \rho^2\mc{S}\paren{\ms{g}; \ms{D}\ms{m}, \ms{w}^1} + \rho\mc{S}\paren{\ms{g}; \ms{m}, \ms{w}^2}\\
        +&\notag\rho\mc{S}\paren{\ms{g}; \ms{m}, \ms{w}^0} + \rho^2\mc{S}\paren{\ms{g}; (\ms{m})^2, \ms{w}^0} + \rho^2\mc{S}\paren{\ms{g};(\ms{w}^0)^2} \\
        +&\notag \rho^2\mc{S}\paren{\ms{g}; \ms{w}^0, \ms{w}^2}+\rho^2\mc{S}\paren{\ms{g}; (\ms{m})^2, \ms{w}^2}+ \rho^2\mc{S}\paren{\ms{g}; (\ms{w}^2)^2}\\
        +&\notag \rho^2\mc{S}\paren{\ms{g}; (\ms{w}^1)^2}+\rho^2\mc{S}\paren{\ms{g}; \ms{D}\ul{\ms{h}}^{2}, {\ms{f}}^0}+\rho^2\mc{S}\paren{\ms{g}; \ms{Dm}, \ms{f}^1, \ms{f}^0}\\
        +&\notag \rho^2\mc{S}\paren{\ms{g}; \ms{D}\ul{\ms{h}}^0, {\ms{f}}^1}+\mc{S}\paren{\ms{g}; ({\ms{f}}^0)^2}+\notag\mc{S}\paren{\ms{g}; ({\ms{f}}^1)^2}+\rho^2\mc{S}\paren{\ms{g}; \ms{w}^0, ({\ms{f}}^1)^2}\\
        +&\notag \rho^2\mc{S}\paren{\ms{g};\ms{w}^2, ({\ms{f}}^0)^2}+ \rho^2\mc{S}\paren{\ms{g};  \ms{w}^1, {\ms{f}}^0, {\ms{f}}^1} +\rho^2\mc{S}\paren{\ms{g};  \ms{w}^2, ({\ms{f}}^1)^2}\\
        +&\notag\rho^2  \mc{S}\paren{\ms{g}; \ms{w}^0, ({\ms{f}}^0)^2}+\rho^2 \mc{S}\paren{\ms{g}; \ms{m}^2, (\ms{f}^0)^2}+\rho^2\mc{S}\paren{\ms{g}; \ms{m}^2, (\ms{f}^1)^2}\\
        +&\notag\rho^2\mc{S}\paren{\ms{g}; (\ul{\ms{h}}^0)^2} + \rho^2\mc{S}\paren{\ms{g}; (\ul{\ms{h}}^{2})^2} +\rho\mc{S}\paren{\ms{g}; \ul{\ms{h}}^{2}, {\ms{f}}^1} +\rho\mc{S}\paren{\ms{g}; \ul{\ms{h}}^0, {\ms{f}}^0}\\
        +&\notag\rho^2\mc{S}\paren{\ms{g}; \ms{m}, \ul{\ms{h}}^{2}, {\ms{f}}^1}+\rho^2\mc{S}\paren{\ms{g}; (\ms{f}^1)^4} + \rho\mc{S}\paren{\ms{g}; \ms{m}, ({\ms{f}}^1)^2} +\rho^2\mc{S}\paren{\ms{g}; \ms{Dm}, \ms{f}^0, \ms{f}^1}\\
        +&\notag  \rho^2\mc{S}\paren{\ms{g};  ({\ms{f}}^1)^2, ({\ms{f}}^0)^2}+\rho^2\mc{S}\paren{\ms{g}; ({\ms{f}}^0)^4}+\rho\mc{S}\paren{\ms{g}; \ms{m}, ({\ms{f}}^0)^2}  +\rho^2\mc{S}\paren{\ms{g}; \ms{m}, \ul{\ms{h}}^0, {\ms{f}}^0} \text{,}\\
        \label{remainder_wave_w_1}\mc{R}_{W, \ms{w}^1}:=& \rho^2\mc{S}\paren{\ms{g}; \ms{m}, \ms{D}\ms{w}^0} + \rho^2\mc{S}\paren{\ms{g}; \ms{m}, \ms{D}\ms{w}^2} + \rho^2\mc{S}\paren{\ms{g}; \ms{D}\ms{m}, \ms{w}^0}\\
        &\notag +\rho^2\mc{S}\paren{\ms{g}; \ms{D}\ms{m}, \ms{w}^2} +\rho\mc{S}\paren{\ms{g}; \ms{m}, \ms{w}^1} + \rho^2\mc{S}\paren{\ms{g}; (\ms{m})^2, \ms{w}^1}\\
        &\notag + \rho^2\mc{S}\paren{\ms{g}; \ms{w}^1, \ms{w}^2} + \rho^2\mc{S}\paren{\ms{g}; \ms{w}^0, \ms{w}^1}+ \rho^2\mc{S}\paren{\ms{g}; \ms{D}\ul{\ms{h}}^0, {\ms{f}^0} }\\
        &\notag  + \rho^2\mc{S}\paren{\ms{g}; \ms{D}\ul{\ms{h}}^{2}, {\ms{f}}^1}+\mc{S}\paren{\ms{g}; {\ms{f}}^0, {\ms{f}}^1}+ \rho^2\mc{S}\paren{\ms{g}; \ul{\ms{h}}^0, \ul{\ms{h}}^{2}}\\
        &\notag +\rho^2\mc{S}\paren{\ms{g};  \ms{w}^0, {\ms{f}}^0, {\ms{f}}^1} + \rho^2\mc{S}\paren{\ms{g}; \ms{w}^1, ({\ms{f}}^0)^2} + \rho^2\mc{S}\paren{\ms{g}; \ms{w}^2, {\ms{f}}^0, {\ms{f}}^1}\\
        &\notag+\rho\mc{S}\paren{\ms{g};  \ms{f}^0, \ul{\ms{h}}^{2}} + \rho\mc{S}\paren{\ms{g}; \ul{\ms{h}}^0, {\ms{f}}^1} + \rho^2\mc{S}\paren{\ms{g};  ({\ms{f}}^1)^3, {\ms{f}}^0}\\
        &\notag + \rho^2\mc{S}\paren{\ms{g};  {\ms{f}}^1, ({\ms{f}}^0)^3}+ \rho^2\mc{S}\paren{\ms{g};  \ms{w}^1, ({\ms{f}}^1)^2}+ \rho^2 \mc{S}\paren{\ms{g}; \ms{m}, \ul{\ms{h}}^{2}, {\ms{f}}^0} +\rho^2\mc{S}\paren{\ms{g}; (\ms{m})^2, \ms{f}^1, \ms{f}^0}\\
        &\notag+ \rho^2\mc{S}\paren{\ms{g}; \ms{m},\ul{\ms{h}}^0, {\ms{f}}^1}+\rho\mc{S}\paren{\ms{g}; \ms{m}, {\ms{f}}^1, {\ms{f}}^0} +\rho^2\mc{S}\paren{\ms{g}; \ms{Dm}, (\ms{f}^0)^2}+\rho^2\mc{S}\paren{\ms{g}; \ms{Dm}, (\ms{f}^1)^2}\text{.}
    \end{align}}
\end{proposition}

\begin{proof}
See Appendix \ref{app:prop_wave_w}.
\end{proof}

Before looking at the difference fields and the equations they satisfy, let us introduce a useful notation that will allow us to handle the error terms in the wave-transport system by simply looking at the asymptotic of the different fields involved. 

\begin{definition}
    Let $(\mc{M}, g)$ be a FG-aAdS segment, $M\geq 0$ and let $\phi\in C^\infty(\mc{M})$ be a positive scalar field. 
    \begin{itemize}
        \item We denote by $\mc{O}_M({\phi})$ any vertical tensor field $\ms{A}$ satisfying: 
        \begin{equation*}
            \abs{\ms{A}}_{M, \varphi}\leq C_{M, \varphi} \phi\text{,}
        \end{equation*}
        with $C_{M, \varphi}$ a positive constant. 
        \item We will also write $\mc{O}_M(\phi; \,\ms{B})$ any expression of the form, with $\ms{B}$ a vertical tensor field: 
        \begin{equation*}
            \mc{S}\paren{\ms{A}, \ms{B}}\text{,}
        \end{equation*}
        and $\ms{A}=\mc{O}_M\paren{\phi}$. 
    \end{itemize}
\end{definition}
The purpose of the following proposition is to write the asymptotics of the different fields involved here using the above notation. 

\begin{proposition}\label{prop_asymptotics_fields}
    Let $(\mc{M}, g, F)$ be a regular Maxwell-FG-aAdS segment of order $M\geq n+2$, as in Definition \ref{def_regular_M}. The following asymptotics hold for the metric $\ms{g}$ and the Maxwell fields ${\ms{f}}^0, {\ms{f}}^1$: 
    \begin{gather}
        \label{asymptotics_g}\ms{g} = \mc{O}_M(1)\text{,}\qquad \ms{g}^{-1}=\mc{O}_M(1)\text{,}\qquad \ms{m} = \mc{O}_{M-2}(\rho) = \mc{O}_{M-1}(1)\text{,}\qquad \Lie_\rho \ms{m} = \mc{O}_{M-2}(1)\text{,}\\
        \notag{\ms{f}}^0 = \mc{O}_{M}(\rho)\text{,} \qquad {\ms{f}}^1 = \mc{O}_{M}(\rho)\text{.}
    \end{gather}
    Furthermore, the following asymptotics hold for the Weyl tensor and the higher-derivative fields:
    \begin{gather*}
        \ms{w}^\star = \mc{O}_{M-2}(1)\text{,}\qquad \Lie_\rho \ms{w}^\star = \mc{O}_{M-3}(1)\text{,}\\
        \notag \ms{w}^1\text{,}\, \ms{w}^2 = \mc{O}_{M-2}(1) = \mc{O}_{M-3}(\rho)\text{,}\qquad \Lie_\rho \ms{w}^1\text{,}\, \Lie_\rho \ms{w}^2 = \mc{O}_{M-3}(1)\\
        \notag\ul{\ms{h}}^0\text{, }\ul{\ms{h}}^{2} = \mc{O}_{M-1}(\rho)\text{,}\qquad
        \Lie_\rho \ul{\ms{h}}^0\text{, }\Lie_\rho \ul{\ms{h}}^2= \mc{O}_{M-2}(1)\text{.}
    \end{gather*}
    \end{proposition}
    \begin{proof}
        See Appendix \ref{app:prop_asymptotics_fields}.
    \end{proof}

    \subsection{Difference relations}\label{sec:difference}

    In order to describe difference relations, we are lead to consider two FG-aAdS segments: 
    \begin{gather*}
        \mc{M}:=(0, \rho_0]\times \mc{I}\text{,}\qquad g := \rho^{-2}\paren{d\rho^2 + \ms{g}_{ab}dx^adx^b}\text{,}\qquad \check{g} := \rho^{-2}\paren{d\rho^2 + \check{\ms{g}}_{ab}dx^adx^b}\text{.}
    \end{gather*}
    \begin{definition}[Notations]
        Let $(\mc{M}, g)$ and $(\mc{M}, \check{g})$ denote two FG-aAdS segment. We will denote with a $\check{}$ operators and fields associated to $(\mc{M}, \check{g})$. Furthermore, we will denote the difference of two associated vertical fields $\ms{T}$ and $\check{\ms{T}}$ by 
        \begin{equation*}
        \delta \ms{T} := \ms{T}- \check{\ms{T}}\text{.}
        \end{equation*}
    \end{definition}
    We now extend $(\mc{M}, g)$ and $(\mc{M}, \check{g})$ to Maxwell-FG-aAdS segments. Namely, we will consider the two triplets $(\mc{M}, g, F)$ and $(\mc{M}, \check{g}, \check{F})$ with $F$ and $\check{F}$ solving the Maxwell and Bianchi equations with respect to $g$ and $\check{g}$, respectively.
    
    \begin{remark}
        We keep the two manifolds $\mc{M}$ and $\mc{I}$ identical and keep the gauge \eqref{FG-gauge} untouched. The two spacetimes therefore differ by the vertical metrics $\ms{g}, \check{\ms{g}}$. This choice of gauge will also allow us to work with the same vertical bundles. As a consequence, the expression of the vertical fields, as well as their transport and wave equations are identical, with only the metric differing. We will describe in Section \ref{sec:full_result} how to deal with more general gauges. 
    \end{remark}

    \begin{definition}
        For any collection of vertical tensor fields $\ms{A}_1, \dots, \ms{A}_N$, $N\geq 0$,  and let: 
        \begin{equation*}
            \mc{A} := \mc{S}\paren{\ms{g}; \ms{A}_1,\dots, \ms{A}_N}\text{,}
        \end{equation*}
        be some expression involving $\ms{g}\text{, }\ms{A}_1\text{, }\dots\text{, }\ms{A}_N$.
        We will write: 
        \begin{equation*}
            \delta \mc{S}\paren{\ms{g}; \ms{A}_1, \dots, \ms{A}_N} := \delta\mc{A}\text{.}
        \end{equation*}
    \end{definition}
    The above definition will allow us to consider difference of remainder terms which have the exact same algebraic structure. The remark below will allows us to treat such differences as remainder of the differences.  

    \begin{remark}\label{rmk_difference_remainders}
        Observe that: 
        \begin{align*}
            \delta\mc{S}\paren{\ms{g}; \ms{A}_1, \dots, \ms{A}_N} =& \sum\limits_{i=1}^N \mc{S}\paren{\ms{g}; \ms{A}_1, \dots, \ms{A}_{i-1},\delta\ms{A}_i, \ms{A}_{i+1}, \dots, \ms{A}_N} + \mc{S}\paren{\ms{g}; \delta\ms{g}, \ms{A}_1, \dots, \ms{A}_N}+\mc{R}\\
            =:&\mc{Q} + \mc{R}\text{,}
        \end{align*}
        where $\mc{R}$ will, at worst, yield the same asymptotics as $\mc{Q}$. This is due to the fact that $\mc{R}$ contains similar terms as $\mc{Q}$, with some of the $\ms{A}_i$'s replaced by $\check{\ms{A}}_i$'s. Note that the vertical fields $\ms{A}_i$ and $\check{\ms{A}}_i$ have the same asymptotics.
        Furthermore, $\mc{R}$ will contain powers of $\ms{g}, \ms{g}^{-1}, \check{\ms{g}}, \check{\ms{g}}^{-1}$, which are all $\mc{O}_{M}(1)$.
    \end{remark}

    The challenge in this section will be to obtain a closed system (for the Carleman estimates) of the difference of fields. Obtaining this system will be enough for us to conclude on the uniqueness, using appropriate Carleman estimates accounting the geometry of FG-aAdS spacetimes.

    This will come at the cost of introducing renormalised fields:
    \begin{definition}\label{def_renormalisation}
        Let $(\mc{M}, g, F)$ and $(\mc{M}, \check{g}, F)$ denote two Maxwell-FG-aAdS segments. We define, with respect to any $(U, \varphi)$ a compact coordinate system on $\mc{I}$: 
        \begin{enumerate}
            \item $\ms{Q}$, an anti-symmetric $(0,2)$ vertical tensor field, solution to the following transport equation: 
            \begin{gather}\label{transport_Q}
                \Lie_\rho \ms{Q}_{ab} = -\ms{g}^{de}\ms{m}_{d[a}(\delta \ms{g}+ \ms{Q})_{b]e}\text{,}\qquad \ms{Q}\rightarrow^0 0\text{.}
            \end{gather}
            \item The $(0, 3)$--vertical tensor field $\ms{B}$, defined by: 
            \begin{equation}\label{def_B}
                \ms{B}_{abc} := 2\ms{D}_{[a}\delta \ms{g}_{b]c} - \ms{D}_c\ms{Q}_{ab}\text{,}
            \end{equation}
            \item The renormalised difference $\Delta$ of vertical fields by the following: for any $\ms{N}_{\bar{a}}$, $\check{\ms{N}}_{\bar{a}}$, vertical tensors of the same rank $(0, l)$: 
            \begin{equation}\label{renormalised_def}
                \Delta \ms{N}_{\bar{a}} := \delta \ms{N}_{\bar{a}} - \frac{1}{2} \ms{g}^{bc}\sum\limits_{j=1}^{l}\ms{N}_{\hat{a}_j[b]}(\delta \ms{g} + \ms{Q})_{a_j c}
            \end{equation}
        \end{enumerate}
    \end{definition}

\begin{proposition}\label{prop_diff_g_gamma}
    Let $(\mc{M}, g, F)$ and $(\mc{M}, \check{g}, \check{F})$ be two Maxwell-FG-aAdS segments regular to order $M\geq n+2$. Then, with respect to any compact coordinate system $(U, \varphi)$ on $\mc{I}$: 
    \begin{itemize}
        \item The following asymptotics are satisfied:
        \begin{equation}\label{asymptotics_difference}
            \delta\ms{g}^{-1} = \mc{O}_{M}(1; \delta\ms{g})\text{,}\qquad \ms{Q} = \mc{O}_{M-1}(\rho) = \mc{O}_{M-2}(\rho^2)\text{,}\qquad \ms{B} = \mc{O}_{M-2}(1). 
        \end{equation}
        \item The difference of the Christoffel symbols satisfy: 
        \begin{align*}
            &\delta \Gamma^{\alpha}_{\rho\rho} = \delta \Gamma^\rho_{\rho\alpha} = 0\text{,}\qquad \delta \Gamma^{\rho}_{ab} = \rho^{-1}\delta\ms{g}_{ab} +\frac{1}{2}\delta\ms{m}_{ab}\text{,}\\
            &\delta \Gamma^a_{\rho b} = \delta \ms{\Gamma}^a_{\rho b} = \mc{O}_{M-2}\paren{\rho; \delta\ms{g}}^a{}_b + \mc{O}_{M}\paren{1; \delta\ms{m}}^a{}_b\\
            &\delta \Gamma^{a}_{bc} = \delta \ms{\Gamma}^a_{bc} = \frac{1}{2}\check{\ms{g}}^{ad}\paren{\ms{D}_{c}\delta \ms{g}_{bd} + \ms{D}_b\delta\ms{g}_{cd}-\ms{D}_d\delta\ms{g}_{bc}}\\
            &\hspace{56pt} = \mc{O}_{M}(1; \ms{D}\delta\ms{g})^a{}_{bc}\text{.}
        \end{align*}
    \end{itemize}
\end{proposition}

\begin{proof} 
    See \cite{Holzegel22}, Proposition 3.12.  
\end{proof}

The next proposition establishes the transport equations satisfied by the difference of vertical fields. 

\begin{proposition}\label{prop_transport_diff}
    Let $(\mc{M}, g, F)$ and $(\mc{M}, \check{g}, \check{F})$ be two Maxwell-FG-aAdS segments regular to order $M\geq n+2$, and let $(U, \varphi)$ be a compact coordinate system on $\mc{I}$. Then,
    \allowdisplaybreaks{
    \begin{align}
        &\Lie_\rho \delta \ms{g} = \delta \ms{m}\text{,}\\
        &\Lie_\rho \paren{\rho^{-1}\delta \ms{m}} = -2\rho^{-1}\Delta\ms{w}^2 + \mc{O}_{M-3}(1; \delta\ms{g}) + \mc{O}_{M-3}(1; \ms{Q}) + \mc{O}_{M-2}(1; \delta\ms{m})\text{,}\label{transport_dm}\\
        &\notag\hspace{40pt} + \mc{O}_{M}\paren{1; \delta\ms{f}^0} + \mc{O}_{M}\paren{1; \delta \ms{f}^1}\text{,}\\
        &\label{Lie_Q}\Lie_\rho \ms{Q} = \mc{O}_{M-2}(\rho; \delta \ms{g}) + \mc{O}_{M-2}(\rho; \ms{Q})\text{,}\\
        &\Lie_\rho \ms{B} = \mc{S}(\Delta\ms{w}^1) + \mc{O}_{M-2}(\rho; \ms{B}) + \mc{O}_{M-3}(\rho; \delta \ms{g}) + \mc{O}_{M-1}(1; \delta\ms{m}) + \mc{O}_{M-3}(\rho;\ms{Q})\text{,} \\
        &\notag \hspace{35pt} +\mc{O}_{M}(\rho;\delta\ms{f}^0) + \mc{O}_{M}(\rho; \delta\ms{f}^1)\text{,}\\
        &\notag\Lie_\rho(\rho^{-(n-2)}\delta\ms{f}^0) =\mc{O}_M(\rho^{-(n-2)};\Delta\ul{\ms{h}}^0) + \mc{O}_{M-1}\paren{\rho^{-(n-3)};\delta\ms{g}} + \mc{O}_{M} \paren{\rho^{-(n-3)}; \delta\ms{m}}\\
        &\label{transport_df0}\hspace{40pt}+ \mc{O}_{M-2}(\rho^{-(n-3)}; \delta \ms{f}^0)+\mc{O}_{M-1}(\rho^{-(n-3)}; \ms{Q})\text{,}\\
        &\label{transport_df1}\Lie_\rho(\rho^{-1}\delta\ms{f}^1) = \mc{O}_M\paren{\rho^{-1}; \Delta \ul{\ms{h}}^{2}} + \mc{O}_{M-1}(1; \delta\ms{g})+\mc{O}_{M}(1;\delta\ms{m})+\mc{O}_{M-2}(1; \delta\ms{f}^1)+\mc{O}_{M-1}(1; \ms{Q})\text{.}
    \end{align}
    Furthermore, the following derivative transport equations are satisfied:
    \begin{align}
        &\label{transport_Ddg}\Lie_\rho {\ms{D}\delta\ms{g}} = \ms{D}\delta \ms{m} + \mc{O}_{M-3}(\rho; \delta\ms{g})\text{,}\\
        &\notag\Lie_\rho \ms{D} (\rho^{-1}\delta\ms{m}) = 2\rho^{-1}\ms{D}\Delta \ms{w}^2 + \mc{O}_{M-3}(1;  \ms{D}\delta\ms{g}) + \mc{O}_{M-3}\paren{1; \ms{DQ}} + \mc{O}_{M-2}(1;  \ms{D}\delta{\ms{m}})\\
        &\notag\hspace{40pt} + \mc{O}_{M-4}(1; \delta\ms{g}) + \mc{O}_{M-4}\paren{1; \ms{Q}} + \mc{O}_{M-3}(1; \delta\ms{m}) + \mc{O}_{M}\paren{1; \Delta\ul{\ms{h}}^{2}}\\
        &\hspace{40pt}\mc{O}_{M}\paren{1; \Delta\ul{\ms{h}}^0} + \mc{O}_{M-1}\paren{1; \delta\ms{f}^0}+\mc{O}_{M-1}\paren{1; \delta \ms{f}^1}\text{,}\\ 
        &\Lie_\rho \ms{D}\ms{Q} = \mc{O}_{M-3}(\rho; \ms{Q}) + \mc{O}_{M-2}(\rho; \ms{D}\delta\ms{g}) + \mc{O}_{M-2}(\rho; \ms{DQ}) + \mc{O}_{M-3}(\rho;\delta\ms{g}) \text{,}\\
        &\notag\Lie_\rho \ms{DB} = \mc{S}(\ms{D}\Delta\ms{w}^1) + \mc{O}_{M-2}\paren{\rho; \ms{DB}} + \mc{O}_{M-3}(\rho; \ms{B}) + \mc{O}_{M-3}(\rho; \ms{D}\delta\ms{g}) + \mc{O}_{M-4}(\rho; \delta\ms{g}) \\
        &\notag\hspace{40pt} + \mc{O}_{M-1}(1; \ms{D}\delta{\ms{m}}) + \mc{O}_{M-2}(1; \delta\ms{m}) + \mc{O}_{M-3}(\rho; \ms{DQ}) + \mc{O}_{M-4}(\rho; \ms{Q})\\
        &\hspace{40pt}+\mc{O}_{M}(\rho; \Delta\ul{\ms{h}}^{2}) + \mc{O}_{M-1}(\rho; \delta\ms{f}^0) + \mc{O}_{M}(\rho; \Delta\ul{\ms{h}}^0) + \mc{O}_{M-1}(\rho; \delta\ms{f}^1)\text{,}\\
        &\notag\Lie_\rho \ms{D}(\rho^{-(n-2)}\delta\ms{f}^0) = \mc{S}(\rho^{2-n}\ms{D}\Delta\ul{\ms{h}}^0) + \mc{O}_{M-3}(\rho^{3-n}; \delta\ms{f}^0) + \mc{O}_{M-1}(\rho^{3-n}; \ms{D}\delta\ms{g}) \\
        &\notag\hspace{40pt} +\mc{O}_{M-2}(\rho^{3-n}: \delta\ms{g})+ \mc{O}_{M}(\rho^{3-n}; \ms{D}\delta\ms{m})+\mc{O}_{M-1}(\rho^{3-n}; \delta\ms{m})+\mc{O}_{M-2}(\rho^{3-n}; \ms{D}\Delta\ul{\ms{h}}^{2})\\
        &\hspace{40pt}+\mc{O}_{M-2}(\rho^{-(n-3)}; \ms{Q})+\mc{O}_{M-2}(\rho^{-(n-3)}; \ms{DQ})\text{,}\\
        &\notag\Lie_\rho \ms{D}(\rho^{-1}\delta\ms{f}^1) = \mc{S}(\rho^{-1}\ms{D}\Delta\ul{\ms{h}}^{2}) + \mc{O}_{M-3}(1; \delta\ms{f}^1) + \mc{O}_{M-1}(1; \ms{D}\delta\ms{g})+\mc{O}_{M-2}(1;\delta\ms{g}) + \mc{O}_{M-1}(1; \delta\ms{m}) \\
        &\hspace{40pt}+ \mc{O}_{M}(1; \ms{D}\delta\ms{m}) +  \mc{O}_{M-2}(1; \Delta \ul{\ms{h}}^0)+\mc{O}_{M-2}(1; \ms{Q})+\mc{O}_{M-1}(1; \ms{DQ})\label{transport_Ddf1}\text{.}
    \end{align}}
\end{proposition}

\begin{proof}
    See Appendix \ref{app:prop_transport_diff}. 
\end{proof}

We will also need the wave equations satisfied by the differences of the Weyl curvature components, as well as ${H}$. 

\begin{proposition}[Difference wave equations]\label{prop_wave_difference}
    Let $(\mc{M}, g, F)$ and $(\mc{M}, \check{g}, \check{F})$ be two FG-aAdS segment and let $(U, \varphi)$ be any compact coordinate system on $\mc{I}$. Then, the following vertical wave equations are satisfied 
        \begin{align}
            &(\ol{\Box} + m_{\ms{w}})\Delta\ms{w} = \ms{F}_{\ms{w}}\text{,}\qquad \ms{w}\in\lbrace \ms{w}^\star, \ms{w}^1, \ms{w}^2\rbrace\text{,}\\
            &(\ol{\Box} + m_{\ms{h}})\Delta\ms{h} = \ms{G}_{\ms{h}}\text{,}\qquad \ms{h}\in\lbrace\ul{\ms{h}}^0, \ul{\ms{h}}^2\rbrace\text{,}
        \end{align}
        with the masses given by: 
        \begin{gather*}
            m_{\ms{w}^\star} = 0\text{,}\qquad m_{\ms{w}^1} = n-1\text{,}\qquad m_{\ms{w}^2} = 2(n-2)\text{,}\\
            m_{\ul{\ms{h}}^0} = n-1\text{,}\qquad m_{\ul{\ms{h}}^{2}} = 2(n-2)\text{,}
        \end{gather*} 
        and $\ms{F}_{\ms{w}}$, $\ms{G}_{\ms{h}}$ of the form: \allowdisplaybreaks{
        \begin{align}
            \label{F_w}\ms{F}_{\ms{w}} =& \, \mc{O}_{M-4}(\rho^2; \delta\ms{g}) + \mc{O}_{M-3}(\rho^2; \ms{Q}) + \mc{O}_{M-3}(\rho; \delta\ms{m}) \\
            &\notag+\mc{O}_{M-3}(\rho^2; \ms{D}\delta\ms{g}) + \mc{O}_{M-3}(\rho^2; \ms{DQ}) + \mc{O}_{M-2}(\rho^2; \ms{D}\delta\ms{m})+\mc{O}_{M-2}(\rho^2; \ms{DB})\\
            &\notag+\mc{O}_{M-3}(\rho; \delta\ms{f}^0) + \mc{O}_{M-3}(\rho; \delta\ms{f}^1)\\
            &\notag+\sum\limits_{\ms{w}\in\lbrace\ms{w}^\star, \ms{w}^1, \ms{w}^2\rbrace} \mc{O}_{M-3}(\rho^2; \Delta\ms{w}) + \mc{O}_{M-2}(\rho^3; \ms{D}\Delta\ms{w})\\
            &\notag+\sum\limits_{\ms{h}\in\lbrace\ul{\ms{h}}^0, \ul{\ms{h}}^2\rbrace} \mc{O}_{M-3}(\rho^2; \Delta\ms{h}) + \mc{O}_{M-2}(\rho^3; \ms{D}\Delta\ms{h})\text{,}\\
            \label{G_h}\ms{G}_{\ms{h}} =&\, \mc{O}_{M-4}(\rho^2; \delta\ms{g}) + \mc{O}_{M-3}(\rho; \delta\ms{m}) + \mc{O}_{M-3}(\rho^3;\ms{D}\delta\ms{g}) + \mc{O}_{M-2}(\rho^2; \ms{D}\delta\ms{m}) \\
            &\notag+\mc{O}_{M-3}(\rho^2;\ms{Q})+\mc{O}_{M-3}(\rho^3; \ms{DQ}) + \mc{O}_{M-2}(\rho^3; \ms{DB})\\
            &\notag+\mc{O}_{M-3}(\rho;\delta\ms{f}^0)+\mc{O}_{M-3}(\rho;\delta\ms{f}^1) \\
            &\notag+\sum\limits_{\ms{w}\in\lbrace\ms{w}^\star, \ms{w}^1, \ms{w}^2\rbrace} \mc{O}_{M-3}(\rho^2; \Delta\ms{w}) + \mc{O}_{M-2}(\rho^3; \ms{D}\Delta\ms{w})\\
            &\notag+\sum\limits_{\ms{h}\in\lbrace\ul{\ms{h}}^0, \ul{\ms{h}}^2\rbrace} \mc{O}_{M-3}(\rho^2; \Delta\ms{h}) + \mc{O}_{M-2}(\rho^3; \ms{D}\Delta\ms{h})\text{.}
        \end{align}}
\end{proposition}

\begin{proof}
    We will write here $(\ms{w}, \check{\ms{w}}, \Delta\ms{w})\in\{(\ms{w}^2, \check{\ms{w}}^2, \Delta\ms{w}^2),(\ms{w}^1, \check{\ms{w}}^1, \Delta\ms{w}^1), (\ms{w}^\star, \check{\ms{w}}^\star, \Delta\ms{w}^\star)\}$ and \\$(\ms{h}, \check{\ms{h}}, \Delta\ms{h})\in \{ (\ul{{\ms{h}}}^0, \ul{\check{{\ms{h}}}}^0, \Delta \ul{\ms{h}}^0), (\ul{\ms{h}}^{2}, \ul{\check{\ms{h}}}^{2}, \Delta\ul{\ms{h}}^{2}) \}$. We will also denote by $\ms{A}$ a general vertical tensor of rank $(0, l)$. Typically, if $\ms{A}$ and $\check{\ms{A}}$ satisfy wave equations of the form: 
    \begin{equation}
        (\ol{\Box} + m_{\ms{A}})\ms{A} = \mc{R}_{W, \ms{A}}\text{,}\qquad (\check{\ol{\Box}} + m_{\ms{A}})\check{\ms{A}} = \check{\mc{R}}_{W, \ms{A}}\text{,}
    \end{equation}
    with $\mc{R}_{W, \ms{A}}, \check{\mc{R}}_{W, \ms{A}}$ lower-order terms, then: 
    \begin{align*}
        (\ol{\Box} + m_{\ms{A}})\Delta \ms{A} =& \left[(\ol{\Box} + m_{\ms{A}})\ms{A}- (\check{\ol{\Box}} + m_{\ms{A}})\check{\ms{A}}\right] - (\ol{\Box}- \check{\ol{\Box}})\check{\ms{A}}\\
        &- \frac{1}{2}\ms{g}^{bc}\sum\limits_{j=1}^{l_2} (\ol{\Box}+m_{\ms{A}})({\ms{A}}_{\hat{a}_j[b]}(\delta\ms{g} + \ms{Q}))_{a_jc}\text{,}\\
        =:&I_{1, \ms{A}} + I_{2, \ms{A}} + I_{3, \ms{A}}\text{.}
    \end{align*}
    The first term, $I_1$ can be easily dealt with since it can simply written  as:
    \begin{equation*}
        I_{1,\ms{A}}=\mc{R}_{W, \ms{A}} - \check{\mc{R}}_{W, \ms{A}}\text{.}
    \end{equation*}
    The goal is therefore to write differences of remainder terms as remainder terms of the difference of fields, inspired by Remark \ref{rmk_difference_remainders}. Typically,  starting with the remainder terms in the wave equations for $\ms{h}$, one has, for example: 
    \begin{align*}
        \rho^2 \delta\mc{S}\paren{\ms{g}; \ms{w}^1, \ul{\ms{h}}^0}=&\, \rho^2\mc{S}(\ms{g}; \delta\ms{g}, \ms{w}^1, \ul{\ms{h}}^0) + \rho^2 \mc{S}(\ms{g}; \delta\ms{w}^1, \ul{\ms{h}}^0)+\rho^2 \mc{S}(\ms{g}; \ms{w}^1, \delta\ul{\ms{h}}^0)+\mc{R} \\
        =&\, \mc{O}_{M-3}(\rho^4; \delta \ms{g}) + \mc{O}_{M-1}(\rho^3; \delta\ms{w}^1) + \mc{O}_{M-3}(\rho^3; \delta\ul{\ms{h}}^0)\text{,}
    \end{align*}
    with $\mc{R}$ as in Remark \ref{rmk_difference_remainders}, i.e. having the same asymptotics as the three first terms terms. 
    The remainder also contains terms with vertical covariant derivatives. These can be treated as in the following example: 
    \begin{align*}
        \rho^2 \delta \mc{S}\paren{\ms{g}; \ms{m}, \ms{D}\ul{\ms{h}}^{2}} =& \rho^2 \mc{S}\paren{\ms{g}; \delta\ms{g}, \ms{m},\ms{D}\ul{\ms{h}}^{2}} + \rho^2 \mc{S}\paren{\ms{g}; \delta\ms{m}, \ms{D}\ul{\ms{h}}^{2}} \\
        &+\rho^2 \mc{S}\paren{\ms{g}; \ms{m}, \ms{D}\delta\ul{\ms{h}}^{2}} + \rho^2\mc{S}\paren{\ms{g}; \ms{m}, (\delta\ms{D})\ul{\ms{h}}^{2}}+\mc{R}\text{,}
    \end{align*}
    where the term $(\delta\ms{D})\ms{h}^{2,\pm}$ can be written as: 
    \begin{align*}
        (\delta\ms{D}_e)\ul{\ms{h}}^{2}_{ab} &= -(\delta\ms{\Gamma})^c_{ea}\ul{\ms{h}}^{2}_{cb}  -(\delta\ms{\Gamma})^c_{eb}\ul{\ms{h}}^{2}_{ac}\\
        &=\mc{O}_{M-1}(\rho; \ms{D}\delta\ms{g})_{eab}\text{,}
    \end{align*}
    where we used Proposition \ref{prop_diff_g_gamma}. As a consequence, 
    \begin{align*}
        \rho^2 \delta\mc{S}\paren{\ms{g}; \ms{m}, \ms{Dh}^{2, \pm}} =&\, \mc{O}_{M-2}\paren{\rho^{4}; \delta\ms{g}} + \mc{O}_{M-2}(\rho^{3}; \delta \ms{m}) + \mc{O}_{M-2}(\rho^3; \ms{D}\delta\ms{h}^{2,\pm})\\
        &+\mc{O}_{M-2}(\rho^{4}, \ms{D}\delta\ms{g})\text{.}
    \end{align*}
    Let $\underline{\ms{f}}\in \lbrace \ms{f}^0, \ms{f}^1\rbrace$, $\underline{\ms{w}}\in \lbrace \ms{w}^\star, \ms{w}^1, \ms{w}^2\rbrace$, and $\underline{\ms{h}}\in \lbrace \ul{\ms{h}}^0, \ul{\ms{h}}^{2}\rbrace$. In order to make the presentation easier, we will simply present the asymptotics, in the differences of $\ms{R}_{W, \ms{h}}$ or $\ms{R}_{W, \ms{w}}$, in terms of the representatives $\underline{\ms{f}}$, $\underline{\ms{w}}$ and $\underline{\ms{h}}$. The exact tensor field will not matter here. Starting with $\underline{\ms{h}}$: 
    \allowdisplaybreaks{
    \begin{align*}
        &\rho\delta\mc{S}\paren{\ms{g}; (\underline{\ms{f}})^3} = \mc{O}_{M}(\rho^{3}; \delta \underline{\ms{f}}) + \mc{O}_{M}(\rho^{4}; \delta \ms{g})\text{,}\\
        &\delta \mc{S}\paren{\ms{g}; \ms{m}, \underline{\ms{f}}} = \mc{O}_{M-2}(\rho^2; \delta\ms{g}) + \mc{O}_{M}(\rho;\delta\ms{m}) + \mc{O}_{M-2}(\rho; \delta\underline{\ms{f}})\text{,}\\
        &\rho^2\delta\mc{S}\paren{\ms{g}; (\underline{\ms{f}})^2, \underline{\ms{h}}} = \mc{O}_{M-1}(\rho^5; \delta \ms{g}) + \mc{O}_{M-1}\paren{\rho^4; \delta\underline{\ms{f}}} + \mc{O}_{M}\paren{\rho^4; \delta\underline{\ms{h}}}\text{,}\\
        &\rho^2 \delta\mc{S}\paren{\ms{g}; \ms{m}, \ms{D}\underline{\ms{h}}} = \mc{O}_{M-2}(\rho^4; \delta\ms{g}) +\mc{O}_{M-2}(\rho^3; \delta\ms{m}) + \mc{O}_{M-2}(\rho^3; \ms{D}\delta\underline{\ms{h}}) \\
        &\hspace{110pt}+\mc{O}_{M-2}(\rho^4; \ms{D}\delta\ms{g})\text{,}\\
        &\rho^2\delta\mc{S}\paren{\ms{g}; \ms{D}\underline{\ms{w}}, \underline{\ms{f}}} = \mc{O}_{M-3}(\rho^3; \delta\ms{g}) + \mc{O}_{M}(\rho^{3}; \ms{D}\delta\underline{\ms{w}})+ \mc{O}_{M-3}(\rho^3; \ms{D}\delta\ms{g}) + \mc{O}_{M-3}(\rho^2;\delta\underline{\ms{f}})\text{,}\\
        &\rho\delta\mc{S}\paren{\ms{g}; \underline{\ms{w}}, \underline{\ms{f}}} = \mc{O}_{M-2}(\rho^{2}; \delta\ms{g})+\mc{O}_{M} (\rho^{2}; \delta\underline{\ms{w}}) + \mc{O}_{M-2}(\rho;\delta\underline{\ms{f}})\text{,}\\
        &\rho^2\delta \mc{S}\paren{\ms{g}; \underline{\ms{w}}, \underline{\ms{h}}} = \mc{O}_{M-2}(\rho^3; \delta\ms{g}) + \mc{O}_{M-1}(\rho^3; \delta\underline{\ms{w}}) + \mc{O}_{M-2}(\rho^2; \delta\underline{\ms{h}})\text{,}\\
        &\rho^2\delta\mc{S}\paren{\ms{g}; \ms{Dm}, \underline{\ms{h}}} = \mc{O}_{M-2}(\rho^3; \delta\ms{g}) + \mc{O}_{M-1}(\rho^3; \ms{D}\delta\ms{m})+\mc{O}_{M-3}(\rho^3; \delta\underline{\ms{h}})+\mc{O}_{M-2}(\rho^4; \ms{D}\delta\ms{g})\text{,}\\
        &\rho^2\delta\mc{S}(\ms{g}; \ms{m}, \underline{\ms{w}}, \underline{\ms{f}}) = \mc{O}_{M-2}(\rho^4; \delta\ms{g}) + \mc{O}_{M-2}(\rho^3; \delta\ms{m}) + \mc{O}_{M-2}(\rho^3; \delta\ms{m}) + \mc{O}_{M-2}(\rho^4; \delta\underline{\ms{f}})\text{,}\\
        &\rho\delta\mc{S}\paren{\ms{g}; \ms{m}, \underline{\ms{h}}} = \mc{O}_{M-2}(\rho^3; \delta\ms{g}) + \mc{O}_{M-1}(\rho^2; \delta\ms{m}) + \mc{O}_{M-2}(\rho^2;\delta\underline{\ms{h}})\text{,}\\
        &\rho^2\delta\mc{S}\paren{\ms{g}; \ms{m}^2, \underline{\ms{h}}} = \mc{O}_{M-2}(\rho^5; \delta\ms{g}) + \mc{O}_{M-2}(\rho^4; \delta\ms{m})+\mc{O}_{M-2}(\rho^4; \delta\underline{\ms{h}})\text{,}\\
        &\rho^2\delta\mc{S}\paren{\ms{g}; \ms{m}^2, \underline{\ms{f}}}=\mc{O}_{M-2}(\rho^5; \delta\ms{g}) + \mc{O}_{M-2}(\rho^4; \delta\ms{m})+\mc{O}_{M-2}(\rho^4; \delta\underline{\ms{f}})\text{,}\\
        &\rho^2\delta\mc{S}\paren{\ms{g}; \ms{m}, (\underline{\ms{f}})^3} = \mc{O}_{M-2}(\rho^6; \delta\ms{g}) + \mc{O}_{M}(\rho^5; \delta\ms{m}) + \mc{O}_{M-2}(\rho^5; \delta\underline{\ms{f}})\text{,}\\
        &\rho\delta \mc{S}\paren{\ms{g}; \ms{Dm}, \underline{\ms{f}}} = \mc{O}_{M-3}(\rho^3; \delta\ms{g}) + \mc{O}_{M}(\rho^2; \ms{D}\delta\ms{m})+\mc{O}_{M-2}(\rho^2; \ms{D}\delta\ms{g}) + \mc{O}_{M-3}(\rho^2; \delta\underline{\ms{f}})\text{,}
    \end{align*}}
    while for $\underline{\ms{w}}$: 
    \allowdisplaybreaks{
    \begin{align*}
        &\rho^2\delta\mc{S}(\ms{g}; (\underline{\ms{w}})^2) = \mc{O}_{M-2}(\rho^2; \delta\ms{g}) + \mc{O}_{M-2}(\rho^2; \delta\underline{\ms{w}})\text{,}\\
        &\rho^2 \delta\mc{S}\paren{\ms{g}; (\underline{\ms{h}})^2} = \mc{O}_{M-1}\paren{\rho^4; \delta\ms{g}} + \mc{O}_{M-1}(\rho^3; \delta\underline{\ms{h}})\text{,}\\
        &\rho^2\delta \mc{S}\paren{\ms{g}; \ms{D}\underline{\ms{h}}, \underline{\ms{f}}} = \mc{O}_{M-2}(\rho^4; \delta\ms{g}) + \mc{O}_{M}(\rho^3; \ms{D}\delta\underline{\ms{h}}) + \mc{O}_{M-1}(\rho^4; \ms{D}\delta\ms{g}) +\mc{O}_{M-2}\paren{\rho^3; \delta\underline{\ms{f}}}\text{,}\\
        &\rho^2\delta\mc{S}\paren{\ms{g}; \underline{\ms{w}}, (\underline{\ms{f}})^2} = \mc{O}_{M-2}(\rho^{4}; \delta\ms{g}) + \mc{O}_{M}(\rho^{4}; \delta\underline{\ms{w}}) + \mc{O}_{M-2}\paren{\rho^{3}; \delta\underline{\ms{f}}}\text{,}\\
        &\rho^2\delta \mc{S}\paren{\ms{g}; (\underline{\ms{f}})^4} = \mc{O}_{M}(\rho^{6}; \delta\ms{g}) + \mc{O}_{M}(\rho^{5}; \delta\underline{\ms{f}})\text{,}\\
        &\rho^2\delta\mc{S}\paren{\ms{g}; \ms{m}, \ms{D}\underline{\ms{w}}} = \mc{O}_{M-2}(\rho^3; \delta\ms{g}) + \mc{O}_{M-2}(\rho^2; \delta\ms{m}) + \mc{O}_{M-2}(\rho^3;\ms{D}\delta\underline{\ms{w}}) + \mc{O}_{M-2}(\rho^3; \ms{D}\delta\ms{g})\text{,}\\
        &\rho^2 \delta\mc{S}\paren{\ms{g}; \ms{Dm}, \underline{\ms{w}}} = \mc{O}_{M-3}(\rho^3; \delta\ms{g}) + \mc{O}_{M-2}(\rho^2; \ms{D}\delta\ms{m}) + \mc{O}_{M-2}(\rho^2; \delta\underline{\ms{w}}) + \mc{O}_{M-2}(\rho^3; \ms{D}\delta\ms{g})\text{,}\\
        &\rho\delta\mc{S}\paren{\ms{g}; \ms{m}, \underline{\ms{w}}} = \mc{O}_{M-2}(\rho^2;\delta\ms{g}) + \mc{O}_{M-2}(\rho; \delta\ms{m}) + \mc{O}_{M-2}(\rho^2;\delta\underline{\ms{w}})\text{,}\\
        &\delta\mc{S}\paren{\ms{g}; (\underline{\ms{f}})^2} = \mc{O}_{M}(\rho^2; \delta\ms{g}) + \mc{O}_{M}(\rho; \delta\underline{\ms{f}})\text{,}\\
        &\rho\delta\mc{S}\paren{\ms{g}; \underline{\ms{h}}, \underline{\ms{f}}} = \mc{O}_{M-1}(\rho^3; \delta\ms{g}) + \mc{O}_{M}(\rho^2; \delta\underline{\ms{h}}) + \mc{O}_{M-1}(\rho^2; \delta\underline{\ms{f}})\text{,}\\
        &\rho^2\delta\mc{S}\paren{\ms{g}; (\ms{m})^2, \underline{\ms{w}}} = \mc{O}_{M-2}(\rho^4; \delta\ms{g}) + \mc{O}_{M-2}(\rho^3: \delta\ms{m}) + \mc{O}_{M-2}(\rho^4; \delta\underline{\ms{w}})\text{,}\\
        &\rho^2\delta\mc{S}\paren{\ms{g}; \ms{m}, \underline{\ms{h}}, \underline{\ms{f}}} = \mc{O}_{M-2}(\rho^5; \delta\ms{g}) + \mc{O}_{M-1}(\rho^4; \delta\ms{m}) +\mc{O}_{M-2}(\rho^4; \delta\underline{\ms{h}}) + \mc{O}_{M-2}(\rho^4;\delta\underline{\ms{f}})\text{,}\\
        &\rho\delta\mc{S}\paren{\ms{g}; \ms{m}, (\underline{\ms{f}})^2} = \mc{O}_{M-2}(\rho^4; \delta\ms{g}) + \mc{O}_{M}(\rho^3; \delta\ms{m}) + \mc{O}_{M-2}(\rho^3; \delta\underline{\ms{f}})\text{,}\\
        &\rho^2\delta\mc{S}\paren{\ms{g}; \ms{Dm}, (\underline{\ms{f}})^2} = \mc{O}_{M-3}(\rho^5; \delta\ms{g}) + \mc{O}_{M}(\rho^4; \ms{D}\delta\ms{m}) + \mc{O}_{M-2}(\rho^5; \ms{D}\delta\ms{g})+ \mc{O}_{M-2}(\rho^4; \delta\underline{\ms{f}})\text{.}
    \end{align*}}
    After a careful analysis, one can show that the $I_{1, \ms{h}}$ and $I_{1, \ms{w}}$ can be written as: 
    \begin{align*}
        &\begin{aligned}
            I_{1, \ms{h}} =& \mc{O}_{M-3}\paren{\rho^2; \delta\ms{g}} + \mc{O}_{M-3}(\rho; \delta\ms{m}) + \mc{O}_{M-3}(\rho; \delta\ms{f}^0) +\mc{O}_{M-3}(\rho;\delta\ms{f}^1)\\
            &+\mc{O}_{M-2}\paren{\rho^3; \ms{D}\delta\ms{g}} + \mc{O}_{M-2}\paren{\rho^2; \ms{D}\delta\ms{m}} \\
            &+ \sum\limits_{\ms{h}\in\{\ul{\ms{h}}^0, \ul{\ms{h}}^{2}\}}\mc{O}_{M-3}(\rho^2; \delta \ms{h}) + \mc{O}_{M-2}(\rho^3;\ms{D}\delta\ms{h}) \\
            &+ \sum\limits_{\ms{w}\in\lbrace \ms{w}^0, \ms{w}^1, \ms{w}^2\rbrace} \mc{O}_{M-2}(\rho^2; \delta\ms{w})+\mc{O}_{M-1}(\rho^3; \ms{D}\delta\ms{w})\text{,}
        \end{aligned}\\
        &\begin{aligned}
        I_{1, \ms{w}} =& \mc{O}_{M-3}(\rho^2; \delta\ms{g}) + \mc{O}_{M-3}(\rho; \delta \ms{m}) + \mc{O}_{M-3}(\rho; \delta\ms{f}^0) + \mc{O}_{M-3}\paren{\rho; \delta\ms{f}^1}\\
        &+\mc{O}_{M-2}(\rho^3; \ms{D}\delta\ms{g}) + \mc{O}_{M-2}(\rho^2; \ms{D}\delta\ms{m}) \\
        &+\sum\limits_{\ms{h}\in\{\ul{\ms{h}}^0, \ul{\ms{h}}^{2}\}} \mc{O}_{M-2}(\rho^2; \delta\ms{h}) + \mc{O}_{M-1}(\rho^3; \ms{D}\delta\ms{h}) \\
        &+\sum\limits_{\ms{w}\in\{\ms{w}^0, \ms{w}^1, \ms{w}^2\}} \mc{O}_{M-3}(\rho^2; \delta \ms{w}) + \mc{O}_{M-2}(\rho^3; \ms{D}\delta\ms{w})\text{.}
        \end{aligned}
    \end{align*}
    We should however consider, instead of $\delta\ms{h}$ and $\delta\ms{w}$, their renormalised differences. One can therefore note, for any vertical tensor $\ms{A}$: 
    \begin{align*}
        &\delta\ms{A} = \Delta\ms{A} + \mc{S}\paren{\ms{g}; \ms{A}, \delta\ms{g}}+\mc{S}\paren{\ms{g}; \ms{A}, \ms{Q}}\text{,}\\
        &\ms{D}\delta\ms{A} = \delta\ms{D}\ms{A} + \mc{S}\paren{\ms{g}; \ms{D}\delta\ms{g}, \check{\ms{A}}}\text{.}
    \end{align*}
    Since $\ms{w}^\star$  solves a ``nice" wave equation, unlike $\ms{w}^0$, one also needs to write: 
    \begin{align*}
        &\begin{aligned}\delta\ms{w}^0 &= \delta\ms{w}^\star + \mc{O}_{M}(1; \delta\ms{w}^2) + \mc{O}_{M-2}(1; \delta\ms{g})\\
        &=\Delta\ms{w}^\star + \mc{O}_{M-2}(1; \ms{Q}) + \mc{O}_{M-2}(1; \delta\ms{g}) + \mc{O}_{M}(1; \Delta\ms{w}^2)\text{,}
        \end{aligned}\\
        &\begin{aligned}
        \ms{D}\delta \ms{w}^0 &= \ms{D}\Delta\ms{w}^\star + \mc{O}_M(1; \ms{D}\Delta\ms{w}^2) + \mc{O}_{M-1}(1; \Delta\ms{w}^2) + \mc{O}_{M-2}(1; \ms{DQ}) \\
        &+ \mc{O}_{M-3}(1; \ms{Q}) + \mc{O}_{M-2}(1; \ms{D}\delta\ms{g}) + \mc{O}_{M-3}(1; \delta\ms{g})\text{,}
        \end{aligned}
    \end{align*}
    giving eventually: 
    \begin{align}
        \label{I_1_h}I_{1, \ms{h}} = &\mc{O}_{M-3}\paren{\rho^2; \delta\ms{g}} + \mc{O}_{M-3}(\rho; \delta\ms{m}) + \mc{O}_{M-3}(\rho; \delta\ms{f}^0) +\mc{O}_{M-3}(\rho;\delta\ms{f}^1)\\
        &\notag+\mc{O}_{M-2}\paren{\rho^3; \ms{D}\delta\ms{g}} + \mc{O}_{M-2}\paren{\rho^2; \ms{D}\delta\ms{m}} +\mc{O}_{M-3}(\rho^2; \ms{Q}) + \mc{O}_{M-2}(\rho^3; \ms{DQ}) \\
        &\notag+ \sum\limits_{\ms{h}\in\{\ul{\ms{h}}^0, \ul{\ms{h}}^{2}\}}\mc{O}_{M-3}(\rho^2; \Delta \ms{h}) + \mc{O}_{M-2}(\rho^3;\ms{D}\Delta\ms{h}) \\
        &\notag+ \sum\limits_{\ms{w}\in\lbrace \ms{w}^\star, \ms{w}^1, \ms{w}^2\rbrace} \mc{O}_{M-2}(\rho^2; \Delta\ms{w})+\mc{O}_{M-1}(\rho^3; \ms{D}\Delta\ms{w})\text{,}\\
        \label{I_1_w}I_{1, \ms{w}} &= \mc{O}_{M-3}(\rho^2; \delta\ms{g}) + \mc{O}_{M-3}(\rho; \delta \ms{m}) + \mc{O}_{M-3}(\rho; \delta\ms{f}^0) + \mc{O}_{M-3}\paren{\rho; \delta\ms{f}^1}\\
        &\notag+\mc{O}_{M-2}(\rho^3; \ms{D}\delta\ms{g}) + \ms{O}_{M-2}(\rho^2; \ms{D}\delta\ms{m})+\mc{O}_{M-3}(\rho^2; \ms{Q}) + \mc{O}_{M-2}(\rho^3; \ms{DQ})\\
        &\notag+\sum\limits_{\ms{h}\in\{\ul{\ms{h}}^0, \ul{\ms{h}}^{2}\}} \mc{O}_{M-2}(\rho^2; \Delta\ms{h}) + \mc{O}_{M-1}(\rho^3; \ms{D}\Delta\ms{h}) \\
        &\notag+\sum\limits_{\ms{w}\in\{\ms{w}^\star, \ms{w}^1, \ms{w}^2\}} \mc{O}_{M-3}(\rho^2; \Delta\ms{w}) + \mc{O}_{M-2}(\rho^3; \ms{D}\Delta\ms{w})\text{.}
    \end{align}
    The term $I_{2, \ms{A}}$ will be very subtle to handle, and the choice of renormalisation from Equation \eqref{renormalised_def} will turn out to be crucial. First of all, by simply computing the wave equation for each metric, using \eqref{wave_op}, one finds, for any vertical tensor $\ms{A}$: 
    \begin{align}\label{diff_wave_A}
        (\ol{\Box} - \check{\ol{\Box}}){\ms{A}} =& \rho^2(\delta\Box_{\ms{g}}){\ms{A}} + \rho^2 \mc{S}(\Lie_\rho\ms{A})\cdot \delta\mc{S}\paren{\ms{g}; \ms{m}}+\rho\mc{S}(\ms{A})\cdot \delta\mc{S}\paren{\ms{g}; \ms{m}}+\rho^2\mc{S}(\ms{A})\cdot \delta\mc{S}\paren{\ms{g}; \Lie_\rho\ms{m}}\\
        &\notag + \rho^2 \mc{S}(\ms{A})\cdot \delta\mc{S}\paren{\ms{g}; \ms{m}^2}\text{,}
    \end{align}
    where we wrote:
    \begin{equation*}
        (\delta\Box_{\ms{g}})\ms{A}=(\ms{g}^{bc}\ms{D}^2_{bc} - \check{\ms{g}}^{bc}\check{\ms{D}}^2_{bc})\ms{A}\text{,}
    \end{equation*}
    will be the only term in \eqref{diff_wave_A} that will be tricky to handle. 
    
    Note that the terms in \eqref{diff_wave_A} have been gathered in pairs with the same algebraic form, for which one can use Remark \ref{rmk_difference_remainders}. Using thus Proposition \ref{prop_asymptotics_fields}, \eqref{transport_dm}, as well Propositions \ref{prop_transport_h} and \ref{prop_transport_w}, these terms can be treated as remainders: 
    \begin{align}\label{I_2_h_intermediate}
        (\ol{\Box}- \check{\ol{\Box}})\check{\ms{h}} =& \rho^2(\delta\Box_{\ms{g}})\check{\ms{h}} + \mc{O}_{M-3}(\rho^3; \delta\ms{g}) + \mc{O}_{M-3}(\rho^2; \delta\ms{m}) + \mc{O}_{M-3}(\rho^3; \ms{Q}) \\
        &\notag+\mc{O}_{M-2}(\rho^4; \delta\ms{f}^0) +\mc{O}_{M-2}(\rho^4; \delta\ms{f}^1) + \mc{O}_{M-2}(\rho^3; \Delta\ms{w}^2)\text{,}\\
        \label{I_2_w_intermediate} (\ol{\Box}- \check{\ol{\Box}})\check{\ms{w}}=&\rho^2(\delta\Box_{\ms{g}})\check{\ms{w}} + \mc{O}_{M-3}(\rho^2; \delta\ms{g}) + \mc{O}_{M-3}(\rho; \delta\ms{m}) + \mc{O}_{M-3}(\rho^3; \ms{Q}) \\
        &\notag+\mc{O}_{M-2}(\rho^3; \delta\ms{f}^0) +\mc{O}_{M-2}(\rho^3; \delta\ms{f}^1) + \mc{O}_{M-2}(\rho^2; \Delta\ms{w}^2)\text{.}
    \end{align}
    The term $(\delta\Box_{\ms{g}})\ms{A}$, expanded, gives: 
    \begin{align}
        \rho^2(\delta\Box_{\ms{g}})\ms{A}_{\bar{a}} =\rho^2\ms{g}^{bc}\ms{D}_b (\delta\ms{D}_c\ms{A}_{\bar{a}}) + \rho^2\ms{g}^{bc}\delta\ms{D}_b(\check{\ms{D}}_c\ms{A}_{\bar{a}}) +\rho^2\delta \ms{g}^{bc}\check{\ms{D}}^2_{bc}\ms{A}_{\bar{a}}\text{.}
    \end{align}
    Now, observe that for any vertical tensor $\ms{A}$ of rank $(0, l)$, the operator $\delta\ms{D}$ can be written in terms of difference of Christoffel symbols, which can be controlled using Proposition \ref{prop_diff_g_gamma}: 
    \begin{align*}
        (\delta\ms{D})_a\ms{A}_{\bar{b}} = -\sum\limits_{i=1}^l \delta \ms{\Gamma}_{ab_i}^c\ms{A}_{\hat{b}_i[c]}\text{,}
    \end{align*}
    giving thus: 
    \begin{align*}
        \rho^2(\delta\Box_{\ms{g}})\ms{A}_{\bar{a}} = &-\rho^2\ms{g}^{bc}\sum\limits_{i=1}^{l}\ms{D}_b(\delta\ms{\Gamma}_{ca_i}^d\ms{A}_{\hat{a}_i[d]}) - \rho^2\ms{g}^{bc}\ms{\delta}\ms{\Gamma}_{bc}^d\check{\ms{D}}_{d}\ms{A}_{\bar{a}} - \rho^2\ms{g}^{bc}\sum\limits_{i=1}^l \delta\ms{\Gamma}_{ba_i}^d\check{\ms{D}}_c\ms{A}_{\hat{a}_i[d]}\\
        &+\rho^2\delta \ms{g}^{bc}\check{\ms{D}}^2_{bc}\ms{A}_{\bar{a}}\\
        =&-\frac{1}{2}\rho^2\ms{g}^{bc}\sum\limits_{i=1}^l \ms{D}_{b}(\check{\ms{g}}^{de}(\ms{D}_c\delta\ms{g}_{ea_i} + \ms{D}_{a_i}\delta\ms{g}_{ec} - \ms{D}_e\delta\ms{g}_{ca_i})\ms{A}_{\hat{a}_i[d]}) \\
        &+\left[\mc{S}(\check{\ms{D}}\ms{A})\mc{O}_M(\rho^2; \ms{D}\delta\ms{g})\right]_{\bar{a}} +\rho^2\delta \ms{g}^{bc}\check{\ms{D}}^2_{bc}\ms{A}_{\bar{a}}\\
        =&-\frac{1}{2}\rho^2\ms{g}^{bc}\ms{g}^{de}\sum\limits_{i=1}^l \ms{D}_{b}(\ms{D}_c\delta\ms{g}_{ea_i} + \ms{D}_{a_i}\delta\ms{g}_{ec} - \ms{D}_e\delta\ms{g}_{ca_i})\ms{A}_{\hat{a}_i[d]} \\
        &+ \left[\mc{S}(\ms{DA})\mc{O}_{M-2}(\rho^2; \ms{D}\delta\ms{g})\right]_{\bar{a}}+ \left[\mc{S}(\ms{A})\mc{O}_{M-2}(\rho^2; \ms{D}\delta\ms{g})\right]_{\bar{a}}+\left[\mc{S}\paren{\check{\ms{D}}\ms{A}}\mc{O}_{M-2}(\rho^2;\delta\ms{g})\right]_{\bar{a}}\\
        & +\rho^2\delta \ms{g}^{bc}\check{\ms{D}}^2_{bc}\ms{A}_{\bar{a}}\text{,}
    \end{align*}
    The heuristics behind those calculations is essentially the realisation that any derivatives of $\delta\ms{g}$ less than two can be controlled. The only problematic ones, and this is the reason why we kept them explicitly, are those of the form $\ms{D}^2\delta\ms{g}$, which cannot be controlled in the Carleman estimates. Furthermore, the reader may keep in mind that any term with a $\check{}$, can be expressed in terms of its counterpart in $(\mc{M}, g)$, up to remainder terms, such as $\check{\ms{g}}^{-1} = \ms{g}^{-1} + \mc{O}_M(1; \delta\ms{g})$. 
    
    Observe also, from the definitions of $\ms{Q}$ and $\ms{B}$ in \eqref{transport_Q}, \eqref{def_B}, the following holds: 
    \begin{align*}
        \ms{D}_b\ms{D}_{[a_i}\delta\ms{g}_{e]c} = \frac{1}{2}(\ms{D}_b \ms{B}_{a_iec} + \ms{D}^2_{bc}\ms{Q}_{a_ie}) \text{,}
    \end{align*}
    such that one can write: 
    \begin{align*}
        &\frac{1}{2}\rho^2\ms{g}^{bc}\ms{g}^{de}\sum\limits_{i=1}^l \ms{D}_{b}(\ms{D}_c\delta\ms{g}_{ea_i} + \ms{D}_{a_i}\delta\ms{g}_{ec} - \ms{D}_e\delta\ms{g}_{ca_i})\ms{A}_{\hat{a}_i[d]} \\
        &=\frac{1}{2}\rho^2\ms{g}^{bc}\ms{g}^{de}\sum\limits_{i=1}^l \ms{D}^2_{bc}(\delta\ms{g} + \ms{Q})_{a_i e} \ms{A}_{\hat{a}_i[d]}+ \rho^2\mc{S}(\ms{g}; \ms{A},\ms{DB}) \text{,}
    \end{align*}
    
    Eventually, for $\check{\ms{h}}$ and $\check{\ms{w}}$, \eqref{I_2_h_intermediate} and \eqref{I_2_w_intermediate} become: 
    \begin{align}
        \label{I_2_h}I_{2,\ms{h}} = &\frac{1}{2}\rho^2\ms{g}^{bc}\ms{g}^{de}\sum\limits_{i=1}^{l} \ms{D}^2_{bc}(\delta\ms{g} + \ms{Q})_{a_ie}\ms{h}_{\hat{a}_i[d]}+\mc{O}_{M-3}(\rho^2;\Delta\ms{h}) \\
        &\notag + \mc{O}_{M-3}(\rho^2; \delta\ms{g}) + \mc{O}_{M-2}(\rho^3; \ms{D}\delta\ms{g}) + \mc{O}_{M-3}(\rho^3; \ms{Q}) + \mc{O}_{M-3}(\rho^2; \delta\ms{m})\\
        &\notag+\mc{O}_{M-2}(\rho^4; \delta\ms{f}^0) + \mc{O}_{M-2}(\rho^4; \delta\ms{f}^1) + \mc{O}_{M-2}(\rho^3; \Delta\ms{w}^2) + \mc{O}_{M-2}(\rho^3; \ms{DB})\text{,}\\
        \label{I_2_w}I_{2, \ms{w}} =&\frac{1}{2}\rho^2\ms{g}^{bc}\ms{g}^{de}\sum\limits_{i=1}^{l} \ms{D}^2_{bc}(\delta\ms{g} + \ms{Q})_{ea_i}\ms{w}_{\hat{a}_i[d]}+\mc{O}_{M-3}(\rho^2; \Delta\ms{w})\\
        &\notag+ \mc{O}_{M-4}(\rho^2; \delta\ms{g}) + \mc{O}_{M-3}(\rho^2; \ms{D}\delta\ms{g}) + \mc{O}_{M-3}(\rho^3; \ms{Q}) + \mc{O}_{M-3}(\rho; \delta\ms{m})\\
        &\notag+\mc{O}_{M-2}(\rho^3; \delta\ms{f}^0) + \mc{O}_{M-2}(\rho^3; \delta\ms{f}^1) + \mc{O}_{M-2}(\rho^2; \Delta\ms{w}^2)+\mc{O}_{M-2}(\rho^2; \ms{DB})\text{,}
    \end{align}
    where we used the fact that: 
    \begin{align*}
        \frac{1}{2}\rho^2\ms{g}^{bc}\ms{g}^{de}\sum\limits_{i=1}^{l} \ms{D}^2_{bc}(\delta\ms{g} + \ms{Q})_{a_ie}\check{\ms{h}}_{\hat{a}_i[d]}=&\frac{1}{2}\rho^2\ms{g}^{bc}\ms{g}^{de}\sum\limits_{i=1}^{l} \ms{D}^2_{bc}(\delta\ms{g} + \ms{Q})_{a_ie}{\ms{h}}_{\hat{a}_i[d]} + \mc{O}_{M-3}(\rho^2; \Delta\ms{h})\\
        &+ \mc{O}_{M-3}(\rho^3;\delta\ms{g}) + \mc{O}_{M-3}(\rho^3; \ms{Q})\text{,}
    \end{align*}
    and similarly for $\ms{w}$. This identity can be obtained from the fact that $\check{\ms{h}} = \ms{h} - \delta \ms{h}$ and from the following informal observations: 
    \begin{gather*}
        \ms{D}^2\delta\ms{g} = \mc{O}_{M-2}(1)\text{,}\qquad \rho^2\mc{S}\paren{\ms{g}; \ms{D}^2\ms{Q}, \delta\ms{h}} = \mc{O}_{M-3}(\rho^2; \Delta\ms{h}) + \mc{O}_{M-3}(\rho^3; \delta\ms{g}) + \mc{O}_{M-3}(\rho^3; \ms{Q})\text{,}\\ 
        \rho^2\mc{S}\paren{\ms{g}; \ms{D}^2\delta \ms{g}, \delta\ms{h}} = \mc{O}_{M-2}(\rho^2; \Delta\ms{h}) + \mc{O}_{M-2}(\rho^3;\delta\ms{g}) + \mc{O}_{M-2}(\rho^3; \ms{Q})\text{.}
    \end{gather*}
    
    Finally, one can look at $I_{3, \ms{A}}$:
    \begin{align*}
        I_{3,\ms{A}} =& -\frac{1}{2} \ms{g}^{bc}\sum\limits_{j=1}^{l_2} (\ol{\Box}+m_{\ms{A}})({\ms{A}}_{\hat{a}_j[b]}(\delta\ms{g} + \ms{Q}))_{a_jc}\\
        =&- \frac{1}{2}\ms{g}^{bc}\sum\limits_{j=1}^{l_2} (\ol{\Box}+m_{\ms{A}})({\ms{A}}_{\hat{a}_j[b]})(\delta\ms{g} + \ms{Q})_{a_jc} - \frac{1}{2}\ms{g}^{bc}\sum\limits_{j=1}^{l_2} {\ms{A}}_{\hat{a}_j[b]}\ol{\Box}(\delta\ms{g} + \ms{Q})_{a_jc}\\
        &-\rho^2\ms{g}^{bc}\sum\limits_{j=1}^{l_2} \ol{\ms{D}}_\rho{\ms{A}}_{\hat{a}_i[b]}\ol{\ms{D}}_\rho(\delta\ms{g}+ \ms{Q})_{a_jc} -\rho^2\ms{g}^{bc}\ms{g}^{de}\sum\limits_{j=1}^{l_2} {\ms{D}}_d{\ms{A}}_{\hat{a}_i[b]}{\ms{D}}_e(\delta\ms{g}- \ms{Q})_{a_jc}\\
        =:&\,I_{3, \ms{A}}^{(1)} + I_{3, \ms{A}}^{(2)}+I_{3, \ms{A}}^{(3)}+I_{3, \ms{A}}^{(4)}
    \end{align*}
    The last term is easy to control:
    \begin{equation*}
        I_{3, \ms{A}}^{(4)} = \mc{S}(\ms{D}\ms{A})(\mc{O}_{M}(\rho^2; \ms{D}\delta\ms{g}) + \mc{O}_{M-1}(\rho^2; \ms{DQ}))\text{,}
    \end{equation*}
    giving for $\ms{h}$ and $\ms{w}$: 
    \begin{align}
        &\label{I_3_4_h}I_{3, \ms{h}}^{(4)}=\mc{O}_{M-2}(\rho^3;\ms{D} \delta\ms{g}) + \mc{O}_{M-2}(\rho^3; \ms{DQ})\text{,}\\
        &\label{I_3_4_w}I_{3, \ms{w}}^{(4)} = \mc{O}_{M-3}(\rho^2; \ms{D}\delta\ms{g}) + \mc{O}_{M-3}(\rho^2; \ms{DQ})\text{,}
    \end{align}
    while for $I_{3, \ms{A}}^{(3)}$, one has to use Proposition \ref{prop_transport_diff}: 
    \begin{align*}
        I_{3, \ms{A}}^{(3)} = (\mc{S}(\Lie_\rho\ms{{A}}) + \mc{S}(\ms{g}; \ms{m}, {\ms{A}}))(\mc{O}_{M}(\rho^2; \delta\ms{m}) +\mc{O}_{M-2}(\rho^3; \delta\ms{g})+ \mc{O}_{M-2}(\rho^3; \ms{Q}))\text{,}
    \end{align*}
    which gives, similarly, after using Proposition \ref{prop_asymptotics_fields}: 
    \begin{align}\label{I_3_3}
         &I_{3, \ms{h}}^{(3)}\text{,} \;I_{3, \ms{w}}^{(3)} = \mc{O}_{M-3}(\rho^3; \delta\ms{g}) + \mc{O}_{M-3}(\rho^2; \delta\ms{m}) + \mc{O}_{M-3}(\rho^3; \ms{Q})\text{.}
    \end{align}
    The term $I_{3, \ms{A}}^{(1)}$ can be controlled by looking at the wave equation satisfied by $\ms{A}$, therefore using Propositions \ref{prop_wave_h} and \ref{prop_wave_w}: 
    \begin{align}\label{intermediate_I_3}
        I_{3, \ms{h}}^{(1)}, I_{3, \ms{w}}^{(1)} = \mc{O}_{M-3}(\rho^2; \delta\ms{g}) + \mc{O}_{M-3}(\rho^2; \ms{Q})\text{.}
    \end{align}
    This is due to the fact that $\ms{h}$ and $\ms{w}$ satisfy vertical wave equations as given in Propositions \ref{prop_wave_h} and \ref{prop_wave_w} with remainder terms satisfying the asymptotics \eqref{intermediate_I_3}. As an example, the remainder terms for \eqref{wave_w} yield terms of the form: 
    \begin{gather*}
        \abs{\mc{S}\paren{\ms{g}; \ms{f}^2}}_{M, \varphi} \lesssim \rho^2\text{,}
    \end{gather*}
    with $\ms{f}\in \lbrace \ms{f}^0, \ms{f}^1\rbrace$. Similarly, \eqref{remainder_h_0} and \eqref{remainder_h_2} contain terms as:
    \begin{gather*}
        \abs{\mc{S}\paren{\ms{g}; \ms{m}, \ms{f}}}_{M-2, \varphi}\lesssim \rho^2\text{.}
    \end{gather*}
    Analysing every term in \eqref{remainder_h_0}, \eqref{remainder_h_2}, \eqref{remainder_wave_w_2}, \eqref{remainder_wave_w_1} yield indeed the $\rho^2$ as minimal decay and $M-3$ as minimal regularity. 

    The treatment of $I_{3, \ms{A}}^{(2)}$, containing terms of the form $\ms{D}^2\delta\ms{g}$, will require a bit more care. More precisely, using \eqref{wave_op} and Proposition \ref{prop_asymptotics_fields}:
    \begin{align*}
        I^{(2)}_{3, \ms{A}} =& -\frac{1}{2}\ms{g}^{bc}\sum\limits_{i=1}^l \ms{A}_{\hat{a}_i[b]}\ol{\Box}(\delta\ms{g} +\ms{Q})_{a_ic}\\
        =&-\frac{1}{2}\rho^2\ms{g}^{bc}\ms{g}^{de}\sum\limits_{i=1}^l \ms{A}_{\hat{a}_i[b]} \ms{D}^2_{de}(\delta\ms{g}+\ms{Q})_{a_i c} + \left[\mc{S}(\ms{A})(\mc{O}_{M-2}(\rho^2;\Lie_\rho^2(\delta\ms{g} + \ms{Q}))\right. \\
        &\left.+\mc{O}_{M-2}(\rho; \Lie_\rho(\delta\ms{g} + \ms{Q})) + \mc{O}_{M-2}(\rho^2; \delta\ms{g} + \ms{Q}))\right]_{\bar{a}}\text{.}
    \end{align*}
    In order to treat the remainders, one can note, using \eqref{transport_Q}:
    \begin{align*}
        \Lie_\rho^2 \ms{Q} &= \Lie_\rho \paren{\mc{S}\paren{\ms{g}; \ms{m}, \delta\ms{g}} + \mc{S}\paren{\ms{g}; \ms{m}, \ms{Q}}}\\
        &=\mc{O}_{M-2}(1; \delta\ms{g}) + \mc{O}_{M-2}(\rho; \delta\ms{m}) + \mc{O}_{M-2}(1; \ms{Q})\text{.}
    \end{align*}
    As a consequence, using Proposition \ref{prop_transport_diff}, one has: 
    \begin{align*}
        &\rho\Lie_\rho(\delta\ms{g}+\ms{Q}) = \rho\delta\ms{m} + \mc{O}_{M-2}(\rho^2; \delta\ms{g}) + \mc{O}_{M-2}(\rho^2; \ms{Q})\text{,}\\
        &\begin{aligned}
            \rho^2\Lie_\rho^2 (\delta\ms{g} +\ms{Q}) =& \mc{O}_{M-2}(\rho; \delta\ms{m}) + \mc{S}\paren{\rho^2 \Delta\ms{w}^2} + \mc{O}_{M-3}(\rho^2; \delta\ms{g}) + \mc{O}_{M-3}(\rho^2; \ms{Q})\\
            &+\mc{O}_{M-2}(\rho^3;\delta\ms{f}^0) + \mc{O}_{M-2}(\rho^3; \delta\ms{f}^1)\text{,}
        \end{aligned}
    \end{align*}
    giving eventually: 
    \begin{align*}
        I_{3, \ms{A}}^{(2)} =& -\frac{}{2}\rho^2\ms{g}^{bc}\ms{g}^{de}\sum\limits_{i=1}^l \ms{A}_{\hat{a}_i[b]} \ms{D}^2_{de}(\delta\ms{g}+\ms{Q})_{a_ic}\\
        &+ \mc{S}(\ms{A})\left[\mc{O}_{M-2}\paren{\rho^2; \Delta\ms{w}^2} + \mc{O}_{M-3}(\rho^2; \delta\ms{g}) + \mc{O}_{M-2}(\rho; \delta\ms{m}) + \mc{O}_{M-3}(\rho^2; \ms{Q}) \right.\\
        &+\left. \mc{O}_{M-2}(\rho^3; \delta\ms{f}^0))+ \mc{O}_{M-2}(\rho^3; \delta\ms{f}^1)\right]\text{.}
    \end{align*}
    For $\ms{h}$ and $\ms{w}$, one eventually gets: 
    \begin{align}
        \label{I_3_2h}I_{3, \ms{h}}^{(2)} =& -\frac{1}{2}\rho^2\ms{g}^{bc}\ms{g}^{de}\sum\limits_{i=1}^l \ms{h}_{\hat{a}_i[d]}\ms{D}^2_{bc}(\delta\ms{g} + \ms{Q})_{a_ie} + \mc{O}_{M-2}(\rho^3; \Delta\ms{w}^2) +\mc{O}_{M-3}(\rho^3; \delta\ms{g})\\
        &\notag + \mc{O}_{M-2}(\rho^2; \delta\ms{m}) + \mc{O}_{M-3}(\rho^3; \ms{Q}) + \mc{O}_{M-2}(\rho^4; \delta\ms{f}^0) + \mc{O}_{M-2}(\rho^4; \delta\ms{f}^1)\text{,}\\
        \label{I_3_2_w}I_{3, \ms{w}}^{(2)} =& -\frac{1}{2}\rho^2\ms{g}^{bc}\ms{g}^{de}\sum\limits_{i=1}^l \ms{w}_{\hat{a}_i[d]}\ms{D}^2_{bc}(\delta\ms{g} + \ms{Q})_{a_ie} + \mc{O}_{M-2}(\rho^2; \Delta\ms{w}^2) +\mc{O}_{M-3}(\rho^2; \delta\ms{g})\\
        &\notag + \mc{O}_{M-2}(\rho; \delta\ms{m}) + \mc{O}_{M-3}(\rho^2; \ms{Q}) + \mc{O}_{M-2}(\rho^3; \delta\ms{f}^0) + \mc{O}_{M-2}(\rho^3; \delta\ms{f}^1)\text{.}
    \end{align}
    Summing \eqref{I_2_h} with \eqref{intermediate_I_3}, \eqref{I_3_2h}, \eqref{I_3_3} and \eqref{I_3_4_h}, and \eqref{I_2_w} with \eqref{intermediate_I_3}, \eqref{I_3_2_w}, \eqref{I_3_3} and \eqref{I_3_4_w} gives:
    \begin{align}
        &\begin{aligned}
             \label{I_2+I_3_h}I_{2,\ms{h}} + I_{3,\ms{h}} = & \mc{O}_{M-3}(\rho^2; \delta\ms{g}) + \mc{O}_{M-2}(\rho^3; \ms{D}\delta\ms{g}) + \mc{O}_{M-3}(\rho^2; \ms{Q})+\mc{O}_{M-2}(\rho^3; \ms{DQ}) + \mc{O}_{M-3}(\rho^2; \delta\ms{m})\\
            &+\mc{O}_{M-2}(\rho^4; \delta\ms{f}^0) + \mc{O}_{M-2}(\rho^4; \delta\ms{f}^1) + \mc{O}_{M-2}(\rho^3; \Delta\ms{w}^2) + \mc{O}_{M-2}(\rho^3; \ms{DB})+\mc{O}_{M-3}(\rho^2;\Delta\ms{h})\text{,}
        \end{aligned}\\
        &\begin{aligned}
            \label{I_2+I_3_w}I_{2,\ms{w}} + I_{3, \ms{w}} = &+ \mc{O}_{M-4}(\rho^2; \delta\ms{g}) + \mc{O}_{M-3}(\rho^2; \ms{D}\delta\ms{g}) + \mc{O}_{M-3}(\rho^2; \ms{Q}) +\mc{O}_{M-3}\paren{\rho^2; \ms{DQ}}+ \mc{O}_{M-3}(\rho; \delta\ms{m})\\
        &+\mc{O}_{M-2}(\rho^3; \delta\ms{f}^0) + \mc{O}_{M-2}(\rho^3; \delta\ms{f}^1) + \mc{O}_{M-2}(\rho^2; \Delta\ms{w}^2)+\mc{O}_{M-2}(\rho^2; \ms{DB}) + \mc{O}_{M-3}(\rho^2; \Delta\ms{w)}\text{.}
        \end{aligned}
    \end{align}
    Observe that the terms involving $\ms{D}^2\delta\ms{g}$, which are not controllable, cancel out. 
    
    Equation \eqref{F_w} is thus obtained by adding \eqref{I_1_w} to \eqref{I_2+I_3_w} while Equation \eqref{G_h} is given by the sum of \eqref{I_1_h} and \eqref{I_2+I_3_h}. This concludes the proof.
\end{proof}

Finally, the last proposition we need for the unique continuation result will allow us to use transport equations for $\delta\ms{w}$ and $\delta{\ms{h}}$ in order to trade regularity for higher order of vanishing easily: \allowdisplaybreaks{
\begin{proposition}\label{prop_transport_diff_w_h}
    Let $(\mc{M}, g, F)$ be a FG-aAds segment and let $(U, \varphi)$ be a compact coordinate system on $\mc{I}$. Then, the following transport equations hold: 
    \begin{align}
        &\begin{aligned}\label{transport_dm_modified}
        \Lie_\rho \paren{\rho^{-1}\delta \ms{m}} =& -2\rho^{-1}\delta\ms{w}^2 + \mc{O}_{M-3}(1; \delta\ms{g}) + \mc{O}_{M-2}(1; \delta\ms{m})\text{,}\\
        &+ \mc{O}_{M}\paren{1; \delta\ms{f}^0} + \mc{O}_{M}\paren{1; \delta \ms{f}^1}\text{,}
        \end{aligned}\\
        &\label{transport_df0_modified}\begin{aligned}
        \Lie_\rho(\rho^{-(n-2)}\delta\ms{f}^0)=&\mc{O}_M(\rho^{-(n-2)};\delta\ul{\ms{h}}^0) + \mc{O}_{M-1}\paren{\rho^{-(n-3)};\delta\ms{g}} + \mc{O}_{M} \paren{\rho^{-(n-3)}; \delta\ms{m}}\\
        &+ \mc{O}_{M-2}(\rho^{-(n-3)}; \delta \ms{f}^0)\text{,}
        \end{aligned}\\
        &\label{transport_df1_modified}\begin{aligned}
        \Lie_\rho(\rho^{-1}\delta\ms{f}^1) &= \mc{O}_M\paren{\rho^{-1};\delta \ul{\ms{h}}^{2}} + \mc{O}_{M-1}(1; \delta\ms{g})+\mc{O}_{M}(1;\delta\ms{m}) + \mc{O}_{M-2}(1; \delta\ms{f}^1)\text{,}
        \end{aligned}\\
        &\label{transport_dw2}\begin{aligned}
        \Lie_\rho\paren{\rho^{2-n}\delta\ms{w}^2} =&\, \mc{O}_M(\rho^{2-n}; \ms{D}\delta\ms{w}^1) + \mc{O}_{M-4}(\rho^{3-n}; \delta\ms{g}) +\mc{O}_{M-2}(\rho^{2-n}; \delta\ms{m}) \\
        &+\mc{O}_{M-3}(\rho^{3-n}; \ms{D}\delta\ms{g}) + \mc{O}_{M-2}(\rho^{3-n}; \delta\ms{w}^0) + \mc{O}_{M-2}(\rho^{3-n}; \delta\ms{w}^2)\\
        &+\mc{O}_{M}(\rho^{3-n}; \delta\ul{\ms{h}}^0) +\mc{O}_{M}(\rho^{3-n}; \delta \ul{\ms{h}}^{2})\\
        &+\mc{O}_{M-1}(\rho^{2-n}; \delta\ms{f}^0) + \mc{O}_{M-1}(\rho^{2-n}; \delta\ms{f}^1)\text{,}
        \end{aligned}\\
        &\label{transport_dw1}\begin{aligned}
        \Lie_\rho\paren{\rho^{-1}\delta\ms{w}^1}=&\,\mc{S}\paren{\rho^{-1}\ms{D}\delta\ms{w}^2}+\mc{O}_{M-3}(1; \delta\ms{g}) + \mc{O}_{M-2}(\rho^{-1}; \delta\ms{m})\\
        &+\mc{O}_{M-3}(1; \ms{D}\delta\ms{g}) + \mc{O}_{M-2}(1; \delta\ms{w}^1) + \mc{O}_{M-1}(\rho^{-1}; \delta\ms{f}^0)\\
        &+ \mc{O}_{M-1}(\rho^{-1}; \delta\ms{f}^1) + \mc{O}_{M-1}(1; \delta\ul{\ms{h}}^0)+\mc{O}_{M}(1; \delta\ul{\ms{h}}^{2})\text{,}
        \end{aligned}\\
        &\begin{aligned}
        \Lie_{\rho}\delta\ms{w}^0 =&\, \mc{S}\paren{\ms{D}\delta\ms{w}^1} + \mc{O}_{M-2}(\rho^{-1}; \delta\ms{w}^2) + \mc{O}_{M-3}(1; \delta\ms{g})+\mc{O}_{M-2}(1; \delta\ms{m}) \\
        & + \mc{O}_{M-3}(\rho; \ms{D}\delta\ms{g}) + \mc{O}_{M-2}(\rho; \delta\ms{w}^0)+ \mc{O}_{M}(\rho; \delta\ul{\ms{h}}^0) \\
        &+\mc{O}_{M}(\rho; \delta\ul{\ms{h}}^{2})+\mc{O}_{M-1}(1; \delta\ms{f}^0)+\mc{O}_{M-1}(1; \delta\ms{f}^1)\text{,}
        \end{aligned}\label{transport_dw0}\\
        &\begin{aligned}
        \Lie_\rho\paren{\rho^{2-n}\delta\ul{\ms{h}}^{2}} &= \mc{S}\paren{\rho^{2-n}\ms{D}\delta\ul{\ms{h}}^0} + \mc{O}_{M-1}\paren{\rho^{3-n}; \ms{D}\delta\ms{g}} + \mc{O}_{M-2}(\rho^{3-n}; \delta\ms{g}))\\
        &+\mc{O}_{M-1}(\rho^{3-n}; \delta\ms{w}^2)+\mc{O}_{M-2}(\rho^{2-n}; \delta\ms{f}^1) + \mc{O}_{M-2}(\rho^{2-n}; \delta\ms{f}^0)   \\
        &+ \mc{O}_{M-1}(\rho^{2-n}; \delta\ms{m})+ \mc{O}_{M}(\rho^{3-n}; \ms{D}\delta\ms{m})+\mc{O}_{M-2}(\rho^{3-n}; \delta\ul{\ms{h}}^{2})
        \end{aligned}\label{transport_dh2}\\
        &\begin{aligned}\label{transport_dh0}
        \Lie_\rho(\rho^{-1}\delta\ul{\ms{h}}^0) &= \mc{O}_M(\rho^{-1};\ms{D}\delta \ul{\ms{h}}^{2}) + \mc{O}_{M-1}(\rho^{-1}; \delta\ms{m}) + \mc{O}_{M-2}(1;\delta\ms{g})  \\
        &+\mc{O}_{M-1}(1; \ms{D}\delta\ms{g})+\mc{O}_{M-1}(1; \ms{D}\delta\ms{m}) + \mc{O}_{M-2}(\rho^{-1}; \delta\ms{f}^1) \\
        &+ \mc{O}_{M-2}(\rho^{-1}; \delta\ms{f}^0) +  \mc{O}_{M-2}(1; \delta\ul{\ms{h}}^0) +\mc{O}_{M-1}(1; \delta\ms{w}^2)\text{.}
        \end{aligned}
    \end{align}
\end{proposition}}

\begin{proof}
    See Appendix \ref{app:prop_transport_diff_w_h}.
\end{proof}

\subsection{Higher-order of vanishing}\label{sec:higher}

In this section, we will use the transport equations satisfied by the vertical fields in order to trade vertical regularity for ``horizontal" (along $\rho$) vanishing. The idea is therefore to use the asymptotics from Corollary \ref{prop_expansion_vertical_fields} and integrate the remainder terms in the transport equations from Propositions \ref{prop_transport_diff} and \ref{prop_transport_diff_w_h} with respect to $\rho$. Typically, the integration will lead to gain of two powers of $\rho$. With this in mind, we will, from now on, consider $M_0-n$ to be even, for convenience\footnote{If not, one can simply take $M_0'=M_0-1$ instead.}. 

Furthermore, let $\mf{f}^{0, ((n-4)_+)}$ and $\mf{f}^{1, (0)}$ be the free coefficients in the expansions \eqref{expansion_f0} and \eqref{expansion_f1}. In the definition of $\ms{f}^0$ and $\ms{f}^1$ used in this chapter, these coefficients are located at the order $n-2$ and $1$ in $\rho$, respectively. 

\begin{remark}
    From Proposition \ref{prop_constraints}, we know that the coefficients $\mf{g}^{(n)}$, $\mf{f}^{0, ((n-4)_+)}$ and $\mf{f}^{0, (0)}$ are not completely free but are constrained by the following: 
\begin{gather}
    \notag\mf{D}_{\mf{g}^{(0)}}\cdot\mf{g}^{(n)} = \mf{G}^1_n\paren{\partial^{\leq n+1}\mf{g}^{(0)}, \partial^{\leq n-3}\mf{f}^{1, (0)}}\text{,}\qquad \mf{tr}_{\mf{g}^{(0)}}\mf{g}^{(n)} = \mf{G}^2_n\paren{\partial^{\leq n}\mf{g}^{(0)}, \partial^{\leq n-4}\mf{f}^{1, (0)}}\text{,}\\
    \label{constraints_f_g}\mf{D}_{\mf{g}^{(0)}}\cdot \mf{f}^{0,((n-4)_+)} = \mf{F}^0_n(\partial^{\leq n-2}\mf{g}^{(0)}, \partial^{\leq n-2}\mf{f}^{1, (0)})\text{,}\qquad \mf{D}_{\mf{g}^{(0)}}{}_{[a} \mf{f}^{1,(0)}{}_{bc]} = 0\text{,}
\end{gather}
where $\mf{G}^1_n$, $\mf{G}^2_n$ and $\mf{F}^0_n$ are universal functions depending only on $n$ and vanishing for $n$ odd. 
\end{remark}

\begin{definition}[Holographic data]
    Let $(\mc{M}, g, F)$ be a Maxwell-FG-aAdS segment of regularity $M_0\geq n+2$. We will refer to the \emph{holographic data} of $(\mc{M}, g, F)$ as the following: 
    \begin{equation}\label{identic_data}
        (\mc{I}, \mf{g}^{(0)}, \mf{g}^{(n)}, \mf{f}^{0, ((n-4)_+)}, \mf{f}^{1, (0)})\text{,}
    \end{equation}
    with $\mf{g}^{(0)}, \mf{g}^{(n)}, \mf{f}^{0, ((n-4)_+)}, \mf{f}^{1,(0)}$ tensor fields on $\mc{I}$.
\end{definition}

\begin{remark}
    The form of the gauge \eqref{FG-gauge} obviously gives room to some residual gauge freedom. As a consequence, two holographic data, although different, may be related by a gauge transformation. This very interesting question of the residual gauge, first addressed in this context in \cite{Holzegel22}, has been known for a long time in the physics literature \cite{deHaro:2000vlm, Imbimbo_2000}. We will address this issue in Section \ref{sec:full_result}. 
\end{remark}

\begin{proposition}[Higher-order of vanishing]\label{prop_higher_order_vanishing}
    Let $(\mc{M}, F, g)$ and $(\mc{M}, \check{F}, \check{g})$ be two FG-aAdS segments of regularity $M_0\geq n+2$, such that: 
    \begin{align*}
        \mf{g}^{(0)} = \check{\mf{g}}^{(0)}\text{,}\qquad \mf{g}^{(n)} = \check{\mf{g}}^{(n)}\text{,}\qquad \mf{f}^{0,((n-4)_+)} = \check{\mf{f}}^{0,((n-4)_+)}\text{,}\qquad \mf{f}^{1, (0)}=\check{\mf{f}}^{1, (0)}\text{.}
    \end{align*}
    Then, the difference of fields satisfy the following improved asymptotics:
    \begin{itemize}
        \item There exist vertical tensor fields $\ms{r}_{\delta\ms{g}}$ and $\ms{r}_{\delta\ms{m}}$ of rank $(0, 2)$ such that: 
        \begin{gather*}
            \delta \ms{g} = \rho^{M_0-2}\ms{r}_{\delta\ms{g}}\text{,}\qquad \ms{r}_{\delta\ms{g}} \rightarrow^2 0\text{,} \qquad \delta \ms{m} = \rho^{M_0-3}\ms{r}_{\delta\ms{m}}\text{,}\qquad \ms{r}_{\delta\ms{m}} \rightarrow^2 0\text{.}
        \end{gather*}
        \item There exist vertical tensor fields $\ms{r}_{\delta\ms{w}^0}$, $\ms{r}_{\delta\ms{w}^1}$, $\ms{r}_{\delta\ms{w}^2}$ of rank $(0, 4), (0, 3)$ and $(0, 2)$, respectively, satisfying: 
        \begin{gather*}
            \delta\ms{w}^0 = \rho^{M_0-4}\ms{r}_{\delta\ms{w}^0}\text{,}\qquad \ms{r}_{\delta\ms{w}^0}\rightarrow^2 0\text{,}\qquad \delta\ms{w}^1 = \rho^{M_0-3}\ms{r}_{\delta\ms{w}^1}\text{,}\qquad \ms{r}_{\delta\ms{w}^1}\rightarrow^1 0\text{,}\\
            \delta\ms{w}^2 = \rho^{M_0-4}\ms{r}_{\delta\ms{w}^2}\text{,}\qquad \ms{r}_{\delta\ms{w}^2}\rightarrow^2 0\text{.}
        \end{gather*}
        \item There exist vertical tensor fields $\ms{r}_{\delta\ms{f}^0}$, $\ms{r}_{\delta\ms{f}^1}$, of rank $(0, 1)$ and $(0, 2)$, respectively, satisfying: 
        \begin{gather*}
            \delta\ms{f}^0 = \rho^{M_0-4}\ms{r}_{\delta\ms{f}^0}\text{,}\qquad \ms{r}_{\delta\ms{f}^0}\rightarrow^4 0\text{,}\qquad 
            \delta\ms{f}^1 = \rho^{M_0-3}\ms{r}_{\delta\ms{f}^1}\text{,}\qquad \ms{r}_{\delta\ms{f}^1}\rightarrow^3 0\text{,}
        \end{gather*}
        \item There exist vertical tensor fields $\ms{r}_{\delta\ul{\ms{h}}^0}$ of rank $(0, 3)$ and $\ms{r}_{\delta\ms{h}^{2, \pm}}$ of rank $(0, 2)$ such that: 
        \begin{gather*}
            \delta\ul{\ms{h}}^0 = \rho^{M_0-3}\ms{r}_{\delta\ul{\ms{h}}^0}\text{,}\qquad \ms{r}_{\delta\ul{\ms{h}}^0}\rightarrow^2 0\text{,}\qquad \delta \ul{\ms{h}}^{2} = \rho^{M_0-4}\ms{r}_{\delta\ul{\ms{h}}^{2}}\text{,}\qquad \ms{r}_{\delta\ul{\ms{h}}^{2}}\rightarrow^{3} 0\text{.}
        \end{gather*}
    \end{itemize}
\end{proposition}
\begin{proof}
    First of all, since the data $(\mf{g}^{(0)}, \mf{g}^{(n)}, \mf{f}^{0, ((n-4)_+)}, \mf{f}^{1, (0)})$ are identical, there exist, from Corollaries \ref{prop_expansion_vertical_fields} and \ref{weyl_expansion}, vertical tensor fields $\ms{r}_{\delta\ms{A}}$, with $\ms{A}\in \lbrace \ms{g}, \ms{m}, \ms{w}, \ms{f} \rbrace$ such that: 
    \begin{gather*}
        \delta \ms{g} = \rho^n \ms{r}_{\delta\ms{g}}\text{,}\qquad\delta \ms{m} = \rho^{n-1}\ms{r}_{\delta\ms{m}}\text{,}\\
        \delta \ms{w}^0 = \rho^{n-2}\ms{r}_{\delta\ms{w}^0}\text{,}\qquad \delta\ms{w}^1 = \rho^{n-1}\ms{r}_{\delta\ms{w}^1}\text{,}\qquad \delta\ms{w}^2 = \rho^{n-2}\ms{r}_{\delta\ms{w}^2}\text{,}\\
        \delta\ms{f}^0 = \rho^{n-2}\ms{r}_{\delta\ms{f}^0}\text{,}\qquad \delta\ms{f}^1=\rho^{n-1}\ms{r}_{\delta\ms{f}^1}\text{,}
    \end{gather*}
    with: 
    \begin{gather*}
        \ms{r}_{\delta\ms{g}}\rightarrow^{M_0-n}0\text{,}\qquad \ms{r}_{\delta\ms{m}}\rightarrow^{M_0-n}0\text{,}\\
        \ms{r}_{\delta\ms{w}^0}\rightarrow^{M_0-n} 0\text{,}\qquad \ms{r}_{\delta\ms{w}^1}\rightarrow^{M_0-n-1} 0\text{,}\qquad \ms{r}_{\delta\ms{w}^2}\rightarrow^{M_0-n} 0\\
        \ms{r}_{\delta\ms{f}^0}\rightarrow^{M_0-(n-2)} 0\text{,}\qquad \ms{r}_{\delta\ms{f}^1}\rightarrow^{M_0-(n-1)} 0\text{.}
    \end{gather*}
    The Fefferman-Graham expansions of Corollary \ref{prop_expansion_vertical_fields} and Corollary \ref{weyl_expansion} involve indeed coefficients depending only on the holographic data. Since we assumed \eqref{identic_data}, the expansion for the difference of fields consists only in the remainder term.  
    
    One can also easily obtain the asymptotics for the fields $\ms{h}$ since, as an example: 
    \begin{align*}
        \delta \ul{\ms{h}}^0 &= \delta (\ms{D}\ms{f}^1) =\rho^{n-1}\ms{r}_{\delta\ul{\ms{h}}^0}\text{,}
    \end{align*}
    with $\ms{r}_{\ul{\ms{h}}^0} \rightarrow^{M_0-n} 0$.
    Similar considerations hold for $\ul{\ms{h}}^2$: 
    \begin{gather*}
        \delta\ul{\ms{h}}^{2} = \rho^{n-2}\ms{r}_{\delta\ul{\ms{h}}^{2}}\text{,}\qquad \ms{r}_{\delta\ms{h}^{2,\pm}}\rightarrow^{M_0-n+1}0\text{.}
    \end{gather*}
    In this proof, we will use a very specific hierarchy that can be summarised as follows: 
    \begin{gather}\label{hierarchy}
        \delta\ms{f}^0\rightarrow \delta\ul{\ms{h}}^{2} \text{, }\delta\ms{w}^2\text{, } \delta\ms{f}^1 \rightarrow \delta\ms{w}^0 \rightarrow \delta\ms{m} \rightarrow \delta\ms{g} \rightarrow \delta\ul{\ms{h}}^0\text{, }\delta\ms{w}^1\text{,}
    \end{gather}
    where we, between each step, we integrate transport equations from Propositions \ref{prop_transport_diff} and \ref{prop_transport_diff_w_h}. Between each arrow, the decay of the vertical fields can be improved regardless of the specific order. After each integration, we will show how one gains two powers of $\rho$ and lose at least two orders of regularity in the vertical directions. 

    Starting thus with $\delta\ms{f}^{0}$, one gets schematically, after integrating \eqref{transport_df0_modified}: 
     \begin{align}
        \notag\delta \ms{f}^0 &=\rho^{n-2}\int_0^\rho \left[\mc{O}_{M_0}\paren{\sigma^{-(n-2)};\delta\ul{\ms{h}}^0} + \mc{O}_{M_0-1}\paren{\sigma^{-(n-3)}; {\delta\ms{g}}} + \mc{O}_{M_0}\paren{\sigma^{-(n-3)}; {\delta\ms{m}}} \right.\\
        &\notag\left.\hspace{50pt}+ \mc{O}_{M_0-2}\paren{\sigma^{-(n-3)}; {\delta\ms{f}^0}}\right]\vert_\sigma d\sigma \\
        &\notag= \rho^{n-2}\int_0^\rho \left[ \mc{O}_{M_0}\paren{\sigma; \ms{r}_{\delta\ul{\ms{h}}^0}} + \mc{O}_{M_0-1}\paren{\sigma^3; \ms{r}_{\delta\ms{g}}} + \mc{O}_{M_0}\paren{\sigma^2;\ms{r}_{\delta\ms{m}}} + \mc{O}_{M_0-2}\paren{\sigma;\ms{r}_{\delta\ms{f}^0}}\right]\vert_\sigma d\sigma \\
        &\label{improved_decay_f_0}=\rho^n \tilde{\ms{r}}_{\delta\ms{f}^0}\text{,}
    \end{align}
    with $\tilde{\ms{r}}_{\delta\ms{f}^0}\rightarrow^{M_0-n}0$\text{,} and where we used the fact the boundary term satisfies: 
    \begin{equation*}
        \rho^{2-n}\delta\ms{f}^0 = \ms{r}_{\delta\ms{f}^0}\rightarrow^{M_0-(n-2)} 0\text{.}
    \end{equation*}

    Next, looking at $\delta\ul{\ms{h}}^{2}$, one can integrate \eqref{transport_dh2}:
    \begin{align}
    &\label{improved_decay_h_2}\begin{aligned}\delta\ms{h}^{2,\pm} &= \rho^{n-2}\int_0^\rho
            &&\left[\mc{O}_{M_0}\paren{\sigma; \ms{D}\ms{r}_{\delta\ul{\ms{h}}^0}} + \mc{O}_{M_0-1}\paren{\sigma^3; \ms{D}\ms{r}_{\delta\ms{g}}} + \mc{O}_{M_0-2}\paren{\sigma^3;\ms{r}_{\delta\ms{g}}}\right. \\
        & &&+\mc{O}_{M_0-1}\paren{\sigma; \ms{r}_{\delta\ms{w}^2}} + \mc{O}_{M_0-2}\paren{\sigma; \ms{r}_{\delta\ms{f}^1}} + \mc{O}_{M_0-2}\paren{\sigma^2; \tilde{\ms{r}}_{\delta\ms{f}^0}} \\
        & &&\Bigl.+\mc{O}_{M_0-1}\paren{\sigma; \ms{r}_{\delta\ms{m}}}+ \mc{O}_{M_0}\paren{\sigma; \ms{Dr}_{\delta\ms{m}}} + \mc{O}_{M_0-2}\paren{\sigma; \ms{r}_{\delta\ul{\ms{h}}^{2}}}\Bigr]\vert_\sigma d\sigma\\
        &=\rho^n \tilde{\ms{r}}_{\delta\ul{\ms{h}}^{2}}\text{,}
        \end{aligned}
    \end{align}
    where $\tilde{\ms{r}}_{\delta\ms{h}^{2,\pm}}\rightarrow^{M_0-n-1}0$, and where we used: 
    \begin{equation*}
        \rho^{2-n}\delta\ms{h}^{2,\pm} = \ms{r}_{\delta\ms{h}^{2,\pm}}\rightarrow^{M_0-n} 0\text{.}
    \end{equation*}
   Note also that we used the improved decay \eqref{improved_decay_f_0} for $\delta\ms{f}^0$. 
    
    This allows us to improve the decay of $\delta\ms{w}^2$ by integrating \eqref{transport_dw2}:
    \begin{align}
        &\begin{aligned}
            \delta\ms{w}^2 &= \rho^{n-2}\int_0^\rho &&\left[\mc{O}_{M_0}\paren{\sigma;\ms{D}\ms{r}_{\delta\ms{w}^1}} + \mc{O}_{M_0-4}\paren{\sigma^3; \ms{r}_{\delta\ms{g}}} + \mc{O}_{M_0-2}\paren{\sigma; \ms{r}_{\delta\ms{m}}} + \mc{O}_{M_0-3}\paren{\sigma^3; \ms{D}\ms{r}_{\delta\ms{g}}}\right.\\
            & &&+\mc{O}_{M_0-2}\paren{\sigma; \ms{r}_{\delta\ms{w}^2}} + \mc{O}_{M_0-2}\paren{\sigma; \ms{r}_{\delta\ms{w}^0}} + \mc{O}_{M_0}\paren{\sigma^2; \ms{r}_{\delta\ul{\ms{h}}^0}}  \\
            & &&\left.+\mc{O}_{M_0}\paren{\sigma; {\ms{r}}_{\delta\ul{\ms{h}}^{2}}} + \mc{O}_{M_0-1}\paren{\sigma^2;\tilde{\ms{r}}_{\delta\ms{f}^0}} +\mc{O}_{M_0-1}\paren{\sigma; \ms{r}_{\delta\ms{f}^1}}\right]\vert_\sigma d\sigma \\
            &\label{improved_decay_w_2}=\rho^n \tilde{\ms{r}}_{\delta\ms{w}^2}\text{,}
        \end{aligned}
    \end{align}
    with $\tilde{\ms{r}}_{\delta\ms{w}^2} \rightarrow^{M_0-n-2} 0$. 
    where we used the improved decay \eqref{improved_decay_f_0} and the fact that the boundary term vanishes: 
    \begin{equation*}
        \rho^{2-n}\delta\ms{w}^2 = \ms{r}_{\delta\ms{w}^2}\rightarrow^{M_0-n}0\text{.}
    \end{equation*} 

    The decay for $\delta\ms{f}^1$ can be obtained simultaneously by integrating \eqref{transport_df1_modified}: 
    \begin{align*}
        \delta\ms{f}^1 &= \rho \int_0^\rho \left[\mc{O}_{M_0}\paren{\sigma^{n-1};\tilde{\ms{r}}_{\delta\ul{\ms{h}}^{2}}} + \mc{O}_{M_0-1}\paren{\sigma^{n};\ms{r}_{\delta\ms{g}}} + \mc{O}_{M_0}\paren{\sigma^{n-1}; \ms{r}_{\delta\ms{m}}}+\mc{O}_{M_0-2}\paren{\sigma^{n-1};\ms{r}_{\delta\ms{f}^1}}\right]\vert_\sigma d\sigma\\
        &=\rho^{n+1}\tilde{\ms{r}}_{\delta\ms{f}^1}\text{,}
    \end{align*}
    with $\tilde{\ms{r}}_{\delta\ms{f}^1}\rightarrow^{M_0-n-1}0$, and where we used the improved asymptotics \eqref{improved_decay_f_0}, \eqref{improved_decay_h_2}, as well as: 
    \begin{equation*}
        \rho^{-1}\delta\ms{f}^1 = \rho^{n-2}\ms{r}_{\delta\ms{f}^1}\rightarrow^{M_0-n+1}0\text{.}
    \end{equation*}
    
    One can use \eqref{improved_decay_f_0} and \eqref{improved_decay_w_2} to obtain: 
    \begin{align*}
        \delta\ms{w}^0 &= \int_0^\rho\left[\right. \mc{S}\paren{\sigma^{n-1};\ms{D}\ms{r}_{\delta\ms{w}^1}} + \mc{O}_{M_0-2}\paren{\sigma^{n-1}; \tilde{\ms{r}}_{\delta\ms{w}^2}} + \mc{O}_{M_0-3}\paren{\sigma^n; \ms{r}_{\delta\ms{g}}}+ \mc{O}_{M_0-2}\paren{\sigma^{n-1};\ms{r}_{\delta\ms{m}}}\\
        &\hspace{40pt} + \mc{O}_{M_0-3}\paren{\sigma^{n+1}; \ms{D}\ms{r}_{\delta\ms{g}}} + \mc{O}_{M_0-2}\paren{\sigma^{n-1};\ms{r}_{\delta\ms{w}^0}} +\mc{O}_{M_0}\paren{\sigma^{n-1};\ms{r}_{\delta\ul{\ms{h}}^{2}}} + \mc{O}_{M_0}\paren{\sigma^{n};\ms{r}_{\delta\ul{\ms{h}}^2}}\\
        &\hspace{40pt} + \mc{O}_{M_0-1}\paren{\sigma^{n};\tilde{\ms{r}}_{\delta\ms{f}^0}} + \mc{O}_{M_0-1}\paren{\sigma^{n-1};\ms{r}_{\delta\ms{f}^1}}\left.\right]\vert_\sigma d\sigma \\
        &= \rho^n \tilde{\ms{r}}_{\delta\ms{w}^0}\text{,}
    \end{align*}
    with $\tilde{\ms{r}}_{\delta\ms{w}^0}\rightarrow^{M_0-n-2} 0$, and where we used \eqref{improved_decay_f_0}, \eqref{improved_decay_w_2}, as well as: 
    \begin{equation*}
        \delta\ms{w}^0 = \rho^{n-2}\ms{r}_{\delta\ms{w}^0}\rightarrow^{M_0-n} 0\text{.}
    \end{equation*}

    One can also improve the decay of $\delta\ms{m}$ by integrating \eqref{transport_dm_modified}: 
    \begin{align}
        \label{improved_decay_delta_m}\delta\ms{m} &= \rho \int_0^\rho \left[-2\sigma^{n-1} \tilde{\ms{r}}_{\delta\ms{w}^2} +\mc{O}_{M_0-3}\paren{ \sigma^n; \ms{r}_{\delta\ms{g}}} +\mc{O}_{M_0-2}\paren{\sigma^{n-1};\ms{r}_{\delta\ms{m}}} + \mc{O}_{M_0}\paren{\sigma^n ;\tilde{\ms{r}}_{\delta\ms{f}^0}} + \mc{O}_{M_0}\paren{\sigma^{n-1};\ms{r}_{\delta\ms{f}^1}}\right]\vert_\sigma d\sigma\\
        &\notag=\rho^{n+1}\tilde{\ms{r}}_{\delta\ms{m}}\text{,}
    \end{align}
    with $\tilde{\ms{r}}_{\delta\ms{m}}\rightarrow^{M_0-n-2}0$, and where we used \eqref{improved_decay_f_0}, \eqref{improved_decay_w_2}, as well as: 
    \begin{equation*}
        \rho^{-1}\delta\ms{m} = \rho^{n-2}\ms{r}_{\delta\ms{m}}\rightarrow^{M_0-n}0\text{.}
    \end{equation*}
    This immediately implies, from the definition of $\ms{m}$: 
    \begin{align*}
        \delta\ms{g} &= \int_0^\rho \sigma^{n+1}\tilde{\ms{r}}_{\delta\ms{m}}\vert_\sigma d\sigma\\
        &=\rho^{n+2}\tilde{\ms{r}}_{\delta\ms{g}}\text{,}\qquad \tilde{\ms{r}}_{\delta\ms{g}}\rightarrow^{M_0-n-2}0\text{.}
    \end{align*}

    Following the hierarchy \eqref{hierarchy}, one can use \eqref{transport_dh0} to get: 
    \begin{align*}
        \begin{aligned}
            \delta\ul{\ms{h}}^0 &= \rho\int_0^\rho \left[\mc{O}_{M_0}\paren{\sigma^{n-1};\ms{D}\tilde{\ms{r}}_{\delta\ul{\ms{h}}^{2}} }+ \mc{O}_{M_0-1}\paren{\sigma^n; \tilde{\ms{r}}_{\delta\ms{m}}} + \mc{O}_{M_0-2}\paren{\sigma^n; \ms{r}_{\delta\ms{g}}} +\mc{O}_{M_0-1}\paren{\sigma^n; \ms{D}\ms{r}_{\delta\ms{g}}}\right. \\
            &\notag \hspace{40pt}+ \mc{O}_{M_0-1}\paren{\sigma^{n-1};\ms{D}\ms{r}_{\delta\ms{m}}}+ \mc{O}_{M_0-2}\paren{\sigma^{n-1};\tilde{\ms{r}}_{\delta\ms{f}^1}} + \mc{O}_{M_0-2}\paren{\sigma^{n-1}; \tilde{\ms{r}}_{\delta\ms{f}^0}} \\
            &\hspace{40pt}\left.+ \mc{O}_{M_0-2}\paren{\sigma^{n-1};\ms{r}_{\delta\ms{h}}^0} + \mc{O}_{M_0-2}\paren{\sigma^{n};\tilde{\ms{r}}_{\delta\ms{w}^2}}\right]\vert_\sigma d\sigma\\
            &= \rho^{n+1}\tilde{\ms{r}}_{\delta\ul{\ms{h}}^0}\text{,}
        \end{aligned}
    \end{align*}
    with $\tilde{\ms{r}}_{\delta\ul{\ms{h}}^0}\rightarrow^{M_0-n-2}0$, and where we used \eqref{improved_decay_f_0}, \eqref{improved_decay_delta_m}, \eqref{improved_decay_w_2}, as well as: 
    \begin{equation*}
        \rho^{-1}\delta\ul{\ms{h}}^0 = \rho^{n-2}\ms{r}_{\delta\ul{\ms{h}}^0} \rightarrow^{M_0-n}0\text{.}
    \end{equation*}

    Finally, one can integrate \eqref{transport_dw1}:
    \begin{align*}
        \delta\ms{w}^1&=\rho\int_0^\rho \left[\mc{S}\paren{\sigma^{n-1}\tilde{\ms{D}\ms{r}}_{\delta\ms{w}^2}}+\mc{O}_{M_0-3}\paren{\sigma^n ; \ms{r}_{\delta\ms{g}} }+\mc{O}_{M_0-2}\paren{\sigma^n; \tilde{\ms{r}}_{\delta\ms{m}}} + \mc{O}_{M_0-3}\paren{\sigma^n ;\ms{D}\ms{r}_{\delta\ms{g}}}\right.\\
        &\hspace{40pt} + \mc{O}_{M_0-2}\paren{\sigma^{n-1};\ms{r}_{\delta\ms{w}^1}} + \mc{O}_{M_0-1}\paren{\sigma^{n-1} ;\tilde{\ms{r}}_{\delta\ms{f}^0}}+ \mc{O}_{M_0-1}\paren{\sigma^{n}\tilde{\ms{r}}_{\delta\ms{f}^1}} \\
        &\hspace{40pt} + \mc{O}_{M_0-1}\paren{\sigma^{n-1};\ms{r}_{\delta\ul{\ms{h}}^0}} + \mc{O}_{M_0}\paren{\sigma^n ;\tilde{\ms{r}}_{\delta\ul{\ms{h}}^{2}}}\left.\right]\vert_\sigma d\sigma\\
        &=\rho^{n+1}\tilde{\ms{r}}_{\delta\ms{w}^1}\text{,}
    \end{align*}
    with $\tilde{\ms{r}}_{\delta\ms{w}^1}\rightarrow^{M_0-n-3} 0$, where we used \eqref{improved_decay_f_0}, \eqref{improved_decay_w_2}, \eqref{improved_decay_h_2}, \eqref{improved_decay_delta_m} as well as: 
    \begin{equation*}
        \rho^{-1}\delta\ms{w}^1 = \rho^{n-2}\ms{r}_{\delta\ms{w}^1}\rightarrow^{M_0-n-1}0\text{.}
    \end{equation*}
    
    In summary, we obtained the following improved asymptotics: 
    \begin{gather}\label{improved_asymptotics}
        \begin{aligned}
        &\delta\ul{\ms{h}}^{2} = \rho^n \tilde{\ms{r}}_{\delta\ul{\ms{h}}^{2}}\text{,}\qquad  &&\tilde{\ms{r}}_{\delta\ul{\ms{h}}^{2}}\rightarrow^{M_0-n-1} 0\text{,}\\
        &\delta\ms{f}^0 = \rho^n\tilde{\ms{r}}_{\delta\ms{f}^0}\text{,}\qquad &&\tilde{\ms{r}}_{\delta\ms{f}^0}\rightarrow^{M_0-n}0\text{,}\\
        &\delta\ms{w}^2 = \rho^n\tilde{\ms{r}}_{\delta\ms{w}^2}\text{,}\qquad &&\tilde{\ms{r}}_{\delta\ms{w}^2}\rightarrow^{M_0-n-2}0\text{,}\\
        &\delta\ms{f}^1 = \rho^{n+1}\tilde{\ms{r}}_{\delta\ms{f}^1}\text{,}\qquad &&\tilde{\ms{r}}_{\delta\ms{f}^1}\rightarrow^{M_0-n-1}0\text{,}\\
        &\delta\ul{\ms{h}}^0 = \rho^{n+1} \tilde{\ms{r}}_{\delta\ms{h}^0}\text{,}\qquad &&\tilde{\ms{r}}_{\delta\ul{\ms{h}}^0}\rightarrow^{M_0-n-2} 0\\
        &\delta\ms{w}^0 = \rho^{n}\tilde{\ms{r}}_{\delta\ms{w}^0}\text{,}\qquad &&\tilde{\ms{r}}_{\delta\ms{w}^2} \rightarrow^{M_0-n-2}0\text{,}\\
        &\delta\ms{m} = \rho^{n+1}\tilde{\ms{r}}_{\delta\ms{m}}\text{,}\qquad &&\tilde{\ms{r}}_{\delta\ms{m}}\rightarrow^{M_0-n-2}0\text{,}\\
        &\delta\ms{g} = \rho^{n+2}\tilde{\ms{r}}_{\delta\ms{g}}\text{,}\qquad &&\tilde{\ms{r}}_{\delta\ms{g}}\rightarrow^{M_0-n-2} 0\text{,}\\
        &\delta\ms{w}^1 = \rho^{n+1}\tilde{\ms{r}}_{\delta\ms{w}^1}\text{,}\qquad &&\tilde{\ms{r}}_{\delta\ms{w}^1}\rightarrow^{M_0-n-3}0\text{.}
        \end{aligned}
    \end{gather}
    The idea is now to iterate the process up to saturation, namely, when one hits the lowest regularity for $\delta\ms{w}^1$. Performing this operation $k$--times, yields: 
    \begin{gather}\label{improved_asymptotics_k}
        \begin{aligned}
        &\delta\ul{\ms{h}}^{2} = \rho^{n-2+2k} \tilde{\ms{r}}^{(k)}_{\delta\ul{\ms{h}}^{2}}\text{,}\qquad  &&\tilde{\ms{r}}^{(k)}_{\delta\ul{\ms{h}}^{2}}
            \rightarrow^{M_0-n+1-2k}0\text{,}\\
        &\delta\ms{f}^0 = \rho^{n-2+2k}\tilde{\ms{r}}^{(k)}_{\delta\ms{f}^0}\text{,}\qquad &&\tilde{\ms{r}}^{(k)}_{\delta\ms{f}^0}\rightarrow^{M_0-n-2+2k}0\text{,}\\
        &\delta\ms{w}^2 = \rho^{n-2+2k}\tilde{\ms{r}}^{(k)}_{\delta\ms{w}^2}\text{,}\qquad &&\tilde{\ms{r}}^{(k)}_{\delta\ms{w}^2}\rightarrow^{M_0-n-2k}0\text{,}\\
        &\delta\ms{f}^1 = \rho^{n-1+2k}\tilde{\ms{r}}^{(k)}_{\delta\ms{f}^1}\text{,}\qquad &&\tilde{\ms{r}}^{(k)}_{\delta\ms{f}^1}\rightarrow^{M_0-n+1-2k}0\text{,}\\
        &\delta\ul{\ms{h}}^0 = \rho^{n-1+2k} \tilde{\ms{r}}^{(k)}_{\delta\ul{\ms{h}}^0}\text{,}\qquad &&\tilde{\ms{r}}^{(k)}_{\delta\ul{\ms{h}}^0}\rightarrow^{M_0-n-2k} 0\\
        &\delta\ms{w}^0 = \rho^{n-2+2k}\tilde{\ms{r}}^{(k)}_{\delta\ms{w}^0}\text{,}\qquad &&\tilde{\ms{r}}^{(k)}_{\delta\ms{w}^2} \rightarrow^{M_0-n-2k}0\text{,}\\
        &\delta\ms{m} = \rho^{n-1+2k}\tilde{\ms{r}}^{(k)}_{\delta\ms{m}}\text{,}\qquad &&\tilde{\ms{r}}^{(k)}_{\delta\ms{m}}\rightarrow^{M_0-n-2k}0\text{,}\\
        &\delta\ms{g} = \rho^{n+2k}\tilde{\ms{r}}^{(k)}_{\delta\ms{g}}\text{,}\qquad &&\tilde{\ms{r}}^{(k)}_{\delta\ms{g}}\rightarrow^{M_0-n-2k} 0\text{,}\\
        &\delta\ms{w}^1 = \rho^{n-1+2k}\tilde{\ms{r}}^{(k)}_{\delta\ms{w}^1}\text{,}\qquad &&\tilde{\ms{r}}^{(k)}_{\delta\ms{w}^1}\rightarrow^{M_0-n-1-2k}0\text{.}
        \end{aligned}
    \end{gather}
    
    For $M_0-n$ even, after 
    \begin{equation*}
        k_\star = (M_0-n)/2-1
    \end{equation*} 
    iterations, the vertical field $\delta\ms{w}^1$ satisfies: 
    \begin{equation*}
        \delta\ms{w}^1 = \rho^{M_0-3}\tilde{\ms{r}}^{(k_\star)}_{\delta\ms{w}^1}\text{,}\qquad \tilde{\ms{r}}^{(k_\star)}_{\delta\ms{w}^1}\rightarrow^{1}0\text{,}
    \end{equation*}
    and one cannot keep improving the decay. 
\end{proof}

We now simply need to apply this proposition to the renormalised and auxiliary fields: 
\begin{corollary}\label{corollary_asymptotics}
    Let $(\mc{M}, F, g)$ and $(\mc{M}, \check{F}, \check{g})$ be two FG-aAdS segments of regularity $M_0\geq n+2$ such that: 
    \begin{align*}
        \mf{g}^{(0)} = \check{\mf{g}}^{(0)}\text{,}\qquad \mf{g}^{(n)} = \check{\mf{g}}^{(n)}\text{,}\qquad \mf{f}^{0,((n-4)_+)} = \check{\mf{f}}^{0,((n-4)_+)}\text{,}\qquad \mf{f}^{1, (0)}=\check{\mf{f}}^{1, (0)}\text{.}
    \end{align*}
    Then, the renormalised and auxiliary fields satisfy the following asymptotics: 
    \begin{gather*}
        \ms{Q} = \rho^{M_0}\ms{r}_{\ms{Q}}\text{,}\qquad \ms{r}_{\ms{Q}}\rightarrow^2 0\text{,}\qquad 
        \ms{B} = \rho^{M_0-2}\ms{r}_{\ms{B}}\text{,}\qquad \ms{r}_{\ms{B}}\rightarrow^1 0\text{,}\\
        \Delta\ms{w}^2 = \rho^{M_0-4}\ms{r}_{\Delta\ms{w}^2}\text{,}\qquad \ms{r}_{\Delta\ms{w}^2} \rightarrow^2 0\text{,}\qquad 
        \Delta\ms{w}^1 = \rho^{M_0-3}\ms{r}_{\Delta\ms{w}^1}\text{,}\qquad \ms{r}_{\Delta\ms{w}^1} \rightarrow^1 0\text{,}\\
        \Delta\ms{w}^\star = \rho^{M_0-4}\ms{r}_{\Delta\ms{w}^0}\text{,}\qquad \ms{r}_{\Delta\ms{w}^0}\rightarrow^2 0\text{,}\\
        \Delta \ul{\ms{h}}^0 = \rho^{M_0-3}\ms{r}_{\Delta\ul{\ms{h}}^0}\text{,}\qquad \ms{r}_{\Delta\ul{\ms{h}}^0}\rightarrow^2 0\text{,}\qquad 
        \Delta\ul{\ms{h}}^{2} = \rho^{M_0-4}\ms{r}_{\Delta\ul{\ms{h}}^{2}}\text{,}\qquad \ms{r}_{\Delta\ul{\ms{h}}^{2}} \rightarrow^2 0\text{.}
    \end{gather*}
\end{corollary}

\begin{proof}
    The proof follows from from Definition \ref{def_renormalisation} and Proposition \ref{prop_higher_order_vanishing}. In particular, the asymptotics for $\ms{Q}$ follow simply from its definition. Integrating indeed \eqref{transport_Q} gives, on a local chart $(U, \varphi)$:  
    \begin{align*}
       \ms{Q} &=\int_0^\rho \left[\mc{O}_{M_0-2}\paren{\sigma; \delta\ms{g}} + \mc{O}_{M_0-2}\paren{\sigma; \ms{Q}}\right]\vert_\sigma d\sigma\\
        &=\int_0^\rho \left[\mc{O}_{M_0-2}\paren{\sigma^{M_0-1}; \ms{r}_{\delta\ms{g}}} + \mc{O}_{M_0-2}\paren{\sigma; \ms{Q}}\right]\vert_\sigma d\sigma\\
        &=\rho^{M_0}\tilde{\ms{r}}_{{\delta\ms{g}}} + \int_0^\rho \mc{O}_{M-2}\paren{\sigma; \ms{Q}}\vert_\sigma d\sigma\text{,}
    \end{align*}
    with $\tilde{\ms{r}}_{\delta\ms{g}}\rightarrow^2 0$ a vertical tensor of order $(0, 2)$. The asymptotic therefore follows after using Grönwall's inequality on the following estimate:
    \begin{equation*}
        \abs{\ms{Q}}_{2,\varphi}\lesssim \rho^{M_0}\abs{\tilde{\ms{r}}_{\delta\ms{g}}}_{2,\varphi} + \int_0^\rho \sigma \abs{\ms{Q}}_{2,\varphi}\vert_\sigma d\sigma\text{.}
    \end{equation*}

    The improved asymptotics for $\ms{B}$ is a consequence of the above and \eqref{def_B}: 
    \begin{align*}
        \ms{B} &= \mc{S}\paren{\ms{D}\delta\ms{g}} +\mc{S}\paren{\ms{DQ}}\\
        &=\rho^{M_0-2}\ms{r}_{\ms{B}}\text{,}\qquad \ms{r}_{\ms{B}}\rightarrow^1 0\text{.}
    \end{align*}

    Note also that for any vertical tensor field $\ms{A}$, the following holds from \eqref{renormalised_def}:
    \begin{align*}
        \Delta \ms{A} = \delta\ms{A} +\mc{S}\paren{\ms{g}; \ms{A}, \delta\ms{g}} + \mc{S}\paren{\ms{g}; \ms{A},\ms{Q}}\text{,}
    \end{align*}
    as well as: 
    \begin{align*}
        \Delta\ms{w}^\star = \Delta \ms{w}^0 + \mc{S}\paren{\ms{g}; \Delta\ms{w}^2} + \mc{S}\paren{\ms{g}; \ms{w}^2 ,\delta\ms{g}} + \mc{S}\paren{\ms{g}; \ms{w}^2, \ms{Q}}\text{.}
    \end{align*}
    This yields, using Proposition \ref{prop_higher_order_vanishing}: 
    \begin{gather*}
        \Delta\ms{w}^2 = \rho^{M_0-4}\ms{r}_{\Delta\ms{w}^2}\text{,}\qquad \ms{r}_{\Delta\ms{w}^2} \rightarrow^2 0\text{,}\qquad 
        \Delta\ms{w}^1 = \rho^{M_0-3}\ms{r}_{\Delta\ms{w}^1}\text{,}\qquad \ms{r}_{\Delta\ms{w}^1} \rightarrow^1 0\text{,}\\
        \Delta\ms{w}^\star = \rho^{M_0-4}\ms{r}_{\Delta\ms{w}^0}\text{,}\qquad \ms{r}_{\Delta\ms{w}^0}\rightarrow^2 0\text{,}\\
        \Delta \ul{\ms{h}}^0 = \rho^{M_0-3}\ms{r}_{\Delta\ul{\ms{h}}^0}\text{,}\qquad \ms{r}_{\Delta\ul{\ms{h}}}\rightarrow^2 0\text{,}\qquad \Delta\ul{\ms{h}}^{2} = \rho^{M_0-4}\ms{r}_{\Delta\ul{\ms{h}}^{2}}\text{,}\qquad \ms{r}_{\Delta\ul{\ms{h}}^{2}}\rightarrow^2 0\text{.}
        \end{gather*}
\end{proof}

\subsection{Carleman estimates}\label{sec_carleman}

In this section, we will establish two Carleman estimates, one for the transport equations and one the wave equations, well-adapted to the geometry of Anti-de Sitter solutions. The validity of those estimates will rely on the so-called \emph{Generalised Null Convexity Criterion}, as defined in \cite{Holzegel22}, generalising the previous \emph{Null Convexity Criterion}, as established in \cite{hol_shao:uc_ads, hol_shao:uc_ads_ns}, in a gauge-invariant fashion. 

First, in order to measure the size of vertical fields: 
\begin{definition}\label{r_metric}
    Let $(\mc{M}, g)$ be a FG-aAdS segment and let $\mf{h}$ be a Riemannian metric on $\mc{I}$. Then, we write, for a vertical field $\ms{A}$ of rank $(k,l)$: 
    \begin{equation*}
        |\ms{A}|^2_{\mf{h}} := \mf{h}_{a_1c_1}\mf{h}_{a_2c_2}\dots \mf{h}_{a_kc_k}\mf{h}^{b_1d_1}\mf{h}^{b_2d_2}\dots \mf{h}^{b_ld_l}\ms{A}^{a_1\dots a_k}{}_{b_1\dots b_l}\ms{A}^{c_1\dots c_k}{}_{d_1\dots d_l}\text{,}
    \end{equation*}
    the point-wise norm with respect to $\mf{h}$. 
\end{definition}
\begin{remark}
    The metric $\mf{h}$ will only be used to measure sizes of vertical fields on sets of compact closure. Our results will, of course, not depend on the specific metric chosen. 
\end{remark}

We are now ready to define the GNCC:
\begin{definition}\label{def_GNCC}
    Let $(\mc{M}, g)$ be a strongly FG-aAdS segment\footnote{See Definition \ref{def_strongly_FG_aAdS}.}, $\mf{h}$ a Riemannian metric on $\mc{I}$ and $\mi{D}\subset\mc{I}$ open, with compact closure. Then, we say $(\mi{D}, \mf{g}^{(0)})$ satisfies the \emph{Generalised Null Convexity Criterion} (GNCC) if and only if there exists a scalar function $\eta\in C^4(\overline{\mi{D}})$ and a strictly positive constant $c>0$ such that, for all $\mf{g}^{(0)}$--null vector $\mf{X}$ on $\mi{D}$: 
    \begin{equation}\label{GNCC}
        \begin{cases}
            \paren{\mf{D}_{\mf{g}^{(0)}}^2 \eta - \eta\cdot\mf{g}^{(2)}}(\mf{X}, \mf{X})>c\eta \mf{h}(\mf{X}, \mf{X})\text{,}\qquad &\text{in }\mi{D}\text{,}\\
            \eta >0\text{,}\qquad &\text{in }\mi{D}\text{,}\\
            \eta = 0\text{,}\qquad &\text{on }\partial\mi{D}\text{.}
        \end{cases}
    \end{equation}
    Furthermore, we say that $(\mi{D}, \mf{g}^{(0)})$ satisfies the \emph{vacuum GNCC} if \eqref{GNCC} is satisfied with $\mf{g}^{(2)}$ replaced by $-\frac{1}{n-2}\mf{Rc}^{(0)}$, the Ricci tensor associated to $\mf{g}^{(0)}$. 
\end{definition}

From Theorem \ref{theorem_main_fg}, one can immediately see that for Maxwell-FG-aAdS segments $(\mc{M}, g, F)$, the vacuum GNCC is equivalent to the GNCC. In other words, the validity of the GNCC depends only on $\mf{g}^{(0)}$ and not on the data for Maxwell fields.

\begin{proposition}
    Let $(\mc{M}, g, F)$ be a Maxwell-FG-aAdS segment and let $\mi{D}\subset \mc{I}$ be open and with compact closure. Then, if $(\mi{D},\mf{g}^{(0)})$ satisfies the vacuum GNCC, $(\mi{D}, \mf{g}^{(0)})$ also satisfies the GNCC. 
\end{proposition}

\begin{proof}
    Immediate, since for any Maxwell-FG-aAdS segment, the coefficient $\mf{g}^{(2)}$ does not depend on the Maxwell data, from Theorem \ref{theorem_main_fg}. As a consequence, it must be identical to the one obtained in the vacuum case \cite{shao:aads_fg}: 
    \begin{equation*}
        -\mf{g}^{(2)}=\frac{1}{n-2}(\mf{Rc}^{(0)}-\frac{1}{2(n-1)}\mf{Rs}^{(0)}\cdot \mf{g}^{(0)})\text{.}
    \end{equation*}
    In particular, for any $\mf{g}^{(0)}$--null vector $\mf{X}$ on  $\mi{D}$: 
    \begin{equation*}
        \mf{g}^{(2)}(\mf{X}, \mf{X}) = -\frac{1}{n-2}\mf{Rc}^{(0)}(\mf{X}, \mf{X})\text{.}
    \end{equation*}
\end{proof}

Before stating the estimates, we will need to to describe the domains on which we will work: 
\begin{definition}
    Let $(\mc{M}, g)$ be a strongly FG-aAdS segment and let $(\mi{D}, \mf{g}^{(0)})$ satisfy the GNCC as in Definition \ref{def_GNCC} with $\eta$ satisfying the conditions \eqref{GNCC}. We define the following function $f\in C^4((0, \rho_0)\times \mi{D})$: 
    \begin{equation}
        f := \rho\cdot \eta^{-1}\text{,}
    \end{equation}
    as well as the following domain: 
    \begin{equation}
        \Omega_\star = \lbrace p \in \mc{M} \, |\, f|_p<f_\star \rbrace\text{,}
    \end{equation}
    where ${f}_\star>0$.
\end{definition}

The following Proposition establishes the appropriate Carleman estimate for the transport equations:

\begin{proposition}[Proposition 4.9 in \cite{Holzegel22}]\label{prop_transport_carleman}
    Let $(\mc{M}, g)$ be a strongly FG-aAdS segment and let: 
    \begin{itemize}
        \item $\mf{h}$ be a Riemannian metric on $\mc{I}$ and $\mi{D}\subset \mc{I}$ be an open set with compact closure.
        \item Assume $(\mi{D}, \mf{g}^{(0)})$ satisfies the GNCC, with $\eta\in C^4(\overline{\mi{D}})$ satisfying \eqref{GNCC}. 
    \end{itemize}
    Then, for any $s\geq 0$, $\kappa\in \R$ and $\lambda, f_\star, p>0$ such that\footnote{Here, $0<f_\star\ll_{\mf{g}^{(0)}, \mi{D}} 1$ means that we assume $f_\star$ to be small enough with respect to $\mf{g}^{(0)}$ and $\mi{D}$. We will keep using this notation throughout the rest of the chapter.}: 
    \begin{gather}
        2\kappa\geq \max(n-2, s-3)\text{,}\qquad 0<2p<1\text{,}\qquad 0<f_\star\ll_{\mf{g}^{(0)}, \mi{D}} 1\text{,}
    \end{gather}
    there exist constants $C, D>0$, depending on $(\ms{g}, \mf{h}, \mi{D})$,  such that, for any vertical field $\ms{A}$ on $\mc{M}$: 
    \begin{align}
        \label{carleman_transport}\int_{\Omega_\star} \omega_\lambda(f)\rho^{s+2}|\Lie_\rho\ms{A}|_{\mf{h}}^2d\mu_g + D\lambda\limsup_{\rho_\star\searrow 0} \int_{\Omega_\star \cap\lbrace \rho=\rho_\star\rbrace} &\rho^s |\rho^{-\kappa -1}\ms{A}|_{\mf{h}}^2 d\mu_{\ms{g}}\\
        &\notag\geq C\lambda \int_{\Omega_\star} \omega_\lambda(f)\rho^{s}f^{p}|\ms{A}|_{\mf{h}}^2d\mu_g\text{,}
    \end{align}
    where $d\mu_g$ is the volume form on $(\mc{M}, g)$, and $d\mu_{\ms{g}}$ the volume form on level sets of $\rho$ induced by the vertical metric $\ms{g}$ and where we wrote: 
    \begin{equation*}
        \omega_\lambda(f) := e^{-\lambda f^p/p}f^{n-2-p-2\kappa}\text{.}
    \end{equation*}
    \begin{proof}
        See Appendix \ref{app:prop_transport_carleman}.
    \end{proof}
\end{proposition}

\begin{remark}
    In fact, one does not the GNCC assumption on $\mi{D}$ in the proof \cite{Holzegel22}. This is because this Carleman estimate can be obtained in full generality, regardless of this condition. The reason for which we incorporate this assumption in the proposition is due to the fact that we will couple it to the following Carleman estimate for waves, where the GNCC plays a crucial role. 
\end{remark}

\begin{theorem}[Theorem 5.11 of \cite{Shao22}]\label{thm_carleman}
    Let $(\mc{M}, g)$ be a strongly FG-aAdS segment and let: 
    \begin{itemize}
        \item $\mf{h}$ be a Riemannian metric on $\mc{I}$ as in Definition \ref{r_metric}, 
        \item $\mi{D}\subset \mc{I}$ an open region with compact closure, with $(\mi{D}, \mf{g}^{(0)})$ satisfying the GNCC, and with $\eta\in C^4(\overline{\mi{D}})$ as in Definition \ref{def_GNCC},
        \item $k, l\geq 0$ integers, fixed, and $\sigma\in \R$.
    \end{itemize}
    Then, there exist $\mc{C}_0\geq 0$ and $C, D>0$ depending on ($\ms{g}, \mi{D}, k, l, \mf{h}$) such that for all $\kappa\in \R$: 
    \begin{gather}
        2\kappa\geq n-1+\mc{C}_0\text{,}\qquad \kappa^2 -(n-2)\kappa + \sigma-(n-1)-\mc{C}_0\geq 0\text{,}
    \end{gather}
    and for all constants $f_\star, \lambda, p>0$ with
    \begin{gather*}
        0<f_\star \ll_{\ms{g}, \mf{h}, \mi{D}, k, l}1\text{,}\qquad \lambda \gg_{\ms{g}, \mf{h}, \mi{D}, k, l}\abs{\kappa} + \abs{\sigma}\text{,}\qquad 0<p<1/2\text{,}
    \end{gather*}
    the following Carleman estimate holds for all vertical tensor field $\ms{A}$ on $\mc{M}$ of rank $(k,l)$ such that both $\ms{A}, \ol{\nabla}\ms{A}$ vanish identically on $\lbrace f=f_\star\rbrace$: 
    \begin{align}\label{carleman_wave}
        &\int_{\Omega_\star}\omega_\lambda(f)\abs{(\ol{\Box}+ \sigma)\ms{A}}_{\mf{h}}^2d\mu_g \\
        &\notag\hspace{50pt}+ D\lambda^3\limsup_{\rho_\star\searrow 0}\int_{\{\rho=\rho_\star\}\cap\Omega_\star}\lbrace |\ol{\ms{D}}_\rho(\rho^{-\kappa}\ms{A})|^2_{\mf{h}}+ |\ms{D}(\rho^{-\kappa}\ms{A})|_{\mf{h}}^2 + |\rho^{-\kappa-1}\ms{A}|_{\mf{h}}^2\rbrace d\mu_{\ms{g}}\\
        &\notag\geq C\lambda \int_{\Omega_\star}\omega_\lambda(f)f^p\paren{\rho^4|\ol{\ms{D}}_\rho\ms{A}|^2_{\mf{h}} + \rho^4|\ms{D}\ms{A}|^2_{\mf{h}} + f^{2p}|\ms{A}|^2_{\mf{h}}}d\mu_g\text{,}
    \end{align}
    where $d\mu_g$ is the volume form on $\mc{M}$ generated by $g$ and $d\mu_{\ms{g}}$ the volume form on level sets of $\rho$ generated by $\ms{g}$. 
\end{theorem}

The Carleman estimate \eqref{carleman_wave} relies on the domain $\mi{D}$ and the metric $\mf{g}^{(0)}$ to satisfy the GNCC. In fact, as pointed out in \cite{McGill20} and \cite{Shao22}, such a property is related to the propagation of near-boundary null geodesics, as the following informal characterisation suggests: 

\begin{remark}[Characterisation of the GNCC, Theorem 4.1 in \cite{Shao22}]
    If an open region $\mi{D}\subset \mc{I}$, with compact closure, satisfies the GNCC, then any near-boundary $g$--null geodesic must either start from $\lbrace 0\rbrace \times \mi{D}$, or end on $\lbrace 0\rbrace \times \mi{D}$. 
\end{remark}

In other words, the existence of near-boundary null geodesics flying over $\mi{D}$ is an obstacle to the GNCC, and thus, to the validity of the Carleman estimates. In light of Theorem 1.4 of \cite{GUISSET2024223}, the existence of such geodesics allows for the construction of counterexamples to unique continuation from $\lbrace 0 \rbrace \times \mi{D}$, using geometric optics. 

\subsection{Unique continuation result}

In this section, we will apply the two Carleman estimates to the wave-transport system obtained in Propositions \ref{prop_transport_diff}, \ref{prop_wave_difference} and \ref{prop_transport_diff_w_h}. The improved asymptotics obtained in the previous sections will turn out to be useful to treat boundary terms. We now have everything in our hands for the following proposition, establishing our unique continuation result: 

\begin{proposition}\label{prop_unique_continuation1}
    Let $n>2$, and $(\mc{M}, g, F)$ and $(\mc{M}, \check{g}, \check{F})$ be two Maxwell-FG-aAdS segments regular to order $M_0$, $M_0$ big enough, with: 

    \begin{itemize}
        \item $\mi{D}\subset \mc{I}$ an open region with compact closure $\ol{\mi{D}}$ such that $(\mi{D}, \mf{g}^{(0)})$ satisfies the GNCC.
        \item Identical holographic data on $\mi{D}$: 
             \begin{gather*}
                \mf{g}^{(0)}\vert_{\mi{D}} = \check{\mf{g}}^{(0)}\vert_{\mi{D}} \text{,}\qquad \mf{g}^{(n)}\vert_{\mi{D}}  = \check{\mf{g}}^{(n)}\vert_{\mi{D}} \text{,}\qquad \mf{f}^{0, ((n-4)_+)}\vert_{\mi{D}}  = \check{\mf{f}}^{0, ((n-4)_+)}\vert_{\mi{D}} \text{,}\qquad \mf{f}^{1, (0)}\vert_{\mi{D}}  = \check{\mf{f}}^{1, (0)}\vert_{\mi{D}} \text{.}
            \end{gather*}
    \end{itemize}
    Then $(g, F) \equiv (\check{g}, \check{F})$ in some neighbourhood of $\lbrace 0 \rbrace \times \mi{D}$ of the form $\Omega_\star$, with $f_\star$ small enough. 
\end{proposition}

\begin{proof}
    The proof is a classical argument, once one has obtained the appropriate Carleman estimates. We will show it here for completeness. 

    First, let us fix $\mf{h}$ a Riemannian metric on $\mi{D}\subset \mc{I}$, as well as the following constants: 
    \begin{gather*}
        0 < p < 1/2\text{,}\qquad f_\star \ll_{\ms{g}, \mi{D}} 1\text{,}\qquad 1\ll_{\ms{g}, \check{\ms{g}}, \mi{D}} \kappa < M_0-5\text{,}
    \end{gather*}
    given $M_0$ sufficiently large. 

    We will also need a cutoff $\chi = \bar{\chi}\circ f$ defined such that: 
    \begin{equation*}
        \bar{\chi}(z) := \begin{cases}
            1\text{,}\qquad s\leq f_i\\
            0\text{,}\qquad s\geq f_o
        \end{cases}\text{,}\qquad  0 < f_i < f_o < f_\star\text{.}
    \end{equation*}
    Let us also define the following regions of interests: 
    \begin{gather*}
        \Omega_i := \lbrace 0 < f < f_i\rbrace \text{,}\qquad \Omega_o := \lbrace f_i < f < f_o \rbrace\text{,}
    \end{gather*}
    as well as the following collections of vertical fields: 
    \begin{gather*}
        \mc{T} := \lbrace \delta\ms{g}, \delta\ms{m}, \ms{B}, \ms{Q}, \delta\ms{f}^0, \delta \ms{f}^1\rbrace\text{,}\qquad \mc{W} := \lbrace \Delta\ms{w}^\star, \Delta \ms{w}^1, \Delta\ms{w}^2, \Delta\ul{\ms{h}}^0, \Delta\ul{\ms{h}}^{2}\rbrace\text{,} 
    \end{gather*}
    associated to the transport and wave equations, respectively. 

    The idea is now to apply Carleman estimates \eqref{carleman_transport} and \eqref{carleman_wave} to fields cutoff in $\Omega_i \cup \Omega_o$, i.e. $\chi \cdot \ms{A}$, for $\ms{A}\in \mc{T} \cup \mc{W}$. This way, the assumptions of the Theorem \ref{thm_carleman} are satisfied, namely that $\ms{A} =\ol{\nabla}\ms{A} = 0$ on $\lbrace f = f_\star \rbrace$. First, the following integral on $\Omega_i$, where $\chi\equiv 1$, can be bounded by above, using Theorem \ref{thm_carleman}:
    \begin{align}
        &\label{estimate_wave1}\lambda\int_{\Omega_i} \omega_\lambda(f)f^p \sum\limits_{\ms{A}\in \mc{W}}\paren{\rho^{2p}\abs{\ms{A}}^2_{\mf{h}} + \rho^4\abs{\ms{DA}}^2_{\mf{h}}} d\mu_g\\
        &\notag\leq \lambda \int_{\Omega_\star} \omega_\lambda(f)f^p \sum\limits_{\ms{A}\in \mc{W}}\paren{\rho^{2p}\abs{\chi\ms{A}}^2_{\mf{h}} + \rho^4\abs{\ms{D(\chi A)}}^2_{\mf{h}}} d\mu_g\\
        &\notag\lesssim \lambda^3 \limsup\limits_{\rho_\star\searrow 0} \int_{\Omega_\star\cap \lbrace \rho = \rho_\star\rbrace}\sum\limits_{\ms{A}\in \mc{W}} \left[|\ol{\ms{D}}_\rho(\rho^{-\kappa}\chi\ms{A})|_{\mf{h}}^2 + |\ms{D}(\rho^{-\kappa}\chi\ms{A})|^2_{\mf{h}} + |\rho^{-\kappa- 1}\chi\ms{A}|^2_{\mf{h}}\right]d\mu_{\ms{g}} \\
        &\notag\hspace{30pt} + \int_{\Omega_i \cup \Omega_o} \omega_\lambda(f)\sum\limits_{\ms{A}\in\mc{W}}|\chi(\ol{\Box} + \sigma_{\ms{A}})\ms{A}|^2_{\mf{h}} d\mu_g\\
        &\notag\hspace{30pt} + \int_{\Omega_i \cup \Omega_o} \omega_\lambda(f) |[\ol{\Box}, \chi]\ms{A}|^2_{\mf{h}}d\mu_g =: I_1 + I_2 + I_3\text{,}
    \end{align}
    where $\sigma_{\ms{A}}$ is the mass for the wave equation of $\ms{A}$, which can be found in Proposition \ref{prop_wave_difference}.

    Note that $I_1=0$, since, from Proposition \ref{prop_higher_order_vanishing}, the vertical fields
    \begin{equation*}
        \ol{\ms{D}}_\rho(\rho^{-\kappa}\ms{A})\text{,}\qquad  \ms{D}(\rho^{-\kappa}\ms{A})\text{,}\qquad \rho^{-\kappa -1}\ms{A}\text{,}\qquad \ms{A}\in \mc{W}\text{,}
    \end{equation*}
    will have a power of $\rho$ strictly positive, for $M_0$ big enough. Note that $I_3$ can be decomposed into two integrals: one on $\Omega_i$ and the other on $\Omega_o$. The former identically vanishes since it contains at least one derivative of $\chi$, as a consequence: 
    \begin{equation*}
        I_3 = \int_{\Omega_o} \omega_\lambda(f) |[\ol{\Box}, \chi]\ms{A}|^2_{\mf{h}}d\mu_g\text{.}
    \end{equation*}
    
    Using Proposition \ref{prop_wave_difference}, one has for ${I}_2$:
    \begin{align}
        \label{I_2}I_2 \lesssim& \int_{\Omega_i} (\omega_\lambda(f)\rho^p) \left[\sum\limits_{\ms{A}\in \mc{W}}(\rho^{4-p}|\ms{A}|_{\mf{h}}^2 + \rho^{6-p}|\ms{DA}|_{\mf{h}}^2) + \sum\limits_{\ms{G}\in \mc{T}}\rho^{2-p}|\ms{G}|^2_{\mf{h}} + \rho^{4-p}|\ms{DG}|_{\mf{h}}^2\right]d\mu_g\\
        \notag&+\int_{\Omega_o} \omega_\lambda(f)\sum\limits_{\ms{A}\in\mc{W}}|\chi(\ol{\Box} + \sigma_{\ms{A}})\ms{A}|^2_{\mf{h}} d\mu_g\text{.}
    \end{align}
    We now need to control the terms from $\mc{T}$. This can be done using the transport Carleman estimates from Proposition \ref{prop_transport_carleman}. Note that the boundary terms vanish using the asymptotics from Proposition \ref{prop_higher_order_vanishing} and Corollary \ref{corollary_asymptotics}. Furthermore, the cutoff is not necessary here since the fields and their derivatives do not need to vanish on $\Omega_\star$. One eventually obtains, using Proposition \ref{prop_transport_carleman} with $s=2$ for $\delta\ms{m}$ and $\delta \ms{f}^1$, $s=2(n-2)$ for $\delta\ms{f}^0$ and $s=0$ for the other fields: 
    \begin{align}
        \lambda\int_{\Omega_i} \omega_\lambda(f) f^p \sum\limits_{\ms{G}\in \mc{T}} |\ms{G}|^2_{\mf{h}}d\mu_g \lesssim &\notag\int_{\Omega_i} \omega_\lambda(f)[{}\sum\limits_{\ms{G}\in \mc{T}\setminus\lbrace \delta\ms{m}, \delta\ms{f}^0, \delta\ms{f}^1\rbrace } \rho^2|\Lie_\rho \ms{G}|^2_{\mf{h}}+\rho^4|\Lie_\rho(\rho^{-1}\delta\ms{m})|^2_{\mf{h}}  \\
        &\notag\left.\hspace{30pt}+\rho^4|\Lie_\rho(\rho^{-1}\delta\ms{f}^1)|_{\mf{h}}^2 + \rho^{2n-2} |\Lie_\rho(\rho^{2-n}\delta\ms{f}^0)|^2_{\mf{h}}\right]d\mu_g \\
        \label{estimate_transport_1}\lesssim& \int_{\Omega_i} (\omega_\lambda(f)\rho^p)[\sum\limits_{\ms{A}\in \mc{W}}\rho^{2-p}|\ms{A}|^2_{\mf{h}} + \sum\limits_{\ms{G}\in \mc{T}}\rho^{2-p}|\ms{G}|^2_{\mf{h}}]d\mu_g\text{,}
    \end{align}
    where we used Proposition \ref{prop_transport_diff} in the last line. 
    Similarly, one can use the same Carleman estimates for the vertical derivatives $\ms{DG}$ (where we use $s=5$ for $\ms{D}\delta\ms{m}$, $\ms{D}\delta\ms{f}^1$, $s=2n-1$ for $\ms{D}\delta\ms{f}^0$ and $s=3$ for the other fields:: 
    \begin{align}
        \label{estimate_transport_2}\lambda\int_{\Omega_i} \omega_\lambda(f)f^p \sum\limits_{\ms{G}\in\mc{T}} \rho^3 |\ms{DG}|^2_{\mf{h}}d\mu_g \lesssim& \int_{\Omega_i}\omega_\lambda(f) [\sum\limits_{\ms{G}\in \mc{T}}\rho^5|\Lie_\rho\ms{DG}|_{\mf{h}}^2 + \rho^7 |\Lie_\rho(\rho^{-1}\ms{D}\delta\ms{m})|^2_{\mf{h}}\\
        &\notag\hspace{30pt}+\rho^7 |\Lie_\rho(\rho^{-1}\ms{D}\delta\ms{f}^1)|^2_{\mf{h}}+\rho^{2n+1}|\Lie_\rho(\rho^{2-n}\ms{D}\delta\ms{f}^0)|^2_{\mf{h}}\left.\right]d\mu_g \\
        \notag\lesssim& \int_{\Omega_i} (\omega_\lambda(f)\rho^p)\left[\sum\limits_{\ms{A}\in\mc{W}} \rho^{5-p}\abs{\ms{A}}_{\mf{h}}^2+\rho^{5-p}|\ms{DA}|^2_{\mf{h}}+ \sum\limits_{\ms{G}\in \mc{T}}\rho^{5-p}|\ms{G}|_{\mf{h}}^2\right.\\
        &\notag \hspace{30pt}\left. + \rho^{5-p}|\ms{DG}|^2_{\mf{h}}\right]d\mu_g\text{,}
    \end{align}
    where we used transport equations \eqref{transport_Ddg}--\eqref{transport_Ddf1}. We can now sum \eqref{estimate_transport_1} and \eqref{estimate_transport_2}, as well as \eqref{estimate_wave1}, \eqref{I_2}, yielding: 
    \begin{align*}
        &\lambda\int_{\Omega_i} \omega_\lambda(f)f^p \left[\sum\limits_{\ms{A}\in \mc{W}}\paren{\rho^{2p}\abs{\ms{A}}^2_{\mf{h}} + \rho^4\abs{\ms{DA}}^2_{\mf{h}} }+\sum\limits_{\ms{G}\in\mc{T}} \paren{|\ms{G}|_{\mf{h}}^2 + \rho^3 |\ms{DG}|_{\mf{h}}^2}\right] d\mu_g\\
        &\lesssim \int_{\Omega_i} (\omega_f(\lambda)\rho^p)\left[\sum\limits_{\ms{A}\in \mc{W}} \left(\rho^{2-p} |\ms{A}|^2_{\mf{h}} + \rho^{5-p}|\ms{DA}|^2_{\mf{h}}\right) + \sum\limits_{\ms{G}\in \mc{T}}\left(\rho^{2-p}|\ms{G}|_{\mf{h}}^2 + \rho^{4-p}|\ms{DG}|^2_{\mf{h}}\right)\right]d\mu_g \\
        &\hspace{40pt}+\int_{\Omega_o} \omega_f(\lambda) \sum\limits_{\ms{A}\in \mc{W}}\left[\chi(\ol{\Box}+\sigma_{\ms{A}})\ms{A} + \left[\ol{\Box}, \chi\right]\ms{A}\right]d\mu_g \text{.}
    \end{align*}
    Observe now that the powers of $\rho$ appearing in the first integral of the right-hand side are higher than those on the left-hand side. Choosing $\lambda$ big enough therefore allows use to insert the former into the latter by noticing $f \gtrsim \rho$ by definition, giving: 
    \begin{align*}
        &\lambda\int_{\Omega_i} \omega_\lambda(f)f^p \left[\sum\limits_{\ms{A}\in \mc{W}}\paren{\rho^{2p}\abs{\ms{A}}^2_{\mf{h}} + \rho^4\abs{\ms{DA}}^2_{\mf{h}} }+\sum\limits_{\ms{G}\in\mc{T}} \paren{|\ms{G}|_{\mf{h}}^2 + \rho^3 |\ms{DG}|_{\mf{h}}^2}\right] d\mu_g\\
        &\lesssim \int_{\Omega_o} \omega_f(\lambda) \sum\limits_{\ms{A}\in \mc{W}}\left[\chi(\ol{\Box}+\sigma_{\ms{A}})\ms{A} + \left[\ol{\Box}, \chi\right]\ms{A}\right]d\mu_g\text{.}
    \end{align*}
    In order to get rid of the weights $\omega_\lambda(f)$, one can observe that $\omega_\lambda(f)f^p \geq \omega_\lambda (f_i)f_i^p$ on $\Omega_i$ and $\omega_\lambda(f)f^p\leq \omega_\lambda(f_i)f_i^p $ on $\Omega_o$, for $\kappa$ big enough. Overall, one obtains: 
    \begin{align*}
        &\int_{\Omega_i} \left[\sum\limits_{\ms{A}\in \mc{W}}\paren{\rho^{2p}\abs{\ms{A}}^2_{\mf{h}} + \rho^4\abs{\ms{DA}}^2_{\mf{h}} }+\sum\limits_{\ms{G}\in\mc{T}} \paren{|\ms{G}|_{\mf{h}}^2 + \rho^3 |\ms{DG}|_{\mf{h}}^2}\right] d\mu_g\\
        &\hspace{40pt}\lesssim \frac{1}{\lambda}\int_{\Omega_o} \rho^{-p} \sum\limits_{\ms{A}\in \mc{W}}\left[\chi(\ol{\Box}+\sigma_{\ms{A}})\ms{A} + \left[\ol{\Box}, \chi\right]\ms{A}\right]d\mu_g\text{,}
    \end{align*}
    where we used, once more, $f^{-p}\lesssim \rho^{-p}$. Since the right-hand side is finite, from Proposition \ref{prop_higher_order_vanishing} and Corollary \ref{corollary_asymptotics}, one obtains, after taking the limit $\lambda\rightarrow \infty$: 
    \begin{gather*}
        \delta \ms{g} \equiv 0\text{,}\qquad \delta \ms{f}^0 \equiv 0\text{,}\qquad \delta\ms{f}^1\equiv 0\text{,}
    \end{gather*}
    in $\Omega_i$. 
    \end{proof}
    
\subsection{Gauge freedom and full result}\label{sec:full_result}

As mentioned earlier, the gauge \eqref{FG-gauge} still leaves room for residual freedom. Leaving the trivial diffeormorphisms along the vertical directions aside, one still has some freedom left when it comes to specifying the metric as in \eqref{FG-gauge}. As a consequence, the assumptions of the Proposition \ref{prop_unique_continuation1} are certainly too strong, since this residual gauge freedom may bring to different holographic data, although they represent the same information. This brings us to the following definition: 

\begin{definition}\label{def_FG_radius}
    Let $(\mc{M}, g)$ be a FG-aAdS segment and let $\Sigma\subset \mc{M}$ be a neighbourhood of the conformal boundary. Then, a positive, smooth function $\check{\rho}$ on $\Sigma$ will be referred to as a FG radius for $(\Sigma, g)$ if: 
    \begin{equation}
        \check{\rho}^{-2}g^{-1}(d\check{\rho}, d\check{\rho}) =1\text{,}\qquad \check{\rho}\vert_\mc{I}=0\text{.}
    \end{equation}
\end{definition}

\begin{remark}
    This definition is of course inspired by the gauge \eqref{FG-gauge}. Obviously, the usual function $\rho\in C^\infty(\mc{M})$ is a FG radius for $(\mc{M}, g)$. 
\end{remark}

One can see the FG radius from Definition  \ref{def_FG_radius} as a change of coordinates preserving the Fefferman-Graham condition \eqref{FG-gauge}. Let us make this more precise\footnote{For a more thorough treatment, see \cite{Holzegel22}. The exposition in this work is intended primarily for illustrative purposes.}. 

Let $\check{\rho}$ be a FG radius as in Definition \ref{def_FG_radius} on $\Sigma\subset \mc{M}$, and let $p\in \mc{I}$. Considering $\gamma_p$ to be the integral curve of the $\check{\rho}g$--gradient of $\check{\rho}$ starting from $p\in \mc{I}$ with parameter $[0,\check{\rho}_0]$, satisfying: 
    \begin{equation*}
        \lim\limits_{\sigma\searrow 0}\gamma_p(\sigma) = (0,p)\text{,}
    \end{equation*}
one can identify the pair $(\sigma, p)$ with $\gamma_p(\sigma)$. As a consequence, given a FG radius $\check{\rho}$ and $\mc{J}\subset \mc{I}$, one can identify a region of the form $(0, \check{\rho}_0]\times \mc{J}$,  $\check{\rho}_0>0$, with a submanifold\footnote{Given $\mc{J}\subset \mc{I}$, the value of $\check{\rho}_0$ will depend on the presence of focal points along the geodesic flow.}:
\begin{equation*}
    \widetilde{\mc{M}}=\bigcup\limits_{p\in\mc{J}, \, \sigma\in (0, \check{\rho}_0)} \gamma_p (\sigma)\cap \Sigma \subset \mc{M}\text{,}
\end{equation*}
through the mapping that associates, to each $(\sigma,p)\in (0,\check{\rho}_0]\times \mc{J}$, the point $\gamma_p(\sigma)\in \widetilde{\mc{M}}$. 
Note that since the gradient is $g$--normal to level sets of $\check{\rho}$, the metric $g$ on $\widetilde{\mc{M}}$ will take the form \eqref{FG-gauge}:
\begin{equation}\label{new_metric}
    g = \check{\rho}^{-2}\paren{d\check{\rho}^2 + {\ms{g}}^\star}\text{,}
\end{equation}
with ${\ms{g}}^\star$ a $\check{\rho}$--vertical Lorentzian metric. In other words, given $\check{\rho}$, one immediately generates a $\check{\rho}$--FG-aAdS region $(\tilde{\mc{M}, g})$, characterised by \eqref{new_metric}, isometric to a subregion of the initial FG-aAdS segment $\mc{M}$. 

The next proposition shows, given such a change of conformal frame, how the holographic data changes accordingly:
\begin{proposition}\label{prop_gauge_equiv}
    Let $(\mc{M}, g, F)$ be a Maxwell-FG-aAdS segment, $n>2$, and let $\mf{a}\in C^\infty(\mc{I})$. Then, there is a neighbourhood $\Sigma\subset \mc{M}$ of the conformal boundary, a unique FG radius $\check{\rho}\in C^\infty(\Sigma)$ such that: 
    \begin{equation}
        \frac{\check{\rho}}{\rho}\rightarrow^0 e^{\mf{a}}\text{.}\label{limit_check_rho}
    \end{equation}
    Furthermore, let $(\mf{g}^{(0)}, \mf{g}^{(n)}, \mf{f}^{0, ((n-4)_+)}, \mf{f}^{1, (0)})$ and $(\check{\mf{g}}^{(0)}, \check{\mf{g}}^{(n)}, \check{\mf{f}}^{0, ((n-4)_+)}, \check{\mf{f}}^{1, (0)})$ denote the holographic data on $(\mc{M}, g)$ with respect to $\rho$ and $\check{\rho}$, respectively. Then, there exist universal functions $\mf{G}_n$ and $\mf{F}_n$, depending only on $n$, such that the two holographic data are related by: 
    \begin{gather}
        \label{check_g_n}\check{\mf{g}}^{(0)} =e^{2\mf{a}}\mf{g}^{(0)}\text{,}\qquad \check{\mf{g}}^{(n)} = \mf{G}_n(\mf{g}^{(n)},\partial^{\leq n-2} \mf{g}^{(0)},\mf{D}^{\leq  n-4}\mf{f}^{1, (0)}, \mf{D}^{\leq n}\mf{a})\text{,}\\
        \label{check_f}\begin{cases}
            \check{\mf{f}}^{1, (0)} = \mf{f}^{1, (0)}\text{,}\qquad \check{\mf{f}}^{0, ((n-4)_+)} = e^{-\mf{a}}\mf{f}^{0, ((n-4)_+)}\text{,}\qquad &n =3\text{,}\\
                \check{\mf{f}}^{1, (0)} =\mf{f}^{1, (0)}\text{,}\qquad \check{\mf{f}}^{0, ((n-4)_+)} = \mf{F}_n\paren{\mf{f}^{0, ((n-4)_+)}, \partial^{\leq n-3}\mf{g}^{(0)}, \mf{D}^{\leq n-3}\mf{f}^{1, (0)},\mf{D}^{\leq n-3} \mf{a}}\text{,}\qquad &n\geq 4\text{.}
        \end{cases}
    \end{gather}
\end{proposition}
    \begin{proof}
        The proof is based on an ansatz for $\check{\rho}$ of the form $\check{\rho} = e^{\ms{a}}\rho$, with $\ms{a} \in C^\infty(\mc{M})$,
         yielding the following one-form: 
         \begin{equation*}
             d\check{\rho} = e^{\ms{a}}d\rho + e^{\ms{a}}\rho\ms{Da}\text{.}
         \end{equation*}
         The Definition \ref{def_FG_radius} implies that $\ms{a}$ satisfies the following non-linear transport equation: 
         \begin{equation}
             0 = (2 + \Lie_\rho\ms{a})\Lie_\rho\ms{a} + \rho\ms{g}^{-1}(\ms{Da}, \ms{Da})\text{,}
         \end{equation}
         where we impose $\lim\limits_{\sigma\searrow 0}\ms{a}(\sigma, \cdot)=\mf{a}$ as initial condition. Such an equation can be solved using the method of characteristics in a neighbourhood $\Sigma$ of the conformal boundary.  Obtaining $\ms{a}$ is equivalent to obtaining $\check{\rho}$, verifying the FG radius definition. 

         The relation between the holographic data, for the metrics, is well known in the physics literature and can be simply obtained in the case of $\mf{g}^{(0)}$:
         \begin{align*}
             \check{\ms{g}} = e^{2\ms{a}}\ms{g} + o(1) \Rightarrow \check{\mf{g}}^{(0)} = e^{2\mf{a}}\mf{g}^{(0)}\text{.}
         \end{align*}
         The computations of the coefficient $\check{\mf{g}}^{(n)}$ are far more involved, mixing expansions of $\ms{g}$ and $\ms{a}, e^{\ms{a}}$. The only observation one can make is that it will depend on the coefficients $\lbrace \mf{g}^{(i)}\rbrace_{0\leq i \leq n}$, as well as on derivatives of $\ms{a}$ up to order $n$. The second identity in \eqref{check_g_n} follows since each coefficient $(\mf{g}^{(k)})$, $4\leq k<n$, depend on $\mf{f}^{1, (0)}$ to order $k-4$, 

         In order to see the boundary limits directly, we will look at $\bar{\ms{f}}^0$ and $\bar{\ms{f}}^1$, as defined in Definition \ref{def_vert_f_w}. The conformal change implies: 
         \begin{gather*}
            \check{\bar{\ms{f}}}^0_a =\begin{cases}
                e^{-\ms{a}}\bar{\ms{f}}^0_a\text{,}\qquad &n=3\text{,}\\
                e^{-2\ms{a}}\bar{\ms{f}}^0_a\text{,}\qquad &n\geq 4\text{,}
            \end{cases} \qquad 
                \check{\bar{\ms{f}}}^1_{ab}= \bar{\ms{f}}^1_{ab} -  2\rho\cdot (\ms{D}_{[a}\ms{a})\bar{\ms{f}}^0_{b]}\text{.}
         \end{gather*}
         For $n=3$, one immediately obtains, from Corollary \ref{prop_expansion_vertical_fields}: 
         \begin{equation}
             \check{\mf{f}}^{0, ((n-4)_+)} = e^{-\mf{a}}\mf{f}^{0, ((n-4)_+)}\text{,}\qquad \check{\mf{f}}^{1, (0)} = \mf{f}^{1}\text{.}
         \end{equation}
         One can observe that $\rho\ms{Da}\wedge\ol{\ms{f}}^0$ vanishes at the boundary regardless of the dimension, while $e^{-\ms{a}}\ol{\ms{f}}^0$ is, as for the metric, a universal function of the derivatives of $\ms{a}$, up to order $n-3$, as well as a function of the coefficients $\mf{f}^{0, (k)}$, $k\leq (n-4)_+$. As a consequence, it will also depend on derivatives of $\mf{f}^{1, (0)}$, up to order $n-3$, from Theorem \ref{theorem_main_fg}.  
     \end{proof}

    Of course, one can compute similar transformation rules for the coefficients $\mf{g}^{(2)}, \mf{g}^{(3)}, \dots$ or $\mf{f}^{1, (2)}, \mf{f}^{1, (3)}, \dots $ and so on. We will not show them explicitly in this work. For more details, see \cite{deHaro:2000vlm}. 

    \begin{remark}
        It is important to note that the GNCC condition from Definition \ref{def_GNCC} is a gauge-invariant definition. Namely, for some domain $\mi{D}\subset \mc{I}$, if $(\mi{D}, \mf{g}^{(0)})$ satisfies the GNCC, then $(\mi{D}, \check{\mf{g}}^{(0)})$ automatically satisfies it too. 
    \end{remark}

    We have now characterised gauge-equivalent initial data. The following definition makes this more precise: 

    \begin{definition}\label{def_gauge_equiv}
        Let $\mc{I}$ be an n-dimensional manifold and $(\mf{g}^{(0)}, \check{\mf{g}}^{(0)})$ be two Lorentzian metrics. Furthermore, let $(\mf{g}^{(n)}, \check{\mf{g}}^{(n)})$ be two symmetric $(0, 2)$ tensors, $(\mf{f}^{0, ((n-4)_+)}, \check{\mf{f}}^{0, ((n-4)_+)})$ be one-forms and $(\mf{f}^{1, (0)}, \check{\mf{f}}^{1, (0)})$ be two-forms on $\mc{I}$. If these fields satisfy the relations \eqref{check_g_n} and \eqref{check_f} on $\mi{D}\subset \mc{I}$ for some $\mf{a}\in C^\infty(\mi{D})$, as well as the constraints \eqref{constraints_f_g}, we say the holographic data are \emph{gauge-equivalent} on $\mi{D}$. 
    \end{definition}

    This leads us to the main theorem, which can be seen in many ways as a covariant generalisation of Proposition \ref{prop_unique_continuation1}: 

    \begin{theorem}[Main theorem]\label{thm_UC_main}
        Let $n>2$, $(\mc{M}, g, F)$ and $(\mc{M}, \check{g}, \check{F})$ be two Maxwell-FG-aAdS segments with holographic data $(\mc{I}, \mf{g}^{(0)}, \mf{g}^{(n)}, {\mf{f}}^{0, ((n-4)_+)}, \mf{f}^{1, (0)})$ and \allowdisplaybreaks{$(\mc{I}, \check{\mf{g}}^{(0)}, \check{\mf{g}}^{(n)}, \check{{\mf{f}}}^{0, ((n-4)_+)}, \check{\mf{f}}^{1, (0)})$}, respectively. Let $\mi{D}\subset \mc{I}$ be an open set on $\mc{I}$ with compact closure. If
        \begin{itemize}
            \item $(\mc{M}, g, F)$ and $(\mc{M}, \check{g}, \check{F})$ are regular to order $M_0$, big enough, depending on $\ms{g}, \check{\ms{g}}, \mi{D}$, 
            \item $(\mi{D}, \mf{g}^{(0)})$ satisfies the GNCC,
            \item the holographic data are gauge-equivalent on $\mi{D}$, as described in Definition \ref{def_gauge_equiv}.
        \end{itemize}
        Then, there exists a neighbourhood of $\lbrace 0 \rbrace \times \mi{D}$ in $\mc{M}$ such that $g$ and $\check{g}$ are isometric and ${F}$, $\check{{F}}$ are equal up to a diffeomorphism preserving the conformal boundary. More precisely, there exists $f_\star>0$ small enough, $\phi \, : \, \Omega_\star \rightarrow \mc{M}$ such that: 
        \begin{gather*}
            \phi_\star(\check{g}, \check{F}) = (g, F)\text{,}\qquad \lim\limits_{\sigma\searrow 0}\phi(\sigma, p) = (0, p)\text{,}\qquad p\in\mc{I}\text{.}
        \end{gather*}
    \end{theorem}

    \begin{proof}
        The proof is identical to the one in Theorem 6.7 in \cite{Holzegel22}, we will show it here for completeness. 

        Since the holographic data are gauge-equivalent, there exists a smooth function $\mf{a}$ on $\mi{D}$ such that the relations \eqref{check_g_n} and \eqref{check_f} hold. From Proposition \ref{prop_gauge_equiv}, there exists $\check{\rho}$, a FG radius, satisfying \eqref{limit_check_rho} on $\mc{D}$. Using this $\check{\rho}$, we can generate a FG-aAdS segment $(\mc{N}, g)$, with $g$ of the form \eqref{new_metric}, with $\check{\rho}$ as a FG-radius and with holographic data \allowdisplaybreaks{$(\mc{I}, \check{\mf{g}}^{(0)}, \check{\mf{g}}^{(n)}, \check{{\mf{f}}}^{0, ((n-4)_+)}, \check{\mf{f}}^{1, (0)})$}. This FG-aAdS segment is isometric to a neighbourhood of $\lbrace 0 \rbrace \times \mi{D}$ in $(\mc{M},{g})$. Mapping the integral curves of the gradient of $\check{\rho}$ in $\mc{N}$ to the integral curves of the gradient of $\check{\rho}$ in $\mc{M}$, one can identify the FG gauges for $(\mc{M}, \check{g})$  and $(\mc{N}, g)$.
        
        As a consequence, one obtains identical data and thus the assumptions of Proposition \ref{prop_unique_continuation1}. Applying now the Proposition \ref{prop_unique_continuation1}, one simply needs to recover the metric ${g}$ and the field ${F}$ by writing them in the gauge \eqref{FG-gauge} with respect to ${\rho}$, yielding the isometry $\phi$ as well. 
    \end{proof}

\newpage
\printbibliography

@article{AB,
  author = "S. Alinhac and M. S. Baouendi",
  title = "A non uniqueness result for operators of principal type",
  journal = "Math. Z.",
  volume = "220",
  number = "1",
  pages = "561--568",
  year = "1995"
}

@article{Mtivier1993CounterexamplesTH,
  title={Counterexamples to H{\"o}lmgren's uniqueness for analytic non linear Cauchy problems},
  author={G. M{\'e}tivier},
  journal={Inventiones mathematicae},
  year={1993},
  volume={112},
  pages={217-222}
}

@article{Holzegel12,
author = {Holzegel, Gustav},
title = {{Well-posedness for the Massive Wave Equation on Asymptotically Anti-de Sitter Spacetimes}},
journal = {Journal of Hyperbolic Differential Equations},
volume = {09},
number = {02},
pages = {239-261},
year = {2012},
doi = {10.1142/S0219891612500087},
URL = {https://doi.org/10.1142/S0219891612500087},
eprint = {https://doi.org/10.1142/S0219891612500087
}
}

@article{Holzegel22,
  author = "G. Holzegel and A. Shao",
  title = "The bulk-boundary correspondence for the {Einstein} equations in asymptotically {Anti-de Sitter} spacetimes",
  journal = "Arch. Ration. Mech. Anal.",
  volume = "247",
  pages = "56",
  year = "2023"
}

@article{Warnick13,
  author = "C. M. Warnick",
  title = "The massive wave equation in asymptotically {AdS} spacetimes",
  journal = "Commun. Math. Phys.",
  volume = "321",
  pages = "85--111",
  year = "2013",
}

@article{hol_shao:uc_ads,
  author = "G. Holzegel and A. Shao",
  title = "Unique continuation from infinity in asymptotically {Anti-de Sitter} spacetimes",
  journal = "Comm. Math. Phys.",
  volume = "347",
  number = "3",
  pages = "1--53",
  year = "2016"
}

@article{hol_shao:uc_ads_ns,
  author = "G. Holzegel and A. Shao",
  title = "Unique continuation from infinity in asymptotically {Anti-de Sitter} spacetimes {II: Non-static} boundaries",
  journal = "Comm. Partial Differential Equations",
  volume = "42",
  number = "12",
  pages = "1871--1922",
  year = "2017"
}

@article{McGill20,
  author = "A. McGill and A. Shao",
  title = "Null geodesics and improved unique continuation for waves in asymptotically {Anti-de Sitter} spacetimes",
  journal = "Class. Quantum Grav.",
  volume = "38",
  pages = "054001",
  year = "2020"
}

@article{Shao22,
  author = "A. Chatzikaleas and A. Shao",
  title = "A gauge-invariant unique continuation criterion for waves in asymptotically {Anti-de Sitter} spacetimes",
	journal = "Commun. Math. Phys.",
  volume = "395",
  pages = "1--50",
  year = "2022"
}

@article{Maldacena_1999,
  author = "J. M. Maldacena",
  title = "The large {$N$} limit of superconformal field theories and supergravity",
  journal = "Int. J. Theor. Phys.",
  volume = "38",
  pages = "1113--1133",
  year = "1999"
}

@article{witten1998anti,
  author = "E. Witten",
  title = "Anti de {Sitter} space and holography",
  journal = "Adv. Theor. Math. Phys.",
  volume = "2",
  pages = "253--291",
  year = "1998"
}

@article{shao:aads_fg,
  author = "A. Shao",
  title = "The near-boundary geometry of {Einstein}-vacuum asymptotically {Anti}-de {Sitter} spacetimes",
  journal = "Class. Quantum Grav.",
  volume = "38",
  pages = "034001",
  year = "2020"
}

@incollection{fef_gra:conf_inv,
  author = "C. Fefferman and C. R. Graham",
  title = "Conformal invariants",
  booktitle = "{\'Elie Cartan} et les math\'ematiques d'aujourd'hui - Lyon, 25-29 juin 1984",
  series = "Ast\'erisque",
  publisher = "Soci\'et\'e math\'ematique de France",
  year = "1985",
  pages = "95-116",
}

@article{biq:uc_einstein,
  author = "O. Biquard",
  title = "Continuation unique \`a partir de l'infini conforme pour les m\'etriques {d'Einstein}",
  journal = "Math. Res. Lett.",
  volume = "15",
  number = "6",
  pages = "1091--1099",
  year = "2008"
}

@article{chru_delay:uc_killing,
  author = "P. Chru\'sciel and E. Delay",
  title = "Unique continuation and extensions of {Killing} vectors for stationary vacuum space-times",
  journal = "J. Geom. Phys.",
  volume = "61",
  number = "8",
  pages = "1249--1257",
  year = "2011"
}

@article{Ionescu:2011wx,
    author = "Ionescu, Alexandru D. and Klainerman, Sergiu",
    title = "{On the local extension of Killing vector-fields in Ricci flat manifolds}",
    eprint = "1108.3575",
    archivePrefix = "arXiv",
    primaryClass = "math.AP",
    month = "8",
    year = "2011"
}

@article{toua:go_eve,
 	year = "2023",
	volume = "402",
	number = {3},
	pages = "3109--3200",
	author = "Arthur Touati",
	title = "Geometric optics approximation for the {Einstein} vacuum equations",
	journal = "Commun. Math. Phys."
}

@article{deHaro:2000vlm,
    author = "de Haro, Sebastian and Solodukhin, Sergey N. and Skenderis, Kostas",
    title = "{Holographic reconstruction of space-time and renormalization in the AdS / CFT correspondence}",
    eprint = "hep-th/0002230",
    archivePrefix = "arXiv",
    reportNumber = "SPIN-2000-05, ITP-UU-00-03, PUTP-1921",
    doi = "10.1007/s002200100381",
    journal = "Commun. Math. Phys.",
    volume = "217",
    pages = "595--622",
    year = "2001"
}

@article{Imbimbo_2000,
   title={Diffeomorphisms and holographic anomalies},
   volume={17},
   ISSN={1361-6382},
   url={http://dx.doi.org/10.1088/0264-9381/17/5/322},
   DOI={10.1088/0264-9381/17/5/322},
   number={5},
   journal={Classical and Quantum Gravity},
   publisher={IOP Publishing},
   author={Imbimbo, C and Schwimmer, A and Theisen, S and Yankielowicz, S},
   year={2000},
   month=feb, pages={1129–1138} }

@article{Holzegel:2015swa,
    author = "Holzegel, Gustav and Luk, Jonathan and Smulevici, Jacques and Warnick, Claude",
    title = "{Asymptotic properties of linear field equations in anti-de Sitter space}",
    eprint = "1502.04965",
    archivePrefix = "arXiv",
    primaryClass = "gr-qc",
    doi = "10.1007/s00220-019-03601-6",
    journal = "Commun. Math. Phys.",
    volume = "374",
    number = "2",
    pages = "1125--1178",
    year = "2019"
}

@inbook{Sachdev_2011,
   title={Condensed Matter and AdS/CFT},
   ISBN={9783642048647},
   ISSN={1616-6361},
   url={http://dx.doi.org/10.1007/978-3-642-04864-7_9},
   DOI={10.1007/978-3-642-04864-7_9},
   booktitle={Lecture Notes in Physics},
   publisher={Springer Berlin Heidelberg},
   author={Sachdev, Subir},
   year={2011},
   pages={273–311} }

@article{Policastro_2002,
   title={From AdS/CFT correspondence to hydrodynamics},
   volume={2002},
   ISSN={1029-8479},
   url={http://dx.doi.org/10.1088/1126-6708/2002/09/043},
   DOI={10.1088/1126-6708/2002/09/043},
   number={09},
   journal={Journal of High Energy Physics},
   publisher={Springer Science and Business Media LLC},
   author={Policastro, Giuseppe and Son, Dam T and Starinets, Andrei O},
   year={2002},
   month=sep, pages={043–043} }

@article{KICHENASSAMY2004268,
title = {On a conjecture of Fefferman and Graham},
journal = {Advances in Mathematics},
volume = {184},
number = {2},
pages = {268-288},
year = {2004},
issn = {0001-8708},
doi = {https://doi.org/10.1016/S0001-8708(03)00145-2},
url = {https://www.sciencedirect.com/science/article/pii/S0001870803001452},
author = {Satyanad Kichenassamy}
}

@article{Hubeny_2015,
   title={The AdS/CFT correspondence},
   volume={32},
   ISSN={1361-6382},
   url={http://dx.doi.org/10.1088/0264-9381/32/12/124010},
   DOI={10.1088/0264-9381/32/12/124010},
   number={12},
   journal={Classical and Quantum Gravity},
   publisher={IOP Publishing},
   author={Hubeny, Veronika E},
   year={2015},
   month=jun, pages={124010} }

@inproceedings{Polchinski:2010hw,
    author = "Polchinski, Joseph",
    title = "{Introduction to Gauge/Gravity Duality}",
    booktitle = "{Theoretical Advanced Study Institute in Elementary Particle Physics}: {String theory and its Applications: From meV to the Planck Scale}",
    eprint = "1010.6134",
    archivePrefix = "arXiv",
    primaryClass = "hep-th",
    doi = "10.1142/9789814350525_0001",
    pages = "3--46",
    month = "10",
    year = "2010"
}

@article{Brown:1986nw,
    author = "Brown, J. David and Henneaux, M.",
    title = "{Central Charges in the Canonical Realization of Asymptotic Symmetries: An Example from Three-Dimensional Gravity}",
    doi = "10.1007/BF01211590",
    journal = "Commun. Math. Phys.",
    volume = "104",
    pages = "207--226",
    year = "1986"
}

@misc{moschidis2018,
      title={A proof of the instability of AdS for the Einstein--massless Vlasov system}, 
      author={Georgios Moschidis},
      year={2018},
      eprint={1812.04268},
      archivePrefix={arXiv},
      primaryClass={math.AP},
      url={https://arxiv.org/abs/1812.04268}, 
}

@article{Moschidis_2020,
   title={A proof of the instability of AdS for the Einstein-null dust system with an inner mirror},
   volume={13},
   ISSN={2157-5045},
   url={http://dx.doi.org/10.2140/apde.2020.13.1671},
   DOI={10.2140/apde.2020.13.1671},
   number={6},
   journal={Analysis \& PDE},
   publisher={Mathematical Sciences Publishers},
   author={Moschidis, Georgios},
   year={2020},
   month=sep, pages={1671–1754} }

@article{Cardoso_2014,
   title={Holographic thermalization, quasinormal modes and superradiance in Kerr-AdS},
   volume={2014},
   ISSN={1029-8479},
   url={http://dx.doi.org/10.1007/JHEP04(2014)183},
   DOI={10.1007/jhep04(2014)183},
   number={4},
   journal={Journal of High Energy Physics},
   publisher={Springer Science and Business Media LLC},
   author={Cardoso, Vitor and Dias, Óscar J. C. and Hartnett, Gavin S. and Lehner, Luis and Santos, Jorge E.},
   year={2014},
   month=apr }

@misc{ionescu2015rigidityresultsgeneralrelativity,
      title={Rigidity Results in General Relativity: a Review}, 
      author={Alexandru Ionescu and Sergiu Klainerman},
      year={2015},
      eprint={1501.01587},
      archivePrefix={arXiv},
      primaryClass={gr-qc},
      url={https://arxiv.org/abs/1501.01587}, 
}

@misc{fefferman2008ambientmetric,
      title={The ambient metric}, 
      author={Charles Fefferman and C. Robin Graham},
      year={2008},
      eprint={0710.0919},
      archivePrefix={arXiv},
      primaryClass={math.DG},
      url={https://arxiv.org/abs/0710.0919}, 
}

@article{SKENDERIS_2001,
   title={Asymptotically Anti-de Sitter spacetimes and their stress energy tensor},
   volume={16},
   ISSN={1793-656X},
   url={http://dx.doi.org/10.1142/S0217751X0100386X},
   DOI={10.1142/s0217751x0100386x},
   number={05},
   journal={International Journal of Modern Physics A},
   publisher={World Scientific Pub Co Pte Lt},
   author={Skenderis, Kostas},
   year={2001},
   month=feb, pages={740–749} }

@article{Mazzeo,
 ISSN = {00029327, 10806377},
 URL = {http://www.jstor.org/stable/2374820},
 author = {Rafe Mazzeo},
 journal = {American Journal of Mathematics},
 number = {1},
 pages = {25--45},
 publisher = {Johns Hopkins University Press},
 title = {Unique Continuation at Infinity and Embedded Eigenvalues for Asymptotically Hyperbolic Manifolds},
 urldate = {2024-08-18},
 volume = {113},
 year = {1991}
}

@article{GUISSET2024223,
title = {On counterexamples to unique continuation for critically singular wave equations},
journal = {Journal of Differential Equations},
volume = {395},
pages = {223-261},
year = {2024},
issn = {0022-0396},
doi = {https://doi.org/10.1016/j.jde.2024.02.031},
url = {https://www.sciencedirect.com/science/article/pii/S0022039624001025},
author = {Simon Guisset and Arick Shao}
}

@phdthesis{thesis_Guisset,
  title        = {Asymptotic Properties of Anti-de Sitter Spacetimes},
  author       = {Simon Guisset},
  year         = 2024,
  month        = {September},
  school       = {Queen Mary University of London},
  type         = {PhD thesis}
}

@misc{ganchev2016superradiantinstabilityads,
      title={Superradiant instability in AdS}, 
      author={Bogdan Ganchev},
      year={2016},
      eprint={1608.01798},
      archivePrefix={arXiv},
      primaryClass={hep-th},
      url={https://arxiv.org/abs/1608.01798}, 
}

@misc{graf2024linearstabilityschwarzschildantidesitterI,
      title={Linear Stability of Schwarzschild-Anti-de Sitter spacetimes I: The system of gravitational perturbations}, 
      author={Olivier Graf and Gustav Holzegel},
      year={2024},
      eprint={2408.02251},
      archivePrefix={arXiv},
      primaryClass={gr-qc},
      url={https://arxiv.org/abs/2408.02251}, 
}

@misc{graf2024linearstabilityschwarzschildantidesitterIII,
      title={Linear Stability of Schwarzschild-Anti-de Sitter spacetimes III: Quasimodes and sharp decay of gravitational perturbations}, 
      author={Olivier Graf and Gustav Holzegel},
      year={2024},
      eprint={2410.21994},
      archivePrefix={arXiv},
      primaryClass={gr-qc},
      url={https://arxiv.org/abs/2410.21994}, 
}

@misc{graf2024linearstabilityschwarzschildantidesitterII,
      title={Linear Stability of Schwarzschild-Anti-de Sitter spacetimes II: Logarithmic decay of solutions to the Teukolsky system}, 
      author={Olivier Graf and Gustav Holzegel},
      year={2024},
      eprint={2408.02252},
      archivePrefix={arXiv},
      primaryClass={gr-qc},
      url={https://arxiv.org/abs/2408.02252}, 
}

@misc{holzegel2013decaypropertieskleingordonfields,
      title={Decay properties of Klein-Gordon fields on Kerr-AdS spacetimes}, 
      author={Gustav Holzegel and Jacques Smulevici},
      year={2013},
      eprint={1110.6794},
      archivePrefix={arXiv},
      primaryClass={gr-qc},
      url={https://arxiv.org/abs/1110.6794}, 
}

@misc{mcgill2021,
      title={Holographic Characterisation of Locally Anti-de Sitter Spacetimes}, 
      author={Alex McGill},
      year={2021},
      eprint={2111.11155},
      archivePrefix={arXiv},
      primaryClass={gr-qc},
      url={https://arxiv.org/abs/2111.11155}, 
}

@misc{huneau2024,
      title={Burnett's conjecture in generalized wave coordinates}, 
      author={Cécile Huneau and Jonathan Luk},
      year={2024},
      eprint={2403.03470},
      archivePrefix={arXiv},
      primaryClass={gr-qc},
      url={https://arxiv.org/abs/2403.03470}, 
}

@misc{huneau2024review,
      title={High-frequency solutions to the Einstein equations}, 
      author={Cécile Huneau and Jonathan Luk},
      year={2024},
      eprint={2404.07659},
      archivePrefix={arXiv},
      primaryClass={gr-qc},
      url={https://arxiv.org/abs/2404.07659}, 
}

@book{Hawking:1973uf,
    author = "Hawking, Stephen W. and Ellis, George F. R.",
    title = "{The Large Scale Structure of Space-Time}",
    doi = "10.1017/9781009253161",
    isbn = "978-1-009-25316-1, 978-1-009-25315-4, 978-0-521-20016-5, 978-0-521-09906-6, 978-0-511-82630-6, 978-0-521-09906-6",
    publisher = "Cambridge University Press",
    series = "Cambridge Monographs on Mathematical Physics",
    month = "2",
    year = "2023"
}

@article{Carranza:2018wkp,
    author = "Carranza, Diego A. and Valiente Kroon, Juan A.",
    title = "{Construction of anti-de Sitter-like spacetimes using the metric conformal Einstein field equations: the vacuum case}",
    eprint = "1807.04212",
    archivePrefix = "arXiv",
    primaryClass = "gr-qc",
    doi = "10.1088/1361-6382/aaeb54",
    journal = "Class. Quant. Grav.",
    volume = "35",
    number = "24",
    pages = "245006",
    year = "2018"
}
\clearpage

\appendix

\section{Proofs of Section \ref{sec:intro}}

\subsection{Proof of Proposition \ref{sec:aads_prop_covariant_D}}\label{app:D_bar}

We will here simply prove the metric compatibility. The other properties hold trivially from the expression \eqref{vertical_derivative}. 

    One has immediately, from the definition of $\ms{D}$:
    \begin{equation*}
        \ol{\ms{D}}_a \ms{g} = 0\text{,}\qquad \ol{\ms{D}}_a\ms{g}^{-1}=0\text{.}
    \end{equation*}
    On the other hand, the $\rho-$derivative gives: 
    \begin{align*}
        \ol{\ms{D}}_\rho \ms{g}_{ab} &= \Lie_\rho \ms{g}_{ab} - \frac{1}{2}\ms{g}^{de}\Lie_\rho \ms{g}_{ad}\ms{g}_{eb}- \frac{1}{2}\ms{g}^{de}\Lie_\rho \ms{g}_{bd}\ms{g}_{ea}\\
        &=0\text{.}
    \end{align*}
    Similarly, since $\Lie_\rho \paren{\ms{g}^{-1}}^{ab}=-\paren{\ms{g}^{-1}}^{ac}\paren{\ms{g}^{-1}}^{bd}\Lie_\rho \ms{g}_{cd}$:
    \begin{align*}
        \ol{\ms{D}}_\rho (\ms{g}^{-1})^{ab} &= \Lie_\rho (\ms{g}^{-1})^{ab} + \frac{1}{2}(\ms{g}^{-1})^{ad}\Lie_\rho \ms{g}_{de}(\ms{g}^{-1})^{eb} + \frac{1}{2}(\ms{g}^{-1})^{bd}\Lie_\rho \ms{g}_{de}(\ms{g}^{-1})^{ea} \\
        &=0\text{.}
    \end{align*}

\subsection{Proof of Proposition \ref{prop_christoffel_FG}}\label{app:christoffel}

The expressions of the Christoffels symbols $\Gamma$ follow directly from the Fefferman-Graham gauge \eqref{FG-gauge}, while $\ms{\Gamma}$ follow from the definition of $\overline{\ms{D}}$, as well as its compatibility with $\ms{g}$. 

Let $\lbrace \partial_\alpha\rbrace_{\alpha = 0, \dots , n}$ denote a basis of $T\mc{M}$ and $\lbrace \partial_a\rbrace_{\alpha = 0, \dots n-1}$ denote a basis on the vertical bundle $\ms{V}\mc{M}$. Similarly, we will denote by $dx^\alpha$ and $dx^a$ their associated co-basis. 

The expression for $\overline{\nabla}\mathbf{A}$ can be computed easily using standard methods. The definition of $\ol{\nabla}$ immediately implies: 
\begin{align*}
    &\overline{\nabla}_\gamma \paren{\mathbf{A}^{\bar{\alpha}}{}_{\bar{\beta}}{}^{\bar{a}}{}_{\bar{b}}\partial_{\bar{\alpha}}\otimes dx^{\bar{\beta}}\otimes \partial_{\bar{a}}\otimes dx^{\bar{b}}}\\
    =&\,\partial_\gamma \mathbf{\mathbf{A}}^{\bar{\alpha}}{}_{\bar{\beta}}{}^{\bar{a}}{}_{\bar{b}}\partial_{\bar{\alpha}}\otimes dx^{\bar{\beta}}\otimes \partial_{\bar{a}}\otimes dx^{\bar{b}} \\
    &+ \sum\limits_{i=1}^\kappa \mathbf{A}^{\bar{\alpha}}{}_{\bar{\beta}}{}^{\bar{a}}{}_{\bar{b}}\Gamma^{\mu}_{\gamma\alpha_i}\partial_{\bar{\alpha}_i[\mu]}\otimes dx^{\bar{\beta}}\otimes \partial_{\bar{a}}\otimes dx^{\bar{b}} -\sum\limits_{j=1}^{\lambda} \mathbf{A}^{\bar{\alpha}}{}_{\bar{\beta}}{}^{\bar{a}}{}_{\bar{b}}\Gamma_{\gamma \mu}^{\beta_j}\partial_{\bar{\alpha}}\otimes dx^{\bar{\beta}_j[\mu]}\otimes \partial_{\bar{a}}\otimes dx^{\bar{b}} \\
    &+\sum\limits_{i=1}^k \mathbf{A}^{\bar{\alpha}}{}_{\bar{\beta}}{}^{\bar{a}}{}_{\bar{b}}\ms{\Gamma}^d_{\gamma a_i}\partial_{\bar{\alpha}}\otimes dx^{\bar{\beta}}\otimes \partial_{\bar{a}_i[d]}\otimes dx^{\bar{b}} - \sum\limits_{j=1}^l \mathbf{A}^{\bar{\alpha}}{}_{\bar{\beta}}{}^{\bar{a}}{}_{\bar{b}}\ms{\Gamma}^d_{\gamma b_j}\partial_{\bar{\alpha}}\otimes dx^{\bar{\beta}}\otimes \partial_{\bar{a}}\otimes dx^{\bar{b}_j[d]}\\
    =&\left(\partial_\gamma \mathbf{\mathbf{A}}^{\bar{\alpha}}{}_{\bar{\beta}}{}^{\bar{a}}{}_{\bar{b}}
    +\sum\limits_{i=1}^\kappa \mathbf{A}^{\bar{\alpha}_i[\mu]}{}_{\bar{\beta}}{}^{\bar{a}}{}_{\bar{b}}\Gamma^{\alpha_i}_{\gamma\mu}-\sum\limits_{j=1}^{\lambda} \mathbf{A}^{\bar{\alpha}}{}_{\bar{\beta}_j[\mu]}{}^{\bar{a}}{}_{\bar{b}}\Gamma_{\gamma b_j}^{\mu}+\sum\limits_{i=1}^k \mathbf{A}^{\bar{\alpha}}{}_{\bar{\beta}}{}^{\bar{a}_i[d]}{}_{\bar{b}}\ms{\Gamma}^{a_i}_{\gamma d}\right.\\
    &\left.\hspace{100pt} - \sum\limits_{j=1}^l \mathbf{A}^{\bar{\alpha}}{}_{\bar{\beta}}{}^{\bar{a}}{}_{\bar{b}_i[d]}\ms{\Gamma}^{b_i}_{\gamma b_j}\right)\partial_{\bar{\alpha}}\otimes dx^{\bar{\beta}}\otimes \partial_{\bar{a}}\otimes dx^{\bar{b}}\text{,}
\end{align*}
proving the proposition. 

\subsection{Proof of Proposition \ref{sec:aads_commutation_Lie_D}}\label{app:aads_commutation_Lie_D}

The proof of \eqref{commutation_L_D} is a simple consequence of the definition of $\ms{D}$. Typically, given a vertical tensor $\ms{A}$ of rank $(p, q)$: 
    \begin{align*}
        [\Lie_\rho, \ms{D}_a] \ms{A}^{\bar{b}}{}_{\bar{c}} =  \sum\limits_{i=1}^p\Lie_\rho\ms{\Gamma}^{b_i}_{ad}\ms{A}^{\hat{b}_i[d]}{}_{\bar{c}} - \sum\limits_{i=1}^q\Lie_\rho\ms{\Gamma}^{d}_{ac_i}\ms{A}^{\bar{b}}{}_{\hat{c}_i[d]}\text{,}
    \end{align*}
    and since: 
    \begin{equation*}
        \Lie_\rho \ms{\Gamma} = \mc{S}\paren{\ms{g}; \ms{Dm}}\text{,}
    \end{equation*}
    the identity \eqref{commutation_L_D} follow. Equation \eqref{commutation_D_D} follows from Proposition \ref{sec:aads_prop_covariant_D}: 
    \begin{equation*}
        \ol{\ms{D}}_\rho \ms{A} = \Lie_\rho \ms{A} + \mc{S}\paren{\ms{g};\ms{m}, \ms{A}}\text{.}
    \end{equation*}
    The identity \eqref{wave_power} follows from Definition \ref{def_mixed_connection}: 
    \begin{align*}
        \ol{\Box}(\rho^p\ms{A}) = \rho^p\ol{\Box}\ms{A} + 2p\rho^{p+1}\ol{\ms{D}}_\rho\ms{A} + p(p-1)\rho^p \ms{A} + p\rho^{p-1}\Box\rho \cdot\ms{A}\text{,}
    \end{align*}
    where: 
    \begin{align*}
        \Box\rho &= -\rho^2 \ms{g}^{ab}\Gamma_{ab}^\rho - \rho^2\Gamma^{\rho}_{\rho\rho}\\
        &=-(n-1)\rho + \frac{1}{2}\rho^2\ms{tr}_{\ms{g}}\ms{m}\text{,}
    \end{align*}
    giving the desired expression.

    The expression of the covariant wave operator \eqref{wave_op} is obtained by expanding \eqref{vertical_derivative}: 
    \begin{align*}
        \ol{\Box} \ms{A} = g^{\alpha\beta}\ol{\nabla}_{\alpha\beta} \ms{A} = & \rho^2 \ol{\nabla}^2_{\rho\rho} \ms{A} + \rho^2 \ms{g}^{ab}\ol{\nabla}_{ab}\ms{A}\text{.}
    \end{align*}
    Assuming $\ms{A}$ to be a vertical tensor of order $(k,l)$, one gets: \begin{align*}
        \ol{\nabla}_{\rho\rho}^2\ms{A}^{\bar{b}}{}_{\bar{c}} = \Lie_\rho\ol{\ms{D}}_\rho\ms{A}^{\bar{b}}{}_{\bar{c}} - \Gamma^\rho_{\rho\rho}\ol{\ms{D}}_\rho\ms{A}^{\bar{b}}{}_{\bar{c}} + \sum\limits_{i=1}^{k}\ms{\Gamma}^{b_i}_{a\rho}\ms{A}^{\hat{b}_i[a]}{}_{\bar{c}} - \sum\limits_{i=1}^l \ms{\Gamma}^a_{c_i\rho}\ms{A}^{\bar{b}}{}_{\hat{c}_i[a]}\text{,}
    \end{align*}
    where: 
    \begin{align*}
        &\Lie_\rho\overline{\ms{D}}_\rho\ms{A} = \Lie_\rho\paren{\Lie_\rho \ms{A} + \mc{S}\paren{\ms{g}; \ms{m}, \ms{A}}} \\
        &\hspace{34pt}=\Lie_\rho^2 \ms{A} + \mc{S}\paren{\ms{g}; \ms{m}^2, \ms{A}} + \mc{S}\paren{\ms{g}; \Lie_\rho\ms{m}, \ms{A}} + \mc{S}\paren{\ms{g}; \ms{m}, \Lie_\rho\ms{A}}\text{,}\\
        &-\Gamma^{\rho}_{\rho\rho}\ol{\ms{D}}_\rho\ms{A} =  \rho^{-1}\Lie_\rho\ms{A} + \rho^{-1}\mc{S}\paren{\ms{g}; \ms{m},\ms{A}}\text{,}\\
        &\sum\limits_{i=1}^{k}\ms{\Gamma}^{b_i}_{a\rho}\ms{A}^{\hat{b}_i[a]}{}_{\bar{c}} - \sum\limits_{i=1}^l \ms{\Gamma}^a_{c_i\rho}\ms{A}^{\bar{b}}{}_{\hat{c}_i[a]} = \mc{S}\paren{\ms{g}; \ms{m}, \ms{A}}\text{.}
    \end{align*}
    Finally, the vertical part of the d'Alembertian yields:
    \begin{align*}
        \ms{g}^{ad}\ol{\nabla}_{ad}\ms{A}^{\bar{b}}{}_{\bar{c}} &= {\Box}_{\ms{g}}\ms{A}^{\bar{b}}{}_{\bar{c}} - \ms{g}^{ad}\Gamma_{ad}^\rho \ol{\ms{D}}_\rho\ms{A}^{\bar{b}}{}_{\bar{c}}\\
        &={\Box}_{\ms{g}}\ms{A}^{\bar{b}}{}_{\bar{c}} - n\rho^{-1}\Lie_\rho \ms{A}^{\bar{b}}{}_{\bar{c}} + \rho^{-1}\mc{S}\paren{\ms{g}; \ms{m}, \ms{A}} + \mc{S}\paren{\ms{g}; \ms{m}, \Lie_\rho\ms{A}} + \mc{S}\paren{\ms{g}; \ms{m}^2, \ms{A}}\text{,}
    \end{align*}
    giving the desired equation \eqref{wave_op}. 

\section{Proofs of Section \ref{sec:FG_expansion}}

\subsection{Proof of Corollary \ref{prop_expansion_vertical_fields}}\label{app:fg_expansion}

 First, let us define the following fields: 
\begin{gather*}
    \hat{\ms{g}} := \ms{g} - \rho^n \log\rho \cdot \mf{g}^\star\text{,}\\
    \hat{\ms{f}}^0 := \begin{cases}
         \rho^{-2}\ms{f}^0- \rho^{n-4} \log \rho \cdot \mf{f}^{0, \star}\text{,}\qquad &n \geq 4\text{,}\\
         \rho^{-1}\ms{f}^0\text{,}\qquad &n=3\text{,}\\
         \ms{f}^0\text{,}\qquad &n=2\text{,}
    \end{cases}\qquad \hat{\ms{f}}^1 :=
        \rho^{-1}\ms{f}^1 - \rho^{n-2}\log\rho\cdot \mf{f}^{1, \star}\text{.}
\end{gather*}

Observe that 
\begin{equation*}
    \begin{aligned}
        &\Lie_\rho^k \hat{\ms{g}}\Rightarrow^{M_0-k} k!\,\mf{g}^{(k)}\text{,}\qquad &&k<n\\
        &\Lie_\rho^{l}\hat{\ms{f}}^0 \Rightarrow^{M_0- l-2}l! \, \mf{f}^{0, (l)}\text{,}\qquad &&l<n-4\text{,}\\
        &\Lie_\rho^{q} \hat{\ms{f}}^1 \Rightarrow^{M_0-q-1}q! \, \mf{f}^{1,          (q)}\text{,}\qquad &&q<n-2\text{.}
    \end{aligned}
\end{equation*}

Now, looking at $\hat{\ms{g}}$, one can easily compute: 
\begin{align*}
    \Lie_\rho^n \hat{\ms{g}}= \underbrace{\Lie_\rho^n \ms{g} - n! \log \rho\cdot \mf{g}^\star}_{\longrightarrow^{M_0-n} n! \mf{g}^\dag}  - n! C_n \mf{g}^\star\text{,}
\end{align*}
with $C_n = 1 + 2^{-1} + \dots + n^{-1}$. Defining now: 
\begin{gather*}
    \mf{g}^{(n)} := \mf{g}^\dag - C_n \mf{g}^\star\text{,}\\
    \mf{f}^{0,((n-4)_+)}:= \begin{cases}
        \mf{f}^{0,\dag} - C_{n-4} \mf{f}^{0,\star}\text{,}\qquad &n>4\text{,}\\
        \mf{f}^{0,\dag}\text{,}\qquad &n=4\text{,}\\
    \end{cases}\qquad \mf{f}^{1,(n-2)} := \begin{cases}
        \mf{f}^{1, \dag} - C_{n-2}\mf{f}^{1,\star}\text{,}\qquad &n>2\text{,}\\
        \mf{f}^{1, \dag}\text{,}\qquad &n=2\text{,}
    \end{cases}
\end{gather*}
one has:
\begin{gather*}
    \Lie_\rho^n \hat{\ms{g}} \rightarrow^{M_0-n} n!\, \mf{g}^{(n)}\text{,}\\
    \Lie_\rho^{(n-4)_+}\hat{\ms{f}}^0 \begin{cases}
        \rightarrow^{M_0-n+2} ((n-4)_+)!\mf{f}^{0,((n-4)_+)}\text{,}\qquad &n\geq 3\text{,}\\
        \rightarrow^{M_0-1} \mf{f}^{0,(0)}\text{,}\qquad &n=2\text{,}
    \end{cases}\\
    \Lie_\rho^{n-2}\hat{\ms{f}}^1 
    \begin{cases}
       \rightarrow^{M_0-n+1} (n-2)!\mf{f}^{1,(n-2)}\text{,}\qquad &n\geq 3\text{,}\\
       \rightarrow^{M_0-2} \mf{f}^{1,(0)}\text{,}\qquad &n=2\text{.}
    \end{cases} 
\end{gather*}

The expansions therefore follows from Taylor's theorem applied to $\hat{\ms{g}}$, $\hat{\ms{f}}^0$ and $\hat{\ms{f}}^1$: 
\begin{align*}
    &\hat{\ms{g}} = \sum\limits_{k=0}^n \rho^k \mf{g}^{(k)} + \rho^n \ms{r}_{\hat{\ms{g}}}= \ms{g} - \rho^n \log \rho \cdot \mf{g}^\star\text{,}\\
    &\hat{\ms{f}}^0=\begin{cases} 
    \sum\limits_{k=0}^{n-4} \rho^k \mf{f}^{0, (k)} + \rho^{n-4} \ms{r}_{\hat{\ms{f}}^0} = 
        \rho^{-2}\ms{f}^0 - \rho^{n-4}\log\rho \cdot \mf{f}^{0, \star} \text{,}\qquad &n\geq 4\text{,}\\
        \mf{f}^{0,(0)} + \ms{r}_{\hat{\ms{f}}^0} = \rho^{-1}\ms{f}^0\text{,}\qquad &n = 3\text{,}\\
        \mf{f}^{0,(0)} + \ms{r}_{\hat{\ms{f}}^0} = \ms{f}^0\text{,}\qquad &n=2\text{,}
    \end{cases}\\
    &\hat{\ms{f}}^1= \sum\limits_{k=0}^{n-2} \rho^k \mf{f}^{1, (k)} + \rho^{n-2} \ms{r}_{\hat{\ms{f}}^1} =
        \rho^{-1}\ms{f}^1-\rho^{n-2}\log\rho\cdot\mf{f}^{1,\star}\text{,}
\end{align*}
with the remainder terms satisfying: 
\begin{gather*}
    \ms{r}_{\hat{\ms{g}}}\rightarrow^{M_0-n} 0\text{,}\\
    \ms{r}_{\hat{\ms{f}}^0}\begin{cases}
        \rightarrow^{M_0-(n-2)} 0 \text{,}\qquad &n \geq 4\text{,}\\
        \rightarrow^{M_0-1} 0 \text{,}\qquad &n =2,3\text{,}
    \end{cases}\qquad \ms{r}_{\hat{\ms{f}}^1}\begin{cases}
        \rightarrow^{M_0-(n- 1)} 0 \text{,}\qquad &n\geq 3\text{,}\\
        \rightarrow^{M_0-2} 0\text{,}\qquad &n =2\text{.}
    \end{cases}
\end{gather*}

\subsection{Proof of Corollary \ref{weyl_expansion}}\label{app_weyl}

\begin{proof}
    The expansion of the Weyl curvature simply follows from Corollary \ref{prop_expansion_vertical_fields} and Proposition \ref{prop_w_FG}. 

    Let $n\geq 3$. First, we will need an expansion of the stress-energy tensor, which is a direct consequence of the expansions \eqref{expansion_f0} and \eqref{expansion_f1}. The vertical decomposition of the stress-energy tensor takes the following form: 
    \begin{gather*}
            T_{ab} = \mc{S}\paren{\ms{g}; (\ms{f}^1)^2} + \rho^2 \mc{S}\paren{\ms{g}; (\ms{f}^0)^2}\text{,}\qquad T_{\rho a} = \mc{S}\paren{\ms{g}; \ms{f}^0, \ms{f}^1}\text{,}\\
            T_{\rho\rho} = \mc{S}\paren{\ms{g}; (\ms{f}^0)^2} + \mc{S}\paren{\ms{g}; (\ms{f}^1)^2}\text{.}
        \end{gather*}
        Now, the expansion \eqref{expansion_f0} gives: 
        \begin{align*}
            \mc{S}\paren{\ms{g}; (\ms{f}^1)^2} =\begin{cases}
                \rho^2\mi{D}^{(0,0)}(\mf{g}^{(0)}, \mf{f}^{1,(0)}) + \sum\limits_{k=2}^{n-4}\rho^{k+2} \mi{D}^{(k, k)}\paren{\mf{g}^{(0)}, \mf{f}^{1, (0)}; k} + \rho^{n-2}\ms{r}_{11}\text{,}\qquad &n\geq 4\text{,}\\
                \rho^2 \mi{D}^{(0,0)}\paren{\mf{g}^{(0)}, \mf{f}^{1,(0)}} + \rho^2\ms{r}_{11}\text{,}\qquad &n=3\text{,}
            \end{cases}
        \end{align*}
        where $\ms{r}_{11}$ is a $C^{M_0-(n-3)}$ vertical tensor field satisfying: 
        \begin{equation*}
            \ms{r}_{11}\rightarrow^{M_0-(n-3)} 0\text{.}
        \end{equation*}
        We will also need the following partial expansions: 
        \begin{align*}
            &\mc{S}\paren{\ms{g}; (\ms{f}^0)^2} = \begin{cases}\rho^4\sum\limits_{k=0}^{n-5} \rho^k \mi{D}^{(k+1, k+1)} \paren{\mf{g}^{(0)}, \mf{f}^{1, (0)}; k}+ \rho^{n-1} \ms{r}_{00}\text{,}\qquad &n\geq 5\text{ odd,}\\
            \rho^4\sum\limits_{k=0}^{n-6}\rho^k \mi{D}^{(k+1, k+1)} \paren{\mf{g}^{(0)}, \mf{f}^{1, (0)}; k} +  \rho^{n-2} \ms{r}_{00}\text{,}\qquad &n \geq 6\text{ even,}\\
            \rho^4\mi{D}^{(1,1)}\paren{\mf{g}^{(0)}, \mf{f}^{1,(0)}} + \rho^4\ms{r}_{00}\text{,}\qquad &n=4\text{,}\\
            \rho^2 \mi{D}^{(0,0)}\paren{\mf{g}^{(0)}, \mf{f}^{0,(0)}} + \rho^2 \ms{r}_{00}\text{,}\qquad &n=3\text{,}
            \end{cases}\\
            &\mc{S}\paren{\ms{g}; \ms{f}^1, \ms{f}^0}=
            \begin{cases}
            \rho^3\sum\limits_{k=0}^{n-4} \rho^k \mi{D}^{(k+1, k+1)} \paren{\mf{g}^{(0)}, \mf{f}^{1,(0)};k} + \rho^{n-1} \ms{r}_{10}\text{,} \qquad &n \geq 5\text{ odd,}\\
            \rho^3\sum\limits_{k=0}^{n-4} \rho^k \mi{D}^{(k+1, k+1)} \paren{\mf{g}^{(0)}, \mf{f}^{1,(0)};k} + \rho^{n-1}\log\rho \cdot \mi{D}^{(n-3, n-3)}\paren{\mf{g}^{(0)}, \mf{f}^{1,(0)}}\\
            \hspace{30pt}+ \rho^{n-1} \ms{r}_{10}\text{,}\qquad &n\geq 4\text{ even,}\\
            \rho^2 \mi{D}^{(0,0,0)}\paren{\mf{g}^{(0)}, \mf{f}^{1,(0)}, \mf{f}^{0,(0)}} + \rho^2  \ms{r}_{10}\text{,}\qquad &n=3\text{.}
            \end{cases}
        \end{align*}
        where $\ms{r}_{00}, \ms{r}_{10}$ and $\ms{r}_{00}$ are remainder terms satisfying: 
        \begin{gather*}
            (\ms{r}_{00}, \ms{r}_{10}) \rightarrow^{M_0-(n-1)} 0\text{.}
        \end{gather*}
    
    Typically, starting with $W_{abcd}$, whose expression can be found in Proposition \ref{prop_w_FG}:
    \begin{align*}
        \rho^2 W_{abcd} = \ms{R}_{abcd} +\mc{S}\paren{\ms{g}; (\ms{m})^2}_{abcd} + \rho^{-1}\mc{S}\paren{\ms{g}; \ms{m}}_{abcd}+ \mc{S}(\ms{g}; \tilde{\ms{T}}^0)_{abcd}\text{,}
    \end{align*}
    one simply needs to expand each term in powers of $\rho$. The near-boundary expansion of the Riemann tensor can be recovered using the fact that it must follow the same expansion as $\ms{g}$ by losing two orders of regularity.
    The one for $\tilde{\ms{T}}_{ab}^0 = T_{ab} - n^{-1}\paren{{T_{\rho\rho} + \ms{g}^{cd}T_{cd}}}\ms{g}_{ab}$ can be easily obtained from the above considerations. One will observe, however, that these expansions will be be lower-order, in comparison with the following ones, which follow from the metric expansion \eqref{expansion_g}. 
    
The expansion for terms quadratic in $\ms{m}$ can be deduced from \eqref{expansion_g}: 
    \begin{align*}
        \ms{m} =\begin{cases}
             \rho \sum\limits_{k=2}^{n-1}\rho^{k-2}\mi{D}^{(k,k-4)}\paren{\mf{g}^{(0)}, \mf{f}^{1, (0)}; k} + \rho^{n-1}\mf{a}_1 + \rho^{n-1}\ms{r}_{\ms{m}}\text{,} \qquad &n \text{ odd,}\\
             \rho \sum\limits_{k=2}^{n-2}\rho^{k-2}\mi{D}^{(k,k-4)}\paren{\mf{g}^{(0)}, \mf{f}^{1, (0)}; k} + \rho^{n-1}\log\rho\cdot \mi{D}^{(n, n-4,)}\paren{\mf{g}^{(0)}, \mf{f}^{1, (0)}} \\
             \hspace{10pt}+ \rho^{n-1} \mf{a}_1 + \rho^{n-1}\ms{r}_{\ms{m}}\text{,}\qquad &n \text{ even,}
        \end{cases}
    \end{align*}
    where $\mf{a}_1$ is a $C^{M_0-n}$ tensor field on $\mc{I}$ and $\ms{r}_{\ms{m}}$ is a vertical tensor field on $\mc{I}$ satisfying $\ms{r}_{\ms{m}}\rightarrow^{M_0-n} 0$. We can therefore derive the appropriate expansion: 
    \begin{align*}
        &\mc{S}\paren{\ms{g}; \ms{m}, \ms{m}} = \begin{cases}
            \rho^2 \sum\limits_{k=2}^{n-1}\rho^{k-2}\mi{D}^{(k, k-4)}\paren{\mf{g}^{(0)}, \mf{f}^{1, (0)}; k} + \rho^{n}\mf{b}_2 + \rho^n \ms{r}_{\ms{m}^2}\text{,}\qquad &n \text{ odd,}\\
            \rho^2 \sum\limits_{k=2}^{n-2}\rho^{k-2}\mi{D}^{(k,k-4)}\paren{\mf{g}^{(0)}, \mf{f}^{1, (0)}; k} + \rho^{n}\log\rho\cdot \mi{D}^{(n, n-4)}\paren{\mf{g}^{(0)}, \mf{f}^{1, (0)}} \\
            \hspace{10pt}+ \rho^{n} \mf{b}_2 + \rho^{n}\ms{r}_{\ms{m}^2}\text{,}\qquad &n \text{ even.} 
        \end{cases}\\
        &\rho^{-1}\mc{S}\paren{\ms{g}; \ms{m}}=\begin{cases}
             \sum\limits_{k=2}^{n-1}\rho^{k-2}\mi{D}^{(k,k-4)}\paren{\mf{g}^{(0)}, \mf{f}^{1, (0)}; k} + \rho^{n-2}\mf{a}_2 + \rho^{n-2}\ms{r}_{\ms{m}}\text{,} \qquad &n \text{ odd,}\\
             \sum\limits_{k=2}^{n-2}\rho^{k-2}\mi{D}^{(k,k-4)}\paren{\mf{g}^{(0)}, \mf{f}^{1, (0)}; k} + \rho^{n-2}\log\rho\cdot \mi{D}^{(n, n-4)}\paren{\mf{g}^{(0)}, \mf{f}^{1, (0)}} \\
             \hspace{10pt}+ \rho^{n-2} \mf{a}_2 + \rho^{n-2}\ms{r}_{\ms{m}}\text{,}\qquad &n \text{ even,}
        \end{cases}
    \end{align*}
    with $\mf{a}_2$, $\mf{b}_2$ tensor fields on $\mc{I}$ of regularity $C^{M_0-n}$ and $\ms{r}_{\ms{m}^2}\text{, }\ms{r}_{\ms{m}}\rightarrow^{M_0-n}0$. This leads to the expansion \eqref{expansion_w_0}. One can perform the same computations on $W_{\rho abc}$ and $W_{\rho a \rho b}$: 
    \begin{align*}
        &\rho^2 W_{\rho abc} = \ms{D}_{[c}\ms{m}_{b]a} - \frac{2}{n-1}T_{\rho [b}\ms{g}_{c]a}\text{,}\\
        &\rho^2 W_{\rho a \rho b} = -\frac{1}{2}\Lie_\rho \ms{m} + \frac{\ms{m}}{2\rho} +\frac{1}{4} \ms{g}^{cd}\ms{m}_{ac}\ms{m}_{bd} - \frac{1}{n-1}\paren{\tilde{\ms{T}}^2\ms{g}_{ab} + \tilde{\ms{T}}^0_{ab}}\text{,}
    \end{align*}
    where: 
    \begin{gather*}
        \tilde{\ms{T}}_{ab}^0 := T_{ab} - \frac{1}{n}\paren{T_{\rho\rho} + \ms{g}^{cd}T_{cd}}\ms{g}_{ab}\text{,}\qquad \tilde{\ms{T}}^2 :=\frac{n-1}{n}T_{\rho\rho} - \frac{1}{n}\ms{g}^{ab}T_{ab}\text{.}
    \end{gather*}
    The only partial expansion we will need to obtain \eqref{expansion_w_1} and \eqref{expansion_w_2} is the following: 
    \begin{align*}
        \Lie_\rho \ms{m} -\rho^{-1}\ms{m}=\begin{cases}
            \rho^2\sum\limits_{k=2}^{n-3} \rho^{k-2} \mi{D}^{(k, k-4)}\paren{\mf{g}^{(0)}, \mf{f}^{1,(0)}; k} + \rho^{n-2}\mf{c}_1 + \rho^{n-2}\ms{r}_{\ms{m}'}\text{,}\qquad n\text{ odd,}\\
            \rho^2\sum\limits_{k=2}^{n-4}\mi{D}^{(k, k-4)}\paren{\mf{g}^{(0)}, \mf{f}^{1, (0)};k} + \rho^{n-2}\log\rho \cdot \mi{D}^{(n, n-4)}\paren{\mf{g}^{(0}, \mf{f}^{1,(0)}}\\
            \hspace{10pt} + \rho^{n-2}\mf{c}_2 + \rho^{n-2}\ms{r}_{\ms{m}'}\text{,}\qquad n \text{ even,}
        \end{cases}
    \end{align*}
    with $\mf{c}_1, \mf{c}_2$ being $C^{M_0-n}$ tensor fields on $\mc{I}$. The partial expansions therefore follow. 
\end{proof}

\subsection{Proof of Proposition \ref{higher_order_m}}\label{app:higher_transport_m}

The proof is a consequence of the following simple identity: 
\begin{equation*}
    \Lie_\rho^k (\rho \ms{A}) = \rho \Lie_\rho^k \ms{A} + k\Lie_\rho^{k-1}\ms{A}\text{,}
\end{equation*}
for any vertical tensor field $\ms{A}$. Such an identity gives the left-hand sides of \eqref{higher_transport_m} and \eqref{higher_transport_tr_m}, where for the latter, one has to use: 
\begin{equation}\label{lower_order_m}
    \Lie_\rho^k \ms{tr}_{\ms{g}}\ms{m} = \ms{tr}_{\ms{g}}\Lie_\rho^k \ms{m} + \sum\limits_{\substack{j_1+\dots+j_\ell = k+1\\j_p<k+1}}\mc{S}\paren{\ms{g}; \Lie_\rho^{j_1-1}\ms{m}, \dots, \Lie_\rho^{j_\ell-1}\ms{m}}\text{.}
\end{equation}

Next, looking at the top-order term in \eqref{higher_transport_m}, given by $\ms{tr}_{\ms{m}}\cdot \ms{g}$, one has two possibilities. Either all the $\rho$--derivatives hit $\ms{m}$, in which case one has: 
\begin{equation*}
    \ms{tr}_{\ms{g}}\Lie_\rho^k \ms{m}\text{,}
\end{equation*}
or at least one derivative hits the trace operator, producing new factors of $\ms{m}$ along the way, giving similar lower--order terms to \eqref{lower_order_m}. 

For terms of the form $\rho\ms{m}\cdot \ms{m}$, either all the derivatives avoid the factor of $\rho$, in which case one has a remainder term of the form: 
\begin{equation*}
    \sum\limits_{\substack{j_1+\dots+j_\ell=k+2\\j_p\leq k+2}}\rho\cdot \mc{S}\paren{\ms{g}; \Lie_\rho^{j_1-1}\ms{m}, \dots, \Lie_\rho^{j_\ell -1}\ms{m}}\text{,}
\end{equation*}
or one derivatives hits $\rho$, giving therefore \eqref{lower_order_m}. The other terms on the right-hand side of \eqref{higher_transport_m} and \eqref{higher_transport_tr_m} can be obtained easily. 

The identities \eqref{T_higher_order} can be understood by the fact that, for $n\geq 4$, a number of $k$ $\rho$--derivatives will typically either hit a factor of $\rho$ or hit the nonlinearities. As an example, if $q$--derivatives, with $q\leq 4$ hit a factor of $\rho$, one will have, for the terms involving $\ol{\ms{f}}^0$: 
\begin{equation*}
    \sum\limits_{j+j'+j_0+\dots+j_\ell=k-q}\rho^{4-q}\cdot \mc{S}\paren{\ms{g}; \Lie_\rho^{j_0-1}\ms{m}, \dots, \Lie_\rho^{j_\ell-1}\ms{m}, \Lie_\rho^j \ol{\ms{f}}^0, \Lie_\rho^{j'}\ol{\ms{f}}^0}\text{,}
\end{equation*}
and similarly for the terms involving $\ol{\ms{f}}^1$. It turns out that, in the case of Maxwell-FG-aAdS, the two components $\ms{T}^0$ and $\ms{T}^2$ behave similarly. 

\subsection{Proof of Proposition \ref{higher_order_f}}\label{app:higher_transport_f}

The left-hand sides of \eqref{higher_transport_f_0} and \eqref{higher_transport_f_1} can simply be obtained from: 
\begin{equation*}
    \Lie_\rho^k\paren{\rho\Lie_\rho \ms{A}- c\cdot \ms{A}} = \rho\Lie_\rho^{k+1} \ms{A} - (c-k)\Lie_\rho^k \ms{A}\text{,}
\end{equation*}
for any vertical tensor fields $\ms{A}$ and constant $c$. 

The right-hand sides of \eqref{higher_transport_f_0} and \eqref{higher_transport_f_1} both contain derivative terms of the form $\ms{DA}$, $\ms{A}\in \lbrace \ol{\ms{f}}^0, \ol{\ms{f}}^1\rbrace$, which can be treated using Proposition \ref{sec:aads_commutation_Lie_D}. In the case of \eqref{higher_transport_f_1}, such a term is multiplied by $\rho$. As a consequence, either one derivative hits $\rho$, and thus $k-1$ derivatives hit $\ms{D}\ol{\ms{f}}^0$, or $\rho$ is untouched and $k$ $\rho$--derivatives hit $\ms{D}\ol{\ms{f}}^0$. 

The same story unfolds for \eqref{higher_transport_f_0}, which contains terms of the form: 
\begin{equation*}
    \rho\mc{S}\paren{\ms{g}; \ms{m}, \ol{\ms{f}}^0}\text{.}
\end{equation*}
Each time a derivative hits the metric, one picks a factor of $\ms{m}$, yielding terms of the form: 
\begin{equation*}
    \sum\limits_{\substack{j+j_0+j_1+\dots+j_\ell=k}} \rho\cdot \mc{S}\paren{\ms{g}; \Lie_\rho^{j_0}\ms{m}, \Lie_\rho^{j_1-1}\ms{m}, \dots, \Lie_\rho^{j_\ell-1}\ms{m}, \Lie_\rho^j \ms{m}}\text{,}
\end{equation*}
if no derivative hits the factor of $\rho$. If one does, one simply needs to replace $k$ with $k-1$ in the sum above and remove the factor of $\rho$. 

\subsection{Proof of Proposition \ref{prop_integration}}\label{app:prop_integration}

For \ref{local_estimate}, see \cite{shao:aads_fg}, Proposition 2.40. 

The case $q=0$ for the limits \ref{lemma_item_1} follows from \ref{local_estimate}. In order to derive the higher-order bounds, one simply needs to note the following, for any $q\geq 0$: 
\begin{equation*}
    \norm{\rho^{q+1}\Lie_\rho^{q+1} \ms{A}}_{M, \varphi} \lesssim_{q} \norm{\rho^q  \Lie_\rho^q \ms{A}}_{M, \varphi} + \norm{\rho^q \Lie_\rho^q \ms{G}}_{M, \varphi}\text{.}
\end{equation*}

Similarly, the cases $q=0$ of \ref{lemma_item_2} is also a consequence of Proposition 2.40 of \cite{shao:aads_fg}. The higher-order limits follow immediately by induction from: 
\begin{equation}
    \rho^{q+1}\Lie_\rho^{q+1}\ms{A} = (c-q) \Lie_\rho^q \ms{A} + \rho^q\Lie_\rho^q\ms{G}\text{.}\label{higher_transport_lemma}
\end{equation}

The ``fractional" convergence in \ref{lemma_item_3}, for $q=0$, can be obtained by differentiating \eqref{lemma_transport} once to get: 
\begin{equation*}
    \rho\Lie_\rho\paren{\rho^{1-c}\Lie_\rho\ms{A}} = \rho^{1-c}\Lie_\rho\ms{G}\Rightarrow^{M}0\text{,}
\end{equation*}
where we used \eqref{bound_G_frac}. Integrating from $\rho_0$ to $\rho$, fixed, yields: 
\begin{equation*}
    \rho^{1-c}\Lie_\rho\ms{A} = \rho_0^{1-c}\Lie_\rho\ms{A}\vert_{\rho_0} + \int_{\rho_0}^\rho \sigma^{-1}\cdot \sigma^{1-c}\Lie_\rho\ms{G}\vert_\sigma d\sigma\text{.}
\end{equation*}
Note also that the following field on $\mc{I}$ is well-defined: 
\begin{equation*}
    \mf{A}^{(c)} := \rho_0^{1-c}\Lie_\rho\ms{A}\vert_{\rho_0} + \int_{\rho_0}^0 \sigma^{-1}\cdot \sigma^{1-c}\Lie_\rho\ms{G}\vert_\sigma d\sigma\text{,}
\end{equation*}
since the integral converges. The following limit thus holds: 
\begin{equation*}
    \sup\limits_{\lbrace \rho \rbrace \times U} \abs{\rho^{1-c}\Lie_\rho\ms{A} - \mf{A}^{(c)}}_{M, \varphi} \lesssim \sup\limits_U \int_0^\rho \sigma^{-1}\cdot \abs{\sigma^{1-c}\Lie_\rho\ms{G}}_{M, \varphi}\vert_\sigma d\sigma\rightarrow 0\text{,}
\end{equation*}
as $\rho\rightarrow 0$. The higher-order limits follow from \eqref{higher_transport_lemma}. 

Finally, the anomalous limits from \ref{lemma_item_4} are, for $q=0$, consequences of Proposition 2.41 of \cite{shao:aads_fg}. The higher-order limits can be obtained by induction from \eqref{higher_transport_lemma}, where $c=0$.  

\section{Proofs of Section \ref{chap:UC}}

\subsection{Proof of Proposition \ref{prop_transport_h_bar}}\label{app:transport_h_bar}

The proof follows from Propositions \ref{commutation_D_D} and \ref{transport_f_UC}. More precisely, starting with $\ul{\ms{h}}^0$: 
\begin{align*}
    \rho\ol{\ms{D}}_\rho\ul{\ms{h}}^0_{abc} &= \ms{D}_a\paren{\rho\ol{\ms{D}}_\rho \ms{f}^1_{bc}}+\rho[\ol{\ms{D}}_\rho, \ms{D}_a\ms{f}^1_{bc}] \\
    &=\ul{\ms{h}}^0_{abc} + 2\rho\ms{D}_a\ul{\ms{h}}^2_{[bc]} + \rho\mc{S}\paren{\ms{g}; \ms{Dm}, \ms{f}^1}_{abc} + \rho\mc{S}\paren{\ms{g}; \ms{m}, \ul{\ms{h}}^0}_{abc}\text{,}
\end{align*}
while for $\ul{\ms{h}}^2$: 
\begin{align*}
    \rho\ol{\ms{D}}_\rho\ul{\ms{h}}^2_{ab} &= \ms{D}_a\paren{\rho\ol{\ms{D}}_\rho \ms{f}^0_{b}}+\rho[\ol{\ms{D}}_\rho, \ms{D}_a\ms{f}^0_{b}] \\
    &=(n-2)\ul{\ms{h}}^2_{ab} - \rho\ms{D}_a\ms{tr}_{\ms{g}}\ul{\ms{h}}^0_{b} + \rho\mc{S}\paren{\ms{g}; \ms{Dm}, \ms{f}^0}_{abc} + \rho\mc{S}\paren{\ms{g}; \ms{m}, \ul{\ms{h}}^2}_{abc}\text{.}
\end{align*}

\subsection{Proof of Proposition \ref{prop_transport_h}}\label{app:transport_h}

The transport equations follow by differentiating the Maxwell and Bianchi equations and commuting derivatives. Starting with the former: 
\begin{align*}
    \rho^2 \nabla_a \nabla_\rho F_{\rho b} + \rho^2 \ms{g}^{cd}\nabla_a\nabla_c
F_{db} = 0 \Longrightarrow \rho^2\nabla_\rho \nabla_a F_{\rho b} + \rho^2 \ms{g}^{cd}\nabla_a\nabla_c F_{db} + \rho^2 \left[\nabla_a, \nabla_\rho \right] F_{\rho b} = 0\text{.} 
\end{align*}
Each of these terms can be decomposed vertically using Proposition \ref{sec:aads_derivatives_vertical} and Equation \eqref{sec:aads_weyl_equation}: 
\begin{align*}
    &\rho^2 \nabla_\rho \nabla_a F_{\rho b} = 2\ms{h}^2_{ab} + \rho \ol{\ms{D}}_\rho \ms{h}^2_{ab}\text{,}\\
    &\rho^2\ms{g}^{cd}\nabla_a \nabla_c F_{db} = \rho \ms{g}^{cd}\ms{D}_a \ms{h}^0_{cdb} - \ms{h}^1_{ab} - \ms{h}^2_{ab} + \underbrace{\ms{tr}_{\ms{g}}\ms{h}^2}_{=0}\ms{g}_{ab}+\rho \mc{S}\paren{\ms{g}; \ms{m}, \ms{h}^2}+\rho \mc{S}\paren{\ms{g}; \ms{m}, \ms{h}^1}\\
    &
    \begin{aligned}
        \rho^2[\nabla_a, \nabla_\rho]F_{\rho b} &= - \rho^2 \operatorname{R}^\sigma{}_{\rho a \rho }F_{\sigma b} - \rho^2 \operatorname{R}^\sigma{}_{ba\rho}F_{\rho \sigma}\\
        &= \rho^{-1}\ms{f}^1_{ab}+\rho\mc{S}\paren{\ms{g}; \ms{w}^1, \ms{f}^0}_{ab} + \rho \mc{S}\paren{\ms{g}; \ms{w}^2, \ms{f}^1}_{ab}\\
        &\; + \rho\mc{S}\paren{\ms{g}; (\ms{f}^1)^3}_{ab} +\rho \mc{S}\paren{\ms{g}; (\ms{f}^0)^2, \ms{f}^1}_{ab}\text{,}
    \end{aligned}
\end{align*}
giving the desired transport equation \eqref{sec:syst_transport_h_2}. 

Next, differentiating the Bianchi equation gives us: 
\begin{equation*}
    \rho^2\nabla_a \nabla_{[\rho}F_{bc]} = 0\Rightarrow \rho^2\nabla_\rho \nabla_a F_{bc}= 2\nabla_a\nabla_{[b}F_{|\rho|c]} -\rho^2[\nabla_a, \nabla_\rho]F_{bc}\text{,}
\end{equation*}
where one can, as above, decompose each of these terms vertically: 
\begin{align*}
    &\rho^2\nabla_\rho \nabla_a F_{bc} = \rho\ol{\ms{D}}_\rho \ms{h}^0_{abc} + 2\ms{h}^0_{abc}\text{,}\\
    &\begin{aligned}2\rho^2\nabla_a\nabla_{[b}F_{|\rho|c]} &= 2\rho \ms{D}_a \ms{h}^2_{[bc]} + 2\ms{h}^0_{[b|a|c]}-2\ms{g}_{a[b}\ms{h}^3_{c]} + \rho\mc{S}\paren{\ms{g}; \ms{m}, \ms{h}^0}_{abc}+\rho\mc{S}\paren{\ms{g}; \ms{m}, \ms{h}^3}_{abc}\\
    &=2\rho \ms{D}_a \ms{h}^2_{[bc]} + \ms{h}^0_{abc} + 2\ms{g}_{a[b}\ms{tr}_{\ms{g}}\ms{h}^0_{c]} + \rho\mc{S}\paren{\ms{g}; \ms{m}, \ms{h}^0}_{abc}
    \end{aligned}\\
    &\begin{aligned}
    -\rho^2[\nabla_a, \nabla_\rho]F_{bc} 
    &=- 2 \rho^2\operatorname{R}^\sigma{}_{[b|a\rho}F_{c]\sigma}\\
    &= 2\rho^{-1}\ms{g}_{a[b}\ms{f}^0_{c]} + \rho\mc{S}\paren{\ms{g}; \ms{w}^1, \ms{f}^1}_{abc} + \rho\mc{S}\paren{\ms{g}; \ms{w}^2, \ms{f}^0}_{abc}+\rho\mc{S}\paren{\ms{g};(\ms{f}^0)^3}_{abc}\\
    &\;  + \rho\mc{S}\paren{\ms{g}; (\ms{f}^1)^2, \ms{f}^0}_{abc}\text{,}
    \end{aligned}
\end{align*}
proving \eqref{sec:syst_transport_h_2}. 

\subsection{Proof of Proposition \ref{prop_transport_w}}\label{app:transport_weyl_UC}

The proof of equations \eqref{transport_w_star} and \eqref{transport_w_1.1} follow from the Bianchi identity for the Riemann curvature, here expressed in terms of the Weyl tensor: 
    \begin{align}
        &\nabla_\rho W_{\alpha\beta\gamma\delta} +\frac{1}{n-1}(g\star\nabla_\rho \tilde{T})_{\alpha\beta\gamma\delta} = 2\nabla_{[\gamma}W _{|\alpha\beta\rho|\delta]} + \frac{2}{n-1}(g\star\nabla_{[\gamma}\tilde{T})_{|\alpha\beta\rho|\delta]}\label{Bianchi_weyl_1}\text{,}\\
        &\label{contracted_bianchi}\nabla_\rho W_{\alpha\beta\rho\delta} +\ms{g}^{cd}\nabla_{c}W_{\alpha\beta d \delta}+ \frac{1}{n-1} \paren{g\star \nabla_\rho \tilde{T}}_{\alpha\beta\rho\delta} + \frac{1}{n-1}{g}^{cd}\paren{g\star \nabla_c \tilde{T}}_{\alpha\beta d\delta}=2\rho^{-2}\cdot \nabla_{[\alpha}\hat{T}_{\beta]\delta}\text{,}
    \end{align}
    with $\tilde{T}, \, \hat{T}$ defined as in Definition \ref{def_modif_T}. 
    These two equations will give rise to one transport equation for $\ms{w}^0$ and $\ms{w}^2$ as well as two (equivalent) transport equations for $\ms{w}^1$. Typically, starting with \eqref{Bianchi_weyl_1}, one gets the two independent transport equations: 
    \begin{align}
        &\nabla_\rho W_{abcd}  = 2\nabla_{[c}W_{|ab\rho|d]} -\frac{1}{n-1}\paren{g\star\nabla_\rho \tilde{T}}_{abcd} + \frac{2}{n-1}\paren{g\star \nabla_{[c}\tilde{T}}_{|ab\rho|d]}\text{,}\label{intermediate_transport_w_0}\\
        &\nabla_\rho W_{\rho abc} = 2 \nabla_{[b}W_{|\rho a \rho|c]} -\frac{1}{n-1}\paren{g\star \nabla_\rho \tilde{T}}_{\rho abc}  + \frac{2}{n-1}\paren{g\star \nabla_{[b}\tilde{T}}_{|\rho a\rho|c]}\label{intermediate_transport_w_2}\text{.}
    \end{align}
    Using Proposition \ref{sec:aads_derivatives_vertical}, the leading terms can be written as:
    \begin{align}
        &\rho^2 \nabla_{[c}W_{|ab\rho|d]} = \ms{D}_{[c}\ms{w}^1_{d]ab}+\rho^{-1}\ms{w}^0_{abcd} - \frac{1}{2\rho}\paren{\ms{g}\star\ms{w}^2}_{abcd} + \mc{S}\paren{\ms{g}; \ms{m}, \ms{w}^0}_{abcd} + \mc{S}\paren{\ms{m}, \ms{w}^2}_{abcd}\text{,}\label{intermediate_vertical_transport_w_0}\\
        &\rho^2 \nabla_{[b}W_{|\rho a\rho|c]} = \ms{D}_{[b}\ms{w}^2_{c]a} + \frac{3}{2\rho}\ms{w}^1_{abc} + \mc{S}\paren{\ms{g}; \ms{m}, \ms{w}^1}_{abc}\text{.}\label{intermediate_vertical_transport_w_1}
    \end{align}
    As well as the following remainder terms:
    \begin{align}
        &\rho^2\paren{g\star \nabla_\rho \tilde{T}}_{abcd} = \notag\mc{S}\paren{\ms{g}; \ul{\ms{h}}^0, \ms{f}^0}_{abcd} + \mc{S}\paren{\ms{g}; \ul{\ms{h}}^{2}, \ms{f}^1}_{abcd} +\rho^{-1}\mc{S}\paren{\ms{g}; (\ms{f}^1)^2}_{abcd} \\
        &\hspace{100pt}+ \rho^{-1}\mc{S}\paren{\ms{g}; (\ms{f}^0)^2}_{abcd}+\mc{S}\paren{\ms{g}; \ms{m}, (\ms{f}^0)^2}_{abcd} + \mc{S}\paren{\ms{g}; \ms{m}, (\ms{f}^1)^2}_{abcd}\label{remainder_w_0.1}\\
        &\notag\rho^2 (g\star {\nabla}_{[c} \tilde{T})_{|ab\rho|d]} = \mc{S}\paren{\ms{g}; \ul{\ms{h}}^0, \ms{f}^0}_{abcd} + \mc{S}\paren{\ms{g}; \ul{\ms{h}}^{2}, \ms{f}^1}_{abcd}+ \rho^{-1}\mc{S}\paren{\ms{g}; (\ms{f}^0)^2}_{abcd}\\
        &\hspace{100pt} +  \rho^{-1}\mc{S}\paren{\ms{g}; (\ms{f}^1)^2}_{abcd}+\mc{S}\paren{\ms{g}; \ms{m}, (\ms{f}^0)^2}_{abcd} + \mc{S}\paren{\ms{g}; \ms{m}, (\ms{f}^1)^2}_{abcd}\label{remainder_w_0.2}\\
        &\label{remainder_w_1.1}\rho^2 (g\star \nabla_\rho \tilde{T})_{\rho abc} = \mc{S}\paren{\ms{g};\ul{\ms{h}}^0, \ms{f}^1}_{abc} + \mc{S}\paren{\ms{g}; \ul{\ms{h}}^{2}, \ms{f}^0}_{abc} + \rho^{-1}\mc{S}\paren{\ms{g}; \ms{f}^1, \ms{f}^0}_{abc}+\mc{S}\paren{\ms{g}; \ms{m}, \ms{f}^1, \ms{f}^0}_{abc}\\
        &\label{remainder_w_1.2}\rho^{2}(g\star\tilde{\nabla}_{[b}\tilde{T})_{|\rho a\rho|c]} = \mc{S}\paren{\ms{g}; \ul{\ms{h}}^{2}, \ms{f}^0}_{abc} + \mc{S}\paren{\ms{g}; \ul{\ms{h}}^0, \ms{f}^1}_{abc} \\
        &\notag\hspace{100pt}+\rho^{-1}\mc{S}\paren{\ms{g}; \ms{f}^1, \ms{f}^0}_{abc}+\mc{S}\paren{\ms{g}; \ms{m}, \ms{f}^0, \ms{f}^1}_{abc}\text{.}
    \end{align}
    After writing:
    \begin{gather*}
        \rho^2\nabla_\rho W_{abcd} = \left(\bar{\ms{D}}_\rho \ms{w}^0\right)_{abcd} + 2\rho^{-1}\ms{w}^0_{abcd}\text{,}\qquad \rho^2\nabla_\rho W_{\rho abc} = (\bar{\ms{D}}_\rho \ms{w}^1)_{abc} + 2\rho^{-1}\ms{w}^1_{abc}\text{,}
    \end{gather*}
    one obtains the following equations for $\ms{w}^0$ and $\ms{w}^1$: 
    \begin{align*}
        &\rho\bar{\ms{D}}_\rho \ms{w}^0_{abcd} = 2\rho\ms{D}_{[c}\ms{w}^1_{d]ab} - (\ms{g}\star\ms{w}^2)_{abcd} + (\mc{R}_{T, \ms{w}^2})_{abcd}\text{,}\\
        &\rho\ol{\ms{D}}_\rho \ms{w}^1_{abc} = 2\rho \ms{D}_{[b}\ms{w}^2_{c]a} + \ms{w}^1_{abc} + (\mc{R}_{T,\ms{w}^1})_{abc}\text{,}
    \end{align*}
    where $\mc{R}_{T, \ms{w}^2}$ is given by \eqref{R_T_w_2}, and is obtained from \eqref{intermediate_vertical_transport_w_0}, \eqref{remainder_w_0.1} and \eqref{remainder_w_0.2}, while $\mc{R}_{T,\ms{w}^1}$ is obtained from \eqref{intermediate_vertical_transport_w_1}, \eqref{remainder_w_1.1} and \eqref{remainder_w_1.2}. Using now the contracted Bianchi identity \eqref{contracted_bianchi}, one gets the two following independent equations: 
    \begin{align}
        &\label{intermediate_transport_w_0_2}\rho^2\nabla_\rho W_{\rho c ab} + \rho^2\ms{g}^{de}\nabla_d W_{abec}+\frac{1}{n-1}\rho^2(g\star\nabla_\rho \tilde{T})_{\rho c ab} +\frac{1}{n-1}\rho^2 \ms{g}^{de}\paren{g\star \nabla_d \tilde{T}}_{abec}= 2\nabla_{[a}\hat{T}_{b]c}\text{,}\\
        &\label{intermediate_contracted_transport_w_2}\rho^2\nabla_\rho W_{\rho a \rho b}+\rho^2\ms{g}^{cd}\nabla_c W_{\rho adb}+ \frac{1}{n-1}\rho^2\paren{g\star \nabla_\rho\tilde{T}}_{\rho a \rho b} +\frac{1}{n-1}\rho^2 \ms{g}^{cd}\paren{g\star \nabla_c \tilde{T}}_{\rho a d b} = 2\nabla_{[\rho}\hat{T}_{a]b}\text{,}
    \end{align}
    where the individual terms can be written in terms of vertical tensors: 
    \begin{align}
        &\label{intermediate_w_1_contracted_1}\rho^2 \ms{g}^{de}\nabla_d W_{abec} = \ms{g}^{de}\ms{D}_d \ms{w}^0_{abec}- n\rho^{-1}\cdot \ms{w}^1_{cab} + \mc{S}\paren{\ms{g}; \ms{m}, \ms{w}^1}_{abc}\text{,}\\
        &\label{intermediate_w_1_contracted_2} \rho^2 \nabla_\rho W_{\rho cab} = \ol{\ms{D}}_\rho \ms{w}^1_{cab} + 2\rho^{-1}\ms{w}^1_{cab}\text{,}\\
        &\label{intermediate_w_2_contracted_1}\rho^2 \ms{g}^{cd}\nabla_c W_{\rho adb} = \ms{g}^{cd}\ms{D}_{c}\ms{w}^1_{adb} + \rho^{-1}\ms{g}^{cd}\ms{w}^0_{cadb} - \rho^{-1}(n-1)\ms{w}^2_{ab} \\
        &\notag\hspace{80pt}+ \mc{S}\paren{\ms{g}; \ms{m}, \ms{w}^0}_{ab} + \mc{S}\paren{\ms{g}; \ms{m}, \ms{w}^2}_{ab}\text{,}\\
        &\label{intermediate_w_2_contracted_2} \rho^2 \nabla_\rho W_{\rho a \rho b}= \ol{\ms{D}}_\rho \ms{w}^2_{ab} + 2\rho^{-1}\ms{w}^2_{ab}\text{,}
    \end{align}
    while for the remainder terms:
    \begin{align}
        &\notag\rho^2\paren{g\star \nabla_\rho \tilde{T}}_{\rho cab} +\rho^2 \ms{g}^{de}\paren{g\star \nabla_d \tilde{T}}_{abec}=\mc{S}\paren{\ms{g}; \ul{\ms{h}}, \ms{f}^1}_{abc}+ \mc{S}\paren{\ms{g}; \ul{\ms{h}}^{2}, \ms{f}^0}_{abc}\text{,}\\
        &\label{remainder_contracted_w_1_1} \hspace{50pt}+ \rho^{-1}\mc{S}\paren{\ms{g}; \ms{f}^0, \ms{f}^1}_{abc} + \mc{S}\paren{\ms{g}; \ms{m}, \ms{f}^0, \ms{f}^1}_{abc}\\
        &\nabla_{[a}\hat{T}_{b]c} = \mc{S}\paren{\ms{g}; \ul{\ms{h}}^0, \ms{f}^1}_{abc}  + \mc{S}\paren{\ms{g}; \ul{\ms{h}}^{2} , \ms{f}^1}_{abc} +\rho^{-1}\mc{S}\paren{\ms{g}; \ms{f}^0, \ms{f}^1}_{abc} +\mc{S}\paren{\ms{g}; \ms{m}, \ms{f}^1, \ms{f}^0}_{abc}\label{remainder_contracted_w_1_2}\text{,}\\
        & \notag \rho^2(g\star \nabla_\rho\tilde{T})_{\rho a \rho b} + \rho^2\ms{g}^{cd}\paren{g\star \nabla_c\tilde{T}}_{\rho adb} = \mc{S}\paren{\ms{g}; \ul{\ms{h}}^0, \ms{f}^0}_{ab} + \mc{S}\paren{\ms{g}; \ul{\ms{h}}^{2}, \ms{f}^1}_{ab}+ \rho^{-1}\mc{S}\paren{\ms{g}; (\ms{f}^1)^2}_{ab} \\
        &\hspace{50pt}+ \rho^{-1}\mc{S}\paren{\ms{g}; (\ms{f}^0)^2}_{ab}+\mc{S}\paren{\ms{g}; \ms{m}, (\ms{f}^0)^2}_{ab} +\mc{S}\paren{\ms{g}; \ms{m}, (\ms{f}^1)^2}_{ab}\label{remainder_contracted_w_2_1} \text{,}\\
        &\notag\nabla_{[\rho}\hat{T}_{a]b} = \mc{S}\paren{\ms{g}; \ul{\ms{h}}^0, \ms{f}^0}_{ab} + \mc{S}\paren{\ms{g}; \ul{\ms{h}}^{2}, \ms{f}^1}_{ab} +\rho^{-1}\mc{S}\paren{\ms{g}; (\ms{f}^0)^2}_{ab} + \rho^{-1}\mc{S}\paren{\ms{g}; (\ms{f}^1)^2}_{ab}+ \mc{S}\paren{\ms{g}; \ms{m}, (\ms{f}^0)^2}_{ab}\\
        &\label{remainder_contracted_w_2_2}\hspace{50pt} + \mc{S}\paren{\ms{g}; \ms{m}, (\ms{f}^1)^2}_{ab}\text{.}
    \end{align}
    Combining now Equations \eqref{intermediate_w_1_contracted_1}, \eqref{intermediate_w_1_contracted_2}, as well as the remainder terms \eqref{remainder_contracted_w_1_1}, \eqref{remainder_contracted_w_1_2} gives the following transport equation for $\ms{w}^1$: 
    \begin{equation*}
        \rho\ol{\ms{D}}_\rho\ms{w}^1_{cab} = - \rho \cdot \ms{g}^{de}\ms{D}_d \ms{w}^0_{ecab} + (n-2)\ms{w}^1_{cab} + (\mc{R}_{T, \ms{w}^1})_{cab}\text{.}
    \end{equation*}

    Gathering Equations \eqref{intermediate_w_2_contracted_1}, \eqref{intermediate_w_2_contracted_2}, as well as \eqref{remainder_contracted_w_2_1} and \eqref{remainder_contracted_w_2_2} gives: 
    \begin{equation*}
        \rho \ol{\ms{D}}_\rho \ms{w}^2_{ab} = - \rho \cdot \ms{g}^{cd}\ms{D}_c\ms{w}^1_{adb} + (n-2)\ms{w}^2_{ab} + (\mc{R}_{T,\ms{w}^2})_{ab}\text{,}
    \end{equation*}
    as desired.

\subsection{Proof of Proposition \ref{wave_weyl}}\label{app:wave_w_maxwell}

The proof of \eqref{sec:system_weyl_wave} follows from \eqref{sec:aads_weyl_equation} and the differentiated Bianchi identity: 
    \begin{equation*}
        \Box_g R_{\alpha \beta \gamma\delta} + \nabla^\mu \nabla_\gamma R_{\alpha\beta\delta\mu} + \nabla^\mu \nabla_\delta R_{\alpha\beta\mu\gamma} = 0\text{,}
    \end{equation*}
    followed by a commutation of the derivatives: 
    \begin{equation}\label{wave_R_commut}
        \Box_g R_{\alpha\beta\gamma\delta}-2 [\nabla^\mu, \nabla_{[\gamma}]R_{|\alpha\beta\mu|\delta]} - 2\nabla_{[\gamma} \nabla^\mu R_{|\alpha\beta\mu|\delta]} = 0\text{.}
    \end{equation}
    The first term can be expanded in terms of the Weyl and stress-energy tensors: 
    \begin{equation}\label{box_term_weyl}
        \Box_g R = \Box_g W + \frac{1}{n-1}g\star \Box_g \tilde{T}\text{,}
    \end{equation}
    while the commutation term can be written as: 
    \begin{align}\label{commutator_term}
        -2[\nabla^\mu, \nabla_{[\gamma}]R_{|\alpha\beta\mu|\delta]} = 2(R\odot R)_{\alpha\beta\gamma\delta} + 2\operatorname{Ric}^\sigma{}_{[\gamma}R_{|\alpha\beta|\delta]\sigma}\text{.}
    \end{align}
    Using now \eqref{sec:aads_weyl_equation}, we can start studying the the first term above, here expanded: 
    \begin{align}\label{R_odot}
        (R\odot R) &= \paren{\paren{W - \frac{1}{2}g\star g + \frac{1}{n-1}\tilde{T}\star g} \odot \paren{W - \frac{1}{2}g\star g + \frac{1}{n-1}\tilde{T}\star g}}\text{.}
    \end{align}

    \begin{lemma}\label{lemma_identity_otimes}
        Let $A$ be a $(0,4)$--tensor satisfying, with respect to any local coordinate system: 
        \begin{itemize}
            \item $A_{\alpha\beta\gamma\delta} = -A_{\beta\alpha\gamma\delta}$
            \item $A_{\alpha\beta\gamma\delta} = A_{\gamma\delta\alpha\beta}$
            \item $A_{\alpha\beta\gamma\delta} + A_{\alpha\delta\beta\gamma}+A_{\alpha\gamma\delta\beta}=0.$
        \end{itemize}
        Then, the following holds: 
        \begin{equation*}
            \frac{1}{2}\paren{A \odot (g\star g) + (g\star g)\odot A} = \operatorname{tr}_g A \star g\text{,}
        \end{equation*}
        where we write: 
        \begin{equation*}
            (\operatorname{tr}_g A)_{\alpha\beta} := g^{\gamma\delta}A_{\alpha \gamma \beta \delta}\text{.}
        \end{equation*}
    \end{lemma}

    \begin{proof}
        Follows from straightforward algebraic manipulations. 
    \end{proof}

    An immediate corollary of Lemma \ref{lemma_identity_otimes} is the following: 
    \begin{equation*}
        W\odot (g\star g) + (g\star g)\odot W = 0\text{.}
    \end{equation*}
    Overall, one has: 
    \begin{equation}\label{odot_term}
        2(R\odot R) = 2\left[\paren{W + \frac{1}{n-1}g \star \tilde{T}} \odot \paren{W + \frac{1}{n-1}g\star \tilde{T}}\right] + n(g\star g) - \frac{2}{n-1}\paren{\operatorname{tr}_g(\tilde{T}\star g)\star g}
    \end{equation}
   
    Next, the Ricci term in \eqref{commutator_term} can be written as: 
    \begin{align}
    \notag2\operatorname{Ric}^\sigma{}_{[\gamma}R_{|\alpha\beta|\delta]\sigma} &= 2(-n\delta + \hat{T})^\sigma{}_{[\gamma}\paren{W - \frac{1}{2}g\star g + \frac{1}{n-1}\tilde{T}\star g}_{|\alpha\beta|\delta]\sigma}\\
    &\notag=2nW_{\alpha\beta\gamma\delta} + \frac{2n}{n-1}(\tilde{T}\star g)_{\alpha\beta\gamma\delta} + 2\hat{T}^\sigma{}_{[\gamma}W_{|\alpha\beta|\delta]\sigma} - (\hat{T}\star g)_{\alpha\beta\delta\gamma} \\
    &\quad+ \frac{2}{n-1}\hat{T}^\sigma{}_{[\gamma}(\tilde{T}\star g)_{|\alpha\beta|\delta]\sigma} -n(g\star g)_{\alpha\beta\gamma\delta}\text{.}\label{mass_weyl_computation}
    \end{align}
    The first term in \eqref{mass_weyl_computation} yields the desired mass term in \eqref{sec:system_weyl_wave}. 
    Finally, we will need to look at the last term in \eqref{wave_R_commut}. Using the contracted second Bianchi identity:
    \begin{align*}
        \nabla^\mu R_{\alpha\beta\mu\delta} &= 2\nabla_{[\alpha}\operatorname{Ric}_{\beta]\delta}\\
        &= 2\nabla_{[\alpha}\hat{T}_{\beta]\delta}\text{,}
    \end{align*} 
    one obtains immediately from \eqref{sec:aads_weyl_equation}: 
    \begin{align}\label{divergence_weyl_ricci}
        -2\nabla_{[\gamma}\nabla^\mu R_{|\alpha\beta\mu|\delta]} = -2\nabla_{[\gamma}\nabla_{|\alpha}\hat{T}_{\beta|\delta]}+2\nabla_{[\gamma}\nabla_{|\beta}\hat{T}_{\alpha|\delta]}\text{.}
    \end{align}
    Gathering now \eqref{box_term_weyl}, \eqref{odot_term}, \eqref{mass_weyl_computation} and \eqref{divergence_weyl_ricci} yields the desired Equation \eqref{sec:system_weyl_wave}. 

    The wave equation \eqref{wave_maxwell} can be obtained by differentiating the Bianchi equation for the Maxwell field: 
    \begin{equation*}
        \Box_g F_{\alpha\beta} + 2\nabla^\mu \nabla_{[\alpha} F_{\beta]\mu} =0\text{.}
    \end{equation*}
    Commuting the derivatives and using Maxwell equations, one gets: 
    \begin{align*}
        \Box_g F_{\alpha\beta} &- 2\paren{W^\sigma{}_{[\beta}{}^\mu{}_{\alpha]} + \frac{1}{n-1}(g\star \tilde{T})^\sigma{}_{[\beta}{}^\mu{}_{\alpha]}}F_{\sigma\mu}-2\delta^\sigma_{[\alpha}\delta_{\beta]}^\mu F_{\sigma\mu}\\
        &+2nF_{\alpha\beta} + \hat{T}^\sigma{}_{[\alpha}F_{\beta]\sigma}=0\text{.}
    \end{align*}
    The Bianchi identity: 
    \begin{equation}
        2W^\sigma{}_{[\beta}{}^\mu{}_{\alpha]} = -W_{\alpha\beta}{}^{\sigma\mu}\text{,}
    \end{equation}
    as well as the following: 
    \begin{equation*}
        (g\star \tilde{T})^\sigma{}_{[\beta}{}^\mu{}_{\alpha]} = -(g\star \tilde{T})_{\alpha\beta}{}^{\sigma\mu}\text{,}
    \end{equation*}
    yields the desired wave equation \eqref{wave_maxwell}.`

\subsection{Proof of Proposition \ref{prop_wave_h_ST}}\label{app:wave_h}

This equation simply follows by application two consecutive commutations from the following: 
    \begin{align*}
        \Box_g H_{\mu\alpha\beta} &= \nabla^\gamma\nabla_\gamma \nabla_\mu F_{\alpha\beta} \\
        & = \nabla^\gamma\paren{[\nabla_\gamma, \nabla_\mu]F_{\alpha\beta} +\nabla^\gamma \nabla_\mu \nabla_\gamma F_{\alpha\beta}}\\
        & = - 2\nabla^\gamma\paren{R^\sigma{}_{[\alpha|\gamma\mu}F_{\sigma|\beta]}} + \operatorname{Ric}^\sigma{}_\mu H_{\sigma\alpha\beta} - 2R^\sigma{}_{[\alpha|}{}^\gamma{}_{\mu}H_{\gamma\sigma|\beta]} + \nabla_\mu\Box_g F_{\alpha\beta}\\
        &=- 2\nabla^\gamma R^\sigma{}_{[\alpha|\gamma\mu}F_{\sigma|\beta]} - 4R^\sigma{}_{[\alpha}{}^\gamma{}_{|\mu}H_{\gamma\sigma|\beta]} + \operatorname{Ric}^\sigma{}_\mu H_{\sigma\alpha\beta} - \nabla_\mu W_{\alpha\beta}{}^{\sigma\lambda} F_{\sigma\lambda}\\
        &\hspace{14pt}-2(n-1)H_{\mu\alpha\beta}-W_{\alpha\beta}{}^{\sigma\lambda}H_{\mu\sigma\lambda} + \nabla_\mu\paren{\frac{2}{n-1}\tilde{T}^\sigma{}_{[\alpha}F_{\beta]\sigma}-2\hat{T}^\sigma{}_{[\alpha}F_{\beta]\sigma}}\text{.}
    \end{align*}
    One can work out each of these terms individually: 
    \begin{align*}
        &- 2\nabla^\gamma R^\sigma{}_{[\alpha|\gamma\mu}F_{\sigma|\beta]} = - 2\nabla^\gamma W^\sigma{}_{[\alpha|\gamma\mu}F_{\sigma|\beta]} - \frac{2}{n-1}(g\star \nabla^\gamma \tilde{T})^\sigma{}_{[\alpha|\gamma\mu}F_{\sigma|\beta]}\text{,}\\
        &\begin{aligned}- 4R^\sigma{}_{[\alpha}{}^\gamma_{|\mu}H_{\gamma\sigma|\beta]} =& -4W^\sigma{}_{[\alpha}{}^\gamma{}_{|\mu}H_{\gamma\sigma|\beta]} + 2(g\star g)^\sigma{}_{[\alpha}{}^\gamma{}_{|\mu}H_{\gamma\sigma|\beta]} -\frac{4}{n-1}(g\star \tilde{T})^\sigma{}_{[\alpha}{}^\gamma{}_{|\mu}H_{\gamma\sigma|\beta]} \\
        =& -4W^\sigma{}_{[\alpha}{}^\gamma{}_{|\mu}H_{\gamma\sigma|\beta]} +4g_{\mu[\alpha|}H^\sigma{}_{\sigma|\beta]} - 4H_{[\alpha|\mu|\beta]}  \\
        &- \frac{4}{n-1} \paren{\tilde{T}_{\mu[\alpha|}H^\sigma{}_{\sigma|\beta]} + g_{\mu[\alpha|}\tilde{T}^{\sigma\gamma}H_{\gamma\sigma|\beta]} - \tilde{T}^\gamma{}_{[\alpha|}H_{\gamma\mu|\beta]} - \tilde{T}^\sigma{}_\mu H_{[\alpha|\sigma|\beta]}}\\
        =& -4W^\sigma{}_{[\alpha}{}^\gamma_{|\mu}H_{\gamma\sigma|\beta]} - 2H_{\mu\alpha\beta}  - \frac{4}{n-1}\paren{g_{\mu[\alpha}\tilde{T}^{\sigma\gamma}H_{|\gamma\sigma|\beta]} -\tilde{T}^{\gamma}{}_{[\alpha|}H_{\gamma\mu|\beta]}} \\
        &+ \frac{2}{n-1}\tilde{T}^\sigma{}_\mu H_{\sigma\alpha\beta} \end{aligned}\\
        &\operatorname{Ric}^\sigma{}_\mu H_{\sigma\alpha\beta} = - nH_{\mu\alpha\beta} + \hat{T}^\sigma{}_\mu H_{\sigma\alpha\beta}\text{,}
    \end{align*}
    where we used Maxwell equations, as well as the Bianchi identity for $H$, leading to the desired wave equation \eqref{wave_higher_derivative}.

\subsection{Proof of Proposition \ref{prop_wave_maxwell}}\label{app:prop_wave_maxwell}
To prove this, we will use Proposition \ref{sec:aads_derivatives_vertical}, allowing one to decompose the wave equation, as well as spacetime derivatives, vertically. Typically, the two wave equations follow from: 
    \begin{align}
        &\notag\Box_g F_{ab} = \rho^{-2}\bar{\Box}(\rho\ms{f}^1)_{ab} -4\ul{\ms{h}}^2_{[ab]} -2\rho^{-1}\ms{f}^1_{ab} + \rho\mc{S}\paren{\ms{g}; \ms{m}, \ul{\ms{h}}^2}_{ab} + \rho\mc{S}\paren{\ms{g}; \ms{Dm}, \ms{f}^0}_{ab}\\
        &\label{intermediate_wave_f_1}\hspace{40pt} +\mc{S}\paren{\ms{g}; \ms{m}, \ms{f}^1}_{ab} +\rho \mc{S}\paren{\ms{g}; (\ms{m})^2, \ms{f}^1}_{ab}\\
        &\notag\Box_g F_{\rho a} = \rho^{-2}\bar{\Box}(\rho \ms{f}^0)_a + 2\ms{tr}_{\ms{g}}\ul{\ms{h}}^0_a -(n-1)\rho^{-1}\ms{f}^{0}_a + \rho \mc{S}\paren{\ms{g}; \ms{m}, \ul{\ms{h}}^0}_a + \rho \mc{S}\paren{\ms{g}; \ms{Dm},\ms{f}^1}_a \\
        & \label{intermediate_wave_f_0}\hspace{40pt} +\mc{S}\paren{\ms{g}; \ms{m}, \ms{f}^0}_a + \rho\mc{S}\paren{\ms{g}; (\ms{m})^2, \ms{f}^0}_a\text{.}
    \end{align}
    Furthermore, one can write from \eqref{wave_power}: 
    \begin{align*}
        \rho^{-2}\bar{\Box}(\rho\ms{f}^i) = \rho^{-1}\bar{\Box} \ms{f}^i + 2\bar{\ms{D}}_\rho \ms{f}^i - (n-1)\rho^{-1}\ms{f}^i + \rho \mc{S}\paren{\ms{g}; \ms{m}, \ms{f}^i}\text{,}\qquad i=1, 2\text{.}
    \end{align*}
    The term $\ol{\ms{D}}_\rho \ms{f}^i$ needs to be rewritten using the transport equations \eqref{transport_f0_FG}, \eqref{transport_f1_FG} to give: 
    \begin{align}
        &\label{intermediate_wave_f_1_2}\rho^{-2}\bar{\Box}(\rho\ms{f}^1)_{ab} = \rho^{-1}\bar{\Box}\ms{f}^1_{ab} + 2\rho^{-1}\ms{f}^1_{ab} + 4\ul{\ms{h}}^2_{[ab]} - (n-1)\rho^{-1}\ms{f}^1_{ab} +\mc{S}\paren{\ms{g}; \ms{m}, \ms{f}^1}_{ab}\text{,}\\
        &\label{intermediate_wave_f_0_2}\rho^{-2}\bar{\Box}(\rho \ms{f}^0)_a = \rho^{-1}\bar{\Box}\ms{f}^0_a + 2(n-2)\rho^{-1}\ms{f}^0_a -2 \ms{tr}_{\ms{g}}\ul{\ms{h}}^0_a - (n-1)\rho^{-1}\ms{f}^0_a +\mc{S}\paren{\ms{g}; \ms{m}, \ms{f}^0}_a\text{.}
    \end{align}
    Injecting \eqref{intermediate_wave_f_0_2} into \eqref{intermediate_wave_f_0}, as well as \eqref{intermediate_wave_f_1_2} into \eqref{intermediate_wave_f_1} yield the left-hand sides of \eqref{wave_f_1} and \eqref{wave_f_0}. 
    The remainder terms $\mc{R}_{W, \ms{f}^i}$ can be fully obtained by decomposing the right-hand side of \eqref{wave_maxwell}: 
    \begin{align*}
        &(\Box_g +  2(n-1))F_{ab} = \rho \mc{S}\paren{\ms{g}; \ms{w}^1, \ms{f}^0}_{ab} + \rho \mc{S}\paren{\ms{g}; \ms{w}^0, \ms{f}^1}_{ab} + \rho\mc{S}\paren{\ms{g};  (\ms{f}^1)^3}_{ab}\\
        &\hspace{110pt}+ \rho \mc{S}\paren{\ms{g}; (\ms{f}^0)^2, \ms{f}^1}_{ab}\text{,}\\
        &(\Box_g + 2(n-1))F_{\rho a} = \rho \mc{S}\paren{\ms{g}; \ms{w}^1, \ms{f}^1}_a + \rho\mc{S}\paren{\ms{g}; \ms{w}^2, \ms{f}^0}_a + \rho\mc{S}\paren{\ms{g}; (\ms{f}^1)^2, \ms{f}^0}_a \\
        &\hspace{110pt}+ \rho \mc{S}\paren{\ms{g}; (\ms{f}^0)^3}_a\text{,}
    \end{align*}
    giving the expected remainder terms. 

    \subsection{Proof of Proposition \ref{prop_wave_h}}\label{app:prop_wave_h}

    The proof of this proposition follows from, in the same spirit as Proposition \ref{prop_wave_maxwell}, decomposing the left- and right-hand side of \eqref{wave_higher_derivative}. In order to control the vertical fields in the Carleman estimates, it is essential to formulate vertical wave equations where only lower-order terms act as sources. We will show that the specific selection of \(\ul{\ms{h}}^0\text{, } \ul{\ms{h}}^2\) leads to such a system.

    Note that an alternative method would be to derive the wave equation by studying the commutator $\left[\ol{\Box}, \ms{D}\right]\ms{f}^i$, $i=0, 1$. This method would indeed use the vertical formalism in a more optimal manner, although the amount of computations would typically remain the same.
    
    Starting with the left-hand side, from Proposition \ref{sec:aads_derivatives_vertical}:
    \begin{align*}
        &(\Box_g + 3n)H_{abc} = \rho^{-3}\ol{\Box}(\rho^2 \ms{h}^0)_{abc} - 2 \ms{D}_{a}\ms{h}^1_{bc} -4\ms{D}_{[b}\ms{h}^2_{|a|c]} + 3(n-1)\rho^{-1}\ms{h}^0_{abc} + 4\ms{g}_{a[b}\ms{h}^3_{c]}\\
        &\hspace{95pt}+ \rho \mc{S}\paren{\ms{g}; \ms{m},\ms{Dh}^1}_{abc} + \rho \mc{S}\paren{\ms{g};\ms{Dm}, \ms{h}^1}_{abc} + \rho \mc{S}\paren{\ms{g}; \ms{m},\ms{Dh}^2}_{abc} \\
        &\hspace{95pt}+ \rho \mc{S}\paren{\ms{g};\ms{Dm}, \ms{h}^2}_{abc}+\mc{S}\paren{\ms{g}; \ms{m}, \ms{h}^0}_{abc} + \rho \mc{S}\paren{\ms{g}; (\ms{m})^2, \ms{h}^0}_{abc} \\
        &\hspace{95pt} + \rho \mc{S}\paren{\ms{g}; (\ms{m})^2, \ms{h}^3}_{abc}+\mc{S}\paren{\ms{g}; \ms{m}, \ms{h}^3}_{abc}\text{,}\\
        &(\Box_g + 3n)H_{a\rho b} = \rho^{-3}\ol{\Box}(\rho^2 \ms{h}^2)_{ab} + 2\ms{g}^{cd}\ms{D}_c \ms{h}^0_{adb} + 2\ms{D}_a \ms{tr}_\ms{g}\ms{h}^0_b +2(n-1)\rho^{-1}\ms{h}^2_{ab} - 2\rho^{-1}\ms{h}^1_{ab} + 2\rho^{-1}\ms{h}^2_{ab}\\
        &\hspace{95pt} +\rho\mc{S}\paren{\ms{g}; \ms{m}, \ms{Dh}^0}_{ab} + \rho \mc{S}\paren{\ms{g}; \ms{Dm}, \ms{h}^0}_{ab} + \rho\mc{S}\paren{\ms{g}; \ms{m}, \ms{D}\ms{h}^3}_{ab} \\
        &\hspace{95pt}+ \rho\mc{S}\paren{\ms{g}; \ms{D}\ms{m}, \ms{h}^3}_{ab}  + \mc{S}\paren{\ms{g};\ms{m}, \ms{h}^2}_{ab} + \rho\mc{S}\paren{\ms{g}; \ms{m}^2, \ms{h}^2}_{ab} \\
        &\hspace{95pt}+\rho\mc{S}\paren{\ms{g}; \ms{m}^2, \ms{h}^1}_{ab}+ \mc{S}\paren{\ms{g}; \ms{m}, \ms{h}^1}_{ab}\text{.}
    \end{align*}
    In order to obtain the correct lower order terms \eqref{remainder_h_0} and \eqref{remainder_h_2}, one has to express $\ms{h}^0$, $\ms{h}^1$, $\ms{h}^2$ and $\ms{h}^3$ in terms of $\ul{\ms{h}}^0$ and $\ul{\ms{h}}^2$ using Proposition \ref{h_as_Df}. Let us justify how this can be done, with the following examples:
    \begin{align*}
        \rho\mc{S}\paren{\ms{g}; \ms{m}, \ms{Dh}^1}, \rho\mc{S}\paren{\ms{g}; \ms{m}, \ms{Dh}^2} &= \rho\mc{S}\paren{\ms{g}; \ms{m}, \ms{D}\ul{\ms{h}}^{2}}  + \mc{S}\paren{\ms{g}; \ms{m}, \ul{\ms{h}}^0} +  \rho\mc{S}\paren{\ms{g};\ms{m}, \ms{Dm}, \ms{f}^1} + \rho\mc{S}\paren{\ms{g}; (\ms{m})^2, \ul{\ms{h}}^0}\text{,}\\
        \rho\mc{S}\paren{\ms{g}; \ms{m}, \ms{Dh}^0}, \rho\mc{S}\paren{\ms{g}; \ms{m}, \ms{D}{\ms{h}}^3}&=\rho\mc{S}\paren{\ms{g}; \ms{m}, \ms{D}\ul{\ms{h}}^0} + \mc{S}\paren{\ms{g}; \ms{m}, \ul{\ms{h}}^0} + \rho\mc{S}\paren{\ms{g}; \ms{m},\ms{Dm}, \ms{f}^0} + \rho\mc{S}\paren{\ms{g}; (\ms{m})^2, \ul{\ms{h}}^2}\text{,}
    \end{align*}
    and so on. 
    From Proposition \ref{sec:aads_derivatives_vertical}:
    \begin{align}
        \notag\rho^{-3}\ol{\Box}(\rho^2\ms{h}^i) = \rho^{-1}\ol{\Box}\ms{h}^i + 4\ol{\ms{D}}_\rho \ms{h}^i -2(n-2)\rho^{-1}\ms{h}^i + \rho\mc{S}\paren{\ms{g}; \ms{m},\ms{h}^i}\qquad i=0,2\text{,}
    \end{align}
    together with the transport equations \eqref{sec:syst_transport_h_0} and \eqref{sec:syst_transport_h_2}, one has: 
    \begin{align}
        \rho^{-3}\ol{\Box}(\rho^2\ms{h}^0)_{abc} = &\,\rho^{-1}\ol{\Box}\ms{h}^0_{abc} -2n\rho^{-1}\ms{h}^0_{abc} + 8 \rho^{-1}\ms{g}_{a[b}\ms{tr}_{\ms{g}}\ms{h}^0_{c]} + 8\ms{D}_a\ms{h}^2_{[bc]} + 8 \rho^{-2}\ms{g}_{a[b}\ms{f}^0_{c]}\\
        &\notag+\mc{S}\paren{\ms{g}; \ms{m}, \ul{\ms{h}}^0}_{abc} + \rho^{-1}\mc{S}\paren{\ms{g}; \ms{m}, \ms{f}^0}_{abc}+ \mc{S}\paren{\ms{g}; \ms{w}^2, \ms{f}^0}_{abc} +\mc{S}\paren{\ms{g}; (\ms{m})^2, \ms{f}^0}_{abc} \\
        &\notag+ \mc{S}\paren{\ms{g}; \ms{w}^1, \ms{f}^1}_{abc}+\mc{S}\paren{\ms{g}; (\ms{f}^0)^3}_{abc} + \mc{S}\paren{\ms{g}; (\ms{f}^1)^2, \ms{f}^0}_{abc}\text{,}\\
        \rho^{-3}\ol{\Box}(\rho^2\ms{h}^2)_{ab} =& \,\rho^{-1}\ol{\Box}\ms{h}^2_{ab} -2n\rho^{-1}\ms{h}^2_{ab} + 8\rho^{-1}\ms{h}^2_{[ab]} - 4\ms{D}_a\ms{tr}_{\ms{g}}\ms{h}^0_b - 4\rho^{-2}\ms{f}^1_{ab}\\
        &\notag +\mc{S}\paren{\ms{g}; \ms{m}, \ul{\ms{h}}^{2}}_{ab} + \mc{S}\paren{\ms{g}; \ms{w}^2, \ms{f}^1}_{ab} + \mc{S}\paren{\ms{g}; \ms{w}^1, \ms{f}^0}+\rho^{-1}\mc{S}\paren{\ms{g}; \ms{m}, \ms{f}^1}_{ab}\\
        &\notag+ \mc{S}\paren{\ms{g};(\ms{f}^0)^2, \ms{f}^1}_{ab}+ \mc{S}\paren{\ms{g};  (\ms{f}^1)^3}_{ab}+\mc{S}\paren{\ms{g}; (\ms{m})^2, \ms{f}^1}_{ab}\text{.}
    \end{align}

    At this stage, ignoring lower-order terms in the equations, the wave equation for both $\ms{h}^0$ and $\ms{h}^2$ still contain top-order source terms: 
    \begin{align}
        &\notag\rho\cdot (\Box_g + 3n)H_{abc}=\ol{\Box}\ms{h}^0_{abc} +(n-3)\ms{h}^{0}_{abc}+ 4\rho\ms{D}_a \ms{h}^2_{[bc]} - 4\rho\ms{D}_{[b}\ms{h}^2_{|a|c]} + 4\ms{g}_{a[b}\ms{tr}_{\ms{g}}\ms{h}^0_{c]} +8\rho^{-1}\ms{g}_{a[b}\ms{f}^0_{c]} + \ms{E}^0_{abc}\text{,} \\
        &\label{intermediate_h_2}\rho\cdot(\Box_g + 3n)H_{a\rho b}=\ol{\Box}\ms{h}^2_{ab} + 4\ms{h}^2_{[ab]} -2\rho \ms{D}_a \ms{tr}_{\ms{g}}\ms{h}^0_{b} +2\rho \ms{g}^{cd}\ms{D}_c\ms{h}^0_{adb} -4
        \rho^{-1}\ms{f}^1_{ab} + \ms{E}^2_{ab}\text{,}
    \end{align}
    where we used the expressions for $\ms{h}^3$ and $\ms{h}^1$ in terms of $\ms{h}^0$ and $\ms{h}^2$, and where $\ms{E}^0$ and $\ms{E}^2$ contain the remainder terms considered above. 
    
    We will therefore need the following identities: 
    \begin{lemma}\label{lemma_identity_h}
        The following identities hold for the vertical derivatives of $\ms{h}^0$ and $\ms{h}^2$, with respect to $(U, \varphi)$: 
         \begin{align}
        &\notag \rho\ms{g}^{cd}\ms{D}_a \ms{h}^0_{cdb} = \rho\ms{g}^{cd}\ms{D}_c \ms{h}^0_{adb} - (n-2)\ms{h}^2_{ab}+(n-2)\rho^{-1}\ms{f}^1_{ab}+  \rho\mc{S}\paren{\ms{g}; \ms{f}^1, (\ms{f}^0)^2}_{ab} +\rho\mc{S}\paren{\ms{g}; \ms{m}, \ms{h}^2}_{ab} \\
        &\hspace{100pt}+\rho\mc{S}\paren{\ms{g};  (\ms{f}^1)^3}+\rho \mc{S}\paren{\ms{g}; \ms{w}^1, \ms{f}^0}_{ab}+\rho\mc{S}\paren{\ms{g}; \ms{w}^0, \ms{f}^1}_{ab}\text{,}\label{lemma_identity_h_0}\\
        &\notag2\rho\ms{D}_{[b} \ms{h}^2_{|a|c]}+\ms{h}^0_{abc} =\; 2\rho\ms{D}_a \ms{h}^2_{[bc]} - 2\rho^{-1}\ms{g}_{a[b}\ms{f}^0_{c]}+\rho\mc{S}\paren{\ms{g}; \ms{w}^1, {\ms{f}}^1}_{abc}+\rho\mc{S}\paren{\ms{g}; \ms{m}, \ms{h}^0}_{abc} \\
        &\hspace{100pt}+ \rho\mc{S}\paren{\ms{g}; \ms{w}^0, \ms{f}^0}_{abc}+ \rho\mc{S}\paren{\ms{g}; ({\ms{f}}^1)^2, {\ms{f}}^0}_{abc} +\rho\mc{S}\paren{\ms{g}; ({\ms{f}}^0)^3}_{abc}\label{lemma_identity_h_2}\text{.}
    \end{align}
    \end{lemma}
    \begin{proof}
        The proof of this lemma is a simple consequence of some appropriate commutations of derivatives. Typically, starting with \eqref{lemma_identity_h_0}: 
        \begin{align*}
            \ms{g}^{cd}\nabla_a \nabla_c F_{db} = \ms{g}^{cd}\paren{[\nabla_a, \nabla_c]F_{db} + \nabla_c \nabla_a F_{db}}\text{.}
        \end{align*}
        The left-hand side can be written in terms of vertical tensors using Proposition \ref{sec:aads_derivatives_vertical}: 
        \begin{align*}
            \rho^2\ms{g}^{cd}\nabla_a \nabla_c F_{db} = \rho\ms{D}_a\ms{tr}_{\ms{g}}\ms{h}^0_{b} - \ms{h}^1_{ab} - \ms{h}^2_{ab} + \rho\mc{S}\paren{\ms{g}; \ms{m}, \ms{h}^{2}}_{ab}+ \rho\mc{S}\paren{\ms{g}; \ms{m}, \ms{h}^{1}}_{ab} \text{,}
        \end{align*}
        and on the right-hand side: 
        \begin{align*}
            \rho^2\ms{g}^{cd}\nabla_c\nabla_aF_{db} = \rho\ms{g}^{cd}\ms{D}_c\ms{h}^0_{adb} - \ms{h}^1_{ab} - (n-1)\ms{h}^2_{ab}+ \rho\mc{S}\paren{\ms{g}; \ms{m}, \ms{h}^2}_{ab}+ \rho\mc{S}\paren{\ms{g}; \ms{m}, \ms{h}^{1}}_{ab}\text{,}
        \end{align*}
        while the commutator term yields:
        \begin{align*}
            \rho^2\ms{g}^{cd}[\nabla_a, \nabla_c]F_{db} =&\rho^{-1}(n-2)\ms{f}^1_{ab} + \rho \mc{S}\paren{\ms{g}; \ms{f}^1, (\ms{f}^0)^2}_{ab}  + \rho\mc{S}\paren{\ms{g};  (\ms{f}^1)^3}_{ab}\\
            &+\rho\mc{S}(\ms{g};\ms{w}^1, \ms{f}^0)_{ab} + \rho\mc{S}\paren{\ms{g}; \ms{w}^0, \ms{f}^1}_{ab}\text{,}
        \end{align*}
        where we used \eqref{sec:aads_weyl_equation}, therefore proving the identity. The proof for $\ms{h}^2$ is based on the commutation of the following derivatives: 
        \begin{equation*}
            2\rho\nabla_{[b}\nabla_{|a} F_{\rho| c]} = 2\rho[\nabla_{[b}, \nabla_{|a}]F_{\rho|c]} + 2\rho\nabla_a\nabla_{[b}F_{|\rho|c]}\text{,} 
        \end{equation*}
        where each of the terms can be written as: 
        \begin{align*}
            &2\rho\nabla_{[b}\nabla_{|a} F_{\rho| c]} = 2\ms{D}_{[b}\ms{h}^2_{|a|c]} + 2\rho^{-1}\ms{h}^0_{a[bc]} +2\rho^{-1}\ms{g}_{a[b}\ms{tr}_{\ms{g}}\ms{h}^0_{c]} + \mc{S}\paren{\ms{g}; \ms{m}, \ms{h}^0}_{abc}\text{,}\\
            &\begin{aligned}2\rho[\nabla_{[b}, \nabla_{|a}]F_{\rho|c]} =& -2\rho^{-2}\ms{g}_{a[b}\ms{f}^0_{c]} + \mc{S}\paren{\ms{g}; (\ms{f}^0)^3}_{abc} + \mc{S}\paren{\ms{g}; (\ms{f}^1)^2, \ms{f}^0}_{abc} \\
            &+\mc{S}\paren{\ms{g}; \ms{w}^1, \ms{f}^1}_{abc} + \mc{S}\paren{\ms{g}; \ms{w}^0, \ms{f}^0}_{abc}\text{,}
            \end{aligned}\\
            &2\rho\nabla_a\nabla_{[b}F_{|\rho|c]} = 2\ms{D}_{a}\ms{h}^2_{[bc]} +2\rho^{-1}\ms{h}^0_{[b|a|c]}+2 \rho^{-1}\ms{g}_{a[b}\ms{tr}_{\ms{g}}\ms{h}^0_{c]} + \mc{S}\paren{\ms{g}; \ms{m}, \ms{h}^0}_{abc}\text{.}
        \end{align*}
        Note that, using the Bianchi identity for $F$: 
        \begin{equation*}
            2\ms{h}^0_{[b|a|c]} = \ms{h}^0_{abc}\text{,}
        \end{equation*}
        yielding the desired Equation \eqref{lemma_identity_h_2}. 
    \end{proof}

    Using Lemma \ref{lemma_identity_h}, one immediately obtains the following wave equation for $\ms{h}^0$: 
    \begin{align*}
        \rho\cdot \paren{\Box_g + 3n}H_{abc} = \ol{\Box}\ms{h}^0_{abc} + (n-1)\ms{h}^0_{abc} +12\rho^{-1}\ms{g}_{a[b}\ms{f}^0_{c]} + 4\ms{g}_{a[b}\ms{tr}_{\ms{g}}\ms{h}^0_{c]} +\tilde{\ms{E}}^0_{abc} \text{,}
    \end{align*}
    with $\tilde{\ms{E}}^0$ containing $\ms{E}^0$, as well as the remainders from \eqref{lemma_identity_h_2}. 
    One is however interested in the wave equation for $\ul{\ms{h}}^0$\footnote{Note that one needs to look at the wave equation satisfied by $\mc{S}\paren{ \ms{m}, \ms{f}^0}$. It can easily be checked that such a contribution is contained in the lower-order terms.}: 
    \begin{align*}
        \ol{\Box}{\ul{\ms{h}}^0}_{abc} &= \ol{\Box}\paren{\ms{h}^0_{abc} + 2\rho^{-1}\ms{g}_{a[b}\ms{f}^0_{c]}+ \mc{S}\paren{\ms{m}, \ms{f}^0}_{abc}}\\
        &=-(n-1){\ms{h}}^0_{abc} - 12\rho^{-1}\ms{g}_{a[b}\ms{f}^0_{c]} -4\ms{g}_{a[b}\ms{tr}_{\ms{g}}\ms{h}^0_{c]}+ 2\ms{g}_{a[b}\ol{\Box}(\rho^{-1}\ms{f}^0)_{c]}\\
        &\hspace{50pt}+ \rho\cdot (\Box_g + 3n)H_{bca} + \tilde{\ms{E}}^0_{abc} \\
        &=-(n-1){\ms{h}}^0_{abc} - 12\rho^{-1}\ms{g}_{a[b}\ms{f}^0_{c]} -4\ms{g}_{a[b}\ms{tr}_{\ms{g}}\ms{h}^0_{c]}+ 2\ms{g}_{a[b}\left(-4(n-2)\rho^{-1}\ms{f}^0_{c]} +2\ms{D}\cdot \ms{f}^1_{c]} + (n+1)\rho^{-1}\ms{f}^0_{c]}\right.\\
        &=-(n-1)\paren{\ms{h}^0_{abc} +2\rho^{-1}\ms{g}_{a[b}\ms{f}^0_{c]}}+\rho \cdot (\Box_g + 3n)H_{abc} + \ms{g}_{a[b}(\mc{R}_{W, \ms{f}^0})_{c]} + \tilde{\ms{E}}^0_{abc}\\
        &=-(n-1)\ul{\ms{h}}^0_{abc} +\rho \cdot (\Box_g + 3n)H_{abc} + \ms{g}_{a[b}(\mc{R}_{W, \ms{f}^0})_{c]} + \tilde{\ms{E}}^0_{abc}\text{,}
    \end{align*}
    where we used \eqref{wave_power}, Proposition \ref{h_as_Df} and the wave equation for $\ms{f}^0$ \eqref{wave_f_0}, as well as the transport equation \eqref{transport_f0_FG}. 
   It remains now to compute the lower order terms found in $\rho\cdot (\Box_g + 3n)H_{abc}$, Equation \eqref{wave_higher_derivative}. We will first look at nonlinearities of the form $\nabla W \cdot F$: 
    \begin{align*}
        &-2\nabla^\gamma W^\sigma{}_{[b|\gamma a}F_{\sigma|c]} - \nabla_a W_{bc}{}^{\sigma\lambda}F_{\sigma \lambda} \\
        &= -2\rho^4\nabla_\rho W_{\rho [b| \rho a} F_{\rho|c]} -2 \rho^4 \ms{g}^{de}\nabla_\rho W_{d[b|\rho a}F_{e|c]} -2 \rho^4 \ms{g}^{de}\ms{g}^{fh} \nabla_d W_{f[b|ea}F_{h|c]}\\
        &\hspace{40pt}-2 \rho^4 \ms{g}^{de}\nabla_d W_{\rho[b|e a}F_{\rho|c]}-2\rho^4\ms{g}^{de}\nabla_a W_{bc\rho d}F_{\rho e}-\rho^4 \ms{g}^{de}\ms{g}^{fh}\nabla_a W_{bcdf}F_{eh}\\
        &=\rho\mc{S}\paren{\ms{g}; \ms{Dw}^1, \ms{f}^0}_{abc} + \rho \mc{S}\paren{\ms{g}; \ms{D\ms{w}^2}, \ms{f}^1}_{abc}+\rho \mc{S}\paren{\ms{g}; \ms{Dw}^0, \ms{f}^1}_{abc}+\mc{S}\paren{\ms{g}; \ms{w}^0, \ms{f}^0}_{abc}\\
        &\hspace{40pt}+\mc{S}\paren{\ms{g}; \ms{w}^1, \ms{f}^1}_{abc}+ \mc{S}\paren{\ms{g}; \ms{w}^2, \ms{f}^0}_{abc} + \rho\mc{S}\paren{\ms{g}; \ms{m}, \ms{w}^0, \ms{f}^0}_{abc} + \rho \mc{S}\paren{\ms{g}; \ms{m}, \ms{w}^2, \ms{f}^0}_{abc}\\
        &\hspace{40pt}+\rho \mc{S}\paren{\ms{g}; \ms{m}, \ms{w}^1, \ms{f}^1}_{abc} + \rho \mc{S}\paren{\ms{g}; \ul{\ms{h}}^0, (\ms{f}^1)^2}_{abc} \\
        &\hspace{40pt}+\rho\mc{S}\paren{\ms{g}; \ul{\ms{h}}^0, (\ms{f}^0)^2}_{abc} + \rho\mc{S}\paren{\ms{g}; \ul{\ms{h}}^{2}, \ms{f}^1, \ms{f}^0}_{abc}+\mc{S}\paren{\ms{g};(\ms{f}^1)^2,\ms{f}^0}_{abc} +\rho\mc{S}\paren{\ms{g}; \ms{m}, (\ms{f}^1)^2, \ms{f}^0}_{abc}\\
        &\hspace{40pt} + \mc{S}\paren{\ms{g}; (\ms{f}^0)^3}_{abc} +\rho\mc{S}\paren{\ms{g}; \ms{m},(\ms{f}^0)^3}_{abc}\text{.}
    \end{align*}

    The other lower-order terms, which do not contain derivatives of $W$ or $H$, can be obtained simply by exhausting the possible terms that would lead to a vertical field with an odd vertical rank. More precisely, a term of the form $(W\cdot H)_{abc}$ can only be obtained through two spacetime contractions, yielding a power of $4$ from the metric. This will give: 
    \begin{align*}
        (W\cdot H)_{odd} &\sim \rho^4 \paren{W_{even} H_{odd} + W_{odd}H_{even}} \\
        &= \rho \mc{S}\paren{\ms{g}; \ms{w}^0, \ms{h}^0} + \rho \mc{S}\paren{\ms{g}; \ms{w}^2, \ms{h}^0} + \rho \mc{S}\paren{\ms{g}; \ms{w}^1, \ms{h}^2} + \rho \mc{S}\paren{\ms{g}; \ms{w}^0, \ms{h}^3} \\
        &\hspace{30pt}+ \rho \mc{S}\paren{\ms{g}; \ms{w}^2, \ms{h}^3} + \rho \mc{S}\paren{\ms{g}; \ms{w}^1, \ms{h}^1} \\
        &=\rho \mc{S}\paren{\ms{g}; \ms{w}^0, \ul{\ms{h}}^0} + \mc{S}\paren{\ms{g}; \ms{w}^0, \ms{f}^0}+\rho\mc{S}\paren{\ms{g}; \ms{w}^1, \ms{h}^{2}}+\rho\mc{S}\paren{\ms{g}; \ms{m}, \ms{w}^0, \ms{f}^1}\\
        &\hspace{30pt}+\rho\mc{S}\paren{\ms{g};\ms{m}, \ms{w}^1, \ul{\ms{h}}^2}+\rho \mc{S}\paren{\ms{g}; \ms{w}^2, \ul{\ms{h}}^0}+ \mc{S}\paren{\ms{g}; \ms{w}^2, \ms{f}^0}+ \mc{S}\paren{\ms{g}; \ms{w}^1, \ms{f}^1} \\
        &\hspace{30pt}+\rho\mc{S}\paren{\ms{g}; \ms{m}, \ms{w}^2, \ms{f}^1}\text{,}
    \end{align*} 
    where for $A$ a spacetime tensor, we denote by $A_{even}$ and $A_{odd}$ components of this tensor with an even and odd number of vertical indices, respectively.
    Similarly,\footnote{In principle, one should consider $\hat{T}$ and $\tilde{T}$ but these will lead to identical lower-order terms.} 
    \begin{align*}
        (T\cdot H)_{odd} + (\nabla T \cdot F)_{odd} \sim& \rho^4((F\cdot F)_{even} H_{odd}+(F\cdot F)_{odd} H_{even})_{abc}\\
        =&\rho\mc{S}\paren{\ms{g}; \ul{\ms{h}}^0, (\ms{f}^0)^2}_{abc} + \mc{S}\paren{\ms{g}; (\ms{f}^0)^3}+\rho \mc{S}\paren{\ms{g}; \ul{\ms{h}}^{2}, \ms{f}^1, \ms{f}^0}+\rho\mc{S}\paren{\ms{g}; \ms{m}, (\ms{f}^0)^3}\\
        &+\rho\mc{S}\paren{\ms{g}; \ul{\ms{h}}^0, (\ms{f}^1)^2}  + \mc{S}\paren{\ms{g}; \ms{f}^0, (\ms{f}^1)^2} + \rho\mc{S}\paren{\ms{g}; \ms{m}, (\ms{f}^1)^2, \ms{f}^0}
    \end{align*}
    
    We can obtain the wave equation for $\ms{h}^2$ in a similar fashion. From Lemma \ref{lemma_identity_h}, the left-hand side of \eqref{intermediate_h_2} becomes: 
    \begin{align}\label{intermediate_wave_h0}
        \rho\cdot (\Box_g + 3n)H_{a\rho b}=\ol{\Box}\ms{h}^2_{ab} +2(n-2)\ms{h}^2_{ab} + 4\ms{h}^2_{[ab]} -2n\rho^{-1}\ms{f}^1_{ab} +\tilde{\ms{E}}^2_{ab}\text{,}
    \end{align}
    where $\tilde{\ms{E}}^2_{ab}$ contains $\ms{E}^2_{ab}$ as well as the remainder terms in \eqref{lemma_identity_h_0}.
    Now, once again, the goal is to get rid of the top-order terms. This can be done by looking at the wave equation satisfied by $\ul{\ms{h}}^2_{ab}$: 
    \begin{align*}
        \ol{\Box}\ul{\ms{h}}^{2}_{ab}&=\ol{\Box}\paren{\ms{h}^2_{ab} -\rho^{-1}\ms{f}^1_{ab}+\rho\mc{S}\paren{\ms{g}; \ms{m}, \ms{f}^1}_{ab}}\\
        &= -2(n-2)\ms{h}^{2}_{ab}-4\ms{h}^2_{[ab]}+2n\rho^{-1}\ms{f}^1_{ab} - \paren{\rho^{-1}\overline{\Box}\ms{f}^1_{ab} - 2\ol{\ms{D}}_\rho\ms{f}^1_{ab}+(n+1)\rho^{-1}\ms{f}^1_{ab}}\\
        &\hspace{30pt} +\tilde{\ms{E}}^2_{ab} + \rho\cdot (\Box_g + 3n)H_{[a|\rho|b]}\\
        &=-2(n-2)\paren{\ul{\ms{h}}^{2}_{ab}+\rho^{-1}\ms{f}^1_{ab}+\mc{S}\paren{\ms{g}; \ms{m}, \ms{f}^1}_{ab}}+2n\rho^{-1}\ms{f}^1_{ab}-\left(-\rho^{-1}(n-1)\ms{f}^1_{ab}-2(\rho^{-1}\ms{f}^1_{ab}+2\ul{\ms{h}}^2_{[ab]})\right.\\
        &\hspace{30pt}\left.+(n+1)\rho^{-1}\ms{f}^1_{ab}\right)+\tilde{\ms{E}}^2_{ab} + \rho\cdot (\Box_g + 3n)H_{[a|\rho|b]} + (\mc{R}_{W, \ms{f}^1})_{ab}\\
        &=-2(n-2)\ul{\ms{h}}^{2}_{ab}+\tilde{\ms{E}}^2_{ab} + \rho\cdot (\Box_g + 3n)H_{[a|\rho|b]} + (\mc{R}_{W, \ms{f}^1})_{ab}\text{,}
    \end{align*}
    up to lower-order terms, and where we used \eqref{wave_power}, Proposition \ref{expression_h_Df} and the wave and transport equations for $\ms{f}^1$ as given in Equations \eqref{wave_f_1} and \eqref{transport_f1_FG}. 

    The remainder terms can be obtained in a similar manner as $\ms{h}^0$. Nonlinearities of the form $\nabla W \cdot F$ yield, using \eqref{transport_w_2}: 
    \begin{align*}
        &-2\nabla^\gamma W^\sigma{}_{[\rho|\gamma a}F_{\sigma|b]} - \nabla_aW_{\rho b}{}^{\sigma\lambda}F_{\sigma\lambda}\\
        &= \rho^{-1} \mc{S}\paren{\ms{g}; \ol{\ms{D}}_\rho (\rho^2\ms{w}^2), \ms{f}^1}_{ab} + \rho\mc{S}\paren{\ms{g}; {\ms{D}}\ms{w}^2, \ms{f}^0}_{ab} + \rho \mc{S}\paren{\ms{g}; \ms{Dw}^1, \ms{f}^1}_{ab} \\
        &\hspace{30pt}+ \mc{S}\paren{\ms{g}; \ms{w}^0, \ms{f}^1}_{ab} + \mc{S}\paren{\ms{g}; \ms{w}^2, \ms{f}^1}_{ab} + \rho\mc{S}\paren{\ms{g}; \ms{Dw}^0, \ms{f}^0}_{ab} + \mc{S}\paren{\ms{g}; \ms{w}^1, \ms{f}^0}_{ab}\\
        &\hspace{30pt}+ \rho\mc{S}\paren{\ms{g}; \ms{m}, \ms{w}^1, \ms{f}^0}_{ab} + \rho\mc{S}\paren{\ms{g}; \ms{m}, \ms{w}^0, \ms{f}^1}_{ab}+\rho\mc{S}\paren{\ms{g}; \ms{m}, \ms{w}^2, \ms{f}^1}_{ab}\\
        &=\mc{S}\paren{\ms{g}; \ms{w}^2, \ms{f}^1}_{ab} + \rho\mc{S}\paren{\ms{g}; \ms{Dw}^1, \ms{f}^1}_{ab} + \rho\mc{S}\paren{\ms{g};\ms{m}, \ms{w}^0, \ms{f}^1}_{ab}+\rho\mc{S}\paren{\ms{g}; \ms{m}, (\ms{f}^0)^2, \ms{f}^1}_{ab} \\
        &\hspace{30pt} + \rho\mc{S}\paren{\ms{g};\ms{m}, \ms{w}^2, \ms{f}^1}_{ab} + \rho \mc{S}\paren{\ms{g}; \ul{\ms{h}}^0, \ms{f}^0, \ms{f}^1}_{ab} + \mc{S}\paren{\ms{g};  (\ms{f}^1)^3)}_{ab}+\rho\mc{S}\paren{\ms{g}; \ms{m}, (\ms{f}^1)^3}_{ab}\\
        &\hspace{30pt} + \rho\mc{S}\paren{\ms{g}; \ul{\ms{h}}^{2}, (\ms{f}^1)^2}_{ab} + \mc{S}\paren{\ms{g}; (\ms{f}^0)^2, \ms{f}^1}_{ab}+\rho\mc{S}\paren{\ms{g}; \ms{Dw}^2, \ms{f}^0}_{ab}\\
        &\hspace{30pt}+ \mc{S}\paren{\ms{g}; \ms{w}^0, \ms{f}^1}_{ab}  + \rho\mc{S}\paren{\ms{g}; \ms{Dw}^0, \ms{f}^0}_{ab} + \mc{S}\paren{\ms{g}; \ms{w}^1, \ms{f}^0}_{ab}\\
        &\hspace{30pt} +\rho\mc{S}\paren{\ms{g}; \ms{m}, \ms{w}^1, \ms{f}^0}_{ab} +  \rho\mc{S}\paren{\ms{g};\ul{\ms{h}}^{2}, (\ms{f}^0)^2}_{ab}\text{.}
    \end{align*}
    Similarly, one gets schematically for the other lower-order terms: 
    \begin{align*}
        &\begin{aligned}
        (W\cdot H)_{even} 
        &\sim \rho^4(W_{even} H_{even}+W_{odd}H_{odd})\\
        &=\rho \mc{S}\paren{\ms{g}; \ms{w}^0, \ms{h}^2}_{ab} + \rho\mc{S}\paren{\ms{g}; \ms{w}^0, \ms{h}^1}_{ab} + \rho \mc{S}\paren{\ms{g}; \ms{w}^2, \ms{h}^2}_{ab}+\rho\mc{S}\paren{\ms{g}; \ms{w}^2, \ms{h}^1}_{ab}\\
        &\hspace{30pt} + \rho \mc{S}\paren{\ms{g}; \ms{w}^1, \ms{h}^0}_{ab}+\rho\mc{S}\paren{\ms{g}; \ms{w}^1, \ms{h}^3}_{ab}\\
        &=\rho\mc{S}\paren{\ms{g}; \ms{w}^0, \ul{\ms{h}}^{2}} +\mc{S}\paren{\ms{g}; \ms{w}^0, \ms{f}^1}_{ab} +\rho\mc{S}\paren{\ms{g}; \ms{w}^2, \ul{\ms{h}}^{2}}_{ab}+\mc{S}\paren{\ms{g}; \ms{w}^2, \ms{f}^1}_{ab}\\
        &\hspace{30pt} + \rho\mc{S}\paren{\ms{g}; \ms{m}, \ms{w}^0, \ms{f}^1}+\rho\mc{S}\paren{\ms{g}; \ms{m}, \ms{w}^2, \ms{f}^1}++ \rho\mc{S}\paren{\ms{g}; \ms{m}, \ms{w}^1, \ms{f}^0}\\
        &\hspace{30pt}+\rho\mc{S}\paren{\ms{g}; \ms{w}^1, \ul{\ms{h}}^0} + \mc{S}\paren{\ms{g}; \ms{w}^1, \ms{f}^0}\text{,}
        \end{aligned}\\
        &\begin{aligned}
            (T\cdot H)_{even}+(\nabla T\cdot H)_{even} &\sim \rho^4 ((F\cdot F)_{even}H_{even} + (F\cdot F)_{odd}H_{odd}) \\
            &=\rho \mc{S}\paren{\ms{g}; \ul{\ms{h}}^{2}, (\ms{f}^0)^2} + \rho\mc{S}\paren{\ms{g}; \ul{\ms{h}}^{2}, (\ms{f}^1)^2} + \mc{S}\paren{\ms{g}; (\ms{f}^1)^3} \\
            &\hspace{30pt} + \rho\mc{S}\paren{\ms{g}; \ms{m}, \ms{f}^1, (\ms{f}^0)^2}+ \rho\mc{S}\paren{\ms{g}; \ms{m}, (\ms{f}^1)^3} \\
            &\hspace{30pt}+ \mc{S}\paren{\ms{g}; \ms{f}^1,(\ms{f}^0)^2} + \rho\mc{S}\paren{\ms{g}; \ul{\ms{h}}^0, \ms{f}^1, \ms{f}^0} \text{,}
        \end{aligned}
    \end{align*}
    concluding the proof.

\subsection{Proof of Proposition \ref{prop_wave_w}}\label{app:prop_wave_w}

    Here, we will describe how one can obtain a system of the form: 
    \begin{equation*}
        (\ol{\Box}+ \sigma_i)\ms{w}^i = \mc{R}_{W, \ms{w}^i}\text{,}\qquad i\in \lbrace \star, 1, 2\rbrace\text{,}
    \end{equation*}
    where the right-hand side will only contain lower-order contributions, and $(\sigma_i)$ some constants depending on $i$ and $n$. In this proof, we will emphasise how to obtain the mass and will informally justify the lower-order terms.

    First of all, from Propositions \ref{sec:aads_derivatives_vertical} and \ref{wave_weyl}, the following decompositions hold: 
    \begin{align}
        \label{intermediate_wave_w_2}(\Box_g W)_{\rho a\rho b} =& \rho^{-4}\ol{\Box}(\rho^2\ms{w}^2_{ab}) + 2\rho^{-1}\ms{g}^{cd}\ms{D}_c\ms{w}^1_{bda} + 2\rho^{-1}\ms{g}^{cd}\ms{D}_c\ms{w}^1_{adb} -2(n+1)\rho^{-2}\ms{w}^2_{ab}\\
        &\notag+4\rho^{-2}\ms{w}^2_{ab}+2\rho^{-2}\ms{g}^{cd}\ms{w}^0_{acbd} + \rho^{-1}\mc{S}\paren{\ms{g};\ms{m}, \ms{w}^2}_{ab} + \mc{S}\paren{\ms{g}; \ms{m}^2, \ms{w}^2}_{ab} \\
        &\notag+\mc{S}\paren{\ms{g}; \ms{Dm}, \ms{w}^1}_{ab} + \mc{S}\paren{\ms{g}; \ms{m}, \ms{Dw}^1}_{ab}+\rho^{-1}\mc{S}\paren{\ms{g}; \ms{m}, \ms{w}^0}_{ab} + \mc{S}\paren{\ms{g}; \ms{m}^2, \ms{w}^0}_{ab}\text{,}
    \end{align}
    \begin{align}
        \notag(\Box_g W)_{\rho abc} =& \rho^{-4}\ol{\Box}(\rho^2\ms{w}^1_{abc}) + 2\rho^{-1}\ms{g}^{de}\ms{D}_d \ms{w}^0_{eabc} - 2\rho^{-1}\ms{D}_b \ms{w}^2_{ac} + 2\rho^{-1}\ms{D}_c\ms{w}^2_{ab} - (n+3)\rho^{-2}\ms{w}^1_{abc}\\
        &\notag+2\rho^{-2}\underbrace{(\ms{w}^1_{abc} + \ms{w}^1_{cab} + \ms{w}^1_{bca})}_{=0} + \mc{S}\paren{\ms{g}; \ms{m}, \ms{Dw}^0}_{abc} + \mc{S}\paren{\ms{g}; \ms{m}, \ms{Dw}^2}_{abc}\\
        &\label{intermediate_wave_w_1}+ \mc{S}\paren{\ms{g}; \ms{Dm}, \ms{w}^0}_{abc} + \mc{S}\paren{\ms{g}; \ms{Dm}, \ms{w}^2}_{abc}+\rho^{-1}\mc{S}\paren{\ms{g}; \ms{m}, \ms{w}^1}_{abc} + \mc{S}\paren{\ms{g}; \ms{m}^2, \ms{w}^1}_{abc}\text{,}
    \end{align}
    \begin{align}
        \notag(\Box_g W)_{abcd} =& \rho^{-4}\ol{\Box}(\rho^2\ms{w}^0_{abcd}) -2\rho^{-1}\ms{D}_a \ms{w}^1_{bcd} + 2\rho^{-1}\ms{D}_b \ms{w}^1_{acd}-2\rho^{-1}\ms{D}_c\ms{w}^1_{dab} + 2\rho^{-1}\ms{D}_d\ms{w}^1_{cab} \\
        &\notag -4\rho^{-2}\ms{w}^0_{abcd} + 2\rho^{-2}(\ms{g}\star\ms{w}^2)_{abcd} + \mc{S}\paren{\ms{g}; \ms{m}, \ms{Dw}^1}_{abcd} + \mc{S}\paren{\ms{g}; \ms{Dm}, \ms{w}^1}_{abcd}\\
        &+\rho^{-1}\mc{S}\paren{\ms{g}; \ms{m}, \ms{w}^0}_{abcd} + \mc{S}\paren{\ms{g}; \ms{m}^2, \ms{w}^0}_{abcd} + \rho^{-1}\mc{S}\paren{\ms{g};\ms{m}, \ms{w}^2}_{abcd} + \mc{S}\paren{\ms{g}; \ms{m}^2, \ms{w}^2}_{abcd}\label{intermediate_wave_w_0}\text{.}
    \end{align}
    Now, using \eqref{wave_power}, one obtains: 
    \begin{align*}
        \ol{\Box}(\rho^2 \ms{w}^i_{\bar{a}}) = \rho^2\ol{\Box}\ms{w}^i_{\bar{a}} + 4\rho^3 \ol{\ms{D}}_\rho \ms{w}^i - 2(n-2)\rho^2\ms{w}^i + \rho^3\mc{S}\paren{\ms{g}; \ms{m}, \ms{w}^i}_{\bar{a}}\text{,}\qquad i=0,1,2\text{,}
    \end{align*}
    which gives, for $\ms{w}^2, \, \ms{w}^1$, using the transport equations from Proposition \ref{prop_transport_w}:
    \begin{align*}
        &\rho^{-2}\ol{\Box}(\rho^2\ms{w}^2)_{ab} = \ol{\Box}\ms{w}^2_{ab} + 2(n-2)\ms{w}^2_{ab}-4\rho\cdot\ms{g}^{cd}\ms{D}_c \ms{w}^1_{adb}+(\mc{R}_{T, \ms{w}^2})_{ab}\text{,}\\
        &\rho^{-2}\ol{\Box}(\rho^2\ms{w}^1)_{abc} = \ol{\Box}\ms{w}^1_{abc} \underbrace{- 2\rho(\ms{D}\cdot \ms{w}^0)_{abc} + 2(n-2)\ms{w}^1_{abc}}_{\eqref{transport_w_1.1}} + \underbrace{4\rho \ms{D}_{[b}\ms{w}^2_{c]a} + 2\ms{w}^1_{abc}}_{\eqref{transport_w_1.2}}-2(n-2)\ms{w}^1_{abc}+(\mc{R}_{T,\ms{w}^1})_{abc}\text{.}
    \end{align*}
    Note that for $\ms{w}^1$, we had to use both \eqref{transport_w_1.1} and \eqref{transport_w_1.2}.
    This yields, using Proposition \ref{wave_weyl}:
    \begin{align*}
        &\rho^2(\Box_g + 2n)W_{\rho a \rho b} =  \ol{\Box}\ms{w}^2_{ab} + 2(n-2)\ms{w}^2_{ab} + \ms{E}^2_{ab}\text{,}\\
        &\rho^2(\Box_g + 2n)W_{\rho a b c} = \ol{\Box}\ms{w}^1_{abc} + (n-1)\ms{w}^1_{ab}+ \ms{E}^1_{abc}\text{,} 
    \end{align*}
    where $\ms{E}^2$, $\ms{E}^1$ are $(0,2)$ and $(0,3)$ vertical tensors containing remainder terms found in \eqref{intermediate_wave_w_2}, \eqref{intermediate_wave_w_1}, \eqref{transport_w_2} and \eqref{transport_w_1.1}:
    \begin{align*}
        &\begin{aligned}
            \ms{E}^2 = &  \rho^2\mc{S}\paren{\ms{g}; \ms{m}^2, \ms{w}^2} 
            +\rho^2\mc{S}\paren{\ms{g}; \ms{Dm}, \ms{w}^1} + \rho^2\mc{S}\paren{\ms{g}; \ms{m}, \ms{Dw}^1}_{ab}+ \rho^2\mc{S}\paren{\ms{g}; \ms{m}^2, \ms{w}^0} + \mc{R}_{T, \ms{w}^2}
        \end{aligned}\\
        &\begin{aligned}
            \ms{E}^1 =& \rho^2\mc{S}\paren{\ms{g}; \ms{m}, \ms{Dw}^0} + \rho^2\mc{S}\paren{\ms{g}; \ms{m}, \ms{Dw}^2}+ \rho^2\mc{S}\paren{\ms{g}; \ms{Dm}, \ms{w}^0} + \rho^2\mc{S}\paren{\ms{g}; \ms{Dm}, \ms{w}^2} + \\
            &\rho^2\mc{S}\paren{\ms{g}; \ms{m}^2, \ms{w}^1}+\mc{R}_{T, \ms{w}^1}\text{.}
        \end{aligned}
    \end{align*}
    We also used the fact that anti-symmetrising Equation \eqref{transport_w_2} yields the following: 
    \begin{equation*}
        \rho\ms{g}^{cd}\ms{D}_c \ms{w}^1_{[a|d|b]} = (\mc{R}_{T, \ms{w}^2})_{ab}\text{.}
    \end{equation*}
    
Looking at $\ms{w}^0$, however, one has, using transport equations \eqref{transport_w_2} and \eqref{transport_w_star}: 
    \begin{align*}
        \rho^{-2}\ol{\Box}(\rho^2\ms{w}^0)_{abcd} = \ol{\Box}\ms{w}^0 + 4\rho \left(2\ms{D}_{[a}\ms{w}^1_{b]cd} -(\ms{g}\star\ms{w}^2)_{abcd}\right) - 2(n-2)\ms{w}^0_{abcd}+(\mc{R}_{T,\ms{w}^2})_{abcd}\text{,}
    \end{align*}
    yielding eventually:
\begin{equation}\label{intermediate_box_w_0}
    \rho^2(\Box_g+2n) W_{abcd} = \ol{\Box}\ms{w}^0_{abcd} +{4\rho\ms{D}_{[a}\ms{w}^1_{b]cd}-4\rho\ms{D}_{[c}\ms{w}^1_{d]ab}} - 2(\ms{g}\star \ms{w}^2)_{abcd} + \ms{E}^0_{abcd} \text{,}
\end{equation}
where $\ms{E}$ is $(0,4)$--tensor containing the remainder terms considered above, as well as those coming from the transport equation \eqref{transport_w_star}:
\begin{align*}
    \ms{E}^0 = &\mc{R}_{T,\ms{w}^2} + \rho^2 \mc{S}\paren{\ms{g}; \ms{m}, \ms{Dw}^1} + \rho^2\mc{S}\paren{\ms{g}; \ms{Dm}, \ms{w}^1}+ \rho^2\mc{S}\paren{\ms{g}; \ms{m}^2, \ms{w}^0}+ \rho^2\mc{S}\paren{\ms{g}; \ms{m}^2, \ms{w}^2}\text{.} 
\end{align*}
Observe now that, from \eqref{transport_w_star} and the symmetries of $\ms{w}^0$, the following holds: 
\begin{align}
    0 &=\notag \ol{\ms{D}}_{\rho}\ms{w}^0_{abcd} - \ol{\ms{D}}_{\rho}\ms{w}^0_{cdab}\\
    &=2\rho\ms{D}_{[a}\ms{w}^1_{b]cd} - 2\rho\ms{D}_{[c}\ms{w}^1_{d]ab} + (\mc{R}_{T, \ms{w}^2})_{abcd}\text{.}\label{intermediate_identity_w_0}
\end{align}
Using \eqref{intermediate_identity_w_0} in \eqref{intermediate_box_w_0} yields: 
\begin{equation*}
     \rho^2(\Box_g+2n) W_{abcd} = \ol{\Box}\ms{w}^0_{abcd} - 2(\ms{g}\star \ms{w}^2)_{abcd} + \ms{E}^0_{abcd}\text{,}
\end{equation*}
since $\mc{R}_{T,\ms{w}^2}$ is contained in $\ms{E}^0$. 

The term $-2\rho^{-2}(\ms{g}\star\ms{w}^2)$ cannot be considered as a lower-order term in the Carleman estimates and therefore has to be treated. The solution is to consider instead the vertical tensor $\ms{w}^\star$, corresponding to the traceless part of $\ms{w}^0$: 
\begin{align*}
    \ol{\Box}\ms{w}^\star_{abcd} &= \ol{\Box}\paren{\ms{w}^0 + \frac{1}{n-2}(\ms{g}\star \ms{w}^2)}_{abcd} \\
    &=+ 2(\ms{g}\star\ms{w}^2)_{abcd}+\frac{1}{n-2}\paren{\ms{g}\star (-2(n-2)\ms{w}^2)}_{abcd} - \ms{E}^0_{abcd} \\
    &\quad + \rho^2(\Box_g+2n) W_{abcd} +  (\mc{R}_{W, \ms{w}^2})_{abcd}\\
    &=-\ms{E}^0_{abcd} + \rho^2(\Box_g+2n) W_{abcd} +  (\mc{R}_{W, \ms{w}^2})_{abcd}\text{,}
\end{align*}
where $\mc{R}_{W, \ms{w}^2}$ corresponds the lower-order terms for the wave equation of $\ms{w}^2$, still to be justified\footnote{Observe that, strictly speaking, the remainder terms for the wave equation of $\ms{w}^2$ form a $(0,2)$--vertical tensor. However, since multiplying by the metric does not change the form of the lower order terms, see Definition \ref{def_remainder}, we will use this abuse of notation and consider $\ms{R}_{W, \ms{w}^2}$ and $\ms{g}\star \ms{R}_{W,\ms{w}^2}$ to be identical.}. 

As we will show below, the source terms of the wave equations from Proposition \ref{prop_wave_w} will yield lower-order terms only. As a consequence, we define: 
\begin{gather*}
    \ms{Q}^0_{abcd} := \rho^2(\Box_g + 2n)W_{abcd}\text{,}\qquad \ms{Q}^1_{abc} := \rho^2(\Box_g +2n)W_{\rho abc}\text{,}\qquad \ms{Q}^2_{ab} = \rho^2(\Box_g +2n)W_{\rho a\rho b}\text{.}
\end{gather*}

From these definitions, we finally obtain the following vertical wave equations: 
\begin{align*}
    &(\ol{\Box}+2(n-2))\ms{w}^2_{ab} = - \ms{E}^2_{ab} + \ms{Q}^2_{ab}\text{,}\\
    &(\ol{\Box}+(n-1))\ms{w}^1_{abc} = -\ms{E}^1_{abc} + \ms{Q}^1_{abc} \text{,}\\
    &\ol{\Box}\ms{w}^\star_{abcd} = -\ms{E}^0_{abcd} + \ms{Q}^0_{abcd} + (\mc{R}_{W, \ms{w}^2})_{abcd}\text{.}
\end{align*}

It remains now to justify the remainder terms $\ms{Q}^0$, $\ms{Q}^1$ and $\ms{Q}^2$. Since only the parity of the number of vertical components matters, we will proceed as in the proof of Proposition \ref{prop_wave_h}. As a reminder, for any spacetime tensor $A$, we will denote by $A_{even}$ and $A_{odd}$ components with an even and odd number of vertical indices, respectively. 

Starting with the higher derivative terms, one has: 
\begin{align*}
    &(\nabla^2T)_{even} \sim \rho^2(H_{even}H_{even} + H_{odd}H_{odd} + (\nabla H)_{even}F_{even} + (\nabla H)_{odd}F_{odd})_{abcd}\text{,}\\
    &(\nabla^2 T)_{odd} \sim \rho^2(H_{odd}H_{even} + (\nabla H)_{even}F_{odd} + (\nabla H)_{odd}F_{even})\text{.}
\end{align*}
Terms of the form $\nabla H$ will require extra-care, since one may need to use Proposition \ref{prop_transport_h}. Let us exhaust the possible terms: 
\begin{align}
    &\begin{aligned}\label{he}
        \rho (\nabla H)_{even} =& \rho^{-2}\ol{\ms{D}}_\rho(\rho^2\ms{h}^1) + \rho^{-2}\ol{\ms{D}}_{\rho}(\rho^2\ms{h}^2) + \mc{S}(\ms{g}; \ms{D}\ms{h}^0) + \mc{S}\paren{\ms{g}; \ms{D}\ms{h}^3}+ \rho^{-1}\mc{S}(\ms{g}; \ms{h}^1)\\
        &+ \rho^{-1}\mc{S}(\ms{g}; \ms{h}^2)  + \mc{S}\paren{\ms{g}; \ms{m}, \ms{h}^1} + \mc{S}\paren{\ms{g}; \ms{m}, \ms{h}^2}\\
        =&\rho^{-1}\mc{R}_{T, \ms{h}^2} + \mc{S}\paren{\ms{g}; \ms{D}\ul{\ms{h}}^0}  + \rho^{-1}\mc{S}\paren{\ms{g}; \ul{\ms{h}}^{2}} + \rho^{-2}\mc{S}\paren{\ms{g}; \ms{f}^1}+\mc{S}\paren{\ms{g}; \ms{Dm}, \ms{f}^0}\text{,}
    \end{aligned}\\
    &\begin{aligned}\label{ho}
        \rho (\nabla H)_{odd} =& \rho^{-2}\ol{\ms{D}}_\rho (\rho^2\ms{h}^0) + \rho^{-2}\ol{\ms{D}}_\rho (\rho^2\ms{h}^3) + \mc{S}\paren{\ms{g}; \ms{Dh}^2} + \mc{S}\paren{\ms{g}; \ms{Dh}^1} \\
        &+\rho^{-1}\mc{S}\paren{\ms{g}; \ms{h}^0} + \rho^{-1}\mc{S}\paren{\ms{g};\ms{h}^3} + \mc{S}\paren{\ms{g}; \ms{m}, \ms{h}^0} + \mc{S}\paren{\ms{g}; \ms{m},\ms{h}^3}\\
        =&\rho^{-1}\mc{R}_{T, \ms{h}^0} + \mc{S}\paren{\ms{g}; \ms{D}\ul{\ms{h}}^{2}} +\rho^{-1}\mc{S}\paren{\ms{g}; \ul{\ms{h}}^0} + \mc{S}\paren{\ms{g}; \ms{Dm}, \ms{f}^1} \\
        &+ \rho^{-2}\mc{S}\paren{\ms{g}; \ms{f}^0}+\mc{S}\paren{\ms{g}; \ms{w}^1, \ms{f}^1}\text{,}
    \end{aligned}
\end{align}
where we used the transport equations in Proposition \ref{prop_transport_h}. Eventually, one obtains: 
\begin{align}
    &\begin{aligned}\label{nabla2Te}
        (\nabla^2 T)_{even} =& \mc{S}\paren{\ms{g}; (\ul{\ms{h}}^{2})^2}  +\mc{S}\paren{\ms{g}; (\ul{\ms{h}}^0)^2} \\
        &+\rho^{-1}\mc{S}\paren{\ms{g}; \ul{\ms{h}}^0, \ms{f}^0}+ \rho^{-2}\mc{S}\paren{\ms{g}; (\ms{f}^0)^2} + \rho^{-1}\mc{S}\paren{\ms{g}; \ul{\ms{h}}^{2}, \ms{f}^1} + \rho^{-2}\mc{S}\paren{\ms{g};(\ms{f}^1)^2}\\
        &+\rho^{-1}\mc{S}\paren{\mc{R}_{T, \ms{h}^2}, \ms{f}^1} + \mc{S}\paren{\ms{g}; \ms{D}\ul{\ms{h}}^0, \ms{f}^1} + \rho^{-1}\mc{S}\paren{\mc{R}_{T, \ms{h}^0}, \ms{f}^0} + \mc{S}\paren{\ms{g}; \ms{D}\ul{\ms{h}}^{2}, \ms{f-}^0}
    \end{aligned}\\
    &\begin{aligned}\label{nabla2To}
        (\nabla^2T)_{odd}=& \mc{S}\paren{\ms{g}; \ul{\ms{h}}^0, \ul{\ms{h}}^{2}}+ \rho^{-1}\mc{S}\paren{\ms{g}; \ul{\ms{h}}^0, \ms{f}^1} \\
        & + \rho^{-1}\mc{S}\paren{\ms{g}; \ul{\ms{h}}^{2}, \ms{f}^0} +\rho^{-2}\mc{S}\paren{\ms{g}; \ms{f}^1, \ms{f}^0} + \rho^{-1}\mc{S}\paren{\mc{R}_{T, \ms{h}^2}, \ms{f}^0} + \rho^{-1}\mc{S}\paren{\mc{R}_{T, \ms{h}^0}, \ms{f}^1}\\
        &+\mc{S}\paren{\ms{D}\ul{\ms{h}}^0, \ms{f}^0} + \mc{S}\paren{\ms{g}; \ms{D}\ul{\ms{h}}^{2}, \ms{f}^1}\text{.}
    \end{aligned}
\end{align}

The following terms from  Proposition \ref{wave_weyl}, which do not admit any derivatives, can be easily be computed: 
\begin{align}
    &\begin{aligned}\label{wwe}
        (W\cdot W)_{even} &\sim \rho^4(W_{even}W_{even} + W_{odd}W_{odd}) \\
        &=\mc{S}\paren{\ms{g}; (\ms{w}^2)^2}+ \mc{S}\paren{\ms{g}; (\ms{w}^1)^2} + \mc{S}\paren{\ms{g}; (\ms{w}^0)^2}+\mc{S}\paren{\ms{g}; \ms{w}^0, \ms{w}^2}\text{,}
    \end{aligned}\\
    &\begin{aligned}\label{wwo}
        (W\cdot W)_{odd} &\sim \rho^4(W_{even}W_{odd})\\
        &=  \mc{S}\paren{\ms{g}; \ms{w}^2, \ms{w}^1}+  \mc{S}\paren{\ms{g}; \ms{w}^1, \ms{w}^0}\text{,}
    \end{aligned}\\
    &\begin{aligned}\label{wte}
        (W\cdot T)_{even} \sim& \rho^2\paren{W_{even}T_{even} + W_{odd}T_{odd}}\\
        \sim& \rho^4\paren{W_{even}F_{even}F_{even}+ W_{even}F_{odd}F_{odd} + W_{odd}F_{even}F_{odd}}\\
        =&\mc{S}\paren{\ms{g}; \ms{w}^0, (\ms{f}^1)^2}+\mc{S}\paren{\ms{g}; \ms{w}^2, (\ms{f}^1)^2} + \mc{S}\paren{\ms{g}; \ms{w}^0, (\ms{f}^0)^2} + \mc{S}\paren{\ms{g}; \ms{w}^2, (\ms{f}^0)^2}\\
        &+\mc{S}\paren{\ms{g}; \ms{w}^1, \ms{f}^1, \ms{f}^0}\text{,}
    \end{aligned}\\
    &\begin{aligned}\label{wto}
        (W\cdot T)_{odd} &\sim \rho^2(W_{even}T_{odd} + W_{odd}T_{even})\\
        &\sim \rho^4(W_{even}F_{even}F_{odd} + W_{odd}F_{even}F_{even}+W_{odd}F_{odd}F_{odd})\\
        &=\mc{S}\paren{\ms{g}; \ms{w}^0, \ms{f}^1, \ms{f}^0}+\mc{S}\paren{\ms{g}; \ms{w}^2, \ms{f}^1, \ms{f}^0} + \mc{S}\paren{\ms{g}; \ms{w}^1, (\ms{f}^1)^2} + \mc{S}\paren{\ms{g}; \ms{w}^1, (\ms{f}^0)^2}
    \end{aligned}\\
    &\begin{aligned}\label{gte}
        (g\cdot T)_{even}& \sim \rho^{-2}T_{even}\\
        &\sim F_{odd}F_{odd} + F_{even}F_{even}\\
        &=\rho^{-2}\mc{S}\paren{\ms{g}; (\ms{f}^0)^2} + \rho^{-2}\mc{S}\paren{\ms{g}; (\ms{f}^1)^2}\text{,}
    \end{aligned}\\
    &\begin{aligned}\label{gto}
        (g\cdot T)_{odd}&\sim \rho^{-2}T_{odd}\\
        &\sim F_{even}F_{odd} \\
        &= \rho^{-2}\paren{\ms{g}; \ms{f}^1, \ms{f}^0}\text{,}
    \end{aligned}\\
    &\begin{aligned}\label{tte}
        ( T \cdot T)_{even} &\sim T_{even}T_{even}+T_{odd}T_{odd}\\
        &\sim \rho^4(F_{even}F_{even}+F_{odd}F_{odd})(F_{even}F_{even}+F_{odd}F_{odd})\\
        &=\mc{S}\paren{\ms{g}; (\ms{f}^0)^4} + \mc{S}\paren{\ms{g}; (\ms{f}^1)^4} + \mc{S}\paren{\ms{g}; (\ms{f}^0)^2, (\ms{f}^1)^2}\text{,}
    \end{aligned}\\
    &\begin{aligned}\label{tto}
        (T \cdot T)_{odd} &\sim T_{even}T_{odd}\\
        &\sim \rho^4(F_{even}F_{even}F_{even}F_{odd}+F_{odd}F_{odd}F_{odd}F_{even})\\
        &=\mc{S}\paren{\ms{g}; (\ms{f}^0)^3, \ms{f}^1} + \mc{S}\paren{\ms{g}; (\ms{f}^1)^3, \ms{f}^0}\text{.}
    \end{aligned}
\end{align}
Finally, one obtains $\ms{Q}^2$ and $\ms{Q}^1$ by summing the even and odd contributions, respectively. More precisely:
\begin{itemize}
    \item $\ms{Q}^0$, $\ms{Q}^2$ is obtained by summing \eqref{he}, \eqref{nabla2Te}, \eqref{wwe}, \eqref{wte}, \eqref{gte}, \eqref{tte}. 
    \item $\ms{Q}^1$ is obtained by summing \eqref{ho}, \eqref{nabla2To},  \eqref{wwo}, \eqref{wto}, \eqref{gto},\eqref{tto}. 
\end{itemize}
The contributions from $\ms{Q}^0$ are indeed the same as $\ms{Q}^2$ since only the parity of the number of vertical components mattered.

Finally, comparing \eqref{remainder_wave_w_2} with $-\ms{E}^0+\ms{Q}^0$ and $- \ms{E}^2 + \ms{Q}^2$ and \eqref{remainder_wave_w_2} with $- \ms{E}^1 + \ms{Q}^1$, one has\footnote{Once again, $-\ms{E}^0+\ms{Q}^0$ and $- \ms{E}^2 + \ms{Q}^2$ are vertical tensors of different ranks. They will however yield the same remainder terms, this the reason why we used this abuse of notation.}: 
\begin{equation*}
    \mc{R}_{W, \ms{w}^2} = -\ms{E}^0+\ms{Q}^0\text{,} \qquad \mc{R}_{W, \ms{w}^2} = - \ms{E}^2 + \ms{Q}^2\text{,}\qquad \mc{R}_{W, \ms{w}^1} = - \ms{E}^1 + \ms{Q}^1\text{,}
\end{equation*}
giving indeed the vertical wave equations \eqref{wave_w}. 

\subsection{Proof of Proposition \ref{prop_asymptotics_fields}}\label{app:prop_asymptotics_fields}

The proof follows from Theorem \ref{theorem_main_fg} as well as Corollaries  \ref{prop_expansion_vertical_fields} and \ref{weyl_expansion}. The asymptotics for $\ms{f}^0$ and $\ms{f}^1$ are consequences of \eqref{thm_bounds_f0_f1}, since:
        \begin{align*}
            &\abs{\ms{f}^0}_{M, \varphi} \lesssim \rho\cdot \norm{\rho^{-1}\ms{f}^0}_{M, \varphi}\lesssim \rho\text{,}\\
            &\abs{\ms{f}^1}_{M, \varphi}\lesssim \rho \cdot \norm{\rho^{-1}\ms{f}^1}_{M, \varphi}\lesssim \rho\text{.}
        \end{align*}
        The only non-trivial asymptotics are given by $\Lie_\rho \ms{w}^i$ and $\ul{\ms{h}}^0, \ul{\ms{h}}^{2}$. We will briefly justify them here. 

        The asymptotics for the vertical decomposition of $H$ follows from Definition \ref{def_h_bar}:
        \begin{gather*}
            \ms{D}\ms{f}^1 = \mc{O}_{M-1}(\rho)\text{,}\qquad
            \ms{D}\ms{f}^0 =\mc{O}_{M-1}(\rho)\text{.}
        \end{gather*}

        For $\Lie_\rho\ms{w}^i$, one can simply use the transport equations from Proposition \ref{prop_transport_w}, as well as \eqref{vertical_derivative}: 
        \begin{align*}
            &\Lie_\rho \ms{w}^0_{abcd} = 2\ms{D}_{[a}\ms{w}^1_{b]cd} - \paren{\ms{g}\star \rho^{-1}\ms{w}^2}_{abcd}+\rho^{-1}(\mc{R}_{T, \ms{w}^2})_{abcd} \\
            &\Lie_\rho \ms{w}^1_{abc} = - (\ms{D}\cdot \ms{w}^0)_{abc} + (n-2)\rho^{-1}\ms{w}^1_{abc} + \rho^{-1}(\mc{R}_{T, \ms{w}^1})_{abc}\text{,}\\
            &\Lie_\rho \ms{w}^2_{ab} = -\ms{g}^{cd}\ms{D}_c \ms{w}^1_{adb} + (n-2)\rho^{-1}\ms{w}^2_{ab}+\rho^{-1}(\mc{R}_{T, \ms{w}^2})_{ab} \text{.}
         \end{align*}
         Now, observe that the right-hand sides satisfy the following asymptotics: 
         \begin{gather*}
             \begin{aligned}
                 &\ms{D}\ms{w}^1 = \mc{O}_{M-3}(1)\text{,}\qquad &&\ms{g}\star \rho^{-1}\ms{w}^2 = \mc{O}_{M-3}(1)\text{,}\qquad &&&\rho^{-1}\mc{R}_{T, \ms{w}^2} = \mc{O}_{M-2}(\rho)\text{,}\\
                 &\ms{D}\ms{w}^0 = \mc{O}_{M-3}(1)\text{,}\qquad &&\rho^{-1}\ms{w}^1 = \mc{O}_{M-3}(1)\text{,}\qquad &&&\rho^{-1}\mc{R}_{T, \ms{w}^1} = \mc{O}_{M-2}(\rho)\text{,}\\
                 &\ms{D}\ms{w}^1 = \mc{O}_{M-3}(1)\text{,}\qquad &&\rho^{-1}\ms{w}^2 = \mc{O}_{M-3}(1)\text{.}
             \end{aligned}
         \end{gather*}
         This implies $\Lie_\rho \ms{w}^i = \mc{O}_{M-3}(1)$, for $i=0,1,2$. In order to obtain the correct asymptotics for $\Lie_\rho\ms{w}^\star$, one can use Definition \ref{def_w_star}:
         \begin{equation*}
             \Lie_\rho \ms{w}^\star = \Lie_\rho\ms{w}^0 + \frac{1}{n-2}\ms{g}\star \Lie_\rho \ms{w}^2 + \mc{S}\paren{\ms{g}; \ms{m}, \ms{w}^2} = \mc{O}_{M-3}(1)\text{.}
         \end{equation*}
         Finally, using Proposition \ref{prop_transport_h_bar}, one has: 
         \begin{align*}
             &\Lie_\rho\ul{\ms{h}}^0 = \rho^{-1}\ul{\ms{h}}^0 + \mc{S}\paren{\ms{g}; \ms{D}\ul{\ms{h}}^{2}} + \mc{S}\paren{\ms{g}, \ms{m}, \ul{\ms{h}}^0} + \mc{S}\paren{\ms{g}; \ms{Dm}, \ms{f}^1}\text{,}\\
             &\Lie_\rho \ul{\ms{h}}^2 = (n-2)\rho^{-1}\ul{\ms{h}}^2 + \mc{S}(\ms{g};\ms{D}\ul{\ms{h}}^0) + \mc{S}\paren{\ms{g}; \ms{m}, \ul{\ms{h}}^2} + \mc{S}\paren{\ms{g}; \ms{Dm},\ms{f}^0}\text{,}
         \end{align*}
         with the right-hand sides satisfying the following asymptotics: 
         \begin{align*}
             \begin{aligned}
                 &\rho^{-1}\ul{\ms{h}}^0 = \mc{O}_{M-1}(1)\text{,}\qquad &&\ms{D}\ul{\ms{h}}^0=\mc{O}_{M-2}(\rho)\text{,}\qquad &&&\rho^{-1}\mc{R}_{T, \ms{h}^0} = \mc{O}_{M-2}(\rho)\text{,}\\
                 &\rho^{-1}\ul{\ms{h}}^{2} = \mc{O}_{M-1}(1)\text{,}\qquad &&\ms{D}\ul{\ms{h}}^{2} = \mc{O}_{M-2}(\rho)\text{,}\qquad &&&\rho^{-1}\mc{R}_{T, \ms{h}^2} = \mc{O}_{M-2}(\rho)\text{,}
             \end{aligned}
         \end{align*}
         concluding the proof.

\subsection{Proof of Proposition \ref{prop_transport_diff}}\label{app:prop_transport_diff}

The transport equation for $\delta\ms{g}$ follows from the definition of $\ms{m}$. The second transport equation can be deduced from \eqref{w_2_FG}. Typically, 
    \begin{equation*}
        \Lie_\rho \delta \ms{m} = \rho^{-1}\delta \ms{m} - 2\delta \ms{w}^2 + \delta(\mc{S}\paren{\ms{g}; \ms{m}^2}) + \delta (\mc{S}\paren{\ms{g}; (\ms{f}^0)^2} + \mc{S}\paren{\ms{g}; (\ms{f}^1)^2})\text{.}
    \end{equation*}
    Since, by \eqref{renormalised_def}, $\delta \ms{w}^2 = \Delta \ms{w}^2 + \mc{O}_{M-3}(\rho; \delta\ms{g}) + \mc{O}_{M-3}(\rho; \ms{Q})$ and 
    \begin{align*}
        &\delta\mc{S}\paren{\ms{g}; (\ms{m}^2)} = \mc{O}_{M-2}(\rho^2; \delta \ms{g}) + \mc{O}_{M-2}(\rho; \delta \ms{m})\text{,}\\
        &\delta \mc{S}\paren{\ms{g}; (\ms{f}^0)^2} + \delta \mc{S}\paren{\ms{g}; (\ms{f}^1)^2} = \mc{O}_{M}\paren{\rho; \delta \ms{f}^0} + \mc{O}_{M}(\rho; \delta\ms{f}^1) + \mc{O}_{M}\paren{\rho^2; \delta\ms{g}}\text{.}
    \end{align*}
    From the above one obtains \eqref{transport_dm}. 

    Equation \eqref{Lie_Q} is a direct consequence of \eqref{transport_Q} and Proposition \ref{prop_diff_g_gamma}. 

    The transport equation for $\ms{B}$ can be expressed using Proposition \ref{sec:aads_commutation_Lie_D}: 
    \begin{align*}
        \Lie_\rho \ms{B}_{abc} &= 2\Lie_\rho \ms{D}_{[a}\delta \ms{g}_{b]c} - \Lie_\rho \ms{D}_c \ms{Q}_{ab}\\
        &=2\ms{D}_{[a}\delta\ms{m}_{b]c} - \ms{D}_c \Lie_\rho \ms{Q}_{ab} + \mc{S}\paren{\ms{g}; \ms{Dm}, \delta\ms{g}}_{abc} + \mc{S}\paren{\ms{g}; \ms{Dm}, \ms{Q}}_{abc}\text{.}
    \end{align*}
    One can use \eqref{w_1_FG} to treat the first term on the right-hand side:
    \begin{align*}
        2\ms{D}_{[a}\delta\ms{m}_{b]c} &= \delta(2\ms{D}_{[a}\ms{m}_{b]c}) - 2(\delta\ms{D})_{[a}\check{\ms{m}}_{b]c}\\
        &=2\delta\paren{\ms{w}^1_{cba} + \frac{2}{n-1}\ms{T}_{\rho [b}\ms{g}_{a]c}} - 2(\delta\ms{D})_{[a}\check{\ms{m}}_{b]c}\text{,}
    \end{align*}
    where, using Proposition \ref{prop_diff_g_gamma}:
    \begin{equation*}
        2(\delta\ms{D})_{[a}\check{\ms{m}}_{b]c} = \check{\ms{g}}^{de}\check{\ms{m}}_{d[a}\paren{\ms{D}_{b]}\delta\ms{g}_{ec} + \ms{D}_{|c|}\delta\ms{g}_{b]e} - \ms{D}_{|e|} \delta\ms{g}_{b]c}}\text{.}
    \end{equation*}
    Eventually, this gives, after using \eqref{transport_Q}: 
    \begin{align}
        \Lie_\rho \ms{B}_{abc} =&\, 2\delta\ms{w}^1_{cba} - \check{\ms{g}}^{de}\check{\ms{m}}_{d[a}\paren{\ms{D}_{b]}\delta\ms{g}_{ec} + \ms{D}_{|c|}\delta\ms{g}_{b]e} - \ms{D}_{|e|} \delta\ms{g}_{b]c}} \\
        &\notag+\ms{g}^{de}\ms{D}_{c}\ms{m}_{d[a}\delta \ms{g}_{b]e} + \ms{g}^{de}\ms{m}_{d[a} \ms{D}_{|c|} \delta \ms{g}_{b]e} + \ms{g}^{de}\ms{D}_{c}\ms{m}_{d[a}\ms{Q}_{b]e} + \ms{g}^{de}\ms{m}_{d[a}\ms{D}_{|c|}\ms{Q}_{b]e} \\
        &\notag + \delta \mc{S}\paren{\ms{g}; \ms{f}^0, \ms{f}^1}_{abc} + \mc{S}\paren{\ms{g}; \ms{Dm}, \delta\ms{g}}_{abc} + \mc{S}\paren{\ms{g}; \ms{Dm}, \ms{Q}}_{abc}\\
        =&\notag\,  2\delta \ms{w}^1_{cba} - \paren{\check{\ms{g}}^{de}\check{\ms{m}}_{d[a} - {\ms{g}}^{de}{\ms{m}}_{d[a} }\ms{D}_{|c|}\delta \ms{g}_{b]e} -\check{\ms{g}}^{de}\check{\ms{m}}_{d[a}\paren{\ms{D}_{b]}\delta \ms{g}_{ec}-\ms{D}_{|e|}\delta \ms{g}_{b]c}} \\
        &\notag +\ms{g}^{de} \ms{m}_{d[a}\ms{D}_{|c|}\ms{Q}_{b]e}+\mc{O}_{M-3}\paren{\rho; \delta \ms{g}}_{abc} + \mc{O}_{M-3}(\rho; \ms{Q})_{abc} + \mc{O}_{M}(\rho; \delta\ms{f}^0)_{abc} + \mc{O}_{M}\paren{\rho; \delta \ms{f}^1}_{abc}\text{.}
    \end{align}
    Choosing $\ms{D}\delta\ms{g}$ as coefficients, the second term can be written as: 
    \begin{align*}
    \paren{\check{\ms{g}}^{de}\check{\ms{m}}_{d[a} - {\ms{g}}^{de}{\ms{m}}_{d[a} }\ms{D}_{|c|}\delta \ms{g}_{b]e} = \mc{O}_{M-2}(\rho; \delta\ms{g}) + \mc{O}_{M}(1; \delta\ms{m})\text{,}
    \end{align*}
    where we used Remark \ref{rmk_difference_remainders}.
    Observe also that one can write: 
    \begin{equation*}
        -\check{\ms{g}}^{de}\check{\ms{m}}_{d[a}\paren{\ms{D}_{b]}\delta \ms{g}_{ec}-\ms{D}_{|e|}\delta \ms{g}_{b]c}} =  -{\ms{g}}^{de}{\ms{m}}_{d[a}\paren{\ms{D}_{b]}\delta \ms{g}_{ec}-\ms{D}_{|e|} \delta\ms{g}_{b]c}} + \mc{O}_{M-2}(\rho; \delta \ms{g})_{abc} + \mc{O}_{M-1}(1; \delta \ms{m})_{abc}\text{,}\\
    \end{equation*}
    yielding, using Definition \ref{renormalised_def}: 
    \begin{align*}
        \Lie_\rho \ms{B}_{abc} =&\, 2\Delta \ms{w}^1_{cba} - \ms{g}^{de}\ms{m}_{d[a}\overbrace{\paren{\ms{D}_{b]}\delta \ms{g}_{ec}-\ms{D}_{|e|} \delta\ms{g}_{b]c}-\ms{D}_{|c|}\ms{Q}_{b]e}}}^{\ms{B}_{b]ec}} + \mc{O}_{M-3}\paren{\rho; \delta \ms{g}}_{abc} +\mc{O}_{M-3}(\rho;\ms{Q})_{abc}\\
        &+\mc{O}_{M-1}(1; \delta\ms{m})_{abc}+ \mc{O}_{M}\paren{\rho; \delta\ms{f}^0}_{abc}+ \mc{O}_{M}\paren{\rho; \delta\ms{f}^1}_{abc}\text{.}
    \end{align*}
    The transport equations for the difference of the Maxwell fields follow from \eqref{transport_f0_FG} and \eqref{transport_f1_FG}. Starting with $\delta \ms{f}^0$: 
    \begin{align*}
        \Lie_\rho (\rho^{-(n-2)}\delta\ms{f}^0) =& \, \rho^{-(n-1)}\paren{-\rho\delta(\ms{tr}_{\ms{g}}\ul{\ms{h}}^0) + \rho \delta \mc{S}\paren{\ms{g}; \ms{m}, \ms{f}^0}}\\
        =&\, \mc{O}_{M-1}\paren{\rho^{-(n-3)}; \delta\ms{g}}+ \mc{O}_{M}\paren{\rho^{-(n-2)};\delta\ul{\ms{h}}^0} + \mc{O}_{M} \paren{\rho^{-(n-3)}; \delta\ms{m}}\\
        &+ \mc{O}_{M-2}(\rho^{-(n-3)}; \delta \ms{f}^0)\text{,}
    \end{align*}
    where we use Proposition \ref{prop_diff_g_gamma}. Finally, the transport equation for $\delta\ms{f}^1$: 
    \begin{align*}
        \Lie_\rho(\rho^{-1} \delta \ms{f}^1) &=\rho^{-1}\delta\mc{S}\paren{\ul{\ms{h}}^{2}}+\rho^{-1}\delta\mc{S}\paren{\ms{g}; \ms{m}, \ms{f}^1}\\
        &= \mc{S}\paren{\rho^{-1} \delta \ul{\ms{h}}^{2}} + \mc{O}_{M-1}(1; \delta\ms{g})+\mc{O}_{M}(1;\delta\ms{m})+\mc{O}_{M-2}(1; \delta\ms{f}^1)\text{.}
    \end{align*}
    The transport equations \eqref{transport_Ddg}--\eqref{transport_Ddf1} can simply be obtained from the following: for any vertical tensor $\ms{A}$, one has 
    \begin{equation*}
        \Lie_\rho (\ms{D}\ms{A}) = \mc{O}_{M-3}(\rho; \ms{A}) + \ms{D}\paren{\Lie_\rho \ms{A}}\text{.}
    \end{equation*}
    One also needs to use Definition \ref{def_h_bar}, in order to write, for example: 
    \begin{align*}
        \ms{D}\delta\ms{f}^1 =& \delta\ul{\ms{h}}^0 -(\delta\ms{D})\check{\ms{f}}^1\\
        =&\delta\ul{\ms{h}}^0+ \mc{O}_{M}(\rho; \ms{D}\delta\ms{g})\\
        =&\Delta \ul{\ms{h}}^0 + \mc{O}_{M-1}(\rho; \delta\ms{g}) + \mc{O}_{M}(\rho; \ms{D}\delta\ms{g}) + \mc{O}_{M-1}(\rho; \ms{Q})\text{,}
     \end{align*}
     and similarly for $\ms{f}^0$: 
     \begin{align*}
         \ms{D}\delta\ms{f}^0 &=\delta \ul{\ms{h}}^2 - (\delta\ms{D})\check{\ms{f}}^0\\
         &=\Delta\ul{\ms{h}}^{2} + \mc{O}_{M-1}(\rho; \ms{Q}) + \mc{O}_{M}(\rho; \ms{D}\delta\ms{g}) + \mc{O}_{M-1}(\rho; \delta\ms{g})\text{.}\qedhere
     \end{align*}

\subsection{Proof of Proposition \ref{prop_transport_diff_w_h}}\label{app:prop_transport_diff_w_h}

Notice that the transport equations \eqref{transport_dm_modified}, \eqref{transport_df0_modified} and \eqref{transport_df1_modified} can formally by obtained by replacing $\Delta$ by $\delta$ and removing the $\ms{Q}$ contributions in \eqref{transport_dm}, \eqref{transport_df0} and \eqref{transport_df1}.

    Let us first look at $\delta\ms{w}^2$. Differentiating the transport equation \eqref{transport_w_2} for both $\ms{w}^2$ and $\check{\ms{w}}^2$ gives: 
    \begin{align*}
        \rho\Lie_\rho \delta\ms{w}^2 -(n-2)\delta\ms{w}^2 =& \rho\delta\mc{S}\paren{\ms{g}; \ms{m}, \ms{w}^0} + \rho\delta\mc{S}\paren{\ms{g}; \ms{m}, \ms{w}^2}+\rho\delta\mc{S}\paren{\ms{g}; \ul{\ms{h}}^0, \ms{f}^0} + \delta\mc{S}\paren{\ms{g}; (\ms{f}^1)^2} \\  
        &+ \rho\delta\mc{S}\paren{\ms{g}; \ul{\ms{h}}^{2}, \ms{f}^1}+\delta\mc{S} \paren{\ms{g}; (\ms{f}^0)^2}+ \rho\delta\mc{S}\paren{\ms{g};\ms{m},  (\ms{f}^0)^2}\notag\\
        &  + \rho\delta\mc{S}\paren{\ms{g}; \ms{m}, (\ms{f}^1)^2}\text{.} 
    \end{align*}
    First, observe that the left-hand side can be written as below: 
    \begin{equation*}
         \rho\Lie_\rho \delta\ms{w}^2 -(n-2)\delta\ms{w}^2 = \rho^{n-1}\Lie_\rho \paren{\rho^{2-n}\delta\ms{w}^2}\text{.}
    \end{equation*}
    On the other hand, the right-hand side, containing only lower-order contributions, can be written as: 
    \begin{align*}
        &\mc{O}_{M-2}(\rho;\delta\ms{m}) + \mc{O}_{M-2}(\rho^2; \delta\ms{w}^0) + \mc{O}_{M-2}(\rho^2; \delta\ms{w}^2) + \mc{O}_M(\rho; \ms{D}\delta\ms{w}^1)\\
        & +\mc{O}_{M}(\rho^2; \delta\ul{\ms{h}}^0) + \mc{O}_{M}(\rho^2; \delta\ul{\ms{h}}^{2}) + \mc{O}_{M-1}(\rho; \delta\ms{f}^1) \\
        &+ \mc{O}_{M-1}(\rho; \delta\ms{f}^0) + \mc{O}_{M-4}\paren{\rho^2; \delta\ms{g}}+\mc{O}_{M-3}(\rho^2; \ms{D}\delta\ms{g})\text{,}
    \end{align*}
    where we used the improved decay of $\ms{w}^1$ in Proposition \ref{prop_asymptotics_fields} for the last two terms. This gives the transport equation \eqref{transport_dw2}. The computations for $\ms{w}^0$ and $\ms{w}^1$ are identical.

    The computations for $\ul{\ms{h}}^0$ and  $\ul{\ms{h}}^2$ are also similar. 

\subsection{Proof of Proposition \ref{prop_transport_carleman}}\label{app:prop_transport_carleman}

First of all, note that the following identities hold: 
        \begin{gather*}
            \Lie_\rho f = f\rho^{-1}\text{,}\qquad \Lie_\rho (\omega_\lambda(f)f^p) = -(\lambda f^p +2\kappa -n + 2)\rho^{-1}f^p\omega_\lambda(f)\text{.}
        \end{gather*}
        As a consequence, for any vertical tensor field $\ms{A}$: 
        \begin{align*}
            \Lie_\rho(\omega_\lambda(f)f^p\rho^{s-n}|\ms{A}|_{\mf{h}}^2) =&\, - (2\kappa+\lambda f^p -s+2)\rho^{-1}\omega_\lambda(f)f^p\rho^{s-n}|\ms{A}|_{\mf{h}}^2+\omega_\lambda(f)f^p\rho^{s-n}\Lie_\rho|\ms{A}|_{\mf{h}}^2\text{.}
        \end{align*}
        Using now Young's inequality on $\Lie_\rho |\ms{A}|_h^2= 2\Lie_\rho \ms{A}\cdot_{\mf{h}}\ms{A}$, where $\cdot_{\mf{h}}$ is understood as an $\mf{h}$--contraction: 
        \begin{equation}
            2\Lie_\rho \ms{A}\cdot_{\mf{h}}\ms{A} \leq \frac{\rho}{\lambda f^p}|{\Lie_\rho\ms{A}}|_{\mf{h}}^2 + \lambda f^p \rho^{-1}|\ms{A}|_{\mf{h}}^2\text{,}
        \end{equation}
        one obtains: 
        \begin{align*}
            \Lie_\rho(\omega_\lambda(f)f^p \rho^{s-n}|\ms{A}|_{\mf{h}}^2) + (2\kappa - s +2)\rho^{-1}\omega_\lambda(f)f^p \rho^{s-n}|\ms{A}|_{\mf{h}}^2 \leq \lambda^{-1}\omega_\lambda(f)\rho^{s-n+1} |\Lie_\rho\ms{A}|_{\mf{h}}^2\text{.}
        \end{align*}

        Integrating now over the region $\Omega_\star\cap \lbrace {\rho > \rho_\star}\rbrace$, for some $\rho_\star \ll f_\star$, first on level sets of $\rho$ with respect to $\mf{h}$ and then with respect to $\rho$, one thus gets: 
        \begin{align*}
            \int_{\Omega_\star \cap \{\rho >\rho_\star\}} \omega_\lambda(f)\rho^{s-n+1}|\Lie_\rho\ms{A}|_{\mf{h}}^2 d\mu_{\mf{h}}d\rho \geq &\; \lambda (2\kappa-s+2)\int_{\Omega_\star \cap \{\rho>\rho_\star\}} \omega_\lambda(f)f^p \rho^{s-n-1}|\ms{A}|_{\mf{h}}^2d\mu_{\mf{h}}d\rho \\
            &+\lambda\int_{\Omega_\star \cap \{\rho > \rho_\star\}} \Lie_\rho (\omega_\lambda(f)f^p\rho^{s-n}|\ms{A}|^2_{\mf{h}})d\mu_{\mf{h}}d\rho\text{,}
        \end{align*}
        Now, the last integral can be easily computed to yield: 
        \begin{align*}
            \int_{\Omega_\star \cap \{\rho > \rho_\star\}} \Lie_\rho (\omega_\lambda(f)f^p\rho^{s-n}|\ms{A}|_{\mf{h}})d\mu_{\mf{h}}d\rho
            \geq - \int_{\Omega_\star \cap \lbrace \rho = \rho_\star\rbrace} \omega_\lambda(f)f^p \rho^{s-n}|\ms{A}|^2_{\mf{h}} d\mu_{\mf{h}}\text{,}
        \end{align*}
        where the inequality holds since the boundary term on the level set $\lbrace f =f_\star \rbrace$ has the right sign. One can therefore write: 
        \begin{align*}
            \int_{\Omega_\star \cap \{\rho >\rho_\star\}} &\omega_\lambda(f)\rho^{s-n+1}|\Lie_\rho\ms{A}|_{\mf{h}}^2 d\mu_{\mf{h}}d\rho + \lambda \int_{\Omega_\star \cap \lbrace \rho = \rho_\star\rbrace} \omega_\lambda(f)f^p \rho^{s-n}|\ms{A}|^2_{\mf{h}} d\mu_{\mf{h}}\\
            &\geq (2\kappa - s + 2)\lambda \int_{\Omega_\star \cap \{\rho>\rho_\star\}} \omega_\lambda(f)f^p \rho^{s-n-1}|\ms{A}|_{\mf{h}}^2d\mu_{\mf{h}}d\rho \text{.}
        \end{align*}
        Since $d\mu_{\ms{g}}$ is comparable to $d\mu_{\mf{h}}$ on $\overline{\mi{D}}$, compact, and: 
        \begin{equation*}
            d\mu_{g} = \rho^{-n-1}d\mu_{\ms{g}}d\rho\text{,}
        \end{equation*}
        then, there exist some constants $C, D>0$, depending on $\mf{h}$, $\mi{D}$ and $\ms{g}$, such that: 
        \begin{align*}
            \int_{\Omega_\star\cap\{\rho>\rho_\star\}}\omega_\lambda(f)\rho^{s+2}|\Lie_\rho\ms{A}|_{\mf{h}}^2 d\mu_g + C\lambda &\int_{\Omega_\star\cap\{\rho=\rho_\star\}}\omega_\lambda(f)f^p\rho^{s-n}|\ms{A}|_{\mf{h}}^2d\mu_{\ms{g}}\\
            &\geq D\lambda \int_{\Omega_\star\cap\{\rho>\rho_\star\}}\omega_\lambda(f)f^p \rho^{s}|\ms{A}|_{\mf{h}}^2 d\mu_g\text{.}
        \end{align*}
        Letting $\rho_\star\searrow 0$, one finds the right estimate \eqref{carleman_transport} since: 
        \begin{gather*}
            \omega_\lambda(f)f^p \rho^{s-n}|\ms{A}|_{\mf{h}} = \eta^{2\kappa-(n-2)}\rho^s|\rho^{-\kappa-1}\ms{A}|_{\mf{h}}^2\text{,}
        \end{gather*}
        and, by assumption, $2\kappa-(n-2)\geq 0$. Furthermore, for $\lambda \gg 1$, $e^{-\lambda f^p /p}\leq 1$\text{.}

\end{document}